\documentclass[
  aps,
  prl,
  reprint,
  groupedaddress,
  nofootinbib,
  longbibliography,
  floatfix,
  nobibnotes
]{revtex4-2}

\usepackage[T1]{fontenc}
\usepackage{microtype}
\usepackage{graphicx}
\usepackage{mathtools}
\usepackage{amssymb}
\usepackage{amsthm}
\usepackage{amsmath}
\usepackage{amsbsy}
\usepackage{bbm}
\usepackage[all]{xy}
\usepackage{bm}
\usepackage{tikz}
\usepackage{array}
\usepackage{colortbl,xcolor}
\usepackage{ytableau}
\usepackage{hhline}
\usepackage{capt-of}
\usepackage{braket}
\usepackage{placeins}

\usepackage{hyperref}

\hypersetup{
  colorlinks=true,
  citecolor=blue,
  linkcolor=blue,
  urlcolor=blue,
  bookmarksopen=true,
  bookmarksnumbered=true,
  pdftitle={Optimal two-mode bosonic loss codes from finite group symmetry},
  pdfauthor={Argyris Giannisis Manes, Mahadevan Subramanian, and Liang Jiang}
}
\allowdisplaybreaks

\newtheorem{theorem}{Theorem}
\newtheorem{lemma}{Lemma}
\newtheorem{proposition}{Proposition}
\newtheorem{corollary}{Corollary}

\theoremstyle{definition}
\newtheorem{definition}{Definition}

\theoremstyle{remark}
\newtheorem{remark}{Remark}
\newtheorem{conjecture}[theorem]{Conjecture}

\theoremstyle{plain}

\newcommand{\norm}[1]{\left\lVert#1\right\rVert}

\newcommand{\tr}[1]{\text{Tr}\left[ #1\right]}

\DeclareMathOperator*{\Tr}{Tr}

\makeatletter
\let\FinalLetterCaption\caption
\let\FinalLetterMakeCaption\@makecaption
\let\FinalLetterAtCaption\@caption
\let\FinalLetterCaptionOf\captionof
\makeatother
\usepackage{subcaption}
\usepackage{float}
\usepackage{enumerate}
\usepackage{docmute}
\usepackage{chapterbib}
\hypersetup{hypertexnames=false}

\makeatletter
\let\auto@bib\@empty
\let\bibliographystyle\bibliographystyle@latex
\let\write@bibliographystyle\relax
\newcommand{\FinalLetterSetup}{%
  \IfFileExists{Candidacy.bib}{}{%
    \let\bibliographystyle\@gobble
    \renewcommand{\bibliography}[1]{
}%
  }%
  \let\caption\FinalLetterCaption
  \let\@makecaption\FinalLetterMakeCaption
  \let\@caption\FinalLetterAtCaption
  \let\captionof\FinalLetterCaptionOf
}
\newcommand{\FinalSupplementSetup}{%
  \IfFileExists{Candidacy.bib}{}{%
    \let\bibliographystyle\@gobble
    \renewcommand{\bibliography}[1]{
\begin{thebibliography}{28}%
\makeatletter
\providecommand \@ifxundefined [1]{%
 \@ifx{#1\undefined}
}%
\providecommand \@ifnum [1]{%
 \ifnum #1\expandafter \@firstoftwo
 \else \expandafter \@secondoftwo
 \fi
}%
\providecommand \@ifx [1]{%
 \ifx #1\expandafter \@firstoftwo
 \else \expandafter \@secondoftwo
 \fi
}%
\providecommand \natexlab [1]{#1}%
\providecommand \enquote  [1]{``#1''}%
\providecommand \bibnamefont  [1]{#1}%
\providecommand \bibfnamefont [1]{#1}%
\providecommand \citenamefont [1]{#1}%
\providecommand \href@noop [0]{\@secondoftwo}%
\providecommand \href [0]{\begingroup \@sanitize@url \@href}%
\providecommand \@href[1]{\@@startlink{#1}\@@href}%
\providecommand \@@href[1]{\endgroup#1\@@endlink}%
\providecommand \@sanitize@url [0]{\catcode `\\12\catcode `\$12\catcode `\&12\catcode `\#12\catcode `\^12\catcode `\_12\catcode `\%12\relax}%
\providecommand \@@startlink[1]{}%
\providecommand \@@endlink[0]{}%
\providecommand \url  [0]{\begingroup\@sanitize@url \@url }%
\providecommand \@url [1]{\endgroup\@href {#1}{\urlprefix }}%
\providecommand \urlprefix  [0]{URL }%
\providecommand \Eprint [0]{\href }%
\providecommand \doibase [0]{https://doi.org/}%
\providecommand \selectlanguage [0]{\@gobble}%
\providecommand \bibinfo  [0]{\@secondoftwo}%
\providecommand \bibfield  [0]{\@secondoftwo}%
\providecommand \translation [1]{[#1]}%
\providecommand \BibitemOpen [0]{}%
\providecommand \bibitemStop [0]{}%
\providecommand \bibitemNoStop [0]{.\EOS\space}%
\providecommand \EOS [0]{\spacefactor3000\relax}%
\providecommand \BibitemShut  [1]{\csname bibitem#1\endcsname}%
\let\auto@bib@innerbib\@empty
\bibitem [{\citenamefont {Noh}\ \emph {et~al.}(2019)\citenamefont {Noh}, \citenamefont {Albert},\ and\ \citenamefont {Jiang}}]{nohQuantumCapacityBounds2019}%
  \BibitemOpen
  \bibfield  {author} {\bibinfo {author} {\bibfnamefont {K.}~\bibnamefont {Noh}}, \bibinfo {author} {\bibfnamefont {V.~V.}\ \bibnamefont {Albert}},\ and\ \bibinfo {author} {\bibfnamefont {L.}~\bibnamefont {Jiang}},\ }\href {https://doi.org/10.1109/TIT.2018.2873764} {\bibfield  {journal} {\bibinfo  {journal} {IEEE Transactions on Information Theory}\ }\textbf {\bibinfo {volume} {65}},\ \bibinfo {pages} {2563} (\bibinfo {year} {2019})}\BibitemShut {NoStop}%
\bibitem [{\citenamefont {Gross}(2021)}]{grossDesigningCodesInteractions2021}%
  \BibitemOpen
  \bibfield  {author} {\bibinfo {author} {\bibfnamefont {J.~A.}\ \bibnamefont {Gross}},\ }\href {https://doi.org/10.1103/PhysRevLett.127.010504} {\bibfield  {journal} {\bibinfo  {journal} {Physical Review Letters}\ }\textbf {\bibinfo {volume} {127}},\ \bibinfo {pages} {010504} (\bibinfo {year} {2021})}\BibitemShut {NoStop}%
\bibitem [{\citenamefont {Dummit}\ and\ \citenamefont {Foote}(2004)}]{dummitAbstractAlgebra2004}%
  \BibitemOpen
  \bibfield  {author} {\bibinfo {author} {\bibfnamefont {D.~S.}\ \bibnamefont {Dummit}}\ and\ \bibinfo {author} {\bibfnamefont {R.~M.}\ \bibnamefont {Foote}},\ }\href@noop {} {\emph {\bibinfo {title} {Abstract Algebra}}},\ \bibinfo {edition} {3rd}\ ed.\ (\bibinfo  {publisher} {Wiley},\ \bibinfo {address} {Hoboken, NJ},\ \bibinfo {year} {2004})\BibitemShut {NoStop}%
\bibitem [{\citenamefont {Serre}(1977)}]{serreLinearRepresentationsFinite1977}%
  \BibitemOpen
  \bibfield  {author} {\bibinfo {author} {\bibfnamefont {J.-P.}\ \bibnamefont {Serre}},\ }\href@noop {} {\emph {\bibinfo {title} {{Linear representations of finite groups}}}},\ \bibinfo {series} {{Graduate texts in mathematics}}\ No.~\bibinfo {number} {42}\ (\bibinfo  {publisher} {Springer-Verlag},\ \bibinfo {address} {New York},\ \bibinfo {year} {1977})\BibitemShut {NoStop}%
\bibitem [{\citenamefont {Fulton}\ and\ \citenamefont {Harris}(2004)}]{fultonRepresentationTheory2004}%
  \BibitemOpen
  \bibfield  {author} {\bibinfo {author} {\bibfnamefont {W.}~\bibnamefont {Fulton}}\ and\ \bibinfo {author} {\bibfnamefont {J.}~\bibnamefont {Harris}},\ }\href {https://doi.org/10.1007/978-1-4612-0979-9} {\emph {\bibinfo {title} {Representation {{Theory}}}}},\ \bibinfo {series} {Graduate {{Texts}} in {{Mathematics}}}, Vol.\ \bibinfo {volume} {129}\ (\bibinfo  {publisher} {Springer},\ \bibinfo {address} {New York, NY},\ \bibinfo {year} {2004})\BibitemShut {NoStop}%
\bibitem [{\citenamefont {Goodman}\ and\ \citenamefont {Wallach}(2009)}]{goodmanSymmetryRepresentationsInvariants2009}%
  \BibitemOpen
  \bibfield  {author} {\bibinfo {author} {\bibfnamefont {R.}~\bibnamefont {Goodman}}\ and\ \bibinfo {author} {\bibfnamefont {N.~R.}\ \bibnamefont {Wallach}},\ }\href {https://doi.org/10.1007/978-0-387-79852-3} {\emph {\bibinfo {title} {Symmetry, {{Representations}}, and {{Invariants}}}}},\ \bibinfo {series} {Graduate {{Texts}} in {{Mathematics}}}, Vol.\ \bibinfo {volume} {255}\ (\bibinfo  {publisher} {Springer},\ \bibinfo {address} {New York, NY},\ \bibinfo {year} {2009})\BibitemShut {NoStop}%
\bibitem [{\citenamefont {Sakurai}\ and\ \citenamefont {Napolitano}(2021)}]{sakuraiModernQuantumMechanics2021}%
  \BibitemOpen
  \bibfield  {author} {\bibinfo {author} {\bibfnamefont {J.~J.}\ \bibnamefont {Sakurai}}\ and\ \bibinfo {author} {\bibfnamefont {J.}~\bibnamefont {Napolitano}},\ }\href@noop {} {\emph {\bibinfo {title} {Modern Quantum Mechanics}}},\ \bibinfo {edition} {third edition}\ ed.\ (\bibinfo  {publisher} {Cambridge University Press},\ \bibinfo {address} {Cambridge},\ \bibinfo {year} {2021})\BibitemShut {NoStop}%
\bibitem [{\citenamefont {Hall}(2015)}]{hallLieGroupsLie2015}%
  \BibitemOpen
  \bibfield  {author} {\bibinfo {author} {\bibfnamefont {B.~C.}\ \bibnamefont {Hall}},\ }\href {https://doi.org/10.1007/978-3-319-13467-3} {\emph {\bibinfo {title} {Lie {{Groups}}, {{Lie Algebras}}, and {{Representations}}: {{An Elementary Introduction}}}}},\ \bibinfo {series} {Graduate {{Texts}} in {{Mathematics}}}, Vol.\ \bibinfo {volume} {222}\ (\bibinfo  {publisher} {Springer International Publishing},\ \bibinfo {address} {Cham},\ \bibinfo {year} {2015})\BibitemShut {NoStop}%
\bibitem [{\citenamefont {Aoude}\ and\ \citenamefont {Ochirov}(2021)}]{aoudeClassicalObservablesCoherentspin2021}%
  \BibitemOpen
  \bibfield  {author} {\bibinfo {author} {\bibfnamefont {R.}~\bibnamefont {Aoude}}\ and\ \bibinfo {author} {\bibfnamefont {A.}~\bibnamefont {Ochirov}},\ }\href {https://doi.org/10.1007/JHEP10(2021)008} {\bibfield  {journal} {\bibinfo  {journal} {Journal of High Energy Physics}\ }\textbf {\bibinfo {volume} {2021}},\ \bibinfo {pages} {8} (\bibinfo {year} {2021})}\BibitemShut {NoStop}%
\bibitem [{\citenamefont {Denys}\ and\ \citenamefont {Leverrier}(2024)}]{denysQuantumErrorCorrectingCodes2024}%
  \BibitemOpen
  \bibfield  {author} {\bibinfo {author} {\bibfnamefont {A.}~\bibnamefont {Denys}}\ and\ \bibinfo {author} {\bibfnamefont {A.}~\bibnamefont {Leverrier}},\ }\href {https://doi.org/10.1103/PhysRevLett.133.240603} {\bibfield  {journal} {\bibinfo  {journal} {Physical Review Letters}\ }\textbf {\bibinfo {volume} {133}},\ \bibinfo {pages} {240603} (\bibinfo {year} {2024})}\BibitemShut {NoStop}%
\bibitem [{\citenamefont {Vale}\ and\ \citenamefont {Waldron}(2016)}]{valeConstructionGInvariantFinite2016}%
  \BibitemOpen
  \bibfield  {author} {\bibinfo {author} {\bibfnamefont {R.}~\bibnamefont {Vale}}\ and\ \bibinfo {author} {\bibfnamefont {S.}~\bibnamefont {Waldron}},\ }\href {https://doi.org/10.1007/s00041-015-9443-9} {\bibfield  {journal} {\bibinfo  {journal} {Journal of Fourier Analysis and Applications}\ }\textbf {\bibinfo {volume} {22}},\ \bibinfo {pages} {1097} (\bibinfo {year} {2016})}\BibitemShut {NoStop}%
\bibitem [{\citenamefont {Kubischta}\ and\ \citenamefont {Teixeira}(2025)}]{kubischtaPermutationInvariantQuantumCodes2025}%
  \BibitemOpen
  \bibfield  {author} {\bibinfo {author} {\bibfnamefont {E.}~\bibnamefont {Kubischta}}\ and\ \bibinfo {author} {\bibfnamefont {I.}~\bibnamefont {Teixeira}},\ }\href {https://doi.org/10.1109/TIT.2024.3487964} {\bibfield  {journal} {\bibinfo  {journal} {IEEE Transactions on Information Theory}\ }\textbf {\bibinfo {volume} {71}},\ \bibinfo {pages} {485} (\bibinfo {year} {2025})}\BibitemShut {NoStop}%
\bibitem [{\citenamefont {Bond}\ \emph {et~al.}(2026)\citenamefont {Bond}, \citenamefont {Min{\'a}{\v r}}, \citenamefont {Ozols}, \citenamefont {{Safavi-Naini}},\ and\ \citenamefont {Visnevskyi}}]{bondPermutationinvariantCodesNumerical2026}%
  \BibitemOpen
  \bibfield  {author} {\bibinfo {author} {\bibfnamefont {L.~J.}\ \bibnamefont {Bond}}, \bibinfo {author} {\bibfnamefont {J.}~\bibnamefont {Min{\'a}{\v r}}}, \bibinfo {author} {\bibfnamefont {M.}~\bibnamefont {Ozols}}, \bibinfo {author} {\bibfnamefont {A.}~\bibnamefont {{Safavi-Naini}}},\ and\ \bibinfo {author} {\bibfnamefont {V.}~\bibnamefont {Visnevskyi}},\ }\href {https://doi.org/10.48550/arXiv.2603.10981} {\bibinfo {title} {Permutation-invariant codes: A numerical study and qudit constructions}} (\bibinfo {year} {2026}),\ \Eprint {https://arxiv.org/abs/2603.10981} {arXiv:2603.10981 [quant-ph]} \BibitemShut {NoStop}%
\bibitem [{Not(2008)}]{NotesCoxeterTransformations2008}%
  \BibitemOpen
  \href {https://doi.org/10.1007/978-3-540-77399-3} {\emph {\bibinfo {title} {Notes on {{Coxeter Transformations}} and the {{McKay Correspondence}}}}},\ Springer {{Monographs}} in {{Mathematics}}\ (\bibinfo  {publisher} {Springer},\ \bibinfo {address} {Berlin, Heidelberg},\ \bibinfo {year} {2008})\BibitemShut {NoStop}%
\bibitem [{\citenamefont {Dieci}\ \emph {et~al.}(2005)\citenamefont {Dieci}, \citenamefont {Gasparo},\ and\ \citenamefont {Papini}}]{dieciContinuationSingularValue2005}%
  \BibitemOpen
  \bibfield  {author} {\bibinfo {author} {\bibfnamefont {L.}~\bibnamefont {Dieci}}, \bibinfo {author} {\bibfnamefont {M.~G.}\ \bibnamefont {Gasparo}},\ and\ \bibinfo {author} {\bibfnamefont {A.}~\bibnamefont {Papini}},\ }\href {https://doi.org/10.1007/s00009-005-0038-6} {\bibfield  {journal} {\bibinfo  {journal} {Mediterranean Journal of Mathematics}\ }\textbf {\bibinfo {volume} {2}},\ \bibinfo {pages} {179} (\bibinfo {year} {2005})}\BibitemShut {NoStop}%
\bibitem [{\citenamefont {Fan}(1949)}]{fanTheoremWeylConcerning1949}%
  \BibitemOpen
  \bibfield  {author} {\bibinfo {author} {\bibfnamefont {K.}~\bibnamefont {Fan}},\ }\href {https://doi.org/10.1073/pnas.35.11.652} {\bibfield  {journal} {\bibinfo  {journal} {Proceedings of the National Academy of Sciences of the United States of America}\ }\textbf {\bibinfo {volume} {35}},\ \bibinfo {pages} {652} (\bibinfo {year} {1949})}\BibitemShut {NoStop}%
\bibitem [{\citenamefont {Chuang}\ \emph {et~al.}(1997)\citenamefont {Chuang}, \citenamefont {Leung},\ and\ \citenamefont {Yamamoto}}]{chuangBosonicQuantumCodes1997}%
  \BibitemOpen
  \bibfield  {author} {\bibinfo {author} {\bibfnamefont {I.~L.}\ \bibnamefont {Chuang}}, \bibinfo {author} {\bibfnamefont {D.~W.}\ \bibnamefont {Leung}},\ and\ \bibinfo {author} {\bibfnamefont {Y.}~\bibnamefont {Yamamoto}},\ }\href {https://doi.org/10.1103/PhysRevA.56.1114} {\bibfield  {journal} {\bibinfo  {journal} {Physical Review A}\ }\textbf {\bibinfo {volume} {56}},\ \bibinfo {pages} {1114} (\bibinfo {year} {1997})}\BibitemShut {NoStop}%
\bibitem [{\citenamefont {Wetherbee}\ \emph {et~al.}(2026)\citenamefont {Wetherbee}, \citenamefont {Xu}, \citenamefont {Albert}, \citenamefont {Royer},\ and\ \citenamefont {Fatemi}}]{wetherbeeBosonicCyclicCodes2026}%
  \BibitemOpen
  \bibfield  {author} {\bibinfo {author} {\bibfnamefont {O.~C.}\ \bibnamefont {Wetherbee}}, \bibinfo {author} {\bibfnamefont {Y.}~\bibnamefont {Xu}}, \bibinfo {author} {\bibfnamefont {V.~V.}\ \bibnamefont {Albert}}, \bibinfo {author} {\bibfnamefont {B.}~\bibnamefont {Royer}},\ and\ \bibinfo {author} {\bibfnamefont {V.}~\bibnamefont {Fatemi}},\ }\href {https://doi.org/10.48550/arXiv.2606.11010} {\bibinfo {title} {Bosonic {{Cyclic Codes}}: {{Trading Stabilizers}} for {{Gaussian Non-Clifford Phase Gates}}}} (\bibinfo {year} {2026}),\ \Eprint {https://arxiv.org/abs/2606.11010} {arXiv:2606.11010 [quant-ph]} \BibitemShut {NoStop}%
\bibitem [{\citenamefont {Omanakuttan}\ and\ \citenamefont {Gross}(2023)}]{omanakuttanMultispinCliffordCodes2023}%
  \BibitemOpen
  \bibfield  {author} {\bibinfo {author} {\bibfnamefont {S.}~\bibnamefont {Omanakuttan}}\ and\ \bibinfo {author} {\bibfnamefont {J.~A.}\ \bibnamefont {Gross}},\ }\href {https://doi.org/10.1103/PhysRevA.108.022424} {\bibfield  {journal} {\bibinfo  {journal} {Physical Review A}\ }\textbf {\bibinfo {volume} {108}},\ \bibinfo {pages} {022424} (\bibinfo {year} {2023})},\ \Eprint {https://arxiv.org/abs/2304.08611} {arXiv:2304.08611 [quant-ph]} \BibitemShut {NoStop}%
\bibitem [{\citenamefont {Bumgardner}(2012)}]{bumgardnerCodesWmetricSpaces2012a}%
  \BibitemOpen
  \bibfield  {author} {\bibinfo {author} {\bibfnamefont {C.}~\bibnamefont {Bumgardner}},\ }\href {https://doi.org/10.48550/arXiv.1205.4517} {\bibinfo {title} {Codes in {$W^*$}-metric spaces: Theory and examples}} (\bibinfo {year} {2012}),\ \Eprint {https://arxiv.org/abs/1205.4517} {arXiv:1205.4517 [quant-ph]} \BibitemShut {NoStop}%
\bibitem [{\citenamefont {Okada}(2025)}]{okadaQuantumAnalogDelsartes2025a}%
  \BibitemOpen
  \bibfield  {author} {\bibinfo {author} {\bibfnamefont {R.}~\bibnamefont {Okada}},\ }\href {https://doi.org/10.48550/arXiv.2502.14165} {\bibinfo {title} {A {{Quantum Analog}} of {{Delsarte}}'s {{Linear Programming Bounds}}}} (\bibinfo {year} {2025}),\ \Eprint {https://arxiv.org/abs/2502.14165} {arXiv:2502.14165 [quant-ph]} \BibitemShut {NoStop}%
\bibitem [{\citenamefont {Kubischta}\ and\ \citenamefont {Teixeira}(2026)}]{kubischtaMacWilliamsIdentitiesIntrinsic2026}%
  \BibitemOpen
  \bibfield  {author} {\bibinfo {author} {\bibfnamefont {E.}~\bibnamefont {Kubischta}}\ and\ \bibinfo {author} {\bibfnamefont {I.}~\bibnamefont {Teixeira}},\ }\href {https://doi.org/10.48550/arXiv.2604.16023} {\bibinfo {title} {{{MacWilliams Identities}} for {{Intrinsic Quantum Codes}}}} (\bibinfo {year} {2026}),\ \Eprint {https://arxiv.org/abs/2604.16023} {arXiv:2604.16023 [quant-ph]} \BibitemShut {NoStop}%
\bibitem [{\citenamefont {Boyd}\ and\ \citenamefont {Vandenberghe}(2023)}]{boydConvexOptimization2023}%
  \BibitemOpen
  \bibfield  {author} {\bibinfo {author} {\bibfnamefont {S.~P.}\ \bibnamefont {Boyd}}\ and\ \bibinfo {author} {\bibfnamefont {L.}~\bibnamefont {Vandenberghe}},\ }\href@noop {} {\emph {\bibinfo {title} {Convex Optimization}}},\ \bibinfo {edition} {version 29}\ ed.\ (\bibinfo  {publisher} {Cambridge University Press},\ \bibinfo {address} {Cambridge New York Melbourne New Delhi Singapore},\ \bibinfo {year} {2023})\BibitemShut {NoStop}%
\bibitem [{\citenamefont {Var{\v s}alovi{\v c}}(1989)}]{varsalovicQuantumTheoryAngular1989}%
  \BibitemOpen
  \bibfield  {author} {\bibinfo {author} {\bibfnamefont {D.~A.}\ \bibnamefont {Var{\v s}alovi{\v c}}},\ }\href@noop {} {\emph {\bibinfo {title} {Quantum Theory of Angular Momentum: {{Irreducible}} Tensors, Spherical Harmonics, Vector Coupling Coefficients 3 Nj Symbols}}}\ (\bibinfo  {publisher} {World Scientific Pub},\ \bibinfo {address} {Singapore ; Philadelphia},\ \bibinfo {year} {1989})\BibitemShut {NoStop}%
\bibitem [{DLMF()}]{NISTDigitalLibrary}%
  \BibitemOpen
  DLMF,\ \href@noop {} {\bibinfo {title} {{{{\emph{NIST}}}}{\emph{ digital library of mathematical functions}}}}\BibitemShut {NoStop}%
\bibitem [{\citenamefont {Gessel}(1995)}]{gesselFindingIdentitiesWZ1995}%
  \BibitemOpen
  \bibfield  {author} {\bibinfo {author} {\bibfnamefont {I.~M.}\ \bibnamefont {Gessel}},\ }\href {https://doi.org/10.1006/jsco.1995.1064} {\bibfield  {journal} {\bibinfo  {journal} {Journal of Symbolic Computation}\ }\textbf {\bibinfo {volume} {20}},\ \bibinfo {pages} {537} (\bibinfo {year} {1995})}\BibitemShut {NoStop}%
\bibitem [{\citenamefont {Koepf}(2014)}]{koepfHypergeometricSummationAlgorithmic2014}%
  \BibitemOpen
  \bibfield  {author} {\bibinfo {author} {\bibfnamefont {W.}~\bibnamefont {Koepf}},\ }\href {https://doi.org/10.1007/978-1-4471-6464-7} {\emph {\bibinfo {title} {Hypergeometric {{Summation}}: {{An Algorithmic Approach}} to {{Summation}} and {{Special Function Identities}}}}},\ Universitext\ (\bibinfo  {publisher} {Springer},\ \bibinfo {address} {London},\ \bibinfo {year} {2014})\BibitemShut {NoStop}%
\bibitem [{\citenamefont {Rao}\ \emph {et~al.}(1992)\citenamefont {Rao}, \citenamefont {der Jeugt}, \citenamefont {Raynal}, \citenamefont {Jagannathan},\ and\ \citenamefont {Rajeswari}}]{raoGroupTheoreticalBasis1992}%
  \BibitemOpen
  \bibfield  {author} {\bibinfo {author} {\bibfnamefont {K.~S.}\ \bibnamefont {Rao}}, \bibinfo {author} {\bibfnamefont {J.~V.}\ \bibnamefont {der Jeugt}}, \bibinfo {author} {\bibfnamefont {J.}~\bibnamefont {Raynal}}, \bibinfo {author} {\bibfnamefont {R.}~\bibnamefont {Jagannathan}},\ and\ \bibinfo {author} {\bibfnamefont {V.}~\bibnamefont {Rajeswari}},\ }\href {https://doi.org/10.1088/0305-4470/25/4/023} {\bibfield  {journal} {\bibinfo  {journal} {Journal of Physics A: Mathematical and General}\ }\textbf {\bibinfo {volume} {25}},\ \bibinfo {pages} {861} (\bibinfo {year} {1992})}\BibitemShut {NoStop}%
\end{thebibliography}%
}%
  }%
  \onecolumngrid
  \@booleanfalse\twocolumn@sw
  \let\title@column\@empty
  \frontmatter@init
  \setcounter{section}{0}%
  \setcounter{subsection}{0}%
  \setcounter{equation}{0}%
  \setcounter{figure}{0}%
  \setcounter{table}{0}%
  \setcounter{theorem}{0}%
  \setcounter{lemma}{0}%
  \setcounter{proposition}{0}%
  \setcounter{corollary}{0}%
  \setcounter{definition}{0}%
  \setcounter{example}{0}%
  \setcounter{remark}{0}%
  \setcounter{footnote}{0}%
  \setcounter{tocdepth}{2}%
  \setcounter{secnumdepth}{2}%
  \renewcommand{\theequation}{S\arabic{equation}}%
  \renewcommand{\thefigure}{S\arabic{figure}}%
  \renewcommand{\thetable}{S\arabic{table}}%
  \renewcommand*\l@section{\@dottedtocline{1}{0em}{2.8em}}%
  \renewcommand*\l@subsection{\@dottedtocline{2}{2.8em}{3.2em}}%
  \let\FinalOriginalLabel\label
  \let\FinalOriginalRef\ref
  \let\FinalOriginalPageref\pageref
  \renewcommand{\label}[1]{\FinalOriginalLabel{supp:##1}}%
  \let\ltx@label\label
  \let\FinalOriginalTOC\tableofcontents
  \renewcommand{\tableofcontents}{%
    \let\FinalCitationContext\the@ipfilectr
    \FinalOriginalTOC
    \global\let\the@ipfilectr\FinalCitationContext
  }%
  \renewcommand{\ref}[1]{\FinalOriginalRef{supp:##1}}%
  \renewcommand{\pageref}[1]{\FinalOriginalPageref{supp:##1}}%
}
\makeatother

\begin{document}
\begingroup
\FinalLetterSetup
\renewcommand{\addcontentsline}[3]{}

\raggedbottom

\title{Optimal two-mode bosonic loss codes from finite group symmetry}

\author{Argyris Giannisis Manes}
\email{argy@uchicago.edu}

\author{Mahadevan Subramanian}

\author{Liang Jiang}
\email{liangjiang@uchicago.edu}

\affiliation{Chicago Quantum Institute and Pritzker School of Molecular Engineering, University of Chicago, Illinois 60637, USA}

\begin{abstract}
 Photon loss is a dominant noise process in bosonic quantum hardware, including superconducting cavities. Fixed-total-photon-number qubit encodings in two bosonic modes retain the loss-detection advantage of dual-rail qubits while supporting photon loss correction. Optimizing entanglement fidelity in this setting for $4\leq n\leq25$ over arbitrary encoders and decoders reveals finite-group structure in every best-found code: 11 correspond to two-dimensional irreducible representations and 11 to reducible ones. Motivated by this emergence, we derive the necessary-and-sufficient Knill--Laflamme conditions for arbitrary finite-group-invariant codes, reducing their construction to equations on representation multiplicity spaces. For every finite subgroup of $SU(2)$ and every tensor rank, we explicitly construct the minimum number of operators required to impose all corresponding loss-correction constraints. These results produce analytic counterparts to most numerical codes and predict constructions missed by the initial search. Exact MacWilliams--Farkas certificates show that our codes achieve the maximum possible loss distance in 21 of the 22 sectors, and we prove that the resulting distance bound is monotone in total photon number. Beyond the scan, we identify a binary-polyhedral sequence with photon number $n_d=\lceil(3d^2+1)/4\rceil$ and construct each corresponding code through distance $d=10$. To our knowledge, the constructed $(n,d)=(28,6),(49,8),(76,10)$ codes give the smallest reported $n$ for their respective distances; all members through $d=9$ attain the fixed-$n$ LP distance bound. Together, these constructions, certificates, and a symmetry-reduced constraint count provide evidence for an infinite code family conjectured to attain every distance at the minimum photon number.
\end{abstract}

\maketitle
\enlargethispage{12pt}

\par\smallskip
\noindent\textit{Introduction.---}
Fault-tolerant quantum computation requires protecting quantum information from noise while retaining the ability to manipulate it efficiently. Quantum error correction (QEC) typically achieves this by encoding logical information across multiple physical qubits, whereas hardware-efficient bosonic or spin encodings use the large Hilbert spaces of bosonic modes \cite{cochraneMacroscopicallyDistinctQuantumsuperposition1999, gottesmanEncodingQubitOscillator2001,michaelNewClassQuantum2016,terhalScalableBosonicQuantum2020,ofekExtendingLifetimeQuantum2016} or spins \cite{grossDesigningCodesInteractions2021,omanakuttanMultispinCliffordCodes2023,omanakuttanFaultTolerantQuantumComputation2024,grossHardwareefficientErrorcorrectingCodes2024} to introduce the necessary redundancy. In both settings, QEC performance can be significantly improved by tailoring the code to the structure of the physical noise, such as strong noise bias \cite{aliferisFaulttolerantQuantumComputation2008,bonillaataidesXZZXSurfaceCode2021} or knowledge of the error locations (erasures) \cite{staceThresholdsTopologicalCodes2009,barrettFaultTolerantQuantum2010}. 

A particularly attractive example of hardware-aware QEC is the dual-rail encoding \cite{knillSchemeEfficientQuantum2001,kokLinearOpticalQuantum2007}, which uses the Fock states $\ket{0_L} = \ket{0} \otimes \ket{1},\, \ket{1_L} = \ket{1} \otimes \ket{0}$ to encode a qubit. Since the total photon number of both logical states is $1$, photon loss can be detected by total-photon-number or parity measurements without corrupting the logical information. This loss-detection mechanism has been demonstrated in superconducting cavities \cite{teohDualrailEncodingSuperconducting2023, chouSuperconductingDualrailCavity2024}, converting dominant photon loss errors into erasures that an outer code can correct with much lower overhead \cite{staceThresholdsTopologicalCodes2009}. 

Despite its many advantages, dual-rail only uses a two-dimensional subspace of its oscillators’ infinite-dimensional Hilbert space and can detect, but not correct, photon loss. Encoding in a higher fixed-total-photon-number subspace preserves this detection capability while adding enough redundancy for loss correction. Codes of this kind were first introduced by Chuang, Leung, and Yamamoto \cite{chuangBosonicQuantumCodes1997}, with later constructions including the two-mode binomial codes \cite{albertPerformanceStructureSinglemode2018}. The fixed-$n$ subspace of the two bosonic modes corresponds to a spin-$j=n/2$ system through the Schwinger representation \cite{sakuraiModernQuantumMechanics2021}. This connection enables a mapping of spin codes \cite{grossDesigningCodesInteractions2021,omanakuttanMultispinCliffordCodes2023,kubischtaFamilyQuantumCodes2023} and permutationally-invariant (PI) codes \cite{pollatsekPermutationallyInvariantCodes2004,ouyangPermutationinvariantQuantumCodes2014,aydinFamilyPermutationallyInvariant2024,kubischtaPermutationInvariantQuantumCodes2025} to fixed-$n$ two-mode bosonic loss codes \cite{aydinQuantumErrorCorrection2026}. This mapping provides us with a rich family of possible constructions, but does not identify which codes perform optimally against loss. 

\begin{figure*}[t!]
  \centering
  \includegraphics[width=\textwidth]{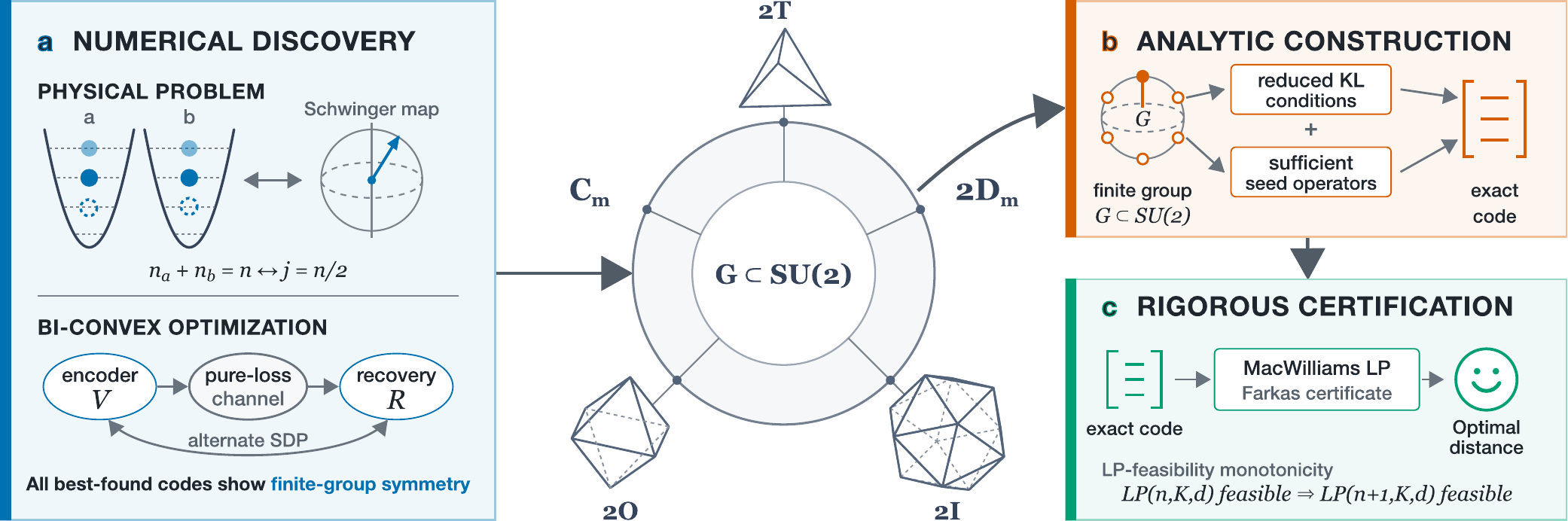}
  \caption{Discovery, construction, and certification of fixed-excitation
  two-mode loss codes. (a) Symmetry-unconstrained bi-convex optimization of
  the encoder and recovery produces numerical codes with finite-group
  representation structure. (b) General representation-theoretic
  Knill--Laflamme reductions and sufficient loss-seed operators turn that
  structure into exact codes. (c) MacWilliams linear programs, exact Farkas
  certificates, and LP-feasibility monotonicity certify optimal code
  distances.}
  \label{fig:discovery-construction-certification}
\end{figure*}
 
In this work, unconstrained entanglement fidelity optimization in the fixed-total-photon-number
subspace reveals finite-group representation structure in all 22 best-found
codes for $4\leq n\leq25$. Building on representation-theoretic QEC
\cite{grossDesigningCodesInteractions2021,omanakuttanMultispinCliffordCodes2023,
denysQuantumErrorCorrectingCodes2024,kubischtaQuantumCodesTwisted2024,
kubischtaIntrinsicQuantumCodes2025,
bradshawSymmetryBasedQuantumCodes2026}, we derive necessary-and-sufficient Knill-Laflamme (KL)
conditions for arbitrary finite-group-invariant codes and construct a
provably minimum set of operators imposing the error-correction constraints at
each loss order. This framework turns the numerically discovered codes into exact
codes and predicts constructions beyond the initial search. MacWilliams-Farkas certificates \cite{bumgardnerCodesWmetricSpaces2012a,okadaQuantumAnalogDelsartes2025a,kubischtaMacWilliamsIdentitiesIntrinsic2026} establish distance optimality in 21 of the 22 scanned sectors, and we prove that the
corresponding LP bound is monotonic in $n$. Finally, we use our formalism to construct binary-polyhedral loss codes of distance $d$ using $n_d=\lceil(3d^2+1)/4\rceil$ photons for every $2\leq d\leq10$. At distances $6$, $8$, and $10$, these use fewer photons than any previously reported code, to our knowledge. A symmetry-reduced constraint count suggests that this sequence extends to all distances, attaining the conjectured minimum photon number \cite{bondPermutationinvariantCodesNumerical2026}.
Figure~\ref{fig:discovery-construction-certification}
summarizes this pipeline.

\par\smallskip
\noindent\textit{Numerical discovery.---}
We start by encoding a $d_L$-dimensional logical system into the fixed-total-excitation subspace of two modes $H_n = \operatorname{span}\{\ket{k} \otimes \ket{n-k}: 0 \leq k \leq n \}$ and write $H_{\leq n} = \bigoplus_{q =0}^n H_q$.  Symmetric pure loss with transmissivity $\eta$ is described by the channel $\mathcal{N}: \mathcal{L}(H_n) \to \mathcal{L}(H_{\leq n})$
\begin{equation}
  \mathcal{N}(\rho) = \sum_{l_1 + l_2 \leq n} K_{l_1 l_2} \rho K_{l_1 l_2}^\dagger 
 \label{two-mode-loss-channel}
\end{equation}
with Kraus operators
\begin{equation}
  K_{l_1 l_2} = \sqrt{\frac{(1-\eta)^{l_1 + l_2}}{l_1! l_2!}} \eta^{\hat{n}_1/2} \hat{a}_1^{l_1} \otimes \eta^{\hat{n}_2/2} \hat{a}_2^{l_2}. \label{loss-kraus}
\end{equation}
Here, $\mathcal{L}(A)$ is the space of linear operators on $A$, and $\hat{a}_i,\, \hat{n}_i$ are the annihilation and number operators on mode $i$.    

We optimize CPTP maps $\mathcal{E}: \mathcal{L}(\mathbb{C}^{d_L}) \to \mathcal{L}(H_n)$ (encoder), $\mathcal{D}: \mathcal{L}\left(H_{\leq n}\right) \to \mathcal{L}(\mathbb{C}^{d_L})$ (decoder)
that maximize the entanglement fidelity of the composite channel given by
\begin{equation}
  \begin{aligned}
  F_e &= \braket{\Phi^+|\left[(\mathcal{D} \circ \mathcal{N} \circ \mathcal{E}) \otimes \mathcal{I}\right]\left(\ket{\Phi^+}\bra{\Phi^+}\right)|\Phi^+},  \\ 
  \ket{\Phi^+} &= \frac{1}{\sqrt{d_L}} \sum_{i =0}^{d_L - 1} \ket{i} \otimes \ket{i},
  \end{aligned}\label{entanglement-fidelity}
\end{equation}
where $\mathcal{I}$ is the identity channel, and $\ket{\Phi^+}$ is the $d_L$-dimensional Bell state \cite{nohQuantumCapacityBounds2019}. We choose entanglement fidelity because it is a key metric of QEC performance: it is affinely related to the average fidelity \cite{horodeckiGeneralTeleportationChannel1999}, and bounds the worst-case fidelity as $F_{\min} \geq 1 - d_L (1 - F_e)$ \cite{schumacherSendingEntanglementNoisy1996}. 

Since $F_e$ is linear in each Choi matrix, fixing one channel makes optimizing the other a semidefinite program (SDP) \cite{reimpellIterativeOptimizationQuantum2005,taghaviChannelOptimizedQuantumError2010,nohQuantumCapacityBounds2019} that is solvable in polynomial complexity in system dimension. The full problem is biconvex, so convergence doesn't guarantee global optimality. Fixed total excitation provides an extra simplification: outputs with different surviving photon numbers occupy orthogonal sectors, so the optimal recovery decomposes exactly into independent SDPs on each $H_q$, yielding a significant speedup over full-decoder optimization (see the Supplemental Material for further discussion).

Using this sector-decomposed alternating optimization, we optimize logical qubits independently for $4\leq n\leq25$ at high transmissivity. All runs start from independent random encoders, without symmetry constraints or analytic-code warm starts. Repeated random starts generally recover the same code structure, although this depends on photon number and transmissivity; secondary basins and lower-transmissivity results are reported in the Supplemental Material.

The resulting codes exhibit a striking common structure: every best-found high-transmissivity code aligns closely with a $G$-invariant rank-two projector for some finite group $G\subset SU(2)$. The overlap of the numerically obtained code projector $P$ with the nearest group-invariant projector $\widehat Q_G$ given by $
\Omega_G(P)=\tfrac1{d_L}\operatorname{Tr}(P\widehat Q_G) $  satisfies $\Omega_G(P)>0.99$ in 21 of the 22 sectors. The exception is $n=15$, with overlap $0.975$, where the random-start search does not fully reach the exact binary-tetrahedral construction. Eleven of the resulting invariant codes carry two-dimensional irreducible representations, while eleven are reducible, including two invariant under non-Abelian groups: $2T$ at $n =12$ and $2D_2$ at $n = 18$. 

The emerging symmetry is surprising: the Schwinger correspondence alone imposes no finite-group symmetry on the codes. To explain its emergence, we first relate the loss-error products entering the KL conditions to irreducible spin tensors.
\par\smallskip
\noindent\textit{Loss errors and spin tensors.---}
For a code projector $P\in\mathcal L(H_n)$, KL products $K_{l_1'l_2'}^\dagger K_{l_1l_2}$ with $l_1 + l_2 \neq l_1' + l_2'$ vanish after projection by $P$, since the Kraus operators map $P$ onto orthogonal photon-number sectors. Moreover $K_{l_1l_2}P\propto
\bigl(\hat a_1^{l_1}\otimes\hat a_2^{l_2}\bigr)P$ \cite{albertPerformanceStructureSinglemode2018}, so the order-$k$ loss KL conditions \cite{knillTheoryQuantumErrorcorrecting1997} are equivalent to $POP\propto P$ for every
\begin{equation}
  \begin{aligned}
    O \in \mathcal{K}_k &= \operatorname{span}\{A^\dagger B:\, A,B \in \mathcal{L}_k\} \\ 
  \mathcal{L}_k &= \operatorname{span}\{\hat{a}_1^t \otimes \hat{a}_2^{k-t}: \ 0 \leq t \leq k\}
  \end{aligned}
\end{equation}

Now, let $S_k$ be the spin-$k$ irreducible representation of $SU(2)$, and denote with $\mathcal{T}_n^\ell \subseteq \mathcal{L}(H_n)$ the space of irreducible spherical tensors of rank $\ell$. We prove the following:

\begin{lemma}
  A code on $H_n$ corrects photon loss errors up to order $k$ iff it satisfies the KL detection conditions for $\mathcal{K}_k$, which we can decompose in terms of $SU(2)$ representations as
  \begin{equation}
    \mathcal{K}_k \cong S_k \oplus S_{k-1} \oplus \dots \oplus S_0
  \end{equation}
  Equivalently, the corresponding spin code detects all $SU(2)$ irreducible tensors $\mathcal{T}_n^0,\, \mathcal{T}_n^1,\, \dots,\, \mathcal{T}_n^k$. \label{loss-tensor-equivalence}
\end{lemma}

 The proof follows from standard $SU(2)$ representation theory. The first step is to show that $\mathcal{L}_k \cong S_{k/2}$, from which it follows that $\mathcal{K}_k \cong S_{k/2}^* \otimes S_{k/2} \cong \bigoplus_{\ell=0}^k S_\ell$. We further show that $\mathcal{K}_{k-1} \subseteq \mathcal{K}_k$, and the result follows. See the Supplemental Material for the detailed proof. 

We note that Aydin, Albert, and Barg map spin codes to two-mode loss codes with the same distance \cite{aydinQuantumErrorCorrection2026}; our lemma gives a direct operator-space proof of the converse by identifying correction of up to $d-1$ photon losses (loss distance $d$) with detection of tensors up to rank $d-1$. We also note that the decomposition $\mathcal K_1\cong S_0\oplus S_1$ is also used in Ref.~\cite{kubischtaIntrinsicQuantumCodes2025} in its analysis of the $n=4$ CLY code.

\par\smallskip
\noindent\textit{Representation-theoretic construction.---}
To turn the loss–spin equivalence into explicit codes, we first derive the KL conditions for arbitrary $G$-invariant codes and then specialize to photon loss.

Let $G$ be a finite group with complex unitary representation $\rho:G \to U(V)$ on a vector space $V$. Write $V \cong \bigoplus_C M_C \otimes C$ as a $G$-representation, with $G$-equivariant isometries $S_C: M_C \otimes C \to V$ of the $C$-isotypic component into $V$, and $M_C$ the corresponding multiplicity space. Any $G$-invariant projector $P \in \mathcal{L}(V)$ has the form
  \begin{equation}
      P = \sum_C S_C(P_C \otimes I_C) S_C^\dagger \label{G-invariant-projector}
  \end{equation}
  with $P_C$ projectors on $M_C$. This includes irreducible covariant encodings as the case of a single nonzero rank-one $P_C$ \cite{denysQuantumErrorCorrectingCodes2024}.

To study the KL conditions for these projectors, we let $W\leq\mathcal L(V)$ \cite{kubischtaIntrinsicQuantumCodes2025} be a $G$-invariant space of KL operators under the conjugation action, and write $W\cong\bigoplus_U U^{\oplus m_U}$, where $m_U$ is the multiplicity of the irrep $U$ of $G$. For the $\alpha$ copy of $U$, we choose basis operators $T_\mu^{U,\alpha}$ transforming as $U$, where $1\leq\mu\leq d_U=\dim U$ and $1\leq\alpha\leq m_U$.

Symmetry implies that their matrix elements between $D$- and $C$-isotypic components have the form
\begin{equation}
    S_C^\dagger T_\mu^{U,\alpha}S_D
=
\sum_{j=1}^{r_U^{C,D}}
B_{j,\alpha}^{U;C,D}\otimes F_{j,\mu}^{U;C,D}, \label{BF-decomposition}
\end{equation}
where $j$ labels the $r_U^{C,D}$ copies of $U$ in $D^*\otimes C$. This can be seen as a finite-group Wigner--Eckart decomposition. The $F_{j,\mu}^{U;C,D}$ can be pre-computed, since they depend only on the representation data; they play the role of Clebsch--Gordan coefficients. On the other hand, the $B_{j,\alpha}^{U;C,D}:M_D\to M_C$ map between multiplicity spaces; $G$-equivariance makes them independent of $\mu$. Since the code projectors $P_C$ also act on the multiplicity spaces, the KL conditions reduce to constraints on these $B$ matrices:

\begin{theorem}[KL-conditions]
   Any $G$-invariant projector $P \in \mathcal{L}(V)$  satisfies the error detection conditions for every operator in $W \cong \bigoplus_U U^{\oplus m_U}$ iff 
  \begin{equation}
     P_C B_{j, \alpha}^{U; C,D} P_D = \begin{cases}
       \lambda_\alpha P_C, & D = C,\, U = \mathbf{1},\, j = 1 \\ 
       0, & \text{otherwise} 
     \end{cases}
  \end{equation}

  for all $1 \leq \alpha \leq m_U, \, 1 \leq j \leq r_U^{C,D}$. \label{general-KL-conditions}
\end{theorem}

Notice that only irreps $U \subset  D^* \otimes C$ contribute non-trivial KL constraints, thereby guiding the choice of representations with a favorable balance of constraints and multiplicity-space freedom; the single-irrep case of this selection rule was established in \cite{kubischtaQuantumCodesTwisted2024}.  Also note that $\lambda_\alpha$ is common to all nonzero $C$-sectors but may depend on the multiplicity label $\alpha$. Proofs, formulas for $B_{j,\alpha}^{U;C,D}$, and normalizations for $F_{j,\mu}^{U;C,D}$ are given in the Supplement.

\looseness=-1
Theorem \ref{general-KL-conditions} reduces the KL conditions to multiplicity-space equations, but constructing a complete $G$-adapted basis of $W$ remains cumbersome. The symmetry structure makes this unnecessary: $G$-invariance gives $POP=\lambda_OP\Longrightarrow P(g\cdot O)P=\lambda_OP$, so one seed operator $O$ enforces the KL condition on its entire $G$-orbit span. Therefore, the remaining problem is to find the smallest set of seeds whose combined $G$-orbits span $W$. The underlying orbit-generator theorem is due to Vale and Waldron \cite{valeConstructionGInvariantFinite2016}; applying it to QEC gives the following result, which we prove in the Supplemental Material.

\begin{theorem}
A $G$-invariant projector $P$ detects
$W\cong\bigoplus_U U^{\oplus m_U}$ iff it detects an appropriate
orbit-generating set of
$s_{\min}(W)=\max_U\lceil m_U/d_U\rceil$ seeds.
No smaller set of seed orbits spans $W$.\label{general-seed-orbit-theorem} 
\end{theorem}

Theorem \ref{general-seed-orbit-theorem} gives the minimum seed count, but not an explicit choice of seeds. For photon loss, we attain this minimum on each $W=\mathcal T_n^\ell$ using rotated powers of the raising operator $J_+$.
\begin{theorem}
  Let $G$ be a finite subgroup of $SU(2)$, define $\widetilde G=G\slash (G\cap\{I,-I\})$, and let $P$ be a $G$-invariant code projector. For every tensor rank $1\leq \ell \leq n$, set $s_\ell =\left\lceil\frac{2\ell+1}{|\widetilde G|}\right\rceil$. There exist rotations $U_1,\ldots,U_{s_\ell}\in SU(2)$ such that $P$ detects every operator in $\mathcal{T}_n^\ell$ iff
  \begin{equation}
    P R_n(U_t)J_+^\ell R_n(U_t)^\dagger P\propto P, \quad 1 \leq t \leq s_\ell.
    \label{minimum-seeds-equation}
  \end{equation}
  Here $R_n$ is the spin-$n/2$ representation, and $s_\ell$ attains the minimum from Theorem \ref{general-seed-orbit-theorem} for $W = \mathcal{T}_n^\ell$. \label{minimum-seeds-theorem}
\end{theorem}

Such rotations are obtained from distinct projective orbits with trivial $\widetilde G$-stabilizers. Writing $(a_t,b_t)^{\mathsf T}$ for the first column of $U_t$, the corresponding seed is
\begin{equation}
R_n(U_t)J_+^\ell R_n(U_t)^\dagger
=\sum_{\mu=-\ell}^{\ell}
\binom{2\ell}{\ell+\mu}^{1/2}
a_t^{\ell+\mu}b_t^{\ell-\mu}T_\mu^{(\ell)}
\label{rotated-seed-tensor-expansion}
\end{equation}
in terms of the irreducible $SU(2)$ tensors of rank $\ell$, with $T_\ell^{(\ell)} = J_+^\ell$ ; see the Supplemental Material for details. 

We note that earlier works developed symmetry-based reductions of the KL conditions for particular finite-group codes, at leading order \cite{grossDesigningCodesInteractions2021} and at higher ranks for real $2O$ and binary-dihedral codes \cite{omanakuttanMultispinCliffordCodes2023,kubischtaPermutationInvariantQuantumCodes2025}. Our construction is distinct: it applies uniformly to all $G\subset SU(2)$ and tensor ranks, requires no reality assumption, and is levelwise minimal.

Together, Theorems 1 and 3 guide the choice of representations with favorable constraint and degree-of-freedom counts and provide levelwise-minimal loss seeds. We compute the corresponding multiplicity-space KL conditions from these seeds and solve them through the maximum possible rank, using any remaining degrees of freedom to minimize the next-order KL residual.
\FloatBarrier
\par\smallskip
\noindent\textit{Exact codes and predictive power.---}
We illustrate the three roles of the analytic construction through Figure~\ref{fig:representative-wigner-functions}. At $n=22$, the aligned numerical projector nearly coincides with the constructed $2D_5$ code [Figs.~\ref{fig:representative-wigner-functions}(a,b)], showing the agreement we observe in almost every sector.  At $n=15$, the analytic $2T$ construction identified a code missed by our initial random-start search and motivated additional runs [Fig. \ref{fig:representative-wigner-functions}(c)]. These subsequently reached comparable fidelity at $\eta=0.95$, although the best numerical projector remains only approximately $2T$-invariant. The full classification and numerical comparisons are given in the Supplemental material.

Last but not least, our construction extends beyond the numerical scan to codes with improved photon-number parameters. For every integer $2\leq d\leq10$, we construct a rank-two loss code of distance $d$ at $n_d=\lceil(3d^2+1)/4\rceil$, with $2O$ symmetry for even $d$ and $2T$ symmetry for odd $d$. To our knowledge, the new $(28,6),(49,8)$, and $(76,10)$ binary-octahedral codes improve the smallest previously reported photon numbers from $36$~\cite{ouyangPermutationinvariantQuantumCodes2014,aydinQuantumErrorCorrection2026}, $61$, and $91$~\cite{kubischtaPermutationInvariantQuantumCodes2025}, respectively.  The $n=28$ code shown in Fig.~\ref{fig:representative-wigner-functions}(d) is the first member beyond the original scan. We note that the $n=7$ \cite{grossDesigningCodesInteractions2021} and $n=37$ members also carry $2I$ symmetry. 

These finite constructions aren't isolated, but suggest an infinite binary-polyhedral family: a proposed uniform $2T$ realization of the previously conjectured odd-distance binary-dihedral sequence \cite{kubischtaPermutationInvariantQuantumCodes2025}, together with a new $2O$ branch at even distances. We justify the conjectured existence of an infinite family by comparing the number of relevant symmetry-reduced KL constraints with the available multiplicity-space degrees of freedom. For suitable two-dimensional $2O$ and $2T$ irreps $C$, the relevant KL--operator space is adjoint-closed and all irreps it shares with $C \otimes C^*$ are real, allowing the reduced $B_{j,\alpha}^{U;C,C}$ matrices to be chosen Hermitian and the KL constraints to be counted over the reals. At every point of the proposed sequence, this formal constraint count exactly matches the $2\dim M_C - 2$ degrees of freedom of a rank-one multiplicity projector, providing structural evidence—but not a proof—for the following conjecture.
\begin{conjecture}[Binary-polyhedral family]
  For every integer $d\geq2$, a rank-two loss code of distance at least $d$ exists at $n_d=\lceil(3d^2+1)/4\rceil$, with $2O$ symmetry for even $d$ and $2T$ symmetry for odd $d$.
  \label{binary-polyhedral-conjecture}
\end{conjecture}

\par
 
\begin{figure}[h!]
  \centering
  \includegraphics[width=\columnwidth]{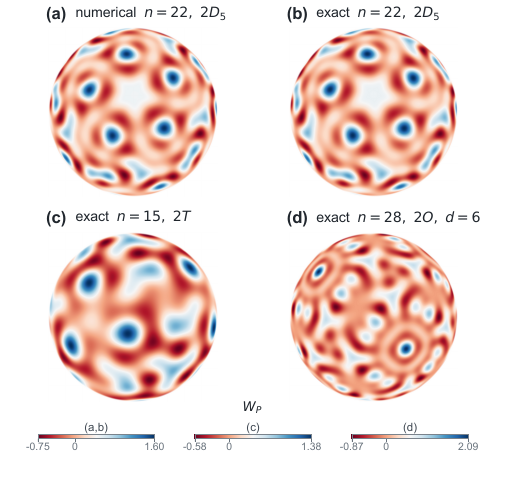}
  \caption{Spherical spin-Wigner functions $W_P(\theta,\phi)$ for rank-two projectors: (a,b) aligned numerical and exact $n=22$ $2D_5$ codes; (c) exact $n=15$ $2T$ code; (d) exact $n=28$ $2O$ code with $d=6$.}
  \label{fig:representative-wigner-functions}
\end{figure}

\par\smallskip
\noindent\textit{Certification and optimality.---}
Since biconvex optimization cannot certify global fidelity optimality, we instead derive exact upper bounds on loss distance and compare them with our constructions. Combining Lemma 1 with the $SU(2)$ MacWilliams linear program \cite{bumgardnerCodesWmetricSpaces2012a,okadaQuantumAnalogDelsartes2025a,kubischtaMacWilliamsIdentitiesIntrinsic2026} gives a necessary linear program $\operatorname{LP}(n,K,d)$ for a rank-$K$, distance-$d$ code on $H_n$. When this program is infeasible, Farkas’ lemma provides an exact certificate ruling out those code parameters \cite{boydConvexOptimization2023}; feasibility alone, however, does not imply code existence.

The LP distance bound is saturated to high precision at $21$ of the $22$ photon numbers $4\leq n\leq25$. The sole gap is $n=18$, where the LP permits $d=5$ but our best construction has $d=4$; Figure~\ref{fig:distance-phase-diagram} summarizes the complete comparison. One noteworthy example is $n = 13$. Gross’s $n=13$ $2O$ family corrects first-order angular-momentum errors while retaining a free phase $\phi$ \cite{grossDesigningCodesInteractions2021}, and later PI, spin, and intrinsic-code treatments reported distance $3$ for members of the same family \cite{omanakuttanMultispinCliffordCodes2023,aydinQuantumErrorCorrection2026,kubischtaIntrinsicQuantumCodes2025}. We choose $\phi$ so that $J_+^3$ is detected; our orbit theorem and KL reduction promote this single condition to detection of the full rank-$3$ sector, yielding distance $4$. The Farkas bound then certifies $13$ is the first photon number where this distance is possible.

Beyond the original scan, every constructed member of our binary-polyhedral family through $d=9$ saturates the fixed-$n_d$ LP distance bound to machine precision; the distance-$10$ code at $n=76$ lies one below the LP bound of $11$. Saturation at fixed photon number does not, however, establish minimum photon number: for distances $6$ and $8$, the LP first permits $n=26$ and $43$, below our constructions at $28$ and $49$. Whether these gaps admit better codes remains open. Indeed, Bond et al. \cite{bondPermutationinvariantCodesNumerical2026} used numerical searches to conjecture that $n_d$ is the minimum photon number required for distance at least $d$, in the equivalent setting of PI codes. Our constructions and constraint count provide further evidence for the optimality of this sequence. If correct, this conjecture would imply systematic non-tightness of the LP bound.

Finally, Figure~\ref{fig:distance-phase-diagram}'s staircase pattern reflects LP-feasibility monotonicity.
\begin{theorem}
   If $\operatorname{LP}(n,K,d)$ is feasible, then $\operatorname{LP}(n+1,K,d)$ is feasible. \label{LP-monotonicity}
\end{theorem}
The proof constructs a nonnegative, column-stochastic, lower-triangular $T_n$ with $M_{n+1}T_n=T_nM_n$ for every MacWilliams matrix $M_n$, mapping LP solutions from $n$ to $n+1$. Thus, LP infeasibility propagates downward, while monotonicity of code constructions remains open.

\par\smallskip
\noindent\textit{Discussion and outlook.---}
We've seen that finite-group symmetry is not merely an elegant ansatz for systematically constructing bosonic and spin codes: it emerges from unconstrained loss optimization, and leads to code constructions that are optimal in loss distance. This surprising connection echoes GKP codes \cite{gottesmanEncodingQubitOscillator2001}: they were originally constructed to correct for small displacements in a single bosonic mode, and yet they are also found through the same biconvex optimization for single-mode loss, while achieving the capacity of the channel \cite{nohQuantumCapacityBounds2019,zhengPerformanceAchievableRates2025,subramanianAchievableRatesConcatenated2025}. Likewise, Gross introduced finite-group spin codes to correct low-order rotation errors \cite{grossDesigningCodesInteractions2021}; here, the same structures emerge from photon-loss optimization and often attain certified optimal loss distance.

Several theoretical directions remain open. Although we have focused on qubits, our representation-theoretic tools apply to arbitrary code dimension and may yield new code families in spin, permutation-invariant, bosonic, and other settings. Most immediately, it would be valuable to prove the conjectured binary-polyhedral family. One possible route is to embed each solution into the next occurrence of its supporting irrep, but constructing such a nesting is unexpectedly subtle and relates to the broader question of whether a spin-$n/2$ code can always be promoted to spin-$(n+1)/2$ while preserving its distance. Lastly, the widening separation between the LP-feasibility and construction frontiers suggests that the intrinsic-code relaxation is not tight. Deriving stronger bounds that either close this gap or reveal new codes within it is therefore an important open problem. 

Finally, our original motivation was physical, and these codes retain the defining hardware advantage of dual rail: photon loss changes the total excitation number and is therefore directly detectable. At the same time, every element of $G\subset SU(2)$ is implementable by passive Gaussian operations in two modes and preserves the codespace, inducing its corresponding logical action; the non-trivial logical gates depend on the kernel of the selected representation. Supplementing these native operations with number-selective control \cite{heeresCavityStateManipulation2015} may provide a route to a universal gate set, which we will investigate in future hardware work. More broadly, because our construction detects the complete space of spin-tensor errors through a prescribed rank, it may be adapted to the native noise processes of other spin-like platforms, including trapped-ion systems \cite{houCoherentCouplingNondestructive2024}. The emergent symmetry therefore provides both a systematic code-design principle and a natural route toward hardware implementation.


\FloatBarrier

\par\smallskip
\noindent\textit{Acknowledgments.---}
AGM dedicates this work to the memory of his beloved grandmother, Kaiti Giannisi.

We thank Takahiro Tsunoda, Gideon Lee, Guo Zheng, Wes Campbell, James Teoh, Kevin Chou, and Kaavya Sahay for helpful discussions. 

We acknowledge support from the ARO(W911NF-23-1-0077), ARO MURI (W911NF-21-1-0325), AFOSR MURI (FA9550-21-1-0209, FA9550-23-1-0338), ONR MURI (N000142612102), DARPA (HR0011-24-9-0361), NSF (OSI-2326767, OSI-2426975, EEC-2550064, OSI-2553578, OSI-2553619). This material is based upon work supported by the U.S. Department of Energy, Office of Science, National Quantum Information Science Research Centers and Advanced Scientific Computing Research (ASCR) program under contract number DE-AC02-06CH11357 as part of the InterQnet quantum networking project. This work was completed with resources provided by the University of Chicago’s Research Computing Center.

\par\smallskip
\noindent\textit{AI disclosure.---}
The construction of the intertwiner $T_n$ was developed with OpenAI’s Codex (GPT-5.6 Sol). Codex and ChatGPT (GPT-6 Astra, GPT-5.6 Sol and earlier GPT-5 models) also assisted with analyzing numerical optimization results, generating figures, writing parts of our code, workshopping proofs, and preparing the manuscript, including minor arguments and expository revisions. The authors independently verified all outputs and take full responsibility for the results and claims.

\par\smallskip
\noindent\textit{Code availability.---}
The software and resulting QEC codes are available in our \href{https://github.com/argy-gm-code/optimal-two-mode-bosonic-loss-codes}{GitHub repository}.
\par

\bibliographystyle{apsrev4-2}
\bibliography{Candidacy}

\onecolumngrid
\section*{End Matter}

\begin{center}
  \begin{minipage}{\textwidth}
  \centering
  \includegraphics[width=\textwidth]{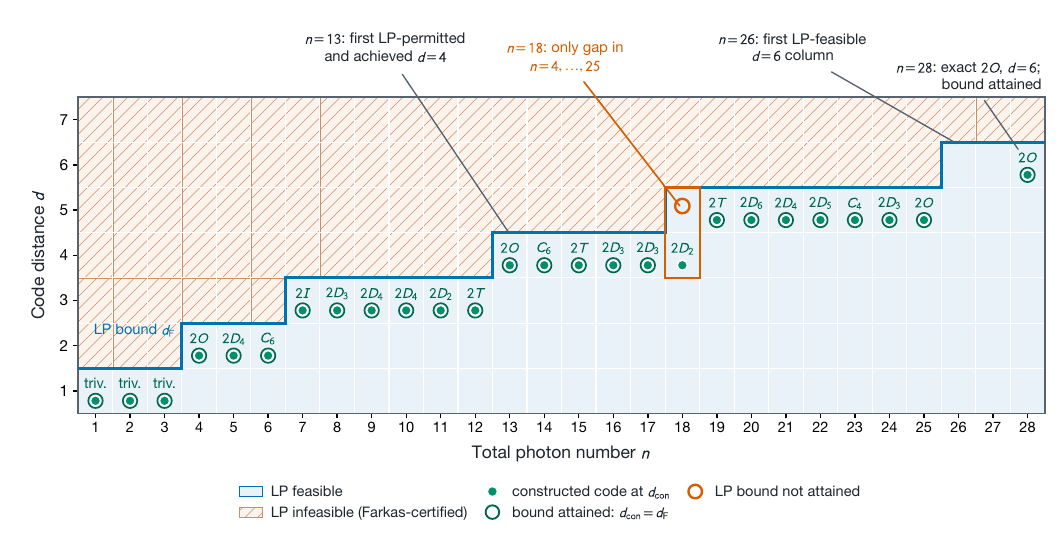}
  \captionof{figure}{Rank-two distance phase diagram. Each cell shows feasibility of
  the intrinsic $SU(2)$ MacWilliams linear-programming bound for total
  photon number $n$ and loss distance $d$: blue cells are LP feasible,
  whereas hatched orange cells are excluded by exact Farkas certificates. The dark blue line marks the boundary between feasible and infeasible regions. 
  LP feasibility is necessary, but not sufficient, for an actual code.
  Green markers show constructed codes at $d_{\mathrm{con}}$ and are labeled
  by their construction group; an outer green ring indicates
  $d_{\mathrm{con}}=d_{\mathrm{F}}$, where $d_{\mathrm{F}}$ is the largest
  LP-feasible distance. At $n=18$, the open orange circle marks the
  unattained LP bound. Constructed codes attain the LP upper bound in $21$
  of the $22$ studied sectors $n=4,\ldots,25$, with the sole gap at $n=18$
  ($d_{\mathrm{F}}=5$, $d_{\mathrm{con}}=4$). The staircase follows
  LP-feasibility monotonicity. Distance $4$ is first permitted and attained
  at $n=13$, and the $n=19$ $2T$ construction attains distance $5$. The
  LP first permits distance $6$ at $n=26$; the exact $n=28$ $2O$
  construction attains distance $6$ and saturates the bound. } 
  \label{fig:distance-phase-diagram}
  \end{minipage}
\end{center}


\clearpage
\endgroup

\begingroup
\FinalSupplementSetup

\raggedbottom

\title{Supplemental Material for ``Optimal two-mode bosonic loss codes from finite group symmetry''}

\author{Argyris Giannisis Manes}

\author{Mahadevan Subramanian}

\author{Liang Jiang}

\affiliation{Chicago Quantum Institute and Pritzker School of Molecular Engineering, University of Chicago, Illinois 60637, USA}

\date{September 2026}

\maketitle

\tableofcontents
\clearpage

\section{Two-mode loss channel and symmetry-unconstrained optimization}

\subsection{Fixed-excitation Hilbert space and two-mode loss}

We denote $H_n = \operatorname{span}\{\ket{k} \otimes \ket{n-k}: \ 0 \leq k \leq n \}$ and let $P_n$ be the projector onto $H_n$. Throughout this subsection, $\rho$ is a density operator supported on $H_n$, so $\rho=P_n\rho P_n$ and $\Tr\rho=1$. The two-mode loss channel is given by

\begin{equation}
\mathcal{N}(\rho) = \sum_{l_1 = 0}^n \sum_{l_2 = 0}^{n-l_1} K_{l_1 l_2} \rho K_{l_1 l_2}^\dagger.
\end{equation}

\begin{equation}
K_{l_1, l_2} = \sqrt{\frac{(1-\eta)^{l_1 + l_2}}{l_1! l_2!}} \eta^{\hat{n}_1/2} \hat{a}_1^{l_1} \otimes \eta^{\hat{n}_2/2} \hat{a}_2^{l_2}.
\end{equation}

Here, $\eta$ is the transmissivity, $\hat{a}_i$ is the single-mode annihilation operator acting on mode $i$, and $\hat{n}_i=\hat{a}_i^\dagger\hat{a}_i$ is the corresponding number operator.

First, we want to calculate the probability that any of the errors $K_{l_1, l_2}$ with $l_1 + l_2 = k$, occur, given that the state is in the $H_n$ subspace.  We know that the probability of a specific error $K_{l_1, l_2}$ is given by $p_{l_1, l_2} = \tr{K_{l_1, l_2}^\dagger K_{l_1, l_2} \rho}$. We see that 
\begin{equation}
K_{l_1, l_2}^\dagger K_{l_1, l_2}  = \frac{(1-\eta)^{l_1 + l_2}}{l_1! l_2!} (\hat{a}_1^\dagger)^{l_1} \eta^{\hat{n}_1} \hat{a}_1^{l_1} \otimes (\hat{a}_2^\dagger)^{l_2}  \eta^{\hat{n}_2} \hat{a}_2^{l_2}.
\end{equation}

It's clear that this is diagonal in the Fock basis. Indeed, we have (given $l_1 \leq m, l_2 \leq t$; otherwise this is zero)

\begin{equation}
\begin{aligned}
\bra{m} \otimes \bra{t} \left(K_{l_1, l_2}^\dagger K_{l_1, l_2} \right) \ket{m} \otimes \ket{t} &= \frac{(1-\eta)^{l_1 + l_2}}{l_1! l_2!} \cdot  \bra{m}(\hat{a}_1^\dagger)^{l_1} \eta^{\hat{n}_1} \hat{a}_1^{l_1}\ket{m} \cdot  \bra{t}(\hat{a}_2^\dagger)^{l_2}  \eta^{\hat{n}_2} \hat{a}_2^{l_2}  \ket{t} \\
&= \frac{(1-\eta)^{l_1 + l_2}}{l_1! l_2!} \cdot  \frac{m! \eta^{m-l_1} }{(m-l_1)!}  \cdot \frac{t! \eta^{t-l_2}}{(t-l_2)!} ={m \choose l_1}{t \choose l_2} \eta^{m+t-l_1 -l_2} (1-\eta)^{l_1 + l_2}.
\end{aligned}
\end{equation}

Thus, we write

\begin{equation}
P_n K_{l_1, l_2}^\dagger K_{l_1, l_2} P_n = \sum_{m = 0}^n {m \choose l_1}{n-m \choose l_2} \eta^{n-l_1 -l_2} (1-\eta)^{l_1 + l_2} \ket{m,n-m}\bra{m,n-m}.
\end{equation}

Since $\rho=P_n\rho P_n$, we may insert these projectors in the trace defining $p_{l_1,l_2}$. Thus,

\begin{equation}
\begin{aligned}
p_{l_1, l_2} &= \sum_{m = 0}^n \tr{{m \choose l_1}{n-m \choose l_2} \eta^{n-l_1 -l_2} (1-\eta)^{l_1 + l_2} \ket{m,n-m}\bra{m,n-m} \rho } \\
&=\sum_{m = 0}^n \rho_{(m,n-m), (m,n-m)} {m \choose l_1}{n-m \choose l_2} \eta^{n-l_1 -l_2} (1-\eta)^{l_1 + l_2}.
\end{aligned}
\end{equation}

Summing over all the combinations, we calculate

\begin{equation}
\begin{aligned}
\Pr[l_1 + l_2= k] &= \sum_{l_1=0}^k p_{l_1, k-l_1} \\
&= \sum_{l_1=0}^k \sum_{m = 0}^n \rho_{(m,n-m), (m,n-m)} {m \choose l_1}{n-m \choose k-l_1} \eta^{n-l_1 -(k-l_1)} (1-\eta)^{l_1 +k - l_1} \\
&= \sum_{l_1=0}^k \sum_{m = 0}^n \rho_{(m,n-m), (m,n-m)} {m \choose l_1}{n-m \choose k-l_1} \eta^{n-k} (1-\eta)^{k } \\
&= \sum_{m = 0}^n \rho_{(m,n-m), (m,n-m)}{n \choose k} \eta^{n-k} (1-\eta)^{k }.
\end{aligned}
\end{equation}

using the Vandermonde identity. Since $\sum_{m=0}^n \rho_{(m,n-m), (m,n-m)}=\Tr\rho=1$, this becomes

\begin{equation}
    \Pr[k \text{ total photons lost}| n \text{ subspace}] = \sum_{m=0}^n \rho_{(m, n-m), (m, n-m)} {n \choose k} \eta^{n-k} (1-\eta)^k ={ n \choose k} \eta^{n-k} (1-\eta)^k,  \label{k-loss-probability}
\end{equation}

which is $\operatorname{Bin}(n, 1-\eta)$. This is a key fact that will later be used in the sector-decoder implementation. 

\subsection{Entanglement fidelity and bi-convex optimization}

Suppose we want to encode a $d_L$-dimensional codespace $\mathbb{C}^{d_L}$, and start with a maximally entangled state

\begin{equation}
\ket{\Phi^+ } = \frac{1}{\sqrt{d_L}} \sum_{i  = 0}^{d_L-1} \ket{i} \otimes \ket{i}.
\end{equation}

Furthermore, assume we have access to an encoding channel $\mathcal{E}:\mathcal{L}(\mathbb{C}^{d_L})\to\mathcal{L}(H_n)$, a noise channel $\mathcal{N}:\mathcal{L}(H_n)\to\mathcal{L}(H_{\leq n})$, and a recovery channel $\mathcal{D}:\mathcal{L}(H_{\leq n})\to\mathcal{L}(\mathbb{C}^{d_L})$, all of which are completely positive and trace-preserving (CPTP), where $H_{\leq n}=\bigoplus_{q=0}^n H_q$. The entanglement fidelity of the composed channel in this case is given by

\begin{equation}
F_e = \braket{\Phi^+|\left[(\mathcal{D} \circ \mathcal{N} \circ \mathcal{E})\otimes \mathcal{I}\right](\ket{\Phi^+}\bra{\Phi^+}) |\Phi^+},
\end{equation}

where $\mathcal{I}$ is the identity channel.

In our case, the only channel that is fixed is the noise channel (pure loss); we are therefore optimizing over both the encoding and recovery channels $\mathcal{E}, \mathcal{D}$, respectively. We show how these can be solved using biconvex optimization, following the derivations in \cite{nohQuantumCapacityBounds2019}. First, for any quantum channel $\mathcal{C}$, we define a Choi matrix $\hat{X}_{\mathcal{C}}$ with matrix elements

\begin{equation}
(\hat{X}_\mathcal{C})_{[ij], [i'j']} = \braket{j|\mathcal{C}(\ket{i}\bra{i'})|j'},
\end{equation}

for all possible input basis states $i,i'$ and output basis states $j,j'$. One can further define a linear map over Choi matrices associated with our noise channel $\mathcal{N}$ as

\begin{equation}
(f_{\mathcal{N}}(\hat{X}))_{[l'i'], [li]} = \sum_{k, k' = 0}^{n}  (\hat{X}_{\mathcal{N}})_{[kl], [k'l']} (\hat{X})_{[ik], [i'k']}.
\end{equation}

The entanglement fidelity of the full channel can be written as \cite{nohQuantumCapacityBounds2019}

\begin{equation}
F_e = \frac{1}{d_L^2} \Tr[\hat{X}_\mathcal{D} f_{\mathcal{N}}(\hat{X}_\mathcal{E})].
\end{equation}

For the maps $\mathcal{D}, \mathcal{E}$ to be CPTP, we enforce the Choi-matrix constraints. Thus, the optimization problem becomes

\begin{equation}
\begin{aligned}
\underset{\hat{X}_\mathcal{E},\hat{X}_{\mathcal{D}}}{\operatorname{maximize}}\quad
&\Tr[\hat{X}_\mathcal{D} f_{\mathcal{N}}(\hat{X}_\mathcal{E})] \\
\text{subject to}\quad
&\hat{X}_\mathcal{D} = \hat{X}_\mathcal{D}^\dagger \geq 0,
\qquad \Tr_{\mathbb{C}^{d_L}} \hat{X}_{\mathcal{D}} = I_{H_{\leq n}}, \\
&\hat{X}_\mathcal{E} = \hat{X}_\mathcal{E}^\dagger \geq 0,
\qquad \Tr_{H_n} \hat{X}_{\mathcal{E}} = I_{\mathbb{C}^{d_L}}.
\end{aligned}
\end{equation}

If we fix one of the two unknown Choi matrices (WLOG $\hat{X}_{\mathcal{D}}$), the problem of optimizing for $\hat{X}_{\mathcal{E}}$ is a semidefinite program (SDP). Since we're trying to optimize both, we alternate between optimizing for $\hat{X}_{\mathcal{E}}$ and $\hat{X}_{\mathcal{D}}$ until we converge to a solution. Note that this approach to biconvex optimization doesn't guarantee convergence to a global maximum. Finally, note that we can also write the fidelity as

\begin{equation}
F_e = \frac{1}{d_L^2} \Tr [T_\mathcal{D} T_\mathcal{N} T_\mathcal{E}], \label{entanglement-fidelity-superoperator}
\end{equation}

where $T_\mathcal{C}$ is the superoperator of the channel defined as $(T_\mathcal{C})_{jj', ii'}= \braket{j|\mathcal{C}(\ket{i}\bra{i'})|j'}$. We note that the main difference with \cite{nohQuantumCapacityBounds2019} is that instead of enforcing an average photon number constraint on a single mode, we directly encode in $H_n$. 

\subsection{Sector-by-sector recovery}

First, we'll use the shorthand $S = \bigoplus_{q = 0}^n H_q$. Consider the canonical inclusion $V_q: H_q \to S$, and define the channel $\Delta(\cdot) = \sum_{q=0}^n P_q (\cdot) P_q$, where the projector onto $H_q$ is given by $P_q = V_q V_q^\dagger$. Note that we have

\begin{equation}
P_q K_{l_1 l_2} P_n = \delta_{q, n-l_1-l_2} K_{l_1 l_2} P_n.
\end{equation}

Therefore, for all inputs $\rho \in \mathcal{L}(H_n)$, we have

\begin{equation}
\begin{aligned}
(\Delta \circ \mathcal{N})(\rho)
&= \sum_{q=0}^n P_q \left( \sum_{l_1 = 0}^n \sum_{l_2 = 0}^{n-l_1} K_{l_1 l_2} \rho K_{l_1 l_2}^\dagger \right) P_q = \sum_{q=0}^n \sum_{l_1 = 0}^n \sum_{l_2 = 0}^{n-l_1} P_q K_{l_1 l_2} P_n \rho P_n K_{l_1 l_2}^\dagger P_q \\
&= \sum_{l_1 = 0}^n \sum_{l_2 = 0}^{n-l_1} \sum_{q=0}^n \delta_{q, n-l_1-l_2} K_{l_1 l_2} \rho K_{l_1 l_2}^\dagger = \sum_{l_1 = 0}^n \sum_{l_2 = 0}^{n-l_1} K_{l_1 l_2} \rho K_{l_1 l_2}^\dagger
= \mathcal{N}(\rho).
\end{aligned}
\end{equation}

Therefore, for any decoder $\mathcal{D}$, we can write

\begin{equation}
\mathcal{D} \circ \mathcal{N} \circ \mathcal{E}
= \mathcal{D} \circ \Delta \circ \mathcal{N} \circ \mathcal{E}
= \widetilde{\mathcal{D}} \circ \mathcal{N} \circ \mathcal{E},
\end{equation}

where $\widetilde{\mathcal{D}} = \mathcal{D} \circ \Delta$. Therefore, every decoder $\mathcal{D}$ has a corresponding decoder that first gets rid of coherences between different $H_q$. Indeed, for any $\rho \in \mathcal{L}(S)$, we can write

\begin{equation}
\widetilde{\mathcal{D}}(\rho)
= \mathcal{D} \circ \Delta(\rho)
= \mathcal{D}\left(\sum_{q=0}^n P_q \rho P_q\right)
= \sum_{q=0}^n \mathcal{D}\left(V_q(V_q^\dagger \rho V_q)V_q^\dagger\right).
\end{equation}

So, we can define $\mathcal{D}_q: \mathcal{L}(H_q) \to \mathcal{L}(\mathbb{C}^{d_L})$, which is given by

\begin{equation}
\mathcal{D}_q(X)
= \mathcal{D}(V_q X V_q^\dagger)
= \mathcal{D} \circ \mathcal{V}_q(X)
\qquad \forall X \in \mathcal{L}(H_q),
\end{equation}

where $\mathcal{V}_q: \mathcal{L}(H_q) \to \mathcal{L}(S)$, $\mathcal{V}_q(\cdot) = V_q (\cdot) V_q^\dagger$. Since $V_q^\dagger V_q = I_{H_q}$, this can be seen as a composition of CPTP maps, and is thus also CPTP. Therefore, we can effectively write any decoder as

\begin{equation}
\widetilde{\mathcal{D}}(\rho)
= \sum_{q=0}^n \mathcal{D}_q \circ \mathcal{V}_q^\dagger(\rho),
\end{equation}

for all $\rho \in \mathcal{L}(S)$, where $\mathcal{V}_q^\dagger: \mathcal{L}(S) \to \mathcal{L}(H_q)$, $\mathcal{V}_q^\dagger(\cdot) = V_q^\dagger (\cdot) V_q$. Note that $\mathcal{V}_q^\dagger$ is completely positive and trace non-increasing.

Likewise, suppose that we're given arbitrary CPTP maps $\mathcal{D}_q: \mathcal{L}(H_q) \to \mathcal{L}(\mathbb{C}^{d_L})$. Then we'll show that

\begin{equation}
\mathcal{D}(\rho)
= \sum_{q=0}^n \mathcal{D}_q \circ \mathcal{V}_q^\dagger(\rho)
\end{equation}

is a CPTP map $\mathcal{D}: \mathcal{L}(S) \to \mathcal{L}(\mathbb{C}^{d_L})$. We have that

\begin{equation}
\begin{aligned}
\tr{\mathcal{D}(\rho)}
&= \sum_{q=0}^n \tr{\mathcal{D}_q(V_q^\dagger \rho V_q)}
= \sum_{q=0}^n \tr{V_q^\dagger \rho V_q}
= \sum_{q=0}^n \tr{\rho V_q V_q^\dagger} \\
&= \tr{\rho \cdot \sum_{q=0}^n V_q V_q^\dagger}
= \tr{\rho \cdot I_S}
= \tr{\rho},
\end{aligned}
\end{equation}

so this is indeed TP. Furthermore, it follows that $\mathcal{D}$ is completely positive as the sum of completely positive maps, since each $\mathcal{D}_q \circ \mathcal{V}_q^\dagger$ is completely positive as the composition of completely positive maps.

We've shown that any global decoder can be mapped to local decoders and vice versa. The only thing that remains to be shown is that the global optimum coincides with the optimum obtained by optimizing each $\mathcal{D}_q$. By definition, we have that

\begin{equation}
F_e
= \braket{\Phi^+|\left[(\mathcal{D} \circ \mathcal{N} \circ \mathcal{E}) \otimes \mathcal{I}\right]\left(\ket{\Phi^+}\bra{\Phi^+}\right)|\Phi^+}.
\end{equation}

Indeed, we start with $\ket{\Phi^+} = \frac{1}{\sqrt{d_L}}\sum_{i=0}^{d_L-1}\ket{i} \otimes \ket{i}$, a sector-$q$ decoder $\mathcal{D}(\rho) = \sum_{q=0}^n \mathcal{D}_q \circ \mathcal{V}_q^\dagger(\rho)$, encoder $\mathcal{E}$, and the noise channel, which we can write as

\begin{equation}
\mathcal{N}(\cdot) = \sum_{q=0}^n \mathcal{N}_q(\cdot),
\qquad
\mathcal{N}_q(\cdot) = \sum_{l=0}^{n-q} K_{l,n-l-q}(\cdot)K_{l,n-l-q}^\dagger,
\qquad
\mathcal{N}_q: \mathcal{L}(H_n) \to \mathcal{L}(S).
\end{equation}

These maps are also CP and trace non-increasing, and, for every $X \in \mathcal{L}(H_n)$, we have

\begin{equation}
(\mathcal{V}_{q'}^\dagger \circ \mathcal{N}_q)(X)
= \delta_{q',q}(\mathcal{V}_q^\dagger \circ \mathcal{N}_q)(X).
\end{equation}

So, we can thus write the whole channel as

\begin{equation}
  \begin{aligned}
    &\left[(\mathcal{D} \circ \mathcal{N} \circ \mathcal{E}) \otimes \mathcal{I}\right]
\left(\ket{\Phi^+}\bra{\Phi^+}\right) \\
&\qquad= \frac{1}{d_L} \sum_{i,j} \sum_{q=0}^n
\mathcal{D}_q \circ \mathcal{V}_q^\dagger \circ \mathcal{N}_q \circ \mathcal{E}(\ket{i}\bra{j}) \otimes \ket{i}\bra{j},
  \end{aligned}
\end{equation}

and the entanglement fidelity becomes

\begin{equation}
\begin{aligned}
F_e
&= \frac{1}{d_L^2} \sum_{i,j} \braket{i|\sum_{q=0}^n \mathcal{D}_q \circ \mathcal{V}_q^\dagger \circ \mathcal{N}_q \circ \mathcal{E}(\ket{i}\bra{j})|j} = \frac{1}{d_L^2} \sum_{i,j} \sum_{q=0}^n
\braket{i|\mathcal{D}_q \circ \mathcal{V}_q^\dagger \circ \mathcal{N}_q \circ \mathcal{E}(\ket{i}\bra{j})|j} \\
&= \frac{1}{d_L^2} \sum_{q=0}^n \sum_{i,j}
\braket{i|\mathcal{D}_q \circ \mathcal{V}_q^\dagger \circ \mathcal{N}_q \circ \mathcal{E}(\ket{i}\bra{j})|j} = \frac{1}{d_L^2} \sum_{q=0}^n \sum_{i,j}
\left(T_{\mathcal{D}_q}T_{\mathcal{V}_q^\dagger}T_{\mathcal{N}_q}T_{\mathcal{E}}\right)_{ij,ij} \\
&= \frac{1}{d_L^2} \sum_{q=0}^n
\Tr\left[T_{\mathcal{D}_q}T_{\mathcal{V}_q^\dagger}T_{\mathcal{N}_q}T_{\mathcal{E}}\right]
= \sum_{q=0}^n F_{q,e}
 \hspace{2em}\text{with}\qquad
F_{q,e} = \frac{1}{d_L^2}\Tr\left[T_{\mathcal{D}_q}T_{\mathcal{V}_q^\dagger}T_{\mathcal{N}_q}T_{\mathcal{E}}\right].
\end{aligned}
\label{Fe-sector-decomposition}
\end{equation}

Thus,

\begin{equation}
\begin{aligned}
\max_{(\mathcal{D}_q)} F_e
&= \frac{1}{d_L^2}\max_{(\mathcal{D}_q)} \sum_{q=0}^n
\Tr\left[T_{\mathcal{D}_q}T_{\mathcal{V}_q^\dagger}T_{\mathcal{N}_q}T_{\mathcal{E}}\right] = \frac{1}{d_L^2}\sum_{q=0}^n \max_{\mathcal{D}_q}
\Tr\left[T_{\mathcal{D}_q}T_{\mathcal{V}_q^\dagger}T_{\mathcal{N}_q}T_{\mathcal{E}}\right].
\end{aligned}
\label{sector-optimization-equivalence}
\end{equation}

So, we can solve each sector optimization problem separately, i.e., solve for each $q$

\begin{equation}
\mathcal{D}_q
= \underset{\mathcal{D}}{\operatorname{arg\,max}}\,
\Tr\left[T_{\mathcal{D}}T_{\mathcal{V}_q^\dagger}T_{\mathcal{N}_q}T_{\mathcal{E}}\right].
\end{equation}

We can put each sector optimization into the form of an optimization over a CPTP noise channel. Let

\begin{equation}
p_q = {n \choose q}\eta^q(1-\eta)^{n-q}
\end{equation}

be the probability that $q$ photons remain after the loss channel acts on any state supported on $H_n$, and, whenever $p_q>0$, define

\begin{equation}
\overline{\mathcal{N}}_q = \frac{1}{p_q}\mathcal{V}_q^\dagger \circ \mathcal{N}_q,
\qquad
\overline{\mathcal{N}}_q: \mathcal{L}(H_n) \to \mathcal{L}(H_q).
\end{equation}

This map is CP because $\mathcal{V}_q^\dagger \circ \mathcal{N}_q$ is CP and $p_q>0$. We now show explicitly that it is TP. Set $r=n-q$, the number of photons lost, and define operators from $H_n$ to $H_q$ by

\begin{equation}
A_{l,q}
= \frac{1}{\sqrt{l!(r-l)!}}V_q^\dagger \hat{a}_1^l \hat{a}_2^{r-l}P_n,
\qquad 0\leq l\leq r=n-q.
\end{equation}

Restricted to $H_n$, the loss Kraus operators satisfy

\begin{equation}
V_q^\dagger K_{l,r-l}P_n
= \sqrt{\eta^q(1-\eta)^r}\,A_{l,q}.
\end{equation}

Indeed, after $r$ annihilations the output has total photon number $q$, so the number-operator part of the Kraus operator contributes $\eta^{q/2}$, while the loss part contributes $(1-\eta)^{r/2}$. It follows that

\begin{equation}
(\mathcal{V}_q^\dagger \circ \mathcal{N}_q)(X)
= \eta^q(1-\eta)^r\sum_{l=0}^r A_{l,q}X A_{l,q}^\dagger.
\end{equation}

We next compute the Kraus completeness relation. The operator $\sum_{l=0}^r A_{l,q}^\dagger A_{l,q}$ is diagonal in the Fock basis of $H_n$. We have that

\begin{equation}
\begin{aligned}
A_{l,q}^\dagger A_{l,q}
&= \frac{1}{l!(r-l)!}P_n(\hat{a}_1^\dagger)^l(\hat{a}_2^\dagger)^{r-l}P_q
\hat{a}_1^l\hat{a}_2^{r-l}P_n \\
&= \frac{1}{l!(r-l)!}P_n(\hat{a}_1^\dagger)^l\hat{a}_1^l
(\hat{a}_2^\dagger)^{r-l}\hat{a}_2^{r-l}P_n.
\end{aligned}
\end{equation}

So we see

\begin{equation}
\begin{aligned}
\braket{m,n-m|\sum_{l=0}^r A_{l,q}^\dagger A_{l,q}|m,n-m}
&= \sum_{l=0}^r \frac{m!(n-m)!}{l!(r-l)!(m-l)!(n-m-r+l)!} \\
&= \sum_{l=0}^r {m \choose l}{n-m \choose r-l} = {n \choose r} = {n \choose q},
\end{aligned}
\end{equation}

which follows from the Chu--Vandermonde identity. Therefore,

\begin{equation}
\sum_{l=0}^r A_{l,q}^\dagger A_{l,q} = {n \choose q}I_{H_n}.
\end{equation}

Combining this identity with $p_q={n \choose q}\eta^q(1-\eta)^r$, we obtain the explicit expression

\begin{equation}
\overline{\mathcal{N}}_q(X)
= \frac{1}{{n \choose q}}\sum_{l=0}^r A_{l,q}X A_{l,q}^\dagger.
\end{equation}

Consequently, for every $X\in\mathcal{L}(H_n)$,

\begin{equation}
\begin{aligned}
\tr{\overline{\mathcal{N}}_q(X)}
&= \frac{1}{{n \choose q}}\tr{X\sum_{l=0}^r A_{l,q}^\dagger A_{l,q}} \\
&= \tr{X},
\end{aligned}
\end{equation}

so $\overline{\mathcal{N}}_q$ is TP and hence CPTP. This explicit expression also shows that $\overline{\mathcal{N}}_q$ is independent of $\eta$. Finally, the original sector objective satisfies

\begin{equation}
\Tr\left[T_{\mathcal{D}}T_{\mathcal{V}_q^\dagger}T_{\mathcal{N}_q}T_{\mathcal{E}}\right]
= p_q\Tr\left[T_{\mathcal{D}}T_{\overline{\mathcal{N}}_q}T_{\mathcal{E}}\right].
\end{equation}

For a fixed encoder and a sector with $p_q>0$, multiplying the objective by $p_q$ does not change which $\mathcal{D}_q$ maximizes it. Since the normalized sector channel is independent of $\eta$, the optimal decoder in that sector is also independent of $\eta$. The probabilities $p_q$ must still be retained when computing the total entanglement fidelity or optimizing the encoder across sectors. If $p_q=0$, the sector never occurs and its decoder has no effect on the physical channel.

\subsection{Optimization details}

 We alternate between encoder and (sector-by-sector) decoder SDPs, solving each one with SCS. After each accepted encoder–recovery update, we use the computed entanglement fidelity $F_e$ to tighten a nonincreasing duality gap target toward $0.01(1-F_e)$, bounded between $10^{-7}$ and $10^{-5}$. The primal and dual residual targets are ten times the duality gap target; the Choi positivity and trace-preservation target is tightened on the same infidelity scale, bounded between $10^{-7}$ and $10^{-6}$. We note that for the sector decoder, the relevant duality gap and dual residual are weighted by the probability of the relevant sector, while the gap is also divided by the square of the logical dimension. Measured solver gaps, residuals, and Choi errors determine whether a solve is accepted or retried at higher accuracy. After the prescribed alternating rounds, a final polishing stage sets nominal gap, residual, and Choi targets of $10^{-8}$, $10^{-7}$, and $10^{-8}$, respectively.

We perform 21000 rounds of bi-convex optimization for all $n \in \mathbb{N}$ with $4 \leq n \leq 25$. All but $5$ runs complete the $21,000$ rounds; $5$ time out.  For the high-eta runs, we perform the optimization for $\eta = 0.9, \, \eta = 0.95$ for all such $n$, and use the same five seeds $0-4$ for the initial random encoder. In the case of $n =15$, we expanded the search with extra seeds $5-9$. For the low-eta sweep, we additionally optimize for $6 \leq n \leq 13$ using two seeds $\{0,1\}$ for each eta in $\{0.7, 0.75, 0.8, 0.85\}$.  

\subsection{Covariance and equivalent code orientations}

We'll characterize all Gaussian unitaries that commute with the two-mode loss of equal strength. In Gaussian phase-space notation \cite{nohQuantumCapacityBounds2019}, two-mode symmetric loss is characterized by

\begin{equation}
X_\eta=\sqrt{\eta}I_4,
\qquad
Y_\eta=(1-\eta)\frac{I_4}{2},
\qquad
\boldsymbol{d}_\eta=0.
\end{equation}

If our system state is characterized by $\overline{x},V_x$ (mean and covariance matrices), these transform as

\begin{equation}
\overline{x}\longmapsto X_\eta\overline{x}+\boldsymbol{d}_\eta,
\qquad
V_x\longmapsto X_\eta V_xX_\eta^T+Y_\eta.
\end{equation}

So, in the case of loss we get

\begin{equation}
\overline{x}\longmapsto\sqrt{\eta}\overline{x},
\qquad
V_x\longmapsto\eta V_x+\frac{1-\eta}{2}I_4.
\end{equation}

Now, suppose that we act on the system with some two-mode Gaussian unitary, described by a symplectic matrix $S_U$ and displacement $\boldsymbol{s}_U$. Let $\Omega$ denote the standard two-mode symplectic form. We know this transforms

\begin{equation}
\overline{x}\longmapsto S_U\overline{x}+\boldsymbol{s}_U,
\qquad
V_x\longmapsto S_UV_xS_U^T,
\qquad
S_U\Omega S_U^T=\Omega.
\end{equation}

For loss and the unitary to commute, we need, assuming $\eta\neq1$,

\begin{equation}
\sqrt{\eta}S_U\overline{x}+\boldsymbol{s}_U
=\sqrt{\eta}S_U\overline{x}+\sqrt{\eta}\boldsymbol{s}_U
\quad\Longrightarrow\quad
\boldsymbol{s}_U=0,
\end{equation}

and

\begin{equation}
\eta S_UV_xS_U^T+\frac{1-\eta}{2}S_US_U^T
=\eta S_UV_xS_U^T+\frac{1-\eta}{2}I_4
\quad\Longrightarrow\quad
S_US_U^T=I_4.
\end{equation}

So the two conditions we get are $S_US_U^T=I_4$ and $\boldsymbol{s}_U=0$, i.e., we have no displacement, and $S_U$ is orthosymplectic. In other words, loss commutes with any combination of beam splitters and phase shifts (passive Gaussian unitaries), as desired. When $\eta=1$, the loss channel is the identity, so the covariance needed below is immediate.

In particular, let $U\in SU(2)$ describe a passive two-mode transformation, and let $R_q(U)$ be its representation on $H_q$. Define

\begin{equation}
R_{\leq n}(U)=\bigoplus_{q=0}^n R_q(U),
\qquad
\mathcal{U}_q(X)=R_q(U)XR_q(U)^\dagger,
\qquad
\mathcal{U}_{\leq n}(X)=R_{\leq n}(U)XR_{\leq n}(U)^\dagger.
\end{equation}

Restricting the covariance relation to an input in $H_n$ gives

\begin{equation}
\mathcal{N}\circ\mathcal{U}_n
=\mathcal{U}_{\leq n}\circ\mathcal{N}.
\label{loss-passive-covariance}
\end{equation}

Now, we'll argue that the entanglement fidelity is invariant if we appropriately update the recovery. Recall that $F_e$ is the entanglement fidelity of the composite channel $\mathcal{D}\circ\mathcal{N}\circ\mathcal{E}$ defined above, and write $F_e(\mathcal{N},\mathcal{D},\mathcal{E})$ when displaying its dependence on these channels. For a fixed encoding, write the optimized fidelity as

\begin{equation}
F_e^{\mathrm{opt}}(\mathcal{N},\mathcal{E})
=\max_{\mathcal{D}}F_e(\mathcal{N},\mathcal{D},\mathcal{E}),
\end{equation}

where the maximum is over CPTP recovery maps $\mathcal{D}:\mathcal{L}(H_{\leq n})\to\mathcal{L}(\mathbb{C}^{d_L})$, and is attained by some $\mathcal{D}_0$. Now, suppose we rotate our encoding by $U$, i.e., $\mathcal{E}'=\mathcal{U}_n\circ\mathcal{E}$. We see that from what we just showed, we get

\begin{equation}
\begin{aligned}
\mathcal{D}\circ\mathcal{N}\circ\mathcal{E}'
&=\mathcal{D}\circ\mathcal{N}\circ\mathcal{U}_n\circ\mathcal{E} \\
&=\mathcal{D}\circ\mathcal{U}_{\leq n}\circ\mathcal{N}\circ\mathcal{E}.
\end{aligned}
\end{equation}

Composition with the unitary channel $\mathcal{U}_{\leq n}$ is a bijection on CPTP recovery maps. In particular,

\begin{equation}
\mathcal{D}_0'=\mathcal{D}_0\circ\mathcal{U}_{\leq n}^\dagger,
\qquad
\mathcal{U}_{\leq n}^\dagger(X)
=R_{\leq n}(U)^\dagger X R_{\leq n}(U),
\end{equation}

attains the same fidelity for $\mathcal{E}'$ that $\mathcal{D}_0$ attains for $\mathcal{E}$. Therefore,

\begin{equation}
F_e^{\mathrm{opt}}(\mathcal{N},\mathcal{U}_n\circ\mathcal{E})
=F_e^{\mathrm{opt}}(\mathcal{N},\mathcal{E}).
\end{equation}

\begin{remark}
The optimized entanglement fidelity is also invariant under complex conjugation in the occupation-number basis. Here $*$ denotes entrywise complex conjugation in this fixed basis. The loss Kraus operators are real in this basis, so $T_\mathcal{N}^*=T_\mathcal{N}$. Thus, if $T_\mathcal{E}'=T_\mathcal{E}^*$ and $T_\mathcal{D}'=T_\mathcal{D}^*$, then

\begin{equation}
T_\mathcal{D}'T_\mathcal{N}T_\mathcal{E}'
=\left(T_\mathcal{D}T_\mathcal{N}T_\mathcal{E}\right)^*,
\end{equation}

and Equation~\eqref{entanglement-fidelity-superoperator} gives $F_e'=F_e$, since entanglement fidelity is real. Conjugation is a bijection on CPTP maps, so the optimized fidelity is unchanged. For an isometric encoder $V$, this conjugates $T_\mathcal{E}=V\otimes V^*$ and maps its code projector $P=VV^\dagger$ to $P^*$.
\end{remark}

\section{Representation-theoretic preliminaries}

We proceed by presenting a few fundamental definitions and theorems from representation theory of finite groups, which are necessary for the code construction in Ref.~\cite{grossDesigningCodesInteractions2021}, as well as our proofs, mostly following Refs.~\cite{dummitAbstractAlgebra2004,serreLinearRepresentationsFinite1977,fultonRepresentationTheory2004}.

\subsection{Representations, intertwiners, and operator spaces}

\begin{definition}
If $V,W$ are finite-dimensional complex vector spaces, $\operatorname{Hom}(V,W)$ is the set of complex-linear maps from $V$ to $W$. The dual of $V$ is $V^*=\operatorname{Hom}(V,\mathbb C)$, and $\operatorname{Hom}(V,W)\cong V^*\otimes W$.
\end{definition}

\begin{definition}
A \textit{linear representation} $\pi$ of a group $G$ in a vector space $V$ is a homomorphism from $G$ to $\operatorname{GL}(V)$, the group of invertible linear transformations from $V$ to itself.
\label{rep-def}
\end{definition}

\begin{definition}
Let $\pi_V:G\to\operatorname{GL}(V)$ and $\pi_W:G\to\operatorname{GL}(W)$ be representations of $G$. A linear map $\phi:V\to W$ is a map of representations precisely when it respects the two group actions, meaning that
\begin{equation}
\phi(\pi_V(g)v)=\pi_W(g)\phi(v)
\qquad \forall g\in G,\ v\in V.
\end{equation}
This is also called a $G$-linear map. 
\end{definition}

\begin{remark}
$\operatorname{Hom}(V,W)$ is, in general, also a representation. We define it by the action $(g\phi)(v)=g\phi(g^{-1}v)$, or, in the strict representation formalism, we have $\pi_{\operatorname{Hom}}:G\to\operatorname{GL}(\operatorname{Hom}(V,W))$, and
\begin{equation}
(\pi_{\operatorname{Hom}}(g)\phi)(v)=\pi_W(g)\phi(\pi_V(g^{-1})v).
\label{hom-rep}
\end{equation}
We can also write the action more generally as $g\cdot\phi=\pi_W(g)\phi\pi_V(g)^{-1}$.
\end{remark}

\begin{remark}
Under the action in Eq.~\eqref{hom-rep}, a linear map is fixed by every element of $G$ exactly when it is $G$-linear; these maps form $\operatorname{Hom}_G(V,W)\subset\operatorname{Hom}(V,W)$ \cite{fultonRepresentationTheory2004}. This is because being fixed under the action of $G$ means $(g\phi)(v)=\phi(v)\implies g\phi(g^{-1}v)=\phi(v)$. Plugging $v\to gv$, we get $\phi(gv)=g\phi(v)$, i.e., $\phi$ is a $G$-linear map.
\end{remark}

\begin{remark}
If we have $W=V$, then we're looking at $\operatorname{Hom}(V,V)=\operatorname{End}(V)$. Using Eq.~\eqref{hom-rep}, we see the representation on $\operatorname{End}(V)$ is defined by the action
\begin{equation}
g\cdot A=\pi(g)A\pi(g)^{-1},
\qquad A\in\operatorname{End}(V),\ g\in G.
\end{equation}
Note that $\operatorname{End}(V)$ is also an algebra over the relevant field.
\end{remark}

\begin{definition}
The dual $V^*$ is also a representation, with matrices
\begin{equation}
\pi_{V^*}(g)=\pi_V(g)^{-T}.
\label{dual-representation}
\end{equation}
\end{definition}

\begin{remark}
It's clear that if the representation is unitary, $\pi_{V^*}(g)=(\pi_V(g)^{-1})^T=(\pi_V(g)^\dagger)^T=\pi_V(g)^*$, i.e., the dual representation is given by the conjugate.
\label{dual-unitary-rep}
\end{remark}

\begin{definition}
Suppose $W$ a subspace of $V$, i.e., $W\leq V$ that's stable under a representation $\pi$. We can then restrict the representation $\pi^W:G\to\operatorname{GL}(W)$. This is also a valid representation and is called a \textit{subrepresentation}.
\label{subrep-def}
\end{definition}

\begin{definition}
A representation $\pi:G\to\operatorname{GL}(V)$ is called \textit{irreducible} if its only invariant subspaces are $V,\{0\}$, i.e., it has no non-trivial subrepresentations.
\label{irrep-def}
\end{definition}

\begin{lemma}
Let $G$ a group, and let $V,W$ be finite-dimensional complex Hilbert spaces carrying unitary representations $\pi_V:G\to U(V)$ and $\pi_W:G\to U(W)$, respectively. Suppose that $\mathcal A\subseteq\operatorname{Hom}(V,W)$ is stable under the action of $G$. Then:
\begin{enumerate}
\item The adjoint space $\mathcal A^\dagger\cong\mathcal A^*$.
\item The product space
\begin{equation}
\mathcal K(\mathcal A)=\operatorname{span}\{A^\dagger B:A,B\in\mathcal A\}
\end{equation}
is an equivariant image of $\mathcal A^*\otimes\mathcal A$, and if the products $A_i^\dagger A_j$ are all linearly independent, for a basis $A_i$ of $\mathcal A$, then we have
\begin{equation}
\mathcal K(\mathcal A)\cong\mathcal A^*\otimes\mathcal A.
\end{equation}
\end{enumerate}
\label{adjoint-KL-rep-lemma}
\end{lemma}

\begin{proof}
We prove the two claims.
\begin{enumerate}
\item First, we have that
\begin{equation}
(g\cdot A)^\dagger=(\pi_W(g)A\pi_V(g)^\dagger)^\dagger=\pi_V(g)A^\dagger\pi_W(g)^\dagger=g\cdot A^\dagger
\qquad \forall A\in\mathcal A.
\end{equation}
Now, define an orthonormal basis $\{A_0,A_1,\dots,A_s\}$ of $\mathcal A$ with respect to the Hilbert--Schmidt inner product $\braket{A,B}=\operatorname{Tr}(A^\dagger B)$, i.e., we have $\operatorname{Tr}(A_i^\dagger A_j)=\delta_{ij}$. Note that this is preserved by the $G$ action. Suppose that the representation of $\mathcal A$ is given by matrices $D_{uv}^{(\mathcal A)}$ in this basis, i.e.,
\begin{equation}
g\cdot A_v=\sum_{u=0}^sD_{uv}^{(\mathcal A)}(g)A_u.
\end{equation}
These matrices are going to be unitary since
\begin{equation}
\begin{aligned}
\delta_{ij}
&=\braket{A_i,A_j}=\braket{g\cdot A_i,g\cdot A_j}
=\braket{\sum_{u=0}^sD_{ui}^{(\mathcal A)}(g)A_u,\sum_{v=0}^sD_{vj}^{(\mathcal A)}(g)A_v}\\
&=\sum_{u=0}^s\overline{D}_{ui}^{(\mathcal A)}(g)D_{uj}^{(\mathcal A)}(g)
=\left((D^{(\mathcal A)})^\dagger(g)D^{(\mathcal A)}(g)\right)_{ij},
\end{aligned}
\end{equation}
for every $g$. So, we have
\begin{equation}
g\cdot A_v^\dagger=(g\cdot A_v)^\dagger=\sum_{u=0}^s\overline{D}_{uv}^{(\mathcal A)}(g)A_u^\dagger,
\end{equation}
which gives us that the induced representation matrices on $\mathcal A^\dagger$ are given by $\overline{D}_{uv}^{(\mathcal A)}$, i.e., we have $\mathcal A^\dagger\cong\mathcal A^*$, using Eq.~\eqref{dual-representation} and Remark~\ref{dual-unitary-rep}.

\item Define the map $\Phi:\mathcal A^\dagger\otimes\mathcal A\to\mathcal K(\mathcal A)$ via $\Phi(A^\dagger\otimes B)=A^\dagger B$. Firstly, we have that, since the action on the tensor representation is $g\cdot(a\otimes b)=(g\cdot a)\otimes(g\cdot b)$ \cite{fultonRepresentationTheory2004}, we have
\begin{equation}
\begin{aligned}
\Phi(g\cdot(A^\dagger\otimes B))
&=\Phi((g\cdot A^\dagger)\otimes(g\cdot B))
=\Phi((g\cdot A)^\dagger\otimes(g\cdot B))\\
&=(g\cdot A)^\dagger(g\cdot B)\\
&=\pi_V(g)A^\dagger\pi_W(g)^\dagger\pi_W(g)B\pi_V(g)^\dagger\\
&=\pi_V(g)A^\dagger B\pi_V(g)^\dagger
=g\cdot\Phi(A^\dagger\otimes B),
\end{aligned}
\end{equation}
and therefore we see that $\Phi$ is $G$-equivariant. It is also by definition surjective. Finally, if the products $A_i^\dagger A_j$ are linearly independent, we see that
\begin{equation}
\Phi\left(\sum_{ij}c_{ij}A_i^\dagger\otimes A_j\right)=0
\implies\sum_{ij}c_{ij}A_i^\dagger A_j=0
\implies c_{ij}=0,
\end{equation}
so we get that $\Phi$ is injective. Thus, it is a $G$-equivariant isomorphism, and we conclude $\mathcal K(\mathcal A)\cong\mathcal A^\dagger\otimes\mathcal A\cong\mathcal A^*\otimes\mathcal A$ as desired.
\end{enumerate}
\end{proof}

\subsection{Isotypic decompositions and invariant operators}

\begin{lemma}[Schur's lemma {\cite{fultonRepresentationTheory2004}}]
If $V,W$ are irreducible representations of $G$ and $\phi:V\to W$ is a $G$-module homomorphism ($G$-linear map), then
\begin{enumerate}
\item Either $\phi$ is an isomorphism, or $\phi=0$.
\item If $V=W$, then $\phi=\lambda I_V$ for some $\lambda\in\mathbb C$.
\end{enumerate}
\label{schur}
\end{lemma}

\begin{lemma}
Let $V,W$ representations of $G$ with $V$ irreducible, and $\phi:V\to W$ a $G$-linear map such that $\phi\neq0$. Then $\operatorname{im}\phi\subseteq W$ is a $G$-subrepresentation of $W$ isomorphic to $V$.
\label{schur-extension}
\end{lemma}

\begin{definition}
The \textit{character} $\chi_\pi$ of a representation is defined as $\chi_\pi(g)=\operatorname{Tr}(\pi(g))$. We can further define an inner product of characters as
\begin{equation}
\braket{\chi,\varphi}=\frac{1}{|G|}\sum_{g\in G}\chi(g)\varphi(g)^*.
\label{char-def}
\end{equation}
\end{definition}

\begin{lemma}[{\cite[Sec.~2.1, Proposition~1]{serreLinearRepresentationsFinite1977}}]
If $\chi$ is the character of a representation $\pi$ of degree $r$, we have
\begin{enumerate}
\item $\chi(1)=r$.
\item $\chi(g^{-1})=\chi(g)^*$ for all $g\in G$.
\item $\chi(hgh^{-1})=\chi(g)$ for all $g,h\in G$. This makes $\chi$ a class function.
\end{enumerate}
\end{lemma}

\begin{lemma}[Isotypic decomposition]
Let $G$ be a finite group and let $\pi:G\to\operatorname{GL}(V)$ be a finite-dimensional complex representation with character $\varphi$. Then $V$ can be written as
\begin{equation}
V\cong W_1^{\oplus m_1}\oplus W_2^{\oplus m_2}\oplus\dots\oplus W_s^{\oplus m_s},
\end{equation}
where $W_1,\dots,W_s$ are pairwise nonisomorphic irreducible representations of $G$ with characters $\chi_1,\dots,\chi_s$, and $m_i=\braket{\varphi,\chi_i}$.
\label{irrep-decomp}
\end{lemma}

\begin{lemma}[Isotypic projector formula]
Let $G$ be a finite group, let $\pi:G\to\operatorname{GL}(V)$ be a finite-dimensional complex representation, and let $W_i$ be an irreducible representation of $G$ occurring in $V$, with character $\chi_i$ and dimension $d_i$.
The projection $\Pi_i$ of $V$ onto the $W_i$-isotypic component is given by
\begin{equation}
\Pi_i=\frac{d_i}{|G|}\sum_{g\in G}\chi_i(g)^*\pi(g),
\label{isotypic-projector-formula}
\end{equation}
\end{lemma}

\begin{lemma}
The projector $\Pi_i$ commutes with the representation $\pi$, i.e.,
\begin{equation}
\Pi_i\pi(g)=\pi(g)\Pi_i.
\label{projector-commutation}
\end{equation}
\end{lemma}

\begin{proof}
Indeed, we have
\begin{equation}
\begin{aligned}
\Pi_i\pi(g)
&=\frac{d_i}{|G|}\sum_{t\in G}\chi_i(t)^*\pi(t)\pi(g)
=\frac{d_i}{|G|}\sum_{t\in G}\chi_i(t)^*\pi(tg)\\
&=\frac{d_i}{|G|}\sum_{t\in G}\chi_i(gtg^{-1})^*\pi(gt)
=\frac{d_i}{|G|}\sum_{t\in G}\chi_i(t)^*\pi(gt)
=\pi(g)\Pi_i,
\end{aligned}
\end{equation}
using the fact that $\chi_i$ is a class function, as desired.
\end{proof}

We prove two more lemmas that will be useful later.

\begin{lemma}
Let $C$ a representation of $G$. Then we can decompose
\begin{equation}
\operatorname{End}(C)\cong\mathbb CI_C\oplus\operatorname{End}_0(C),
\qquad
\operatorname{End}_0(C)=\{A\in\operatorname{End}(C):\operatorname{Tr}(A)=0\},
\end{equation}
and $\mathbb CI_C=\{\lambda I_C:\lambda\in\mathbb C\}$ carries the trivial representation.
\label{end-decomp}
\end{lemma}

\begin{proof}
Indeed, for any operator $X\in\operatorname{End}(C)$, we can write
\begin{equation}
X=\frac{\operatorname{Tr}(X)}{\dim C}I_C+\left(X-\frac{\operatorname{Tr}(X)}{\dim C}I_C\right).
\end{equation}
This is by definition unique, and it's clear that $\frac{\operatorname{Tr}(X)}{\dim C}I_C\in\mathbb CI_C$ (the trivial representation), and $X-\frac{\operatorname{Tr}(X)}{\dim C}I_C\in\operatorname{End}_0(C)$. Furthermore, it's clear that since
\begin{equation}
\operatorname{Tr}(\pi(g)A\pi(g)^{-1})=\operatorname{Tr}(\pi(g)^{-1}\pi(g)A)=\operatorname{Tr}(A),
\end{equation}
both of these spaces are $G$-stable, i.e., they're both subrepresentations. So we can indeed write
\begin{equation}
\operatorname{End}(C)\cong\mathbb CI_C\oplus\operatorname{End}_0(C),
\end{equation}
as desired.
\end{proof}

\begin{lemma}
Let $\operatorname{End}_0(C)^G$ denote the $G$-invariant subspace of $\operatorname{End}_0(C)$ from Lemma~\ref{end-decomp}. If $C$ is an irrep of $G$, then
\begin{equation}
\operatorname{End}_0(C)^G=\{0\}.
\end{equation}
\label{unique-trivial-rep-end}
\end{lemma}

\begin{proof}
If $A\in\operatorname{End}_0(C)^G$ we have $g\cdot A=A$ and $\operatorname{Tr}(A)=0$. But we have that
\begin{equation}
g\cdot A=\pi(g)A\pi(g)^{-1}=A
\implies\pi(g)A=A\pi(g),
\end{equation}
i.e., $A$ commutes with the representation. But that means that $A\in\operatorname{Hom}_G(C,C)$. So from Schur's lemma we have that $A=\lambda I_C$. And since $\operatorname{Tr}(A)=0$ we get $\lambda=0\implies\operatorname{End}_0(C)^G=\{0\}$.
\end{proof}

\begin{remark}
This means that if $C$ is an irrep, then the trivial representation doesn't appear in $\operatorname{End}_0(C)$. So the decomposition $\operatorname{End}(C)\cong\mathbb CI_C\oplus\operatorname{End}_0(C)$ has only one instance of the trivial representation.
\end{remark}

\begin{remark}
Also note that $\operatorname{End}_G(C)=\operatorname{End}(C)^G$. The former refers to linear transformations $\phi:C\to C$ such that $\phi(\pi(g)v)=\pi(g)\phi(v)$ for all $v$, whereas the latter refers to linear operators $\psi:C\to C$ such that $g\cdot\psi=\psi$. This is also true more generally for $\operatorname{Hom}_G(A,B)=\operatorname{Hom}(A,B)^G$.
\end{remark}

\begin{lemma}\label{hom-tensor-prod}
Let $A,B$ representations of a finite group $G$, and $M,N$ finite-dimensional vector spaces carrying the trivial $G$ action. Then, we have that
\begin{equation}
\operatorname{Hom}_G(M\otimes A,N\otimes B)
\cong\operatorname{Hom}(M,N)\otimes\operatorname{Hom}_G(A,B).
\end{equation}
\end{lemma}

\begin{proof}
Firstly, it's easy to see that
\begin{equation}
\begin{aligned}
\operatorname{Hom}(M\otimes A,N\otimes B)
&\cong(M\otimes A)^*\otimes(N\otimes B)\\
&\cong(M^*\otimes N)\otimes(A^*\otimes B)\\
&\cong\operatorname{Hom}(M,N)\otimes\operatorname{Hom}(A,B).
\end{aligned}
\end{equation}
These natural identifications respect the $G$ action. And since $M,N$ carry the trivial $G$ action, we see that $\operatorname{Hom}_G(M\otimes A,N\otimes B)\cong\operatorname{Hom}(M,N)\otimes\operatorname{Hom}_G(A,B)$, as desired.
\end{proof}

\begin{lemma}\label{EndGU-decomposition}
Let $V$ a finite-dimensional complex representation of a finite group $G$ which decomposes as $V\cong\bigoplus_U M_U\otimes U$ in terms of irreps $U$ of $G$ and multiplicity spaces $M_U$. Then, we have that
\begin{equation}
\operatorname{End}_G(V)\cong\bigoplus_U\operatorname{End}(M_U)\otimes I_U.
\end{equation}
\end{lemma}

\begin{proof}
It is clear that
\begin{equation}
\begin{aligned}
\operatorname{Hom}_G(V,V)
&=\operatorname{Hom}_G\left(\bigoplus_U M_U\otimes U,\bigoplus_{U'}M_{U'}\otimes U'\right)\\
&\cong\bigoplus_{U,U'}\operatorname{Hom}_G(M_U\otimes U,M_{U'}\otimes U')
\cong\bigoplus_U\operatorname{End}(M_U)\otimes I_U,
\end{aligned}
\end{equation}
where we've used Lemma~\ref{hom-tensor-prod}, and the last isomorphism follows from Schur's lemma.
\end{proof}

We also use Burnside's theorem, see Ref.~\cite{goodmanSymmetryRepresentationsInvariants2009}.

\begin{lemma}[Burnside's theorem]
The only irreducible algebra of linear transformations on the finite-dimensional vector space $V$ over $\mathbb C$ of dimension greater than $1$ is the algebra of all linear transformations mapping $V$ into $V$, i.e., $\operatorname{End}(V)$.
\label{burnside}
\end{lemma}

\begin{remark}
The notion of reducibility for algebras is the same as for representations, i.e., a collection of linear transformations is reducible if it has a non-trivial invariant subspace \cite{goodmanSymmetryRepresentationsInvariants2009}.
\end{remark}

\subsection{The Schwinger representation and \texorpdfstring{$SU(2)$}{SU(2)} tensor operators}

The Schwinger model identifies the fixed total photon number subspace $H_n$ with the spin-$n/2$ representation \cite{sakuraiModernQuantumMechanics2021}. We denote this representation by $R_n:SU(2)\to U(H_n)$. Under this identification, the basis states are mapped from Fock to spin states as
\begin{equation}
\ket{m}\otimes\ket{n-m}\longmapsto\ket{j=n/2,m_j=(2m-n)/2}.
\label{schwinger-basis-map}
\end{equation}
The same identification realizes the angular momentum operators in terms of the single-mode annihilation operators $\hat a_1,\hat a_2$ as
\begin{equation}
J_+=\hat a_1^\dagger\otimes\hat a_2,
\qquad
J_-=\hat a_1\otimes\hat a_2^\dagger,
\qquad
J_z=\frac{1}{2}\left(\hat n_1\otimes I_2-I_1\otimes\hat n_2\right).
\label{schwinger-angular-momentum}
\end{equation}
These operators preserve $H_n$ because they commute with the total number operator
\begin{equation}
\hat N=\hat n_1\otimes I_2+I_1\otimes\hat n_2,
\qquad
\hat N\big|_{H_n}=nI_{H_n}.
\label{total-number-operator}
\end{equation}
The operators in Eq.~\eqref{schwinger-angular-momentum} satisfy the usual angular momentum commutation relations and have the standard action on the spin basis in Eq.~\eqref{schwinger-basis-map}.

The representation $R_n$ describes the action of $SU(2)$ on states in $H_n$. For the Knill--Laflamme conditions, however, we need the induced action on operators in $\operatorname{End}(H_n)$.  The group $SU(2)$ acts on $\operatorname{End}(H_n)$ by conjugation,
\begin{equation}
U\cdot A=R_n(U)AR_n(U)^\dagger,
\qquad
U\in SU(2),\quad A\in\operatorname{End}(H_n).
\label{operator-conjugation-action}
\end{equation}

It remains to determine how this operator representation decomposes into irreducible representations of $SU(2)$. Let's denote by $S_j$ the spin-$j$ irreducible representation of $SU(2)$, with $\dim S_j=2j+1$. We first state the Clebsch--Gordan theorem in its standard representation-theoretic form \cite{hallLieGroupsLie2015}.

\begin{lemma}[Clebsch--Gordan theorem]
Let $r,s\in\mathbb Z$, $r\geq s>0$. Then we have that
\begin{equation}
S_{r/2}\otimes S_{s/2}
\cong S_{(r+s)/2}\oplus S_{(r+s-2)/2}\oplus\dots\oplus S_{(r-s+2)/2}\oplus S_{(r-s)/2}.
\label{Clebsch-gordan-equation}
\end{equation}
\label{Clebsch-gordan}
\end{lemma}

\begin{lemma}
For any irreducible representation $S_j$ of $SU(2)$, $S_j^*\cong S_j$.
\label{su2-self-dual}
\end{lemma}

\begin{proof}
See Ref.~\cite{fultonRepresentationTheory2004}, Sec.~11.3, footnote 2.
\end{proof}

The conjugation representation on $\operatorname{End}(H_n)$ is naturally isomorphic to $H_n\otimes H_n^*$. Since $H_n\cong S_{n/2}$, the Clebsch--Gordan theorem and Lemma~\ref{su2-self-dual} give
\begin{equation}
\operatorname{End}(H_n)
\cong H_n\otimes H_n^*
\cong S_{n/2}\otimes S_{n/2}
\cong\bigoplus_{\ell=0}^{n}S_\ell.
\label{end-spin-decomposition}
\end{equation}
This decomposition identifies the irreducible angular momentum sectors of the operator space. In particular, each spin-$\ell$ irrep occurs exactly once in $\operatorname{End}(H_n)$. We denote this unique invariant subspace by
\begin{equation}
\mathcal T_n^\ell
=\operatorname{span}\left\{T_\mu^{(\ell)}:-\ell\leq\mu\leq\ell\right\}
\cong S_\ell,
\qquad
0\leq\ell\leq n.
\label{tensor-operator-space}
\end{equation}
The basis operators $T_\mu^{(\ell)}$ are the rank-$\ell$ irreducible spherical tensors.

To motivate the defining properties of these operators, we first need to consider the representation of the Lie algebra $\mathfrak{su}(2)$. We denote the usual generators of $\mathfrak{su}(2)$ by $X_a=-i\sigma_a/2$, with the usual Pauli matrices $p_a, \, a\in\{x,y,z\}$, and define the Lie algebra representation \cite{hallLieGroupsLie2015}
\begin{equation}
dR_n:\mathfrak{su}(2)\longrightarrow\mathfrak{u}(H_n),
\qquad
dR_n(X)=\left.\frac{d}{dt}R_n(e^{tX})\right|_{t=0}.
\label{differentiated-spin-representation}
\end{equation}
The Schwinger angular momentum operators in Eq.~\eqref{schwinger-angular-momentum} satisfy $J_a=i\,dR_n(X_a)$.
The action in the operator representation for some $X_a\in\mathfrak{su}(2)$ gives
\begin{equation}
\begin{aligned}
J_a\cdot A
&=i\left.\frac{d}{dt}\right|_{t=0}
\left(R_n(e^{tX_a})AR_n(e^{tX_a})^\dagger\right)\\
&=i[dR_n(X_a),A]\\
&=[J_a,A].
\end{aligned}
\label{operator-lie-algebra-action}
\end{equation}

Thus, the familiar commutator action of angular momentum is precisely the infinitesimal form of the $SU(2)$ conjugation action on operators.  Under the infinitesimal conjugation action in Eq.~\eqref{operator-lie-algebra-action}, the rank-$\ell$ irreducible tensors satisfy the commutation relations (setting $\hbar=1$) \cite{sakuraiModernQuantumMechanics2021}
\begin{equation}
\begin{aligned}
[J_z,T_\mu^{(\ell)}]&=\mu T_\mu^{(\ell)},\\
[J_\pm,T_\mu^{(\ell)}]
&=\sqrt{(\ell\mp\mu)(\ell\pm\mu+1)}T_{\mu\pm1}^{(\ell)}.
\end{aligned}
\label{spherical-tensor-commutators}
\end{equation}
These relations determine the elementary tensors up to a common nonzero scalar. We fix their overall normalization by taking
\begin{equation}
T_\ell^{(\ell)}=J_+^\ell.
\label{elementary-tensor-normalization}
\end{equation}

For the proof of our seed construction theorem, we also need a concrete model of the abstract spin-$\ell$ representation $\mathcal T_n^\ell$. We first recall the symmetric-power construction.

\begin{definition}
Let $V$ be a vector space and let $r$ be a positive integer. The symmetric power $\operatorname{Sym}^r V$ is the quotient of $V^{\otimes r}$ by the subspace generated by
\begin{equation}
v_1\otimes v_2\otimes\dots\otimes v_r
-v_{\sigma(1)}\otimes v_{\sigma(2)}\otimes\dots\otimes v_{\sigma(r)},
\end{equation}
for all $(v_1,\dots,v_r)\in V^{\times r}$ and all permutations $\sigma$ of $\{1,\dots,r\}$. We denote the equivalence class of $v_1\otimes\dots\otimes v_r$ by
\begin{equation}
v_1\odot v_2\odot\dots\odot v_r\in\operatorname{Sym}^r V.
\end{equation}
\end{definition}

\begin{remark}
The map $V^r\to\operatorname{Sym}^r V$ sending $(v_1,\dots,v_r)$ to $v_1\odot\dots\odot v_r$ is symmetric and multilinear. It has the universal property that every symmetric multilinear map $m:V^r\to W$ into a vector space $W$ factors uniquely through a linear map $\widetilde m:\operatorname{Sym}^r V\to W$. If $V$ carries a representation of a group $G$, then $\operatorname{Sym}^r V$ carries the induced representation
\begin{equation}
g\cdot(v_1\odot v_2\odot\dots\odot v_r)
=(g\cdot v_1)\odot(g\cdot v_2)\odot\dots\odot(g\cdot v_r).
\end{equation}
\end{remark}

\begin{lemma}[Symmetric-power realization]
Any irreducible representation of $SU(2)$ is a symmetric power of the standard representation $V\cong\mathbb C^2\cong S_{1/2}$. Specifically, $S_j\cong\operatorname{Sym}^{2j}(V)$, and thus $S_j$ can be represented on the polynomial space
\begin{equation}
\mathcal P_j=\operatorname{span}\{x^{2j},x^{2j-1}y,\dots,y^{2j}\}.
\label{su2-polynomial-space}
\end{equation}
\label{su2-symmetric-irreps}
\end{lemma}

\begin{proof}
See Ref.~\cite{fultonRepresentationTheory2004}, Sec.~11, or Ref. ~\cite{hallLieGroupsLie2015}.
\end{proof}

We now make this symmetric-power realization explicit in our notation. Here $x,y$ are taken to be the basis vectors of $V$, and polynomial monomials are shorthand for their symmetric products; in particular, $xy$ denotes $x\mathbin{\odot}y$. For
\begin{equation}
U=
\begin{pmatrix}
a&-\overline b\\
b&\overline a
\end{pmatrix}
\in SU(2),
\end{equation}
the action on the basis vectors is
\begin{equation}
U\cdot x=ax+by,
\qquad
U\cdot y=-\overline b x+\overline a y.
\label{polynomial-direct-action}
\end{equation}
Setting $n=2j$, the action of the irreducible representation $\Pi_n:SU(2)\to\operatorname{GL}(\mathcal P_{n/2})$ is therefore
\begin{equation}
(\Pi_n(U)f)(x,y)
=f(ax+by,-\overline b x+\overline a y),
\qquad
f\in\mathcal P_{n/2}.
\label{homogeneous-polynomial-representation}
\end{equation}
Differentiating this representation gives the polynomial angular momentum generators
\begin{equation}
J_z^{\operatorname{poly}}
=\frac{1}{2}(x\partial_x-y\partial_y),
\qquad
J_x^{\operatorname{poly}}
=\frac{1}{2}(y\partial_x+x\partial_y),
\qquad
J_y^{\operatorname{poly}}
=\frac{1}{2}(iy\partial_x-ix\partial_y).
\label{J-generators-polynomial}
\end{equation}
We can likewise define
\begin{equation}
J_+^{\operatorname{poly}}
=J_x^{\operatorname{poly}}+iJ_y^{\operatorname{poly}}
=x\partial_y,
\qquad
J_-^{\operatorname{poly}}
=J_x^{\operatorname{poly}}-iJ_y^{\operatorname{poly}}
=y\partial_x.
\label{Jpm-definition-polynomial}
\end{equation}

We have now realized the same spin-$\ell$ representation both as the tensor-operator space $\mathcal T_n^\ell$ and as the polynomial space $\mathcal P_\ell$. To identify these two realizations explicitly, we define the polynomials
\begin{equation}
p_\mu^{(\ell)}
=\sqrt{\binom{2\ell}{\ell+\mu}}x^{\ell+\mu}y^{\ell-\mu},
\qquad
-\ell\leq\mu\leq\ell,
\end{equation}
and we see
\begin{equation}
\begin{aligned}
J_z^{\operatorname{poly}}p_\mu^{(\ell)}
&=\mu p_\mu^{(\ell)},\\
J_+^{\operatorname{poly}}p_\mu^{(\ell)}
&=\sqrt{(\mu+\ell+1)(\ell-\mu)}p_{\mu+1}^{(\ell)},\\
J_-^{\operatorname{poly}}p_\mu^{(\ell)}
&=\sqrt{(\ell+\mu)(\ell-\mu+1)}p_{\mu-1}^{(\ell)}.
\end{aligned}
\label{polynomial-weight-relations}
\end{equation}
The action of $J_a^{\operatorname{poly}}$ on $p_\mu^{(\ell)}$ in the $\mathcal P_\ell$ representation is the same as the action of $J_a$ on $T_\mu^{(\ell)}$ in the operator representation. Since $SU(2)$ is connected, with our chosen normalization this defines the $SU(2)$-equivariant isomorphism $\Phi_n^\ell$ \cite[Proposition~4.5(2)]{hallLieGroupsLie2015} such that
\begin{equation}
\Phi_n^\ell\left(T_\mu^{(\ell)}\right)
=\sqrt{\binom{2\ell}{\ell+\mu}}x^{\ell+\mu}y^{\ell-\mu}.
\label{lie-algebra-to-polynomial-isomorphism}
\end{equation}
Equivariance means explicitly that
\begin{equation}
\Phi_n^\ell\left(R_n(U)AR_n(U)^\dagger\right)
=\Pi_{2\ell}(U)\Phi_n^\ell(A),
\qquad
U\in SU(2),\quad A\in\mathcal T_n^\ell.
\label{tensor-polynomial-equivariance}
\end{equation}
The map $\Phi_n^\ell$ is the bridge used in the seed construction, mapping operators to their representative polynomials in the same irrep of $SU(2)$.  In particular, Eq.~\eqref{elementary-tensor-normalization} gives $\Phi_n^\ell(J_+^\ell)=x^{2\ell}$, which is the highest-weight correspondence used below. This normalization need not be Hilbert--Schmidt normalization. When an orthonormal tensor basis is required below, we rescale the entire rank-$\ell$ multiplet by its common Hilbert--Schmidt norm.

\section{Loss products and irreducible spin tensors}

\subsection{Loss operators and the Knill--Laflamme product space}

Let $R_m(g)$ be the spin-$m/2$ representation on $H_m$. A photon-loss operator $L:H_n\to H_{n-k}$ carries the $SU(2)$ action
\begin{equation}
g\cdot L=R_{n-k}(g)LR_n(g)^\dagger.
\label{loss-operator-action}
\end{equation}
On the input sector $H_n$, the Kraus operators of total loss order $k$ are scalar multiples of the corresponding annihilation monomials. For $0<\eta<1$ and $l_1+l_2=k$,
\begin{equation}
K_{l_1l_2}P_n
=\sqrt{\frac{(1-\eta)^k\eta^{n-k}}{l_1!l_2!}}
\left(\hat a_1^{l_1}\otimes\hat a_2^{l_2}\right)P_n.
\label{kraus-loss-monomial}
\end{equation}
Moreover, different total-loss orders have orthogonal output sectors, so
\begin{equation}
P_nK_{l_1l_2}^\dagger K_{l_1'l_2'}P_n=0
\qquad\text{if}\qquad
l_1+l_2\neq l_1'+l_2'.
\label{different-loss-orders-orthogonal}
\end{equation}
Thus the Knill--Laflamme products for loss order $k$ are represented by products of degree-$k$ annihilation monomials.

We first present the following lemmas.

\begin{lemma}\label{multi-op-action}
Given composable loss operators $L_1,L_2,\dots,L_t$, the $SU(2)$ action on their product is given by
\begin{equation}
g\cdot(L_1L_2\dots L_t)
=(g\cdot L_1)(g\cdot L_2)\dots(g\cdot L_t).
\end{equation}
\end{lemma}

\begin{proof}
Let $\operatorname{ord}(L_i)=p_i$ correspond to the number of photons lost per operator, with $p=\sum_i p_i$. Then we have, by definition,
\begin{equation}
\begin{aligned}
g\cdot(L_1L_2\dots L_t)
&=R_{n-p}(g)L_1L_2\dots L_tR_n^\dagger(g)\\
&=R_{n-p}(g)L_1R_{n-p+p_1}^\dagger(g)
R_{n-p+p_1}(g)L_2R_{n-p+p_1+p_2}^\dagger(g)\dots
R_{n-p_t}(g)L_tR_n^\dagger(g)\\
&=(g\cdot L_1)(g\cdot L_2)\dots(g\cdot L_t),
\end{aligned}
\end{equation}
which was to be shown.
\end{proof}

\begin{lemma}
Under the action $L\mapsto R_{n-1}(g)LR_n(g)^\dagger$, the space $\operatorname{span}\{\hat a_1\otimes I_2,I_1\otimes\hat a_2\}$ carries a representation isomorphic to $S_{1/2}$.
\label{annihilation-irrep}
\end{lemma}

\begin{proof}
See Ref.~\cite{aoudeClassicalObservablesCoherentspin2021}, adapted to our convention. The operators transform as
\begin{equation}
\begin{pmatrix}
\hat a_1\otimes I_2\\
I_1\otimes\hat a_2
\end{pmatrix}
\longmapsto
U^\dagger
\begin{pmatrix}
\hat a_1\otimes I_2\\
I_1\otimes\hat a_2
\end{pmatrix},
\qquad U\in SU(2).
\end{equation}
So the operator space carries the $S_{1/2}^*$ representation, which is isomorphic to the fundamental $SU(2)$ representation from Lemma \ref{su2-self-dual}.
\end{proof}

Now, we consider the representation of the loss operators of interest.

\begin{definition}
For $0\leq k\leq n$, we denote by $\mathcal L_k$ the set of two-mode loss operators corresponding to losing $k$ photons, i.e.,
\begin{equation}
\mathcal L_k
=\operatorname{span}\{\hat a_1^t\otimes\hat a_2^{k-t}:0\leq t\leq k\}.
\label{loss-subspaces}
\end{equation}
\end{definition}

\begin{lemma}
The space $\mathcal L_k$ is isomorphic to $S_{k/2}$ as an $SU(2)$ representation.
\label{k-loss-operator-irrep}
\end{lemma}

\begin{proof}
Let $\mathcal L_1=\operatorname{span}\{\hat a_1\otimes I_2,I_1\otimes\hat a_2\}\cong S_{1/2}$ from Lemma~\ref{annihilation-irrep}. Define a multilinear map
\begin{equation}
m:\underbrace{\mathcal L_1\times\cdots\times\mathcal L_1}_{k\text{ factors}}\longrightarrow\mathcal L_k
\end{equation}
such that $m(L_1,L_2,\dots,L_k)=L_1L_2\dots L_k$, where $L_i\in\mathcal L_1$. This is symmetric since the single-mode annihilation operators on distinct tensor factors commute, and thus permuting the input operators gives the same output. By the universal property, there exists some linear $\widetilde m:\operatorname{Sym}^k\mathcal L_1\to\mathcal L_k$ such that $\widetilde m(L_1\odot L_2\odot\dots\odot L_k)=L_1L_2\dots L_k$.

Note that the left-hand side is the output of the universal symmetric mapping, whereas the right-hand side is regular operator multiplication. The map $\widetilde m$ is clearly a bijection. Furthermore, it is $SU(2)$-equivariant. Indeed, since the action on the symmetric power representation is given by
\begin{equation}
g\cdot(L_1\odot L_2\odot\dots\odot L_k)
=(g\cdot L_1)\odot(g\cdot L_2)\odot\dots\odot(g\cdot L_k),
\end{equation}
we have that
\begin{equation}
\begin{aligned}
\widetilde m(g\cdot(L_1\odot L_2\odot\dots\odot L_k))
&=\widetilde m((g\cdot L_1)\odot(g\cdot L_2)\odot\dots\odot(g\cdot L_k))\\
&=(g\cdot L_1)(g\cdot L_2)\dots(g\cdot L_k)\\
&=g\cdot(L_1L_2\dots L_k)
=g\cdot\widetilde m(L_1\cdot L_2\cdot\dots\cdot L_k),
\end{aligned}
\end{equation}
where we've also used Lemma~\ref{multi-op-action}. Therefore, we see that
\begin{equation}
\mathcal L_k\cong\operatorname{Sym}^k\mathcal L_1
\cong\operatorname{Sym}^kS_{1/2}
\cong S_{k/2},
\end{equation}
where the last isomorphism follows from Lemma~\ref{su2-symmetric-irreps}.
\end{proof}

Now, we quickly show the following lemma.

\begin{lemma}
The space $\mathcal L_k^\dagger=\{L^\dagger:L\in\mathcal L_k\}$ is isomorphic to $S_{k/2}^*$ as an $SU(2)$ representation. Furthermore, the relevant space of Knill--Laflamme operators,
\begin{equation}
\mathcal K_k
=\operatorname{span}\{A^\dagger B:A,B\in\mathcal L_k\}
\cong S_{k/2}^*\otimes S_{k/2}.
\label{loss-kl-representation-equation}
\end{equation}
\label{loss-kl-representation}
\end{lemma}

\begin{proof}
We're going to invoke Lemma~\ref{adjoint-KL-rep-lemma}, together with Lemma~\ref{k-loss-operator-irrep}. The first part follows immediately. For the second part, we need only show that the relevant operators are linearly independent. Indeed, let
\begin{equation}
L_t=\hat a_1^t\otimes\hat a_2^{k-t}.
\end{equation}
We need to show that $\{L_i^\dagger L_j\}$ are linearly independent. The product $L_i^\dagger L_j$ maps each occupation state $\ket{m}\otimes\ket{n-m}$ to a multiple of
\begin{equation}
\ket{m+i-j}\otimes\ket{n-m-i+j}.
\label{loss-product-matrix-diagonal}
\end{equation}
Therefore, basis elements with different $i-j$ are automatically linearly independent since they have non-overlapping matrix elements. For fixed $\delta\geq 0$, the operators on the corresponding matrix diagonal are $L_{j+\delta}^\dagger L_j$ for $0\leq j\leq k-\delta$. Suppose there exist $c_j$ such that
\begin{equation}
A=\sum_{j=0}^{k-\delta}c_jL_{j+\delta}^\dagger L_j=0.
\end{equation}
Applying this to $\ket{0}\otimes\ket{n}$ leaves only the $j=0$ term, so
\begin{equation}
A(\ket{0}\otimes\ket{n})
=c_0L_\delta^\dagger L_0(\ket{0}\otimes\ket{n})=0.
\end{equation}
Since $L_\delta^\dagger L_0(\ket{0}\otimes\ket{n})\neq 0$, we obtain $c_0=0$. Now suppose inductively that $c_0,\dots,c_{s-1}=0$ for some $1\leq s\leq k-\delta$. When $A$ is applied to $\ket{s}\otimes\ket{n-s}$, every term with $j>s$ vanishes because $\hat a_1^j\ket{s}=0$, while the terms with $j<s$ vanish by the induction hypothesis. Therefore,
\begin{equation}
A(\ket{s}\otimes\ket{n-s})
=c_sL_{s+\delta}^\dagger L_s(\ket{s}\otimes\ket{n-s})=0.
\end{equation}
The remaining vector is nonzero because $n\geq k$ and $0\leq s\leq k-\delta$, so $c_s=0$. Induction gives $c_j=0$ for every $0\leq j\leq k-\delta$.

The operators with $i-j=-\delta<0$ are the adjoints of the operators with $i-j=\delta>0$. Taking the adjoint of any linear relation among them therefore reduces it to the case just proved. Hence all the products $L_i^\dagger L_j$ are linearly independent, as desired.
\end{proof}

\begin{lemma}
Let $1\leq k\leq n$, and define $L_i^{(r)}=\hat a_1^i\otimes\hat a_2^{r-i}$ for $0\leq i\leq r$. As operators on $H_n$, we have
\begin{equation}
\left(L_i^{(k-1)}\right)^\dagger L_j^{(k-1)}
=\frac{1}{n-k+1}\left(
\left(L_{i+1}^{(k)}\right)^\dagger L_{j+1}^{(k)}
+\left(L_i^{(k)}\right)^\dagger L_j^{(k)}
\right),
\label{lower-loss-KL-identity}
\end{equation}
for all $0\leq i,j\leq k-1$. Therefore, $\mathcal K_{k-1}\subseteq\mathcal K_k$.
\label{lower-loss-KL-inclusion}
\end{lemma}

\begin{proof}
We have
\begin{equation}
L_{i+1}^{(k)}=(\hat a_1\otimes I_2)L_i^{(k-1)},
\qquad
L_i^{(k)}=(I_1\otimes\hat a_2)L_i^{(k-1)}.
\end{equation}
Therefore,
\begin{equation}
\begin{aligned}
\left(L_{i+1}^{(k)}\right)^\dagger L_{j+1}^{(k)}
+\left(L_i^{(k)}\right)^\dagger L_j^{(k)}
&=\left(L_i^{(k-1)}\right)^\dagger\hat N L_j^{(k-1)}\\
&=(n-k+1)\left(L_i^{(k-1)}\right)^\dagger L_j^{(k-1)},
\end{aligned}
\end{equation}
since $L_j^{(k-1)}$ maps $H_n$ into $H_{n-k+1}$, where the total number operator $\hat N$ acts as $n-k+1$. Rearranging gives the claimed identity. Thus, every operator spanning $\mathcal K_{k-1}$ belongs to $\mathcal K_k$, and hence $\mathcal K_{k-1}\subseteq\mathcal K_k$.
\end{proof}

\subsection{Equivalence between loss distance and tensor-rank detection}

\begin{lemma}
Protecting against loss up to order $k$ reduces to satisfying the Knill--Laflamme conditions for $\mathcal K_k$, which we can decompose in terms of $SU(2)$ representations as
\begin{equation}
\mathcal K_k\cong S_k\oplus S_{k-1}\oplus\dots\oplus S_1\oplus S_0.
\label{loss-product-spin-decomposition}
\end{equation}
This is equivalent to having a code that detects every operator in $\mathcal T_n^\ell$ for each $0\leq\ell\leq k$.
\label{loss-rotation-equivalence}
\end{lemma}

\begin{proof}
Products between errors of different total loss orders vanish by Eq.~\eqref{different-loss-orders-orthogonal}. For a common total loss order $p\leq k$, the corresponding products lie in $\mathcal K_p\subseteq\mathcal K_k$ by repeated application of Lemma~\ref{lower-loss-KL-inclusion}. The result then follows from Lemmas~\ref{loss-kl-representation} and \ref{su2-self-dual}, together with the Clebsch--Gordan theorem, Lemma~\ref{Clebsch-gordan}.
\end{proof}

\section{Knill--Laflamme conditions for finite-group-invariant codes}

\begin{corollary}
Let $P:V\to V$ be a Hermitian projector that is $G$-invariant with respect to a unitary representation $\pi_V$ of $G$ on $V$. For each irrep $C$ occurring in $V$, let $\pi_C$ denote its representation, let $M_C$ be its multiplicity space, and let $S_C:M_C\otimes C\to V$ be an equivariant isometry mapping the $C$-isotypic component to the physical representation. Here $G$ acts trivially on $M_C$, and equivariance means
\begin{equation}
\pi_V(g)S_C=S_C\bigl(I_{M_C}\otimes\pi_C(g)\bigr)
\qquad \forall g\in G.
\label{isotypic-embedding-equivariance}
\end{equation}
Then the code projector can be written as
\begin{equation}
P=\sum_C S_C(P_C\otimes I_C)S_C^\dagger,
\label{reducible-irrep-projector}
\end{equation}
where $P_C$ is a Hermitian projector on $M_C$.
\label{invariant-projector-form}
\end{corollary}

\begin{proof}
If $P$ is $G$-invariant, i.e., $\pi_V(g)P\pi_V(g)^{-1}=P$, we have that by definition $P\in\operatorname{End}_G(V)$, and thus by Lemma~\ref{EndGU-decomposition}, we know the projector takes the form $\bigoplus_C P_C\otimes I_C$, i.e., is block diagonal with respect to the irreps of the group. As a result we're also going to have $P_C^2=P_C$ and $P_C=P_C^\dagger$ for all $C$. Mapping the isotypic components to the physical space $V$ with the isometries $S_C$ gives Eq.~\eqref{reducible-irrep-projector}.
\end{proof}

\subsection{General invariant-code Knill--Laflamme theorem}

We prove the following necessary and sufficient KL conditions for codes of this form. Let $W\subseteq\operatorname{End}(V)$ be a $G$-invariant space of operators we wish to detect, and write
\begin{equation}
W\cong\bigoplus_U U^{\oplus m_U}
\label{KL-operator-space-decomposition}
\end{equation}
as a $G$-representation under the conjugation action. For each irrep $U$, let $d_U=\dim U$, let $\pi_U$ denote the representation of $G$ on $U$, and choose a $G$-adapted basis $T_\mu^{U,\alpha}$, with $1\leq\mu\leq d_U$ and $1\leq\alpha\leq m_U$, satisfying
\begin{equation}
\pi_V(g)T_\mu^{U,\alpha}\pi_V(g)^{-1}
=\sum_{\nu=1}^{d_U}[\pi_U(g)]_{\nu\mu}T_\nu^{U,\alpha}.
\label{G-adapted-KL-basis}
\end{equation}
For each pair of irreps $C,D$ occurring in $V$, let $\pi_C$ and $\pi_D$ denote their representations, and define
\begin{equation}
r_U^{C,D}=\dim\operatorname{Hom}_G(U,D^*\otimes C).
\label{cross-irrep-multiplicity}
\end{equation}
Choose a basis $\Phi_j^{U;C,D}:U\to\operatorname{Hom}(D,C)$ of this intertwiner space, with $1\leq j\leq r_U^{C,D}$, and define
\begin{equation}
F_{j,\mu}^{U;C,D}=\Phi_j^{U;C,D}(\ket{\mu}).
\label{F-from-intertwiner}
\end{equation}
For the unique trivial intertwiner in $\operatorname{Hom}_G(\mathbf 1,C^*\otimes C)$, we choose the normalization
\begin{equation}
F_{1,1}^{\mathbf 1;C,C}=I_C.
\label{trivial-F-normalization}
\end{equation}

\begin{theorem}
Let $P$ be the $G$-invariant projector in Eq.~\eqref{reducible-irrep-projector}. For each $U,C,D,\alpha$, there are unique maps $B_{j,\alpha}^{U;C,D}:M_D\to M_C$, independent of $\mu$, such that
\begin{equation}
S_C^\dagger T_\mu^{U,\alpha}S_D
=\sum_{j=1}^{r_U^{C,D}}B_{j,\alpha}^{U;C,D}\otimes F_{j,\mu}^{U;C,D}.
\label{general-compressed-basis-operator}
\end{equation}
Then $P$ detects every operator in $W$ if and only if
\begin{equation}
P_CB_{j,\alpha}^{U;C,D}P_D
=
\begin{cases}
\lambda_\alpha P_C, & C=D,\ U=\mathbf 1,\ j=1,\\
0, & \text{otherwise},
\end{cases}
\label{general-invariant-code-KL-conditions}
\end{equation}
for every $U,C,D,\alpha,j$. The scalar $\lambda_\alpha$ is common to every $C$ for which $P_C\neq0$, although it may depend on the multiplicity label $\alpha$ of the trivial representation in $W$.
\label{general-invariant-code-KL}
\end{theorem}

\begin{proof}
Fix two irreps $C,D$ of $G$, and define the map
\begin{equation}
X_\alpha^{U;C,D}:U\longrightarrow
\operatorname{Hom}(M_D\otimes D,M_C\otimes C)
\label{cross-block-intertwiner}
\end{equation}
by
\begin{equation}
X_\alpha^{U;C,D}(\ket{\mu})
=S_C^\dagger T_\mu^{U,\alpha}S_D,
\label{cross-block-intertwiner-action}
\end{equation}
where $\ket{\mu}$ is the basis of $U$ used in Eq.~\eqref{G-adapted-KL-basis}. We first show that this map is $G$-equivariant. Indeed,
\begin{equation}
\begin{aligned}
X_\alpha^{U;C,D}(\pi_U(g)\ket{\mu})
&=X_\alpha^{U;C,D}\left(\sum_\nu[\pi_U(g)]_{\nu\mu}\ket{\nu}\right)\\
&=\sum_\nu[\pi_U(g)]_{\nu\mu}S_C^\dagger T_\nu^{U,\alpha}S_D.
\end{aligned}
\label{cross-block-equivariance-first}
\end{equation}
On the other hand, the equivariance of $S_C$ and $S_D$ gives
\begin{equation}
\begin{aligned}
g\cdot X_\alpha^{U;C,D}(\ket{\mu})
&=(I_{M_C}\otimes\pi_C(g))S_C^\dagger T_\mu^{U,\alpha}S_D
(I_{M_D}\otimes\pi_D(g)^{-1})\\
&=S_C^\dagger\pi_V(g)T_\mu^{U,\alpha}\pi_V(g)^{-1}S_D\\
&=\sum_\nu[\pi_U(g)]_{\nu\mu}S_C^\dagger T_\nu^{U,\alpha}S_D.
\end{aligned}
\label{cross-block-equivariance-second}
\end{equation}
Thus,
\begin{equation}
X_\alpha^{U;C,D}(\pi_U(g)\ket{\mu})
=g\cdot X_\alpha^{U;C,D}(\ket{\mu}),
\label{cross-block-equivariance}
\end{equation}
i.e., $X_\alpha^{U;C,D}$ is $G$-equivariant. Therefore, using Lemma~\ref{hom-tensor-prod},
\begin{equation}
\begin{aligned}
X_\alpha^{U;C,D}
&\in\operatorname{Hom}_G\bigl(U,\operatorname{Hom}(M_D\otimes D,M_C\otimes C)\bigr)\\
&\cong\operatorname{Hom}(M_D,M_C)\otimes
\operatorname{Hom}_G\bigl(U,\operatorname{Hom}(D,C)\bigr)\\
&\cong\operatorname{Hom}(M_D,M_C)\otimes
\operatorname{Hom}_G(U,D^*\otimes C).
\end{aligned}
\label{cross-block-Hom-decomposition}
\end{equation}
Since the $\Phi_j^{U;C,D}$ form a basis of the second tensor factor, there are unique maps $B_{j,\alpha}^{U;C,D}:M_D\to M_C$ such that
\begin{equation}
X_\alpha^{U;C,D}
=\sum_{j=1}^{r_U^{C,D}}B_{j,\alpha}^{U;C,D}\otimes\Phi_j^{U;C,D}.
\label{cross-block-B-decomposition}
\end{equation}
Evaluating Eq.~\eqref{cross-block-B-decomposition} at $\ket{\mu}$ gives Eq.~\eqref{general-compressed-basis-operator}.

We're now ready to look at the QEC matrix of our code for each elementary tensor $T_\mu^{U,\alpha}$. Indeed, we have
\begin{equation}
\begin{aligned}
PT_\mu^{U,\alpha}P
&=\sum_{C,D}S_C(P_C\otimes I_C)S_C^\dagger
T_\mu^{U,\alpha}S_D(P_D\otimes I_D)S_D^\dagger\\
&=\sum_{C,D}S_C\left(
\sum_{j=1}^{r_U^{C,D}}
P_CB_{j,\alpha}^{U;C,D}P_D\otimes F_{j,\mu}^{U;C,D}
\right)S_D^\dagger.
\end{aligned}
\label{general-QEC-matrix}
\end{equation}
Since the isometries $S_C$ map to mutually orthogonal isotypic components of $V$, we have
\begin{equation}
S_C^\dagger S_D=\delta_{C,D}I_{M_C}\otimes I_C,
\label{isotypic-isometry-orthogonality}
\end{equation}
and hence
\begin{equation}
S_C^\dagger P S_D=\delta_{C,D}P_C\otimes I_C.
\label{invariant-projector-block}
\end{equation}

Suppose first that $P$ detects every operator in $W$. For $C\neq D$, compressing the KL equation for $T_\mu^{U,\alpha}$ between $S_C^\dagger$ and $S_D$ gives
\begin{equation}
\sum_{j=1}^{r_U^{C,D}}
P_CB_{j,\alpha}^{U;C,D}P_D\otimes F_{j,\mu}^{U;C,D}=0
\qquad\text{for every }\mu.
\label{cross-isotypic-KL-block}
\end{equation}
Equivalently, the corresponding map in
$\operatorname{Hom}(M_D,M_C)\otimes\operatorname{Hom}_G(U,D^*\otimes C)$ is zero. Since the $\Phi_j^{U;C,D}$ form a basis, this is equivalent to
\begin{equation}
P_CB_{j,\alpha}^{U;C,D}P_D=0
\label{cross-isotypic-KL-condition}
\end{equation}
for every $j$.

Now take $C=D$. If $U$ is nontrivial, the image of every $\Phi_j^{U;C,C}$ lies in the nontrivial isotypic component of $\operatorname{End}(C)$. By Lemmas~\ref{end-decomp} and \ref{unique-trivial-rep-end}, this component has zero intersection with $\mathbb CI_C$. Thus the diagonal block of $PT_\mu^{U,\alpha}P$ can be proportional to $P_C\otimes I_C$ only if
\begin{equation}
P_CB_{j,\alpha}^{U;C,C}P_C=0
\label{nontrivial-diagonal-KL-condition}
\end{equation}
for every $j$. Indeed, the equality must hold for every $\mu$, so the corresponding intertwiner $\sum_jP_CB_{j,\alpha}^{U;C,C}P_C\otimes\Phi_j^{U;C,C}$ must vanish. The linear independence of the $\Phi_j^{U;C,C}$ then gives Eq.~\eqref{nontrivial-diagonal-KL-condition}. For $U=\mathbf 1$, evaluation at $1\in\mathbf 1\cong\mathbb C$ gives the canonical isomorphism
\begin{equation}
\operatorname{Hom}_G(\mathbf 1,C^*\otimes C)
\cong(C^*\otimes C)^G,
\qquad \Phi\longmapsto\Phi(1).
\label{trivial-Hom-fixed-subspace}
\end{equation}
Indeed, since $G$ acts trivially on $\mathbf 1$, equivariance gives $g\cdot\Phi(1)=\Phi(g\cdot1)=\Phi(1)$, so $\Phi(1)\in(C^*\otimes C)^G$. Conversely, every $X\in(C^*\otimes C)^G$ defines an equivariant map $\Phi_X:\mathbf 1\to C^*\otimes C$ by $\Phi_X(z)=zX$. Under the identification $C^*\otimes C\cong\operatorname{End}(C)$, Lemma~\ref{end-decomp} gives $\operatorname{End}(C)\cong\mathbb CI_C\oplus\operatorname{End}_0(C)$, while Lemma~\ref{unique-trivial-rep-end} gives $\operatorname{End}_0(C)^G=\{0\}$. Therefore,
\begin{equation}
\operatorname{Hom}_G(\mathbf 1,C^*\otimes C)
\cong\operatorname{End}(C)^G
=\mathbb CI_C,
\label{trivial-intertwiner-space}
\end{equation}
so $r_{\mathbf 1}^{C,C}=1$. With the normalization in Eq.~\eqref{trivial-F-normalization}, the diagonal block is proportional to $P_C\otimes I_C$ exactly when
\begin{equation}
P_CB_{1,\alpha}^{\mathbf 1;C,C}P_C=\lambda_\alpha P_C.
\label{trivial-diagonal-KL-condition}
\end{equation}
The same scalar $\lambda_\alpha$ must occur for every nonzero $C$-sector because the full operator $PT_1^{\mathbf 1,\alpha}P$ must be a single scalar multiple of $P$.

Conversely, suppose Eq.~\eqref{general-invariant-code-KL-conditions} holds. Substituting those conditions into Eq.~\eqref{general-QEC-matrix} gives
\begin{equation}
PT_\mu^{U,\alpha}P
=
\begin{cases}
\lambda_\alpha P, & U=\mathbf 1,\\
0, & U\neq\mathbf 1.
\end{cases}
\label{adapted-basis-KL-conclusion}
\end{equation}
Thus the KL detection condition holds for every basis operator $T_\mu^{U,\alpha}$, and hence by linearity for every operator in $W$.
\end{proof}

\begin{remark}
The preceding proof gives the selection rule
\begin{equation}
U\not\subseteq D^*\otimes C
\implies S_C^\dagger T_\mu^{U,\alpha}S_D=0.
\label{reducible-selection-rule}
\end{equation}
Thus, we need only check the relevant conditions for $U\subseteq D^*\otimes C$.
\end{remark}

\subsection{Single-isotypic and pure-irrep corollaries}

We now specialize the general theorem to a code supported on a single isotypic component. Let $C$ be an irrep occurring in $V$, and let $P_S$ be a Hermitian projector on its multiplicity space $M_C$. The code projector is
\begin{equation}
P=S_C(P_S\otimes I_C)S_C^\dagger.
\label{single-isotypic-code-projector}
\end{equation}

\begin{corollary}
The projector in Eq.~\eqref{single-isotypic-code-projector} detects every operator in $W$ if and only if
\begin{equation}
P_SB_{j,\alpha}^{U;C,C}P_S
=
\begin{cases}
\lambda_\alpha P_S, & U=\mathbf 1,\ j=1,\\
0, & U\neq\mathbf 1,
\end{cases}
\label{single-isotypic-KL-conditions}
\end{equation}
for every irrep $U$ shared by $W$ and $\operatorname{End}(C)$, every $1\leq\alpha\leq m_U$, and every $1\leq j\leq r_U^{C,C}$.
\label{single-isotypic-KL}
\end{corollary}

\begin{proof}
In the setting of Eq.~\eqref{single-isotypic-code-projector}, the only nonzero multiplicity projector in Eq.~\eqref{reducible-irrep-projector} is $P_C=P_S$. Thus, all cross-isotypic conditions in Theorem~\ref{general-invariant-code-KL} disappear, and Eq.~\eqref{general-invariant-code-KL-conditions} reduces exactly to Eq.~\eqref{single-isotypic-KL-conditions}.
\end{proof}

We now prove a corollary that will be very useful in practice.

\begin{corollary}
Let $C$ be an irrep of $G$ that occurs with multiplicity one in $V$, and let $P=S_CS_C^\dagger$ be the projector onto its unique copy. Let $O\in\operatorname{End}(V)$, and let
\begin{equation}
W_O=\operatorname{span}\{\pi_V(g)O\pi_V(g)^{-1}:g\in G\}
\label{single-operator-orbit-span}
\end{equation}
be its $G$-orbit span. If the only irrep of $W_O$ that appears in $\operatorname{End}(C)\cong C\otimes C^*$ is the trivial representation, then
\begin{equation}
POP=c(O)P
\label{trivial-overlap-automatic-detection}
\end{equation}
for some scalar $c(O)$.
\label{trivial-irrep-qecc}
\end{corollary}

\begin{proof}
By the selection rule in Eq.~\eqref{reducible-selection-rule}, only irreps shared by $W_O$ and $\operatorname{End}(C)$ can contribute to $S_C^\dagger O S_C$. By assumption, the only such irrep is $\mathbf 1$. Remember that
\begin{equation}
\operatorname{End}(C)\cong\mathbb CI_C\oplus\operatorname{End}_0(C),
\label{EndC-trivial-nontrivial-split-repeated}
\end{equation}
and the trivial representation does not appear in $\operatorname{End}_0(C)$ by Lemma~\ref{unique-trivial-rep-end}. Thus,
\begin{equation}
S_C^\dagger O S_C=c(O)I_C
\label{multiplicity-one-compression-scalar}
\end{equation}
for some scalar $c(O)$. Mapping back to $V$ gives
\begin{equation}
POP=S_C\bigl(c(O)I_C\bigr)S_C^\dagger=c(O)P,
\label{multiplicity-one-automatic-KL-proof}
\end{equation}
as desired.
\end{proof}

\begin{remark}
More generally, suppose $P_S=\ket{s}\bra{s}$ has rank one, where $\ket{s}\in M_C$ is a unit vector~\cite{denysQuantumErrorCorrectingCodes2024}. For every operator $B\in\operatorname{End}(M_C)$,
\begin{equation}
P_SBP_S=\braket{s|B|s}P_S.
\label{rank-one-multiplicity-compression}
\end{equation}
Therefore, the trivial-irrep conditions in Eq.~\eqref{single-isotypic-KL-conditions} are automatic. If $W$ and $\operatorname{End}(C)$ share nontrivial irreps, the pure-irrep code detects all of $W$ if and only if
\begin{equation}
\braket{s|B_{j,\alpha}^{U;C,C}|s}=0
\label{pure-irrep-nontrivial-KL-condition}
\end{equation}
for every shared nontrivial $U$ and every $j,\alpha$. In the multiplicity-one case, this is the statement that all coefficients of the shared nontrivial components vanish. 
\end{remark}

\subsection{Constructing isotypic embeddings and multiplicity matrices}

We now describe how to generally construct the isotypic embeddings $S_C$, the representation matrices $F_{j,\mu}^{U;C,D}$, and the multiplicity matrices $B_{j,\alpha}^{U;C,D}$ appearing in Theorem~\ref{general-invariant-code-KL}. Let $d_C=\dim C$, $d_D=\dim D$, and $d_V=\dim V$. Throughout this subsection, $\operatorname{vec}$ denotes column-stacking vectorization, so that
\begin{equation}
\operatorname{vec}(AXB)=(B^T\otimes A)\operatorname{vec}(X).
\label{column-vectorization-identity}
\end{equation}

We first construct $S_C:M_C\otimes C\to V$. Since $m_C=\dim M_C$ is the multiplicity of $C$ in $V$, the intertwiner space $\operatorname{Hom}_G(C,V)$ has dimension $m_C$. Choose a basis $\{\widetilde E_{C,a}:C\to V:1\leq a\leq m_C\}$ satisfying
\begin{equation}
\pi_V(g)\widetilde E_{C,a}=\widetilde E_{C,a}\pi_C(g)
\qquad\forall g\in G.
\label{isotypic-embedding-intertwiner-equation}
\end{equation}
By Eq.~\eqref{column-vectorization-identity}, these maps can be found by solving
\begin{equation}
\left(I_{d_C}\otimes\pi_V(g)-\pi_C(g)^T\otimes I_{d_V}\right)
\operatorname{vec}(\widetilde E_{C,a})=0
\qquad\forall g\in G.
\label{vectorized-isotypic-embedding-equation}
\end{equation}
It is enough to impose Eq.~\eqref{vectorized-isotypic-embedding-equation} on any generating set of $G$.

The maps $\widetilde E_{C,a}$ need not give an isometric embedding. Define their Gram matrix $\Gamma^{(C)}$ by
\begin{equation}
\Gamma^{(C)}_{ab}
=\frac{1}{d_C}\operatorname{Tr}\left(\widetilde E_{C,a}^\dagger\widetilde E_{C,b}\right).
\label{isotypic-embedding-Gram-matrix}
\end{equation}
Because $\widetilde E_{C,a}^\dagger\widetilde E_{C,b}\in\operatorname{End}_G(C)$, Schur's lemma gives
\begin{equation}
\widetilde E_{C,a}^\dagger\widetilde E_{C,b}
=\Gamma^{(C)}_{ab}I_C.
\label{embedding-intertwiner-Schur-reduction}
\end{equation}
The linear independence of the $\widetilde E_{C,a}$ implies that $\Gamma^{(C)}$ is positive definite. We therefore define
\begin{equation}
E_{C,a}=\sum_{b=1}^{m_C}\widetilde E_{C,b}
\left[(\Gamma^{(C)})^{-1/2}\right]_{ba}.
\label{orthonormalized-embedding-intertwiners}
\end{equation}
Then
\begin{equation}
\begin{aligned}
E_{C,a}^\dagger E_{C,b}
&=\sum_{c,d=1}^{m_C}
\left[(\Gamma^{(C)})^{-1/2}\right]_{ac}
\Gamma^{(C)}_{cd}
\left[(\Gamma^{(C)})^{-1/2}\right]_{db}I_C\\
&=\delta_{ab}I_C.
\end{aligned}
\label{orthonormalized-embedding-isometry}
\end{equation}
For an orthonormal basis $\{\ket{a}:1\leq a\leq m_C\}$ of $M_C$, define
\begin{equation}
S_C(\ket{a}\otimes c)=E_{C,a}c
\qquad(c\in C).
\label{isotypic-embedding-from-intertwiners}
\end{equation}
Equations~\eqref{isotypic-embedding-intertwiner-equation} and \eqref{orthonormalized-embedding-isometry} imply
\begin{equation}
S_C^\dagger S_C=I_{M_C}\otimes I_C,
\qquad
\pi_V(g)S_C=S_C\bigl(I_{M_C}\otimes\pi_C(g)\bigr).
\label{constructed-isotypic-embedding-properties}
\end{equation}
Thus $S_C$ is the required equivariant isometry, and $S_CS_C^\dagger$ is the orthogonal projector onto the $C$-isotypic component of $V$.

We next construct the $F_{j,\mu}^{U;C,D}$. Fix bases of $U$, $C$, and $D$. For an intertwiner $\Phi^{U;C,D}:U\to\operatorname{Hom}(D,C)$, define
\begin{equation}
M[\Phi^{U;C,D}]
=\begin{pmatrix}
\operatorname{vec}(F_1)&\operatorname{vec}(F_2)&\cdots&\operatorname{vec}(F_{d_U})
\end{pmatrix},
\qquad
F_\mu=\Phi^{U;C,D}(\ket{\mu}).
\label{intertwiner-matrix-definition}
\end{equation}
The equivariance of $\Phi^{U;C,D}$ is equivalent to
\begin{equation}
\left(\pi_D(g)^*\otimes\pi_C(g)\right)M[\Phi^{U;C,D}]
-M[\Phi^{U;C,D}]\pi_U(g)=0
\qquad\forall g\in G,
\label{intertwiner-matrix-equation}
\end{equation}
where $*$ denotes entrywise complex conjugation in the chosen basis. Vectorizing Eq.~\eqref{intertwiner-matrix-equation} gives
\begin{equation}
\left(
I_{d_U}\otimes\pi_D(g)^*\otimes\pi_C(g)
-\pi_U(g)^T\otimes I_{d_Cd_D}
\right)
\operatorname{vec}\left(M[\Phi^{U;C,D}]\right)=0.
\label{vectorized-cross-intertwiner-equation}
\end{equation}
Solving these equations gives $r_U^{C,D}$ linearly independent intertwiners $\widetilde\Phi_a^{U;C,D}$, with matrices $\widetilde M_a=M[\widetilde\Phi_a^{U;C,D}]$. As above, it is enough to impose the equations on a generating set of $G$.

For a nontrivial irrep $U$, define
\begin{equation}
\Gamma^{(U;C,D)}_{ab}
=\frac{1}{d_U}\operatorname{Tr}(\widetilde M_a^\dagger\widetilde M_b).
\label{cross-intertwiner-Gram-matrix}
\end{equation}
The matrices $\widetilde M_a^\dagger\widetilde M_b$ commute with $\pi_U(g)$, so Schur's lemma gives $\widetilde M_a^\dagger\widetilde M_b=\Gamma^{(U;C,D)}_{ab}I_U$. Applying the same symmetric orthogonalization as in Eq.~\eqref{orthonormalized-embedding-intertwiners} produces a basis $\Phi_j^{U;C,D}$ satisfying
\begin{equation}
\operatorname{Tr}\left[
\left(F_{j,\mu}^{U;C,D}\right)^\dagger
F_{j',\nu}^{U;C,D}
\right]
=\delta_{jj'}\delta_{\mu\nu}.
\label{F-Hilbert-Schmidt-orthonormality}
\end{equation}
For $U=\mathbf 1$ and $C=D$, we instead use the fixed normalization $F_{1,1}^{\mathbf 1;C,C}=I_C$ from Eq.~\eqref{trivial-F-normalization}; if $C\neq D$, Schur's lemma gives $r_{\mathbf 1}^{C,D}=0$.

Finally, the multiplicity matrices are obtained by Hilbert--Schmidt contraction of the representation factor. For any $1\leq\mu\leq d_U$,
\begin{equation}
B_{j,\alpha}^{U;C,D}
=\frac{
\operatorname{Tr}_D\!\left[
\left(I_{M_C}\otimes(F_{j,\mu}^{U;C,D})^\dagger\right)
S_C^\dagger T_\mu^{U,\alpha}S_D
\right]
}{
\operatorname{Tr}\!\left[
(F_{j,\mu}^{U;C,D})^\dagger F_{j,\mu}^{U;C,D}
\right]
},
\label{general-B-matrix-formula}
\end{equation}
where $\operatorname{Tr}_D$ contracts the common $D$ factor and leaves a map from $M_D$ to $M_C$. The right-hand side is independent of $\mu$ by Eq.~\eqref{general-compressed-basis-operator}.
For $U\neq\mathbf 1$, the Hilbert--Schmidt orthonormality in Eq.~\eqref{F-Hilbert-Schmidt-orthonormality} makes the denominator in Eq.~\eqref{general-B-matrix-formula} equal to one, so
\begin{equation}
B_{j,\alpha}^{U;C,D}
=\operatorname{Tr}_D\!\left[
\left(I_{M_C}\otimes(F_{j,\mu}^{U;C,D})^\dagger\right)
S_C^\dagger T_\mu^{U,\alpha}S_D
\right],
\qquad U\neq\mathbf 1.
\label{nontrivial-B-matrix-formula}
\end{equation}
For the trivial diagonal component, the fixed normalization $F_{1,1}^{\mathbf 1;C,C}=I_C$ instead makes the denominator equal to $d_C$, and Eq.~\eqref{general-B-matrix-formula} becomes
\begin{equation}
B_{1,\alpha}^{\mathbf 1;C,C}
=\frac{1}{d_C}\operatorname{Tr}_C\left(S_C^\dagger T_1^{\mathbf 1,\alpha}S_C\right),
\label{trivial-B-matrix-formula}
\end{equation}
which is consistent with the identity normalization in Eq.~\eqref{trivial-F-normalization}. Equations~\eqref{vectorized-isotypic-embedding-equation}, \eqref{vectorized-cross-intertwiner-equation}, and \eqref{general-B-matrix-formula} provide all the data needed to impose the conditions in Theorem~\ref{general-invariant-code-KL}. When $r_U^{C,D}=0$, the selection rule in Eq.~\eqref{reducible-selection-rule} shows that no matrix needs to be constructed for that triple.

\section{Orbit generators and levelwise minimal loss seeds}

\subsection{General orbit-generator theorem}

We first note that, since the code projector commutes with the group action, detecting one operator implies detecting its entire $G$-orbit span.

\begin{lemma}
Let $P$ be a $G$-invariant code projector on a unitary representation $V$, and suppose that
\begin{equation}
POP=\lambda_OP
\end{equation}
for some operator $O\in\operatorname{End}(V)$. Then
\begin{equation}
PO'P=\lambda_{O'}P
\end{equation}
for every $O'$ in the $G$-orbit span
\begin{equation}
W_O=\operatorname{span}\{\pi_V(g)O\pi_V(g)^{-1}:g\in G\}.
\label{orbit-KL-span}
\end{equation}
\label{orbit-KL-extensions}
\end{lemma}

\begin{proof}
Since we can commute the representation through the projector, we have
\begin{equation}
\begin{aligned}
P\pi_V(g)O\pi_V(g)^{-1}P
&=\pi_V(g)POP\pi_V(g)^{-1}\\
&=\pi_V(g)\lambda_OP\pi_V(g)^{-1}
=\lambda_OP.
\end{aligned}
\end{equation}
Therefore, every operator in $W_O$ satisfies the detection condition by linearity.
\end{proof}

So we can see we don't need to test all relevant operators; we need only find a set of operators whose $G$-orbits span the operator space. The following general orbit-generator theorem is due to Vale and Waldron~\cite{valeConstructionGInvariantFinite2016}.

\begin{theorem}
Let $G$ be a finite group and let $W$ be a finite-dimensional complex representation with representation map $\pi_W$ and decomposition
\begin{equation}
W\cong\bigoplus_U U^{\oplus m_U}
\end{equation}
in terms of irreps of $G$. The minimum number of seeds $w_1,w_2,\dots,w_s\in W$ satisfying
\begin{equation}
\operatorname{span}\{\pi_W(g)w_a:g\in G,\ 1\leq a\leq s\}=W
\end{equation}
is
\begin{equation}
s_{\min}(W)=\max_U\left\lceil\frac{m_U}{d_U}\right\rceil,
\qquad d_U=\dim U.
\label{general-orbit-generator-count}
\end{equation}
\label{spanning-set-G-orbit}
\end{theorem}

\begin{proof}
Let's first consider a single isotypic component of an irrep $U$, i.e., $W[U]\cong U\otimes\mathbb C^{m_U}$. We start by defining
\begin{equation}
\mathcal A_U=\operatorname{span}\{\pi_U(g):g\in G\}\subseteq\operatorname{End}(U).
\end{equation}
This is clearly a vector space over $\mathbb C$, contains the identity element $\pi_U(e)=I_U$, and is a ring. So it's by definition a $\mathbb C$-algebra. Furthermore, any invariant subspace of $\mathcal A_U$ is an invariant subspace of $\{\pi_U(g):g\in G\}$, and therefore $\mathcal A_U$ is irreducible, since $U$ is an irrep. It follows from Burnside's theorem, Lemma~\ref{burnside}, that $\operatorname{End}(U)=\mathcal A_U$. Note that the statement also holds trivially in the case that $d_U=1$, and both spaces are equal to $\mathbb CI_U$.

Now, choose a basis $u_1,\dots,u_{d_U}$ of $U$. We can write an element of $W[U]$ as
\begin{equation}
w_t=\sum_{i=1}^{d_U}u_i\otimes v_{t,i}.
\end{equation}
The action of the group on this vector is given by
\begin{equation}
g\cdot w_t=\sum_{i=1}^{d_U}\pi_U(g)u_i\otimes v_{t,i}.
\end{equation}
Because $\mathcal A_U=\operatorname{End}(U)$, for any $i,j\in\{1,\dots,d_U\}$ there exist parameters $c_g\in\mathbb C$ such that $E_{ji}=\sum_{g\in G}c_g\pi_U(g)$, where $E_{ji}$ is the linear map satisfying $E_{ji}u_i=u_j$ and $E_{ji}u_r=0$ for every $r\neq i$. Thus, we see that
\begin{equation}
\begin{aligned}
\left(\sum_{g\in G}c_g\pi_U(g)\right)\cdot w_t
&=\sum_{r=1}^{d_U}\left(\sum_{g\in G}c_g\pi_U(g)\right)u_r\otimes v_{t,r}\\
&=\sum_{r=1}^{d_U}E_{ji}u_r\otimes v_{t,r}
=u_j\otimes v_{t,i}.
\end{aligned}
\end{equation}
Thus, all the tensor products $u_j\otimes v_{t,i}$ exist in the $G$-orbit span of $w_t$. So, we see that
\begin{equation}
\operatorname{span}\{g\cdot w_t:1\leq t\leq s,\ g\in G\}
=U\otimes\operatorname{span}\{v_{t,i}:1\leq t\leq s,\ 1\leq i\leq d_U\}.
\end{equation}
We want these seeds to span $U\otimes\mathbb C^{m_U}$. So, it's clear that we need
\begin{equation}
s d_U\geq m_U
\quad\Longrightarrow\quad
s_{\min,U}=\left\lceil\frac{m_U}{d_U}\right\rceil.
\end{equation}
It's further clear that we can construct vectors $w_{1,U},w_{2,U},\dots,w_{s_{\min,U},U}$ satisfying this, simply by choosing a basis $e_1,e_2,\dots,e_{m_U}$ for the multiplicity space and defining
\begin{equation}
w_{t,U}=\sum_{\substack{1\leq i\leq d_U\\(t-1)d_U+i\leq m_U}}
u_i\otimes e_{(t-1)d_U+i}.
\label{spanning-vector-u}
\end{equation}

Finally, we need to satisfy this for every irrep $U$. It's clear that we need at least $s_{\min}=\max_U\lceil m_U/d_U\rceil$ seeds to span the entire space. We'll briefly show that this is sufficient. Define, for $1\leq t\leq s_{\min}$, and taking $w_{t,U}=0$ when $s_{\min,U}<t\leq s_{\min}$,
\begin{equation}
w_t=\bigoplus_U w_{t,U}.
\end{equation}
By Eq.~\eqref{isotypic-projector-formula}, the isotypic projector $\Pi_U$ is a linear combination of the matrices $\pi_W(g)$. Therefore, $\Pi_Uw_t$ is in the $G$-orbit span of the $w_t$, and furthermore $\Pi_Uw_t=w_{t,U}$. Therefore, each isotypic component $W[U]$ is in the orbit span of the $w_t$, and we thus conclude that
\begin{equation}
\operatorname{span}\{g\cdot w_t:g\in G,\ 1\leq t\leq s_{\min}\}=W,
\end{equation}
which was to be shown.
\end{proof}

Combining Lemma~\ref{orbit-KL-extensions} with Theorem~\ref{spanning-set-G-orbit}, it suffices to impose the KL conditions on $s_{\min}(W)$ appropriately chosen seed operators.

\subsection{Free projective orbits and rotated highest-weight seeds}

The explicit loss construction below requires finitely many free, pairwise disjoint projective orbits. The following elementary lemma guarantees that they are always available.

\begin{lemma}
Let $G$ be a finite subgroup of $SU(2)$ and let
\begin{equation}
\widetilde G=G/(G\cap\{I,-I\})
\end{equation}
act on $\mathbb{CP}^1$ in the natural way. That is, if $p = [a:b] \in \mathbb{CP}^1$ and $g \in G$ has the standard representation 

\begin{equation}
    R_1(g) = \begin{pmatrix}
        u & -\overline{v} \\ 
        v & \overline{u}
    \end{pmatrix}
\end{equation}

then the projective action is 

\begin{equation}
    \tilde{g} \cdot [a:b] = [ua - \overline{v}b: va + \overline{u}b]
\end{equation}

Then, for every positive integer $s$, there exist points $p_1,\dots,p_s\in\mathbb{CP}^1$ such that
\begin{equation}
\operatorname{Stab}_{\widetilde G}(p_t)=\{e\}
\end{equation}
for every $t$, and their $\widetilde G$-orbits are pairwise disjoint.
\label{free-disjoint-projective-orbits}
\end{lemma}

\begin{proof}
Every nonidentity element of $\widetilde G$ is represented by a nonscalar $2\times2$ matrix. Such a matrix fixes at most two points of $\mathbb{CP}^1$, namely its eigenlines. Because $\widetilde G$ is finite, the union
\begin{equation}
\mathcal F=\bigcup_{\widetilde g\neq e}\operatorname{Fix}(\widetilde g)
\end{equation}
is finite. Choose $p_1\in\mathbb{CP}^1\setminus\mathcal F$. Inductively, after choosing $p_1,\dots,p_r$, exclude
\begin{equation}
\mathcal F\cup\bigcup_{t=1}^r\widetilde G\cdot p_t.
\end{equation}
This is still a finite set, whereas $\mathbb{CP}^1$ is infinite, so we may choose $p_{r+1}$ outside it. Since $p_{r+1}\notin\mathcal F$, its stabilizer is trivial, and since it is not in any previously chosen orbit, its orbit is disjoint from all the previous ones. Continuing inductively proves the result.
\end{proof}

\begin{theorem}
Let $G$ be a finite subgroup of $SU(2)$, and denote
\begin{equation}
\widetilde G=G/(G\cap\{I,-I\}).
\end{equation}
For any $1\leq\ell\leq n$, let
\begin{equation}
s_\ell=\left\lceil\frac{2\ell+1}{|\widetilde G|}\right\rceil,
\label{levelwise-seed-count}
\end{equation}
and choose $s_\ell$ points $p_t=[a_t:b_t]\in\mathbb{CP}^1$ such that
\begin{equation}
\operatorname{Stab}_{\widetilde G}(p_t)=\{e\},
\qquad
(\widetilde G\cdot p_t)\cap(\widetilde G\cdot p_j)=\varnothing
\quad\text{for every }j\neq t,
\label{projective-seed-conditions}
\end{equation}
whose existence is guaranteed by Lemma~\ref{free-disjoint-projective-orbits}. Choose normalized representatives with $|a_t|^2+|b_t|^2=1$ and define
\begin{equation}
U_t=
\begin{pmatrix}
a_t&-\overline b_t\\
b_t&\overline a_t
\end{pmatrix}
\in SU(2).
\label{projective-point-rotation}
\end{equation}
Then
\begin{equation}
\operatorname{span}\left\{
g\cdot\left(R_n(U_t)J_+^\ell R_n(U_t)^\dagger\right):
g\in G,\ 1\leq t\leq s_\ell
\right\}
=\mathcal T_n^\ell.
\label{rotated-highest-weight-orbit-span}
\end{equation}
Consequently, if $P$ is a $G$-invariant code projector, satisfying the detection conditions for the operators
\begin{equation}
O_{\ell,t}=R_n(U_t)J_+^\ell R_n(U_t)^\dagger,
\qquad 1\leq t\leq s_\ell,
\label{rotated-highest-weight-seeds}
\end{equation}
is sufficient to detect every operator in $\mathcal T_n^\ell$.
\label{minimum-seeds-su2}
\end{theorem}

\begin{proof}

By Eq.~\eqref{lie-algebra-to-polynomial-isomorphism} and our normalization $T_\ell^{(\ell)}=J_+^\ell$, the $SU(2)$-equivariant isomorphism $\Phi=\Phi_n^\ell:\mathcal T_n^\ell\to\mathcal P_\ell=\operatorname{Sym}^{2\ell}(\mathbb C^2)$ satisfies $\Phi(J_+^\ell)=x^{2\ell}$.

Now, let $U=U(a,b)\in SU(2)$ in the convention above, with $|a|^2+|b|^2=1$. Its action on the basis vectors is $U\cdot x=ax+by$ and $U\cdot y=-\overline b x+\overline a y$, and thus $U\cdot x^{2\ell}=(ax+by)^{2\ell}$ from the action of the representation on the symmetric power. Note that this is equivalent to rotating $J_+^\ell$ by $U$ in the physical space, since by $SU(2)$-equivariance
\begin{equation}
\Phi(U\cdot J_+^\ell)
=U\cdot\Phi(J_+^\ell)
=U\cdot x^{2\ell}
=(ax+by)^{2\ell}.
\label{rotated-highest-weight-polynomial}
\end{equation}

We now have $s_\ell$ disjoint orbits, each with $|\widetilde G|$ elements by the orbit--stabilizer theorem. We denote by
\begin{equation}
\mathcal O_{\ell,G}
=\{\widetilde g\cdot[a_t:b_t]:1\leq t\leq s_\ell,\ \widetilde g\in\widetilde G\}
\end{equation}
the union of all these orbits, and choose any $2\ell+1$ elements $\{[\widetilde a_i:\widetilde b_i]\}$ to form the set $\widetilde{\mathcal O}_{\ell,G}$. We're going to show these are linearly independent, and thus they span the entire operator space. Indeed, suppose that
\begin{equation}
\begin{aligned}
\sum_{[\widetilde a_i:\widetilde b_i]\in\widetilde{\mathcal O}_{\ell,G}}
c_i(\widetilde a_i x+\widetilde b_i y)^{2\ell}=0
&\Longrightarrow
\sum_{[\widetilde a_i:\widetilde b_i]\in\widetilde{\mathcal O}_{\ell,G}}
c_i\sum_{r=0}^{2\ell}\binom{2\ell}{r}
\widetilde a_i^{2\ell-r}\widetilde b_i^r y^r x^{2\ell-r}=0\\
&\Longrightarrow
\sum_{[\widetilde a_i:\widetilde b_i]\in\widetilde{\mathcal O}_{\ell,G}}
c_i\widetilde a_i^{2\ell-r}\widetilde b_i^r=0
\qquad\text{for every }0\leq r\leq2\ell.
\end{aligned}
\label{coherent-power-linear-relation}
\end{equation}

Now, we'll briefly assume that $\widetilde a_i\neq0$ for every $i$. In that case, define $\widetilde z_i=\widetilde b_i/\widetilde a_i$ and rescale $\widetilde c_i=c_i\widetilde a_i^{2\ell}$. The equation becomes
\begin{equation}
\begin{pmatrix}
1&1&\dots&1\\
\widetilde z_1&\widetilde z_2&\dots&\widetilde z_{2\ell+1}\\
\widetilde z_1^2&\widetilde z_2^2&\dots&\widetilde z_{2\ell+1}^2\\
\vdots&\vdots&\ddots&\vdots\\
\widetilde z_1^{2\ell}&\widetilde z_2^{2\ell}&\dots&\widetilde z_{2\ell+1}^{2\ell}
\end{pmatrix}
\begin{pmatrix}
\widetilde c_1\\
\widetilde c_2\\
\widetilde c_3\\
\vdots\\
\widetilde c_{2\ell+1}
\end{pmatrix}
=0.
\label{transposed-Vandermonde-system}
\end{equation}
This matrix is the transpose of the standard square Vandermonde matrix
\begin{equation}
V=
\begin{pmatrix}
1&\widetilde z_1&\widetilde z_1^2&\dots&\widetilde z_1^{2\ell}\\
1&\widetilde z_2&\widetilde z_2^2&\dots&\widetilde z_2^{2\ell}\\
\vdots&\vdots&\ddots&\vdots\\
1&\widetilde z_{2\ell+1}&\widetilde z_{2\ell+1}^2&\dots&\widetilde z_{2\ell+1}^{2\ell}
\end{pmatrix},
\label{standard-Vandermonde-matrix}
\end{equation}
whose determinant is given by
\begin{equation}
\det V=\prod_{1\leq i<j\leq2\ell+1}(\widetilde z_j-\widetilde z_i).
\label{Vandermonde-determinant}
\end{equation}
Since the $\widetilde z_i$ are distinct, we have
\begin{equation}
\det V^T=\det V\neq0.
\end{equation}
Therefore, the only element of the kernel is the all-zero vector, $\widetilde c_i=0$, which implies $c_i=0$ since $\widetilde a_i\neq0$. Thus these elements are all linearly independent and span $\operatorname{Sym}^{2\ell}(\mathbb C^2)$, and therefore $\mathcal T_n^\ell$.

We can apply the same idea if there exists an element $[0:\widetilde b_i]$. Note that there exists only one such equivalence class. Assume without loss of generality that $\widetilde a_{2\ell+1}=0$. Then the kernel equation becomes, after doing the same rescaling for every $i<2\ell+1$,
\begin{equation}
\begin{pmatrix}
1&1&\dots&1&0\\
\widetilde z_1&\widetilde z_2&\dots&\widetilde z_{2\ell}&0\\
\widetilde z_1^2&\widetilde z_2^2&\dots&\widetilde z_{2\ell}^2&0\\
\vdots&\vdots&\ddots&\vdots&\vdots\\
\widetilde z_1^{2\ell-1}&\widetilde z_2^{2\ell-1}&\dots&\widetilde z_{2\ell}^{2\ell-1}&0\\
\widetilde z_1^{2\ell}&\widetilde z_2^{2\ell}&\dots&\widetilde z_{2\ell}^{2\ell}&\widetilde b_{2\ell+1}^{2\ell}
\end{pmatrix}
\begin{pmatrix}
\widetilde c_1\\
\widetilde c_2\\
\widetilde c_3\\
\vdots\\
\widetilde c_{2\ell}\\
c_{2\ell+1}
\end{pmatrix}
=0.
\label{Vandermonde-system-point-at-infinity}
\end{equation}
Taking the determinant down the last column, we see
\begin{equation}
\det V'=\widetilde b_{2\ell+1}^{2\ell}\det V'_{2\ell},
\end{equation}
where $V'_{2\ell}$ is the same matrix with the $(2\ell+1)$st row and column removed. But that's a $2\ell\times2\ell$ Vandermonde determinant, and since $\widetilde b_{2\ell+1}\neq0$, we have that $\det V'\neq0$. So, in both cases, we conclude that the set of operators is linearly independent.

Finally, applying $\Phi^{-1}$ to the binomial expansion of $(a_tx+b_ty)^{2\ell}$ gives the decomposition of each seed into elementary tensors:
\begin{equation}
O_{\ell,t}
=R_n(U_t)J_+^\ell R_n(U_t)^\dagger
=\sum_{\mu=-\ell}^{\ell}
\sqrt{\binom{2\ell}{\ell+\mu}}
a_t^{\ell+\mu}b_t^{\ell-\mu}T_\mu^{(\ell)}.
\label{rotated-seed-tensor-expansion}
\end{equation}
The final result about error detection follows directly from Lemma~\ref{orbit-KL-extensions}.

\end{proof}

\subsection{Levelwise minimality at fixed tensor rank}

\begin{lemma}
Let $N=G\cap\{I,-I\}$ and $\widetilde G=G/N$. For each irrep $U$, let $m_U^{(\ell)}$ denote its multiplicity in $\mathcal T_n^\ell$, and let $d_U=\dim U$. Then
\begin{equation}
\max_{U\in\operatorname{Irr}(G)}
\left\lceil\frac{m_U^{(\ell)}}{d_U}\right\rceil
=
\max_{U\in\operatorname{Irr}(\widetilde G)}
\left\lceil\frac{m_U^{(\ell)}}{d_U}\right\rceil.
\label{normal-subgroup-irreps-su2-equation}
\end{equation}
\label{normal-subgroup-irreps-su2}
\end{lemma}

\begin{proof}
The action of $G$ on $\mathcal T_n^\ell$ is given by $g\cdot O=R_n(g)OR_n(g)^\dagger$. We have that $-I$ acts on $H_n$ as $R_n(-I)=(-1)^nI_{H_n}$, and thus, for any operator $O\in\mathcal T_n^\ell$, we have $R_n(-I)OR_n(-I)^\dagger=O$. Therefore, $N$ acts trivially on $\mathcal T_n^\ell$, and the representation factors through $\widetilde G=G/N$. Representations of $\widetilde G$ are naturally identified with representations of $G$ on which $N$ acts trivially~\cite[p.~20, footnote~1]{fultonRepresentationTheory2004}. Corresponding irreps have the same dimensions and multiplicities in $\mathcal T_n^\ell$, while every other irrep of $G$ has multiplicity zero. Therefore, the nonzero ratios $m_U^{(\ell)}/d_U$ are the same in the two maxima, which proves the result.
\end{proof}

\begin{corollary}
Let $G$ be a finite subgroup of $SU(2)$, let $\widetilde G=G/(G\cap\{I,-I\})$, and let $1\leq\ell\leq n$. For each irrep $U$, let $m_U^{(\ell)}$ denote its multiplicity in $\mathcal T_n^\ell$, and let $d_U=\dim U$. Then the number of seeds in Theorem~\ref{minimum-seeds-su2} is minimal. In particular,
\begin{equation}
\max_{U\in\operatorname{Irr}(G)}
\left\lceil\frac{m_U^{(\ell)}}{d_U}\right\rceil
=
\max_{U\in\operatorname{Irr}(\widetilde G)}
\left\lceil\frac{m_U^{(\ell)}}{d_U}\right\rceil
=
\left\lceil\frac{2\ell+1}{|\widetilde G|}\right\rceil.
\label{su2-minimum-matches-general-equation}
\end{equation}
\label{su2-minimum-matches-general}
\end{corollary}

\begin{proof}
Let $s_{\min}$ denote the minimum number of seeds whose $G$-orbits span $\mathcal T_n^\ell$. By Theorem~\ref{spanning-set-G-orbit} and Lemma~\ref{normal-subgroup-irreps-su2},
\begin{equation}
s_{\min}
=\max_{U\in\operatorname{Irr}(G)}
\left\lceil\frac{m_U^{(\ell)}}{d_U}\right\rceil
=\max_{U\in\operatorname{Irr}(\widetilde G)}
\left\lceil\frac{m_U^{(\ell)}}{d_U}\right\rceil.
\end{equation}
It follows that $m_U^{(\ell)}\leq s_{\min}d_U$ for every irrep $U$ of $\widetilde G$. Therefore,
\begin{equation}
2\ell+1
=\dim\mathcal T_n^\ell
=\sum_{U\in\operatorname{Irr}(\widetilde G)}m_U^{(\ell)}d_U
\leq s_{\min}\sum_{U\in\operatorname{Irr}(\widetilde G)}d_U^2
=s_{\min}|\widetilde G|,
\end{equation}
where the final equality follows from $|\widetilde G|=\sum_{U\in\operatorname{Irr}(\widetilde G)}d_U^2$. Thus,
\begin{equation}
s_{\min}\geq\left\lceil\frac{2\ell+1}{|\widetilde G|}\right\rceil.
\end{equation}
Conversely, Theorem~\ref{minimum-seeds-su2} constructs $\lceil(2\ell+1)/|\widetilde G|\rceil$ seeds whose $G$-orbits span $\mathcal T_n^\ell$, and therefore
\begin{equation}
s_{\min}\leq\left\lceil\frac{2\ell+1}{|\widetilde G|}\right\rceil.
\end{equation}
Combining the two inequalities proves the result.
\end{proof}

\begin{remark}
The construction is rank-preserving: every seed in Eq.~\eqref{rotated-highest-weight-seeds} is an $SU(2)$ rotation of $J_+^\ell$ and therefore remains in $\mathcal T_n^\ell$. Corollary~\ref{su2-minimum-matches-general} proves minimality separately at each tensor rank. It does not prove that the combined number of seeds is globally minimal for $\bigoplus_{\ell=0}^k\mathcal T_n^\ell$, since a seed with components in several tensor ranks could generate more than one level at once.
\end{remark}

\subsection{Finite-group code-construction procedure}

We can now summarize the construction of a general $G$-invariant loss-correcting code. Equivalently, this constructs an arbitrary subrepresentation of $H_n$ that satisfies the desired KL conditions.

\begin{enumerate}
\item \textbf{Choose the group and decompose the physical representation.} Choose a finite subgroup $G\subset SU(2)$. For each irrep $C$ that appears in $H_n$ with multiplicity space $M_C$, construct a $G$-equivariant isometry
\begin{equation}
S_C:M_C\otimes C\longrightarrow H_n
\end{equation}
that embeds the $C$-isotypic component. The goal is to find Hermitian projectors $P_C$ on the multiplicity spaces $M_C$ such that
\begin{equation}
P=\sum_C S_C(P_C\otimes I_C)S_C^\dagger
\label{constructed-general-invariant-projector}
\end{equation}
detects the desired loss-error space. If $\kappa_C=\operatorname{rank}P_C$, then the resulting code carries the $G$-representation
\begin{equation}
\bigoplus_C C^{\oplus\kappa_C},
\qquad
d_L=\operatorname{rank}P=\sum_C\kappa_Cd_C,
\label{constructed-code-representation-content}
\end{equation}
where $d_C=\dim C$. Thus a prescribed logical dimension can be imposed through the ranks $\kappa_C$.

\item \textbf{Reduce the loss-error space.} To correct loss up to order $k$, we need to satisfy the error-detecting conditions for
\begin{equation}
\mathcal K_k\cong\mathcal T_n^k\oplus\mathcal T_n^{k-1}\oplus\dots\oplus\mathcal T_n^0.
\end{equation}
Thus, it suffices to impose the error-detecting conditions separately on each $\mathcal T_n^\ell$, with $0\leq\ell\leq k$. This follows from Lemmas~\ref{loss-kl-representation} and \ref{lower-loss-KL-inclusion}, together with Lemma~\ref{Clebsch-gordan} and Lemma~\ref{su2-self-dual}.

\item \textbf{Determine the number of seeds.} For each $1\leq\ell\leq k$, define
\begin{equation}
s_\ell=\left\lceil\frac{2\ell+1}{|\widetilde G|}\right\rceil,
\qquad
\widetilde G=G/(G\cap\{I,-I\}),
\end{equation}
which, by Theorem~\ref{minimum-seeds-su2} and Corollary~\ref{su2-minimum-matches-general}, is the minimum number of seeds whose $G$-orbits can span $\mathcal T_n^\ell$.

\item \textbf{Choose the seed operators.} Choose rotations $U_{\ell,t}\in SU(2)$ as in Theorem~\ref{minimum-seeds-su2}, and define
\begin{equation}
O_{\ell,t}=R_n(U_{\ell,t})J_+^\ell R_n(U_{\ell,t})^\dagger,
\qquad 1\leq t\leq s_\ell.
\end{equation}
The $G$-orbits of these seed operators span $\mathcal T_n^\ell$. Therefore, by Lemma~\ref{orbit-KL-extensions}, it suffices to impose the KL conditions on the operators $O_{\ell,t}$.

\item \textbf{Determine the multiplicity-space projectors.} Find the projectors $P_C$ by solving the KL conditions. This can be done in either of the following two ways.

\begin{enumerate}[(a)]
\item \textbf{Direct formulation.} Solve
\begin{equation}
PO_{\ell,t}P=\lambda_{\ell,t}P
\end{equation}
for the projectors $P_C$ and scalars $\lambda_{\ell,t}$, for every $1\leq\ell\leq k$ and $1\leq t\leq s_\ell$, subject to any prescribed ranks $\kappa_C$.

\item \textbf{Using multiplicity matrices.} Take the $G$-invariant KL-operator space to be
\begin{equation}
W=\bigoplus_{\ell=0}^k\mathcal T_n^\ell
\end{equation}
and construct the matrices $B_{j,\alpha}^{U;C,D}$ as above. By Theorem~\ref{general-invariant-code-KL}, the necessary and sufficient conditions are
\begin{equation}
P_CB_{j,\alpha}^{U;C,D}P_D
=
\begin{cases}
\lambda_\alpha P_C, & C=D,\ U=\mathbf 1,\ j=1,\\
0, & \text{otherwise},
\end{cases}
\label{construction-general-invariant-KL-conditions}
\end{equation}
for every $U,C,D,\alpha,j$. In particular, these equations include all cross-irrep blocks. The scalar $\lambda_\alpha$ must be the same for every $C$ with $P_C\neq0$.
\end{enumerate}

\item \textbf{Construct the physical codewords.} For each $C$, choose an orthonormal basis $\{\ket{s_{C,a}}:1\leq a\leq\kappa_C\}$ of $\operatorname{ran}P_C$ and an orthonormal basis $\{\ket{c_{C,\mu}}:1\leq\mu\leq d_C\}$ of $C$. The physical codewords are
\begin{equation}
\ket{\overline{C,a,\mu}}
=S_C(\ket{s_{C,a}}\otimes\ket{c_{C,\mu}}),
\qquad
1\leq a\leq\kappa_C,
\quad
1\leq\mu\leq d_C.
\label{constructed-general-physical-codewords}
\end{equation}
\end{enumerate}

For the single-irrep constructions used below, only one irrep $C$ is occupied and $P_C=P_S=\ket{s}\bra{s}$ has rank one, where $\ket{s}\in M_C$ is normalized. In this case,
\begin{equation}
P=S_C(P_S\otimes I_C)S_C^\dagger
\label{constructed-single-irrep-projector}
\end{equation}
defines one copy of $C$. The direct formulation reduces to solving
\begin{equation}
PO_{\ell,t}P=\lambda_{\ell,t}P
\end{equation}
for every $1\leq\ell\leq k$ and $1\leq t\leq s_\ell$. Alternatively, for each seed $O_{\ell,t}$, expand $S_C^\dagger O_{\ell,t}S_C$ in the representation-factor basis $F_{j,\mu}^{U;C,C}$ and denote the corresponding multiplicity matrices by $\widetilde B_{j,\mu}^U(O_{\ell,t})$. For every shared nontrivial irrep $U$, they are calculated by
\begin{equation}
\widetilde B_{j,\mu}^U(O_{\ell,t})
=\operatorname{Tr}_C\!\left[
\left(I_{M_C}\otimes(F_{j,\mu}^{U;C,C})^\dagger\right)
S_C^\dagger O_{\ell,t}S_C
\right].
\label{seed-specific-multiplicity-matrix}
\end{equation}
The necessary and sufficient conditions for every shared nontrivial irrep are
\begin{equation}
\bra{s}\widetilde B_{j,\mu}^U(O_{\ell,t})\ket{s}=0
\end{equation}
for every $j$ and $\mu$, together with $\braket{s|s}=1$. This is the rank-one specialization of Corollary~\ref{single-isotypic-KL}. The trivial-irrep condition is automatic because $P_S$ has rank one. Once a normalized vector $\ket{s}$ satisfying these conditions is found, choose an orthonormal basis $\{\ket{c_\mu}\}$ of $C$ and define
\begin{equation}
\ket{\overline\mu}=S_C(\ket{s}\otimes\ket{c_\mu}).
\label{constructed-physical-codewords}
\end{equation}

\begin{remark}
The rank-zero component $\mathcal T_n^0$ is proportional to the identity and gives no additional condition.
\end{remark}

\begin{remark}
Theorem~\ref{minimum-seeds-su2} attains the minimum number of seeds necessary, which agrees with the general minimum in Theorem~\ref{spanning-set-G-orbit} by Corollary~\ref{su2-minimum-matches-general}.
\end{remark}

\begin{remark}
If a single occupied irrep $C$ has $m_C=1$, there is no multiplicity vector to find, so the unique copy of $C$ either satisfies the seed conditions or it does not.
\end{remark}

\section{Hermitian reductions and the constraint-balanced polyhedral family}

\subsection{Hermitian adapted operator bases}

We now prove a lemma that will be crucial in simplifying our constraint count and supports our conjectured construction.

Throughout this section, when a single irrep $C$ is fixed, we suppress the repeated $C,C$ labels and write $F_{j,\mu}^U=F_{j,\mu}^{U;C,C}$, $B_{j,\alpha}^U=B_{j,\alpha}^{U;C,C}$, and $r_U=r_U^{C,C}$.

\begin{lemma}
Suppose the $G$-invariant operator space $W$ satisfies $W^\dagger=W$, and write $W\cong\bigoplus_U U^{\oplus m_U}$ as usual. Fix an irrep $C$, and assume that every irrep shared between $W$ and $\operatorname{End}(C)$ is real. Then, for each such shared irrep $U$, there exist Hermitian, $G$-adapted operator bases $T_\mu^{U,\alpha}\in W$ and Hermitian representation-factor bases $F_{j,\mu}^U\in\operatorname{End}(C)$ such that
\begin{equation}
S_C^\dagger T_\mu^{U,\alpha}S_C
=\sum_{j=1}^{r_U}B_{j,\alpha}^U\otimes F_{j,\mu}^U,
\qquad
(B_{j,\alpha}^U)^\dagger=B_{j,\alpha}^U,
\label{Hermitian-adapted-B-decomposition}
\end{equation}
where $S_C:M_C\otimes C\to V$ is the usual isometric embedding and $r_U=\operatorname{mult}_U(\operatorname{End}(C))$. For nontrivial $U$, the $F_{j,\mu}^U$ can also be chosen Hilbert--Schmidt orthonormal. For $U=\mathbf 1$, we retain the normalization $F_{1,1}^{\mathbf 1}=I_C$.
\label{real-BF-lemma}
\end{lemma}

\begin{proof}
Fix a shared real irrep $U$. Choose a basis of $U$ in which every representation matrix is real, and let $K_U:U\to U$ denote the conjugate-linear map given by conjugation in that basis. For $\Phi\in\operatorname{Hom}_G(U,\operatorname{End}(C))$, we define
\begin{equation}
\Phi^\#(u)=\Phi(K_Uu)^\dagger.
\label{adjoint-involution-on-intertwiners}
\end{equation}
We first show this is also a map in $\operatorname{Hom}_G(U,\operatorname{End}(C))$. Indeed,
\begin{equation}
\begin{aligned}
\Phi^\#(g\cdot u)
&=\Phi^\#(\pi_U(g)u)
=\Phi(K_U(\pi_U(g)u))^\dagger\\
&=\Phi(\pi_U(g)(K_Uu))^\dagger
=\left[\pi_C(g)\Phi(K_Uu)\pi_C(g)^\dagger\right]^\dagger\\
&=\pi_C(g)\Phi(K_Uu)^\dagger\pi_C(g)^\dagger
=g\cdot\Phi^\#(u),
\end{aligned}
\label{adjoint-involution-equivariance}
\end{equation}
as desired. Thus, we can decompose each intertwiner as
\begin{equation}
\Phi=\frac{\Phi+\Phi^\#}{2}
+i\frac{\Phi-\Phi^\#}{2i}.
\label{intertwiner-real-form-decomposition}
\end{equation}
We also have $(\Phi^\#)^\#=\Phi$. Thus, we can choose a basis $\{\widetilde\Phi_j^U\}_{j=1}^{r_U}$ such that $(\widetilde\Phi_j^U)^\#=\widetilde\Phi_j^U$. Now, for a real basis $\ket{\mu}$ of $U$, we define $\widetilde F_{j,\mu}^U=\widetilde\Phi_j^U(\ket{\mu})$, which will be Hermitian.

For nontrivial $U$, let $\widetilde M_j^U=M[\widetilde\Phi_j^U]$. In the previous derivation of these matrices, we saw that
\begin{equation}
\begin{aligned}
\Gamma_{ab}
&=\frac{1}{d_U}\operatorname{Tr}\left[(\widetilde M_a^U)^\dagger\widetilde M_b^U\right]\\
&=\frac{1}{d_U}\sum_{\mu=1}^{d_U}
\operatorname{Tr}\left[(\widetilde F_{a,\mu}^U)^\dagger
\widetilde F_{b,\mu}^U\right]
=\frac{1}{d_U}\sum_{\mu=1}^{d_U}
\operatorname{Tr}\left[\widetilde F_{a,\mu}^U
\widetilde F_{b,\mu}^U\right]
\in\mathbb R.
\end{aligned}
\label{Hermitian-intertwiner-Gram-matrix}
\end{equation}
Since the intertwiners are linearly independent, $\Gamma$ is positive definite, and because it is real and symmetric, $\Gamma^{-1/2}$ is real. Thus, we can use the same orthonormalization as before,
\begin{equation}
M_a^U=\sum_b\widetilde M_b^U(\Gamma^{-1/2})_{ba},
\label{Hermitian-intertwiner-orthonormalization}
\end{equation}
and because $\Gamma^{-1/2}$ is real, the corresponding $F_{j,\mu}^U$ are both Hermitian and Hilbert--Schmidt orthonormal, as desired. For the trivial irrep, the established choice $F_{1,1}^{\mathbf 1}=I_C$ is already Hermitian.

Finally, since $W=W^\dagger$, we can likewise find a Hermitian basis $T_\mu^{U,\alpha}$ for each shared real irrep $U$ by applying the exact same argument to $\operatorname{Hom}_G(U,W)$. We then write
\begin{equation}
\begin{aligned}
S_C^\dagger T_\mu^{U,\alpha}S_C
&=\sum_{j=1}^{r_U}B_{j,\alpha}^U\otimes F_{j,\mu}^U,\\
(S_C^\dagger T_\mu^{U,\alpha}S_C)^\dagger
&=\sum_{j=1}^{r_U}(B_{j,\alpha}^U)^\dagger\otimes F_{j,\mu}^U.
\end{aligned}
\label{Hermitian-B-matrix-comparison}
\end{equation}
Since the $T_\mu^{U,\alpha}$ and $F_{j,\mu}^U$ are Hermitian and the $F_{j,\mu}^U$ are linearly independent, we conclude that
\begin{equation}
B_{j,\alpha}^U=(B_{j,\alpha}^U)^\dagger,
\label{Hermitian-B-matrices}
\end{equation}
which was to be shown.
\end{proof}

\subsection{Real symmetry-reduced constraint count}

The orbit theorem minimizes the number of physical operators that must be tested and requires no reality assumptions. To count the resulting real equations on the multiplicity-space projector, we now use additional structure.

\begin{lemma}
Let
\begin{equation}
P=S_C(P_S\otimes I_C)S_C^\dagger,
\label{real-count-single-isotypic-projector}
\end{equation}
where $P_S$ is a rank-$\kappa$ projector on the multiplicity space $M_C$. Suppose that $W^\dagger=W$, let $m_U=\operatorname{mult}_U(W)$ and $r_U=\operatorname{mult}_U(\operatorname{End}(C))$, and suppose that each irrep $U$ shared between $W$ and $\operatorname{End}(C)$ is real. Then the formal real-equation count for the symmetry-reduced KL conditions is
\begin{equation}
N_{\operatorname{sym}}^{(\kappa)}
=\kappa^2\sum_{U\neq\mathbf 1}m_Ur_U
+(\kappa^2-1)m_{\mathbf 1}.
\label{real-constraint-count-equation}
\end{equation}
\label{real-constraint-count}
\end{lemma}

\begin{proof}
Indeed, from Lemma~\ref{real-BF-lemma}, we know that all the $B_{j,\alpha}^U$ can be made Hermitian. Now, we need
\begin{equation}
P_SB_{j,\alpha}^UP_S=0
\quad(U\neq\mathbf 1),
\qquad
P_SB_{1,\alpha}^{\mathbf 1}P_S=\lambda_\alpha P_S.
\label{Hermitian-reduced-KL-equations}
\end{equation}
Since $P_S$ is Hermitian, $P_SB_{j,\alpha}^UP_S$ is also Hermitian. In an orthonormal basis of $\operatorname{ran}P_S$, each such matrix has $\kappa$ real diagonal entries and $\kappa(\kappa-1)/2$ complex off-diagonal entries, giving a total of
\begin{equation}
\kappa+\kappa(\kappa-1)=\kappa^2
\label{Hermitian-matrix-real-dimension}
\end{equation}
real expressions. Thus, the condition $P_SB_{j,\alpha}^UP_S=0$ for $U\neq\mathbf 1$ gives us $\kappa^2$ equations for each $B$ matrix. The trivial-irrep equation gives us $\kappa^2-1$ equations, since we have a free constant $\lambda_\alpha$. The number of nontrivial $B$ matrices is $\sum_{U\neq\mathbf 1}m_Ur_U$, while the number of trivial $B$ matrices is $m_{\mathbf 1}$. This gives Eq.~\eqref{real-constraint-count-equation}, as desired.
\end{proof}

The equations counted in Lemma~\ref{real-constraint-count} are not required to be independent or to have a solution, since the KL conditions are quadratic. Thus, Eq.~\eqref{real-constraint-count-equation} is a formal equation count and not an existence theorem. We also get the following corollary.

\begin{corollary}
If $\kappa=1$ in the setting of Lemma~\ref{real-constraint-count}, then
\begin{equation}
N_{\operatorname{sym}}^{(1)}
=\sum_{U\neq\mathbf 1}m_Ur_U.
\label{rank-one-real-constraint-count}
\end{equation}
\label{real-constraint-count-rank-1}
\end{corollary}

To apply this count to the loss tensor spaces, we first use the following fact.

\begin{proposition}
The spaces of irreducible $SU(2)$ tensor operators $\mathcal T_n^\ell$ satisfy
\begin{equation}
(\mathcal T_n^\ell)^\dagger=\mathcal T_n^\ell.
\label{tensor-space-adjoint-closure}
\end{equation}
\label{tensor-spaces-adjoint-closed}
\end{proposition}

\begin{proof}
For any standard phase convention for the elementary tensors, $(T_\mu^{(\ell)})^\dagger$ is a nonzero scalar multiple of $T_{-\mu}^{(\ell)}$. Since $\mathcal T_n^\ell$ is the span of these irreducible tensors, it is unchanged under the adjoint, and the result follows.
\end{proof}

The preceding results suggest a scalable polyhedral family. Define
\begin{equation}
n_d=\left\lceil\frac{3d^2+1}{4}\right\rceil.
\label{n-of-d}
\end{equation}
For this photon number, define the KL-operator space
\begin{equation}
W_d=\bigoplus_{\ell=0}^{d-1}\mathcal T_{n_d}^\ell.
\label{family-KL-space}
\end{equation}
Proposition~\ref{tensor-spaces-adjoint-closed} implies that $W_d^\dagger=W_d$.
This photon-number sequence was conjectured for the odd-distance binary-dihedral family in Refs.~\cite{kubischtaPermutationInvariantQuantumCodes2025,bondPermutationinvariantCodesNumerical2026}. The following propositions identify specific $2O$ and $2T$ two-dimensional irreps for which the formal real-equation count exactly matches the free-variable count at each corresponding $(n_d,d)$. This count motivates, but does not prove, the existence of the codes.
Before we proceed with our proposed sequences, we note that for a given group (e.g., $2O, \, 2I,\, 2T$), we use $\rho_i$ to denote its $i$-th irrep. The numbering follows the ordering in the GAP-exported representation data supplied with our code. We also note that the numbering generally matches that from the supplemental material of Ref.~\cite{grossDesigningCodesInteractions2021}, with the exception that the three-dimensional irreps $\rho_4, \, \rho_5$ of $2I$ are interchanged. 

\subsection{Even-distance binary-octahedral construction}

\begin{proposition}
For even distance, we have $d=2r$ and $n_d=3r^2+1$. Take $G=2O$, and consider a two-dimensional irrep code using the irrep
\begin{equation}
C_r=
\begin{cases}
\rho_3, & r\text{ odd},\\
\rho_5, & r\equiv2\pmod 4,\\
\rho_4, & r\equiv0\pmod 4.
\end{cases}
\label{even-family-irrep-choice}
\end{equation}
The multiplicity of the irrep is
\begin{equation}
m_r=\left\lfloor\frac{r^2}{4}\right\rfloor+1,
\label{even-family-irrep-multiplicity}
\end{equation}
so a rank-one projector $P_S=\ket{s}\bra{s}$, with normalized $\ket{s}\in M_{C_r}$, has $2m_r-2$ degrees of freedom, matching the formal number of real equations $N_{\operatorname{sym}}^{(1)}$ from Corollary~\ref{real-constraint-count-rank-1}.
\label{even-distance-2O-family}
\end{proposition}

\begin{proof}
Firstly, using Gross's irrep numbering (See Supplemental material of \cite{grossDesigningCodesInteractions2021}), we have
\begin{equation}
\operatorname{End}(\rho_4)
\cong\operatorname{End}(\rho_5)
\cong\mathbf 1\oplus\rho_6,
\qquad
\operatorname{End}(\rho_3)
\cong\mathbf 1\oplus\rho_2\oplus\rho_3.
\label{2O-endomorphism-decompositions}
\end{equation}
One can check from the characters that all three of $\rho_2$, $\rho_3$, and $\rho_6$ have Frobenius--Schur indicator $1$, and thus they can all be made real, satisfying the lemma's conditions.

Thus, it suffices to show that the counts match. Firstly, $r^2$ is congruent to $0$, $1$, or $4$ modulo $8$. For each possible residue class of $r^2$, Gross's Table VII~\cite{grossDesigningCodesInteractions2021} gives
\begin{equation}
\begin{array}{c|c|c}
\text{condition} & n_d+1 & \text{relevant multiplicity}\\ \hline
r^2=8h+1 & 24h+5 & \operatorname{mult}_{\rho_3}=2h+1\\
r^2=8h+4 & 24h+14 & \operatorname{mult}_{\rho_5}=2h+2\\
r^2=8h & 24h+2 & \operatorname{mult}_{\rho_4}=2h+1
\end{array}
\label{multiplicity-table-2O-family}
\end{equation}
for a nonnegative integer $h$. So, if we solve for the multiplicity in terms of $r$, we see that in every case we get $m_r=\lfloor r^2/4\rfloor+1$. This gives us the count $2m_r-2$ of free variables in each case, since each entry of $\ket{s}$ is complex, and we subtract two degrees of freedom for normalization and an arbitrary global phase.

Now, we need to find $\operatorname{mult}_U(W_d)$ for each relevant $U\in\{\rho_2,\rho_3,\rho_6\}$. We notice that for $\rho_2$ and $\rho_3$,
\begin{equation}
\operatorname{mult}_{\rho_2}(\mathcal T_{n_d}^\ell)
+\operatorname{mult}_{\rho_3}(\mathcal T_{n_d}^\ell)
=\left\lfloor\frac{\ell+2}{4}\right\rfloor.
\label{2O-rho2-rho3-level-multiplicity}
\end{equation}
Thus, keeping in mind that in this case $r$ is odd, so $r=2p+1$ and $2r-1=4p+1$, we see that
\begin{equation}
\begin{aligned}
\operatorname{mult}_{\rho_2}(W_d)
+\operatorname{mult}_{\rho_3}(W_d)
&=\sum_{\ell=0}^{d-1}
\left[
\operatorname{mult}_{\rho_2}(\mathcal T_{n_d}^\ell)
+\operatorname{mult}_{\rho_3}(\mathcal T_{n_d}^\ell)
\right]\\
&=\sum_{\ell=0}^{2r-1}\left\lfloor\frac{\ell+2}{4}\right\rfloor
=\sum_{\ell=0}^{4p+1}\left\lfloor\frac{\ell+2}{4}\right\rfloor.
\end{aligned}
\label{2O-odd-r-constraint-sum}
\end{equation}
Each of the ranks $\ell=4j-2,4j-1,4j,4j+1$ contributes $j$ to the sum. The sum goes up to $\ell=4p+1$. So, we get
\begin{equation}
\begin{aligned}
N_{\operatorname{sym}}^{(1)}
&=4\sum_{j=0}^{p}j
=4\frac{p(p+1)}{2}
=2p(p+1)\\
&=\frac{(r-1)(r+1)}{2}
=\frac{r^2-1}{2}
=2\left(\left\lfloor\frac{r^2}{4}\right\rfloor+1\right)-2
=2m_r-2,
\end{aligned}
\label{2O-odd-r-constraint-balance}
\end{equation}
as desired.

For $\rho_6$, we have $r\equiv0\pmod 2$, $d=2r$, and so $d=4p$ and $r=2p$. From the table we see that
\begin{equation}
\sum_{s=0}^3
\operatorname{mult}_{\rho_6}(\mathcal T_{n_d}^{4j+s})
=4j+2.
\label{2O-rho6-block-multiplicity}
\end{equation}
Thus, since $d=4p$, we see
\begin{equation}
\operatorname{mult}_{\rho_6}(W_{4p})
=\sum_{j=0}^{p-1}(4j+2)
=2p(p-1)+2p
=2p^2
=2\frac{r^2}{4}
=2m_r-2,
\label{2O-even-r-constraint-balance}
\end{equation}
exactly matching our free variables, as desired.
\end{proof}

\subsection{Odd-distance binary-tetrahedral construction}

We similarly show the following.

\begin{proposition}
For odd distance $d=2r+1$, $n_d=3r(r+1)+1$. Take $G=2T$, and consider the two-dimensional irrep $C_r=\rho_4$. Its multiplicity in $H_{n_d}$ is
\begin{equation}
m_r=\frac{r(r+1)}{2}+1,
\label{odd-family-irrep-multiplicity}
\end{equation}
so a rank-one projector has $2m_r-2=r(r+1)$ real degrees of freedom, matching the formal real symmetry-reduced KL equation count from Corollary~\ref{real-constraint-count-rank-1}.
\label{odd-distance-2T-family}
\end{proposition}

\begin{proof}
We have
\begin{equation}
\operatorname{End}(\rho_4)\cong\mathbf 1\oplus\rho_7,
\label{2T-endomorphism-decomposition}
\end{equation}
and $\rho_7$ is real, as can be checked by its Frobenius--Schur indicator. From the table in Ref.~\cite{grossDesigningCodesInteractions2021}, we see that
\begin{equation}
\operatorname{mult}_{\rho_7}(\mathcal T_{n_d}^\ell)
=\left\lfloor\frac{\ell+1}{2}\right\rfloor.
\label{2T-rho7-level-multiplicity}
\end{equation}
So, we see that
\begin{equation}
\operatorname{mult}_{\rho_7}(W_{2r+1})
=\sum_{\ell=0}^{2r}\left\lfloor\frac{\ell+1}{2}\right\rfloor
=2\sum_{j=0}^{r}j
=r(r+1).
\label{2T-constraint-count}
\end{equation}

Likewise, we consider $\rho_4$ for $n_d=3r^2+3r+1$. Define $Q=r(r+1)/2$. We have $n_d=6Q+1$, and from Gross's table, we see that $\operatorname{mult}_{\rho_4}(H_{6Q+1})=Q+1$. Therefore,
\begin{equation}
2m_r-2
=2(Q+1)-2
=2Q
=r(r+1)
=\operatorname{mult}_{\rho_7}(W_{2r+1}),
\label{2T-constraint-balance}
\end{equation}
which was to be shown.
\end{proof}

\subsection{Family conjectures and logical status}

Propositions~\ref{even-distance-2O-family} and \ref{odd-distance-2T-family}, together with our constructed examples, support the following conjecture.

\begin{conjecture}[Binary-polyhedral family]
For every integer $d\geq2$, let $n_d$ and $W_d$ be defined by Eqs.~\eqref{n-of-d} and \eqref{family-KL-space}. There exists a unit multiplicity vector $\ket{s_d}\in M_{C_d}$ such that the projector
\begin{equation}
P_d
=S_{C_d,n_d}\left(\ket{s_d}\bra{s_d}\otimes I_{C_d}\right)
S_{C_d,n_d}^\dagger
\label{binary-polyhedral-family-projector}
\end{equation}
detects $W_d$, where
\begin{equation}
S_{C_d,n_d}:M_{C_d}\otimes C_d\longrightarrow H_{n_d}
\label{family-isotypic-embedding}
\end{equation}
is an isometric embedding of the $C_d$-isotypic component and
\begin{equation}
(G_d,C_d)=
\begin{cases}
(2O,C_r), & d=2r,\\
(2T,\rho_4), & d=2r+1.
\end{cases}
\label{binary-polyhedral-family-data}
\end{equation}
Here $C_r$ is the binary-octahedral irrep defined in Proposition~\ref{even-distance-2O-family}. Consequently, $P_d$ is a rank-two loss code of distance at least $d$.
\label{binary-polyhedral-family-conjecture}
\end{conjecture}

\begin{conjecture}[Minimum photon number]
For every integer $d\geq2$, the value $n_d$ in Eq.~\eqref{n-of-d} is the minimum photon number at which a rank-two loss code of distance at least $d$ exists.
\label{binary-polyhedral-minimum-photon-conjecture}
\end{conjecture}

The minimum-photon statement matches the conjecture of Ref.~\cite{bondPermutationinvariantCodesNumerical2026}. A similar counting argument was also made in Ref.~\cite{kubischtaPermutationInvariantQuantumCodes2025} for the respective odd-distance binary-dihedral codes. Of course, matching the formal number of KL equations to the number of free variables does not guarantee that the nonlinear equations have a solution. Nevertheless, we have numerically constructed the corresponding members through $d=10$, up to the KL tolerances reported below; this finite set is evidence for the family-existence conjecture. Lastly, the uniform odd-distance family is formulated using $2T:\rho_4$, but we also see an occasional $2I$ symmetry enhancement, as noted below.

\section{Constructed codes and numerical classification}

\subsection{Construction and verification conventions}

For each member of the binary-polyhedral family, we fix the group $G_d$ and irrep $C_d$ in Eq.~\eqref{binary-polyhedral-family-data}, construct the isotypic embedding $S_{C_d,n_d}$, and solve the symmetry-reduced KL equations for a rank-one multiplicity projector. We impose the equations rank by rank using the levelwise-minimal seeds constructed in Theorem~\ref{minimum-seeds-su2}, and use any remaining degrees of freedom to minimize the KL residual at the first uncorrected rank. This is a nonconvex, multistart numerical optimization, so failure to solve the next rank is not by itself a proof that no other code can do so.

After rescaling each rank-$\ell$ multiplet introduced above by its common Hilbert--Schmidt norm, we retain the notation $\{T_\mu^{(\ell)}:-\ell\leq\mu\leq\ell\}$ for the resulting Hilbert--Schmidt-orthonormal basis. After constructing the physical rank-two projector $P$, we verify the KL conditions against this complete basis at every corrected tensor rank, rather than only against the orbit seeds. We report the maximum normalized residual
\begin{equation}
\epsilon_{\mathrm{KL}}(P;d)
=\max_{\substack{0\leq\ell<d\\-\ell\leq\mu\leq\ell}}
\norm{
P T_\mu^{(\ell)}P
-\frac{\Tr\!\left(P T_\mu^{(\ell)}\right)}{2}P
}_{\mathrm F},
\qquad
\norm{T_\mu^{(\ell)}}_{\mathrm F}=1,
\label{numerical-KL-residual}
\end{equation}
where $\norm{A}_{\mathrm F}=\sqrt{\Tr(A^\dagger A)}$ is the Frobenius norm. Double-precision reconstructions of all the family members in Table~\ref{tab:constraint-balanced-family-codes} satisfy $\epsilon_{\mathrm{KL}}(P;d)<10^{-15}$. For each code, the maximum residual at the first uncorrected tensor rank $\ell=d$ lies between $7.3\times10^{-2}$ and $6.1\times10^{-1}$. For $d=2$ and $d=3$, the KL conditions vanish exactly by the representation-theoretic construction, so their residuals provide a useful double-precision baseline. The comparable residuals for the higher-distance members are consistent with algebraic vanishing obscured by floating-point roundoff, but do not by themselves constitute rigorous certificates.

Here and in the Letter, ``exact'' and ``constructed'' refer to codes obtained through the finite-group isotypic construction, whose projector form and Knill--Laflamme reduction are exact. Unless a separate algebraic verification is stated, a particular code's claimed loss distance is checked numerically using the expression  Eq.~\eqref{numerical-KL-residual}. Most constructions have normalized KL residuals at the level of double-precision roundoff, and all constructions have residuals below $3.3\times10^{-13}$. The exact MacWilliams--Farkas bounds reported below have a different logical status: they are rigorous upper bounds on the attainable distance (we obtain exact rational certificates), whereas numerical KL satisfaction supplies numerical evidence for the corresponding lower bound.

\subsection{Constructed family members through distance ten}

Table~\ref{tab:constraint-balanced-family-codes} collects the members of Conjecture~\ref{binary-polyhedral-family-conjecture} that we have constructed through distance $10$. The Farkas bound $d_{\mathrm F}$ is the largest distance permitted by the exact MacWilliams LP at the displayed $n_d$. Thus, when $d=d_{\mathrm F}$, the numerical construction attains the exact fixed-$n_d$ upper bound up to the KL residual reported in the table. An exact algebraic verification of the KL conditions is additionally required to turn this numerical saturation into a rigorous fixed-$n_d$ optimality statement.

\begin{table}[tb]
    \centering
    \small
    \renewcommand{\arraystretch}{1.12}
    \setlength{\tabcolsep}{4pt}
    \begin{tabular}{c c l l c c}
        \hline
        $d$ & $n_d$ & \shortstack{Family\\construction} & \shortstack{Additional\\realization} & $\epsilon_{\mathrm{KL}}$ & $d_{\mathrm F}$ \\
        \hline
        $2$  & $4$  & $2O:\rho_3$ & --- & $2.886\times10^{-16}$ & $2$ \\
        $3$  & $7$  & $2T:\rho_4$ & Same projector: $2I:\rho_3$ & $9.734\times10^{-16}$ & $3$ \\
        $4$  & $13$ & $2O:\rho_5$ & --- & $2.958\times10^{-16}$ & $4$ \\
        $5$  & $19$ & $2T:\rho_4$ & --- & $8.156\times10^{-16}$ & $5$ \\
        $6$  & $28$ & $2O:\rho_3$ & --- & $3.249\times10^{-16}$ & $6$ \\
        $7$  & $37$ & $2T:\rho_4$ & \shortstack[l]{Same projector: $2I:\rho_3$;\\distinct solution: $2O:\rho_5$} & $4.722\times10^{-16}$ & $7$ \\
        $8$  & $49$ & $2O:\rho_4$ & --- & $2.265\times10^{-16}$ & $8$ \\
        $9$  & $61$ & $2T:\rho_4$ & --- & $2.285\times10^{-16}$ & $9$ \\
        $10$ & $76$ & $2O:\rho_3$ & --- & $5.585\times10^{-16}$ & $11$ \\
        \hline
    \end{tabular}
    \caption{Numerically constructed members of the constraint-balanced binary-polyhedral families. The residual $\epsilon_{\mathrm{KL}}$ is defined in Eq.~\eqref{numerical-KL-residual}. For the odd-family members at $n=7$ and $n=37$, the quoted residuals were evaluated using the saved $2I:\rho_3$ realizations. The exact LP upper bound is attained numerically through $d=9$; the distance-$10$ construction at $n=76$ lies one below its LP bound.}
    \label{tab:constraint-balanced-family-codes}
\end{table}
\FloatBarrier

The $2T$ family can occasionally acquire a larger symmetry. We can also construct the $n=7$ and $n=37$ projectors from the irrep $\rho_3$ of $2I$. The separately constructed $2O:\rho_5$ code at $n=37$ is a distinct distance-$7$ solution. We nevertheless formulate the uniform odd-distance family using $2T:\rho_4$, since this is the symmetry and irrep for which the constraint count of Proposition~\ref{odd-distance-2T-family} applies at every $r$.

The equalities $d=d_{\mathrm F}$ in Table~\ref{tab:constraint-balanced-family-codes} concern optimality only at the displayed fixed $n_d$; they do not prove that $n_d$ is the smallest photon number attaining distance $d$. Indeed, the MacWilliams relaxation first permits distance $6$ at $n=26$ and distance $8$ at $n=43$, below $n_6=28$ and $n_8=49$. At $n=76$, our binary-polyhedral construction has $d=10$, while the MacWilliams LP remains feasible at $d=11$ and is exactly infeasible at $d=12$, giving $d_{\mathrm F}=11$.

\subsection{Exact algebraic examples at photon numbers 13 and 28}

\begin{itemize}
    \item $n=13$, $2O$, $d=4$ construction. 
    
    We first show how Gross's binary-octahedral (2O) $n = 13$ code can be made $d=4$ with an appropriate choice of the free phase $\phi$. Indeed, Gross defines the 2O generators as \cite{grossDesigningCodesInteractions2021}
\begin{equation}
    \begin{aligned}
        S &= \frac{1}{\sqrt{2}}(I_2 - i \sigma_z) = \begin{pmatrix}
        e^{-i \pi/4} & 0 \\ 
        0& e^{i \pi/4}
    \end{pmatrix} \\ 
    H &= \frac{1}{\sqrt{2}}(-i \sigma_x - i \sigma_z) = -\frac{i}{\sqrt{2}} \begin{pmatrix}
        1 & 1 \\
        1 & -1
    \end{pmatrix}
    \end{aligned} \label{2O-generators}
\end{equation}

The qubit code is defined by the codeword (in the two-mode Fock basis)
\begin{equation}
    \begin{aligned}
        \ket{\overline{0}_{5/2}} &= \frac{\sqrt{910}}{56} \ket{13,0} - \frac{3 \sqrt{154}}{56} \ket{9,4} - \frac{\sqrt{770}}{56} \ket{5,8} + \frac{\sqrt{70}}{56} \ket{1,12} \\ 
        \ket{\overline{0}_{-13/6}} &= \frac{\sqrt{231}}{84} \ket{13,0} + \frac{\sqrt{1365}}{84} \ket{9, 4} - \frac{\sqrt{273}}{28} \ket{5,8} - \frac{\sqrt{3003}}{84} \ket{1,12} \\ 
        \ket{\overline{0}_{\phi}} &= \frac{\sqrt{91}}{14} \ket{\overline{0}_{5/2}} + e^{i\phi} \frac{\sqrt{105}}{14} \ket{\overline{0}_{-13/6}}, 
    \end{aligned}
\end{equation}
where we've corrected an interchange of the subscripts and corresponding prefactors in Ref.~\cite{grossDesigningCodesInteractions2021}, as follows directly from the $J_z$ expectation values. With this correction, the code satisfies the error-detection conditions up to rank $2$. Note that the logical $1$ state here is obtained by simply exchanging the Fock number occupations in the two modes. 
We will first show that enforcing the error-detection condition on $J_+^3$ suffices to impose the full rank-$3$ condition. Let $\widetilde G=2O/\{I,-I\}\cong O$. Under the identification in Theorem~\ref{minimum-seeds-su2}, $J_+^3$ is represented by $x^6$, and hence by $[1:0]\in\mathbb{CP}^1$. One can check that the projective orbit of $[1:0]$ has $6$ distinct elements, namely 
\begin{equation}
   \widetilde G \cdot [1:0] = \left\{ [0:1],\, [1:-i],\, [1:1],\, [1:-1],\, [1:i],\, [1:0] \right\}.
\end{equation}
So, $[1:0]$ clearly has a nontrivial $\widetilde G$ stabilizer. Nevertheless, it suffices to impose the KL condition. Indeed, the six corresponding coherent powers are linearly independent by the Vandermonde argument in the proof of Theorem~\ref{minimum-seeds-su2}, so the orbit span has dimension $6$. The rank-$3$ space satisfies $\mathcal{T}_{13}^3 \cong \rho_2 \oplus \rho_6 \oplus \rho_7$ in Gross's notation, where $\dim \rho_2 = 1$ and $\dim \rho_6 = \dim \rho_7 = 3$. Therefore, the orbit span is isomorphic to $\rho_6 \oplus \rho_7$. Since the code carries $\rho_5$, its operator space satisfies $\operatorname{End}(\rho_5) \cong \rho_1 \oplus \rho_6$. So, the only irrep shared by $\mathcal T_{13}^3$ and $\operatorname{End}(\rho_5)$ is $\rho_6$, which is contained in the orbit span. Lemma~\ref{orbit-KL-extensions} and the selection rule therefore show that $J_+^3$ is a sufficient seed for the full rank-$3$ KL conditions. 

Now, recall that $J_+ = \hat{a}_1^\dagger \otimes \hat{a}_2$ in the Schwinger model. Thus,
\begin{equation}
    J_+^3 \ket{m ,13-m} = \sqrt{\frac{(m+3)!(13-m)!}{m!(10-m)!}} \ket{m+3, 10-m}, \quad 0 \leq m \leq 10
\end{equation}
and $0$ otherwise. Therefore, it's clear that the $\braket{\overline{0}_\phi|J_+^3|\overline{0}_\phi} = \braket{\overline{1}_\phi|J_+^3|\overline{1}_\phi} = \braket{\overline{0}_\phi|J_+^3|\overline{1}_\phi} = 0$ from the support of the codewords. Thus, it suffices to satisfy $\braket{\overline{1}_\phi|J_+^3|\overline{0}_\phi} = 0$. First, we calculate 
\begin{equation}
    \begin{aligned}
        J_+^3 \ket{9,4} &= 24\sqrt{55} \ket{12, 1} \\
        J_+^3 \ket{5,8} &= 336\ket{8, 5} \\
        J_+^3 \ket{1, 12} &= 24 \sqrt{55}  \ket{4,9}
    \end{aligned}
\end{equation}
which are the only terms that contribute. Then, we find the matrix elements 
\begin{equation}
    \begin{aligned}
        \braket{\overline{1}_{5/2}|J_+^3|\overline{0}_{5/2}} &=- \frac{2 \cdot 3\cdot 24}{56^2}\sqrt{154 \cdot 55 \cdot 70}  + \frac{770 \cdot 336}{56^2}  = \frac{330}{7} \\
        \braket{\overline{1}_{-13/6}|J_+^3|\overline{0}_{5/2}} &= \frac{99}{7}\sqrt{70} \cdot \frac{\sqrt{3003}}{84} + \frac{336\sqrt{770}}{56} \cdot \frac{\sqrt{273}}{28} + \frac{15\sqrt{154}}{7} \cdot \frac{\sqrt{1365}}{84} = \frac{20 \sqrt{4290}}{7} \\ 
        \braket{\overline{1}_{5/2}|J_+^3|\overline{0}_{-13/6}} &= \frac{20 \sqrt{4290}}{7} \\
        \braket{\overline{1}_{-13/6}|J_+^3|\overline{0}_{-13/6}} &= -2\frac{24\sqrt{1365 \cdot 55 \cdot 3003}}{84^2} + 336 \cdot \frac{273}{28^2} = \frac{104}{7}
\end{aligned}
\end{equation}
Combining these, we see that imposing the KL condition gives

\begin{equation}
    \begin{aligned}
        \braket{\overline{1}_\phi|J_+^3|\overline{0}_\phi} &= 196^{-1} \left( \frac{91 \cdot 330}{7} + \frac{105 \cdot 104}{7} + \frac{20 \sqrt{4290}}{7} \sqrt{91 \cdot 105} \left(e^{i\phi} + e^{-i\phi} \right) \right) \\ 
        &= \frac{975}{98}\left(3 + 4\sqrt{22} \cos \phi \right) \implies \phi = \arccos\left(-\frac{3}{4\sqrt{22}} \right) \pmod {2 \pi} 
    \end{aligned}
\end{equation}

which was to be found. This concludes our proof for the construction of the $(n,K,d) = (13,2,4)$ code.

\item $n=28,\, 2O,\, d=6$ construction: 

We further show how to analytically construct our $n = 28$ $2O$ code with distance $6$. The full construction is implemented symbolically in Python, but we outline the construction here. 

The code uses the irrep $\rho_3$ of $2O$ in the notation of Ref.~\cite{grossDesigningCodesInteractions2021}. This is not a faithful representation of $2O$; its kernel is the quaternion group $Q_8 := 2D_2 = \{\pm I,\, \pm i \sigma_x,\, \pm i \sigma_y,\, \pm i \sigma_z \}$, and we have $2O\slash Q_8 \cong S_3$, where $S_3$ is the symmetric group on $3$ elements \cite{NotesCoxeterTransformations2008}. We first note how the generators $S,H$ of equation \eqref{2O-generators} act by conjugation on the unoriented Pauli axes, i.e., on $\{\mathbb{C} \sigma_x,\, \mathbb{C} \sigma_y,\, \mathbb{C} \sigma_z \}$. Note that each of these axes is invariant under conjugation by $Q_8$.  We see that 
\begin{equation}
    \begin{aligned}
        S \sigma_z S^\dagger &= \sigma_z,\, S \sigma_xS^\dagger = \sigma_y,\, S \sigma_y S^\dagger = -\sigma_x \\
        H \sigma_z H^\dagger &= \sigma_x,\, H\sigma_x H^\dagger = \sigma_z,\, H \sigma_y H^\dagger = - \sigma_y 
    \end{aligned}
\end{equation}
So, we see that $S$ permutes the $x,y$ axes, whereas $H$ permutes the $x,z$ axes (the minus signs are irrelevant). Thus, using the standard two-dimensional  representation of $S_3$, which we'll denote as $\pi_2: S_3 \to GL(V)$ on the plane $V = \{(x,y,z) \in \mathbb{C}^3:  x+ y + z = 0 \}$ \cite{fultonRepresentationTheory2004}, and the homomorphism $h: 2O \to S_3 \cong 2O/Q_8$,  we can represent $\pi_2 \circ h (S) = (12)$, and $\pi_2 \circ h(H) = (13)$. Here $(12)$ represents the permutation $(x,y,z) \to (y,x,z)$ acting on $V$ and similarly for $(13)$. The representation $\pi_2 \circ h$ corresponds to our representation $\rho_3$. 

Now, let's define a basis for $V$ given by 
\begin{equation}
        b_0 = \frac{(1,1,-2)}{\sqrt{6}}, \quad b_1 = \frac{(1, -1, 0)}{\sqrt{2}}
\end{equation}
We see that, applying $H,S$ gives
\begin{equation}
    \begin{aligned}
        Hb_0  &= -\frac{1}{2}b_0 - \frac{\sqrt{3}}{2}b_1, \quad Hb_1  = - \frac{\sqrt{3}}{2}b_0 + \frac{1}{2}b_1 \\ 
        Sb_0 &= b_0, \quad Sb_1 = - b_1
    \end{aligned}
\end{equation}
Thus, if we map $b_0 \to \ket{0}, \, b_1 \to \ket{1}$, we can write 
\begin{equation}
    \rho_3(S) = \begin{pmatrix}
        1 &0 \\
        0 & -1
    \end{pmatrix} = \sigma_z,\, \rho_3(H) = \begin{pmatrix}
        -\frac{1}{2} & - \frac{\sqrt{3}}{2} \\
        -\frac{\sqrt{3}}{2} & \frac{1}{2}
    \end{pmatrix} = -\frac{1}{2}\sigma_z - \frac{\sqrt{3}}{2}\sigma_x \label{2O-SH}
\end{equation}
for the desired representation. Now, we wish to find the relevant subspace in the physical representation. Firstly, $\rho_3$ factors through the quotient by $Q_8$, which is generated by $i\sigma_z,\, i\sigma_x$. So, we see that since we can write (in terms of the $\mathfrak{su}(2)$ generators $X_a = -i\sigma_a/2$)
\begin{equation}
    \begin{aligned}
        R_{28}(-i\sigma_z) &= R_{28}(e^{-\pi i \sigma_z/2}) = e^{\pi dR_{28}(X_z)} = e^{-i \pi J_z} \\ 
        R_{28}(-i \sigma_x) &= R_{28}(e^{-\pi i \sigma_x/2}) = e^{\pi dR_{28}(X_x)} = e^{-i\pi J_x},
    \end{aligned}
\end{equation}

and these acts on the Fock states as
\begin{equation}
    R_{28}(-i \sigma_z) \ket{m, 28-m} =  (-1)^m \ket{m, 28-m}, \quad R_{28}(-i \sigma_x) \ket{m, 28-m} = \ket{28-m, m},
\end{equation}

we can define a basis 
\begin{equation}
    f_m = \begin{cases}
        \frac{\ket{m, 28-m} + \ket{28-m, m}}{\sqrt{2}} &  m=0,2,4,\dots, 12 \\
        \ket{14,14}, & m = 14
    \end{cases}
\end{equation}
fixed by the quaternion representatives. The relevant $\rho_3$ subspace is therefore a subspace of $\operatorname{span} \{ f_m \}$. However, the $Q_8$-invariant subspace of $H_{28}$ can be decomposed as $\rho_1 \oplus \rho_2 \oplus \rho_3^{\oplus 3}$ in the notation of Ref.~\cite{grossDesigningCodesInteractions2021}. We therefore need to distinguish our desired $\rho_3^{\oplus 3}$ subspace from the trivial ($\rho_1$) and sign ($\rho_2$) representation vectors. We note that the $\rho_2$ representation corresponds to the sign of the permutation in the $\rho_3$ representation, giving us $\rho_2(S) = \rho_2(H) = -1$. We see
\begin{equation}
    R_{28}(S) \ket{m, 28-m} = (e^{-i\pi/4})^{m - 28 + m} \ket{m, 28-m} = -e^{-i m \pi/2} \ket{m, 28-m}. 
\end{equation}
But, we know in the $\rho_3$ representation that $\rho_3(S)$ has a $+1$ and $-1$ eigenvector. We will construct a basis of $+1$  $R_{28}(S)$ eigenvectors from $f_m$ and obtain the $-1$ eigenvectors from them, giving us a full basis. Indeed, we get $-e^{-im\pi/2} = 1$ for $m = 2,6,10,14$. So the subspace $\operatorname{span} \{f_2, f_6, f_{10}, f_{14}\}$ corresponds to the span of the $+1$ eigenvectors of $\rho_3$ and the trivial representation direction $\rho_1$. Note that we know this subspace doesn't include the $\rho_2$ representation subspace, since $S$ is an odd permutation. 

It's easy to see that, since the projector on the trivial representation (in the $\operatorname{span} \{f_2, f_6, f_{10}, f_{14} \}$ subspace) is given by $P_{\mathbf{1}} = \frac{1}{3}\left(I + R_{28}(H) + R_{28}(S) R_{28}(H)\right)$,
\begin{equation}
    t_0 = f_{14} + R_{28}(H) f_{14} +  R_{28}(S) R_{28}(H) f_{14}, \quad t = \frac{t_0}{\braket{t_0, f_{14}}}
\end{equation}
is the trivial representation vector. So, we can now take 
\begin{equation}
    c_m = -\braket{f_m, t}, \quad v_m = f_m + c_m f_{14}, \quad w_m = -\frac{v_m  + 2R_{28}(H)v_m}{\sqrt{3}}, \quad m = 2,\,6,\,10 \label{rho-3-basis-28}
\end{equation}
The $(v_m, w_m)$ provide a basis for the $\rho_3$ isotypic component. Indeed, we first check that all the $v_m$ satisfy $\braket{v_m, t} = \braket{f_m, t} - \braket{f_m, t} \braket{f_{14}, t} = 0 $, i.e., they are orthogonal to the trivial representation (the same follows for $w_m$ by unitarity of the representation). We can re-write this orthogonality condition as $P_{\mathbf{1}} v_m  = 0 \implies R_{28}(S) R_{28}(H) v_m = -v_m - R_{28}(H) v_m  $.  So, we easily see
\begin{equation}
    \begin{aligned}
        R_{28}(S) v_m &= v_m, \quad R_{28}(H) v_m =- \frac{\sqrt{3}}{2} w_m - \frac{1}{2}v_m  \\  
        R_{28}(S) w_m &=  - \frac{v_m - 2v_m - 2 R_{28}(H) v_m }{\sqrt{3}} = - w_m, \\ 
         R_{28}(H) w_m &= - \frac{R_{28}(H) v_m + 2 v_m }{\sqrt{3}} = \frac{1}{2}w_m - \frac{\sqrt{3}}{2}v_m 
    \end{aligned}
\end{equation}

Exactly matching our representation \eqref{2O-SH}. It's also easy to see that $-\braket{v_a,w_b} = \braket{v_a, -w_b} =  \braket{R_{28}(S)v_a, R_{28}(S) w_b} = \braket{v_a, w_b} \implies \braket{v_a, w_b}= 0 $.  This also gives us that $\braket{R_{28}(H) v_a, v_b} = - \frac{1}{2}\braket{v_a, v_b}$.

The pairs $(v_m,w_m)$ give an equivariant basis, but the $v_m$ are not mutually orthogonal. We use their Gram matrix to define an isometry $S_{\rho_3}$ in the convention of our general formalism. Let
\begin{equation}
    V = \begin{pmatrix}
        v_2 & v_{6} & v_{10} 
    \end{pmatrix} \quad W = \begin{pmatrix}
        w_2 & w_{6} & w_{10} 
    \end{pmatrix} \quad c = \begin{pmatrix}
        c_{2} \\ c_6 \\ c_{10}
    \end{pmatrix}
\end{equation}
Now, define maps $\tilde{E}_{m}: \mathbb{C}^2 \to H_{28}, \ \tilde{E}_m \ket{0} = v_m,\, \tilde{E}_m \ket{1} = w_m$, which is by definition $2O$-equivariant. We have that by Schur's lemma
\begin{equation}
    \tilde{E}_{m'}^\dagger \tilde{E}_m = \Gamma_{m'm} I_{2}
\end{equation}
We find that
\begin{equation}
    \Gamma_{m'm} = \braket{v_{m'}, v_m} = \delta_{m' m} + c_{m'} c_m  \implies \tilde{E}_{m'}^\dagger \tilde{E}_m = (I +  c c^T)_{m'm} I_2
\end{equation}
where we note that the $c$ are real since $R_{28}(H)$ is real and $R_{28}(S)$ restricted to $f_m$ is also real. Then the inverse square root is given by \cite{dieciContinuationSingularValue2005}
\begin{equation}
    \Gamma^{-1/2} = I - \beta cc^T, \quad \beta = \frac{\alpha }{1 + \alpha c^T c}, \quad \alpha = \frac{\sqrt{1 + c^T c}-1}{c^T c} 
\end{equation}
which we can use to make the bases orthonormal and define the desired $2O$-equivariant isometry as
\begin{equation}
    E_m = \sum_{m'} \tilde{E}_{m'} (\Gamma^{-1/2})_{m'm}, \quad S_{\rho_3}: \mathbb{C}^3 \otimes \rho_3 \to H_{28}, \, S_{\rho_3} (\ket{m} \otimes u) = E_m u
\end{equation}
where $\ket{m}$ labels the multiplicity space basis state.

For the symbolic Knill--Laflamme calculation, we keep the original paired basis $(v_m, w_m)$. Let $a=(a_2,a_6,a_{10})^T\in\mathbb C^3$ be nonzero; this defines the logical states as $\ket{0_L} \propto \sum_m a_m v_m, \,  \ket{1_L} \propto \sum_m a_m w_m $. Its corresponding unit multiplicity vector $\ket{s}$ in the orthonormal basis is
\begin{equation}
    s=\frac{\Gamma^{1/2}a}{\sqrt{a^\dagger\Gamma a}},\qquad
    S_{\rho_3}(\ket{s}\otimes u)
    =\frac{1}{\sqrt{a^\dagger\Gamma a}}\sum_{m\in\{2,6,10\}}a_m\tilde E_m u
    \quad (u\in\rho_3).
\end{equation}
Thus, for any operator $O$ on $H_{28}$, the compression to this copy of $\rho_3$ satisfies
\begin{equation}
    (\bra{s}\otimes I_{\rho_3})S_{\rho_3}^\dagger O S_{\rho_3}(\ket{s}\otimes I_{\rho_3})
    =\frac{1}{a^\dagger\Gamma a}
    \sum_{m,m'\in\{2,6,10\}}\overline{a_{m'}}a_m\tilde E_{m'}^\dagger O\tilde E_m.
\end{equation}
The denominator is positive and common to all representation-factor components. Hence their nontrivial Knill--Laflamme equations can be solved using the numerator, as in the symbolic script, with the scale fixed by $a_2=1$. The normalized codewords are then $Va/\sqrt{a^\dagger\Gamma a}$ and $Wa/\sqrt{a^\dagger\Gamma a}$. 

Finally, we need to algebraicallly find a basis of operators $F_{j,\mu}^{U; \rho_3, \rho_3}$. In this case, $\operatorname{End}(\rho_3) \cong \rho_1 \oplus \rho_2 \oplus \rho_3$, so we need to construct these for $U \in \{\rho_1, \rho_2, \rho_3\}$ (note that we have $j = 1$). For the trivial representation we simply have $F^{(\mathbf{1})} = I_2$. 

For the sign representation $\rho_2$, we choose $F^{(\mathbf{2})} = \frac{\sigma_y}{\sqrt{2}}$. It's easy to see that 

\begin{equation}
    \begin{aligned}
        \rho_3(S) F^{(\mathbf{2})} \rho_3(S)^\dagger &= \frac{\sigma_z \sigma_y \sigma_z}{\sqrt{2}} = -\frac{\sigma_y}{\sqrt{2}}  = - F^{\mathbf{2}} \\ 
        \rho_3(H) F^{(\mathbf{2})} \rho_3(H)^\dagger &= \frac{1}{\sqrt{2}}(\frac{1}{2}\sigma_z + \frac{\sqrt{3}}{2}\sigma_x) \sigma_y (\frac{1}{2}\sigma_z + \frac{\sqrt{3}}{2}\sigma_x) = - \frac{\sigma_y}{\sqrt{2}} = - F^{\mathbf{2}}
    \end{aligned}
\end{equation}
so this indeed transforms as the sign representation, since both generators correspond to odd representations. Finally, for $\rho_3$, we define the operators $F^{(\mathbf{3})}_1 = \frac{\sigma_z}{\sqrt{2}}, \, F^{(\mathbf{3})}_2 = - \frac{\sigma_x}{\sqrt{2}}$. We see that 
\begin{equation}
    \begin{aligned}
        \rho_3(S) F^{(\mathbf{3})}_1 \rho_3(S)^\dagger &= \frac{\sigma_z \sigma_z \sigma_z}{\sqrt{2}}  = \frac{\sigma_z}{\sqrt{2}} = F^{(\mathbf{3})}_1 \\ 
        \rho_3(S) F^{(\mathbf{3})}_2 \rho_3(S)^\dagger &= -\frac{\sigma_z \sigma_x \sigma_z}{\sqrt{2}}  = \frac{\sigma_x}{\sqrt{2}}  =-F^{(\mathbf{3})}_2
    \end{aligned}
\end{equation}
and likewise 
\begin{equation}
    \begin{aligned}
        \rho_3(H) F^{(\mathbf{3})}_1 \rho_3(H)^\dagger &= (\frac{1}{2}\sigma_z + \frac{\sqrt{3}}{2}\sigma_x) \frac{\sigma_z }{\sqrt{2}} (\frac{1}{2}\sigma_z + \frac{\sqrt{3}}{2}\sigma_x)  = -\frac{F_1^{(\mathbf{3})}}{2}  - \frac{\sqrt{3}}{2} F^{(\mathbf{3})}_2 \\
        \rho_3(H) F^{(\mathbf{3})}_2 \rho_3(H)^\dagger &= -(\frac{1}{2}\sigma_z + \frac{\sqrt{3}}{2}\sigma_x) \frac{\sigma_x }{\sqrt{2}} (\frac{1}{2}\sigma_z + \frac{\sqrt{3}}{2}\sigma_x)  = \frac{F_2^{(\mathbf{3})}}{2}  - \frac{\sqrt{3}}{2} F^{(\mathbf{3})}_1 \\
    \end{aligned}
\end{equation}
matching the $\rho_3$ representation we derived in equation \eqref{2O-SH}. Finally, the KL condition is trivial for ranks $0,1$, and for ranks $2-5$ we use the seed rotation 
\begin{equation}
    U = \frac{1}{\sqrt{5}} \begin{pmatrix}
        1 & -2 \\ 
        2 & 1
    \end{pmatrix}
\end{equation}
which we symbolically verify to be sufficient. We construct all these objects symbolically in python, and solve the relevant KL equations, providing the coefficients of the logical states in the $(v_m, w_m)$ basis. We verify symbolically that the resulting code detects all elementary tensors up to rank $5$ but fails at rank $6$, proving our $d = 6$ claim. The SymPy code used for the symbolic derivation and verification is available for reproducibility purposes. 
\end{itemize}

\subsection{Numerical symmetry-classification procedure}

We classify all optimized numerical codes in the range $n = 4-25$. Firstly, let's define the commutator loss of a projector $P$ with a group $G$ as

\begin{equation}
    \operatorname{Comm}_L(P, G) = \frac{1}{2d_L|G|} \sum_g \norm{[P,\pi(g)]}_{\mathrm F}^2 \label{commutator-loss}
\end{equation}

We start by proving the following lemma

\begin{lemma}
    Let $\pi:G \to U(V)$ a unitary representation of a finite group $G$ with representatives $U_g = \pi(g)$, and $P \in \operatorname{End}(V)$ a rank-$d_L$ orthogonal projector.  Then the commutator loss $\operatorname{Comm}_L(P,G)$ of equation \eqref{commutator-loss} satisfies 
    \begin{equation}
        0 \leq \operatorname{Comm}_L(P,G) \leq 1
    \end{equation}
    and $\operatorname{Comm}_L(P,G) = 0$ iff $P$ commutes with every $\pi(g)$. 
\end{lemma}

\begin{proof}
We note that 

\begin{equation}
    \begin{aligned}
        \operatorname{Comm}_L(P, G) &= \frac{1}{2d_L|G|} \sum_g \norm{[P,\pi(g)]}_{\mathrm F}^2 \\
        &= \frac{1}{2d_L|G|} \sum_g \norm{PU_g-U_gP}_{\mathrm F}^2 \\ 
        &= \frac{1}{2d_L|G|} \sum_g \Tr\!\left[(PU_g - U_gP)^\dagger (PU_g - U_g P)\right] \\ 
        &= \frac{1}{2d_L|G|} \sum_g \Tr\!\left[(U_g^\dagger P - P U_g^\dagger) (PU_g - U_g P) \right] \\ 
        &= \frac{1}{2d_L|G|} \sum_g \left( \Tr\!\left[U_g^\dagger P U_g\right] - \Tr\!\left[U_g^\dagger P U_g P\right] - \Tr\!\left[PU_g^\dagger P U_g\right] + \Tr\!\left[P U_g^\dagger U_g P\right]  \right) \\ 
        &= \frac{1}{d_L|G|} \sum_g \left(d_L -  \Tr\!\left[U_gPU_g^\dagger P  \right]    \right)
    \end{aligned}
\end{equation}

Now, $U_g$ is unitary, so we can see that $Q = U_gP U_g^\dagger$ is also a projector. Indeed,

\begin{equation}
    \begin{aligned}
        Q^\dagger &= (U_g P U_g^\dagger)^\dagger = U_g P U_g^\dagger = Q, \\
        \text{and}\qquad (Q)^2 &= (U_g P U_g^\dagger)(U_g P U_g^\dagger) = U_g P U_g^\dagger  = Q.
    \end{aligned}
\end{equation}

Furthermore, $P$ is PSD since we have

\begin{equation}
    \braket{x, Px} = \braket{x, P^2x} = \braket{Px, Px} = \norm{Px}^2 \geq 0,
\end{equation}

and likewise for $Q$. We also have that $PQP \geq 0$ since

\begin{equation}
    \braket{x, PQPx} = \braket{Px, QPx} = \braket{Px, Q^2Px} = \braket{QPx, QPx} = \norm{QPx}^2 \geq 0.
\end{equation}

We also show that $0 \leq PQP \leq P$. Indeed, $P - PQP  = P(I-Q)P$, and

\begin{equation}
    \begin{aligned}
        (I-Q)^2 &= (I-Q)(I-Q) = I - Q - Q + Q^2 = I-Q, \\
        \text{and}\qquad (I-Q)^\dagger &= I-Q,
    \end{aligned}
\end{equation}

i.e., $I-Q$ is also a projector. So, we see

\begin{equation}
    \braket{w, P(I-Q)Pw} = \braket{Pw, (I-Q)Pw}= \braket{(I-Q)Pw, (I-Q)Pw} = \norm{(I-Q)Pw}^2 \geq 0,
\end{equation}

so we get

\begin{equation}
    0 \leq PQP \leq P \implies  0 \leq \Tr\!\left[P U_g P U_g^\dagger P\right] \leq \Tr\!\left[P\right] \implies 0 \leq \Tr\!\left[U_g P U_g^\dagger P\right] \leq d_L.
\end{equation}

So, we see that the difference $0 \leq d_L - \Tr\!\left[U_g P U_g^\dagger P\right] \leq d_L$, and thus we conclude that $0 \leq \operatorname{Comm}_L(P, G) \leq 1$, as desired. So we see that this is indeed a normalized measure of non-invariance of the projector with the group, and it's equal to zero iff the projector commutes with the group. 
\end{proof}

Now, for a projector $P$ and each irrep $C$ of group $G$, define $X_C = S_C^\dagger P S_C$, and 

\begin{equation}
    A_C = \frac{1}{d_C} \Tr_C(X_C) = \frac{1}{d_C} \Tr_C \left(S_C^\dagger P S_C\right) \label{numerical-multiplicity-operator}
\end{equation}

where $S_C$ the usual embedding of the isotypic component of $C$ into $V$. Also, denote the eigenvalues of $A_C$ as $\lambda_{C,1} \geq \lambda_{C,2} \geq \dots \geq \lambda_{C, m_C}$. Further let $\mathcal{P}_{d_L}^G$ the set of rank-$d_L$ $G$-invariant projectors. We now prove the following. 

\begin{proposition}
    The maximum overlap of a rank-$d_L$ projector $P$ with a rank-$d_L$, $G$-invariant projector is 

    \begin{equation}
        \max_{Q \in \mathcal{P}_{d_L}^G} \Tr\!\left[PQ\right] = \max_{\{r_C \}} \left\{ \sum_{C} d_C \left( \sum_{j = 1}^{r_C} \lambda_{C,j} \right): \quad 0 \leq r_C \leq m_C,\, \sum_C r_C d_C = d_L \right\}
    \end{equation}
    where $m_C$ the multiplicity of $C$ in the decomposition of $V$. 
\end{proposition}

\begin{proof}
    We know that any $Q \in \mathcal{P}_{d_L}^G$ has the form 

    \begin{equation}
        Q = \sum_C S_C (Q_C \otimes I_C) S_C^\dagger
    \end{equation}

    from Corollary~\ref{invariant-projector-form}, where $Q_C$ projectors. If $\operatorname{rank}Q = d_L$, and each $\operatorname{rank}Q_C = r_C$, then it's clear that we need $\sum_C r_C d_C = d_L$. And of course, we have the upper bound $r_C \leq m_C$. 

    Now we expand, taking $Q_C = V_CV_C^\dagger$ isometries $V_C: \mathbb{C}^{r_C} \to M_C$, 

    \begin{equation}
        \begin{aligned}
            \Tr\!\left[PQ\right] &= \sum_C \Tr \left[S_C(Q_C \otimes I_C) S_C^\dagger  P  \right] \\ 
            &= \sum_C \Tr\!\left[ (Q_C \otimes I_C) (S_C^\dagger P S_C)\right] \\ 
            &= \sum_C \Tr\!\left[ (Q_C \otimes I_C) X_C\right] \\ 
            &= \sum_C \left(\sum_{c,m} \bra{m} \otimes \bra{c} \left[ (Q_C \otimes I_C) X_C \right] \ket{m} \otimes \ket{c} \right) \\ 
            &= \sum_C \sum_m \bra{m} \Big( \sum_c I_{M_C} \otimes \bra{c} \left[ (Q_C \otimes I_C) X_C \right] I_{M_C} \otimes \ket{c} \Big)  \ket{m}  \\ 
            &= \sum_C \sum_m \bra{m}  (Q_C d_C A_C)  \ket{m}  \\ 
            &= \sum_C d_C \Tr\!\left[V_C V_C^\dagger A_C\right] = \sum_C d_C \Tr\!\left[V_C^\dagger A_C V_C\right]
         \end{aligned}
    \end{equation}

    Now, $V_C^\dagger V_C = I_{r_C}$ (since it's an isometry), so, from the Ky Fan maximum principle \cite{fanTheoremWeylConcerning1949}, we get that 

    \begin{equation}
        \max_{Q \in \mathcal{P}_{d_L}^G,\ \operatorname{rank}Q_C=r_C} \Tr\!\left[PQ\right] = \sum_C d_C (\sum_{j = 1}^{r_C} \lambda_{C,j})
    \end{equation}

    as desired. Finally, this gives us the maximum for a specific allocation of ranks $\{ r_C \}$; to find the maximum over all possible codes, we consider this quantity for every possible allocation $0 \leq r_C \leq m_C,\, \sum_C r_C d_C = d_L$.

\end{proof}

\begin{remark}
    Suppose that $Q_* \in \mathcal{P}_{d_L}^G$ attains this maximum. Then we have that 
    \begin{equation}
    \norm{P-Q_*}_{\mathrm F}^2 = \Tr\!\left[(P - Q_*)^\dagger (P-Q_*)\right] =  \Tr\!\left[P - PQ_* - Q_*P + Q_*\right] = 2(d_L - \Tr\!\left[PQ_*\right]),
    \end{equation}
    so we see that it minimizes the Frobenius distance from the physical projector. 
\end{remark}

In practice, we first minimize the commutator loss for each candidate group over $SU(2)$ rotations and complex conjugations of the projector in the number basis (which provably preserve performance), and choose the group with minimal loss; this allows us to identify the group most compatible with the physical code. We then use dynamic programming to find the optimal allocation $\{r_C\}$. We report the optimized commutator loss $\operatorname{Comm}_L(P,G)$, as well as the normalized overlap 

\begin{equation}
    \Omega_G(P) = \frac{1}{d_L}\Tr\!\left[P\widehat Q_G\right] \label{best-overlap-projector}
\end{equation}

with the nearest $G$-invariant projector $\widehat Q_G$. Note that $0 \leq \Omega_G(P)  \leq 1$, since $0 \leq P\widehat Q_GP \leq P \implies 0 \leq \Tr\!\left[P\widehat Q_GP\right] \leq d_L \implies 0 \leq \frac{\Tr\!\left[P\widehat Q_G\right]}{d_L} \leq 1 $. 

Finally, if a code is invariant under a group $G$, it's also invariant for all its subgroups. In general, if there is more than one group that aligns well ($>0.99$ overlap) with a numerical code, we choose to represent the code with the larger group, even if the overlap is technically a bit smaller. This is because the above metric is naturally more stringent for larger groups, while at the same time  larger groups tend to capture the code structure better.

\begingroup
\setlength{\intextsep}{6pt plus 1pt minus 1pt}
\setlength{\textfloatsep}{8pt plus 2pt minus 2pt}
\captionsetup[figure]{skip=4pt,belowskip=0pt}
\captionsetup[subfigure]{skip=2pt,belowskip=0pt}
\subsection{High-transmissivity classification for photon numbers 4 through 25}

We present the best saved-fidelity numerical code in the reproducibility sweep at each displayed parameter point for $4\leq n\leq25$. We use $\eta=0.90$ except at $n=7,15$, where we use $\eta=0.95$. All $22$ selected codes show finite-group structure: $11$ are associated with two-dimensional irreducible representations and $11$ with reducible representations. Among the latter, two have non-Abelian symmetry and the others use one-dimensional irreps of cyclic groups. Table~\ref{tab:numerical-code-symmetry-classification} reports the group commutator loss \eqref{commutator-loss} and nearest-invariant-projector overlap \eqref{best-overlap-projector} for these same numerical files. The bi-convex search does not certify global fidelity optimality, and differences comparable to decoder tolerances should not be interpreted as rigorous performance orderings.

\begin{table}[!t]
    \centering
    \small
    \renewcommand{\arraystretch}{1.08}
    \begin{tabular}{c c l c c}
        \hline
        $n$ & $\eta$ & Group and irrep content & $\operatorname{Comm}_L(P,G)$ & $\Omega_G(P)$ \\
        \hline
        $4$ & $0.90$ & $2O:\rho_{3}$ & $1.179501\times 10^{-2}$ & $0.9940$ \\
        $5$ & $0.90$ & $C_{6}:\rho_{2}\oplus\rho_{6}$ & $3.238326\times 10^{-12}$ & $0.999999999998$ \\
        $6$ & $0.90$ & $C_{6}:\rho_{1}\oplus\rho_{3}$ & $3.648599\times 10^{-8}$ & $0.99999998$ \\
        $7$ & $0.95$ & $2I:\rho_{3}$ & $4.789343\times 10^{-7}$ & $0.9999997$ \\
        $8$ & $0.90$ & $C_{6}:\rho_{1}\oplus\rho_{3}$ & $7.243801\times 10^{-10}$ & $0.9999999996$ \\
        $9$ & $0.90$ & $C_{6}:\rho_{2}\oplus\rho_{4}$ & $1.798730\times 10^{-9}$ & $0.9999999991$ \\
        $10$ & $0.90$ & $2D_{4}:\rho_{5}$ & $2.741807\times 10^{-4}$ & $0.9998$ \\
        $11$ & $0.90$ & $C_{4}:\rho_{3}\oplus\rho_{4}$ & $3.989942\times 10^{-9}$ & $0.999999998$ \\
        $12$ & $0.90$ & $2T:\rho_{1}\oplus\rho_{2}$ & $1.592746\times 10^{-2}$ & $0.9919$ \\
        $13$ & $0.90$ & $2O:\rho_{5}$ & $1.013267\times 10^{-6}$ & $0.9999994$ \\
        $14$ & $0.90$ & $C_{6}:\rho_{1}\oplus\rho_{5}$ & $2.902749\times 10^{-9}$ & $0.999999998$ \\
        $15$ & $0.95$ & $2T:\rho_{5}$ & $4.888450\times 10^{-2}$ & $0.9751$ \\
        $16$ & $0.90$ & $2D_{3}:\rho_{5}$ & $8.770528\times 10^{-9}$ & $0.999999995$ \\
        $17$ & $0.90$ & $2D_{3}:\rho_{6}$ & $3.459936\times 10^{-9}$ & $0.999999998$ \\
        $18$ & $0.90$ & $2D_{2}:\rho_{3}\oplus\rho_{4}$ & $1.472431\times 10^{-2}$ & $0.9925$ \\
        $19$ & $0.90$ & $2T:\rho_{4}$ & $9.657661\times 10^{-6}$ & $0.999995$ \\
        $20$ & $0.90$ & $C_{6}:\rho_{1}\oplus\rho_{5}$ & $1.086515\times 10^{-2}$ & $0.9945$ \\
        $21$ & $0.90$ & $2T:\rho_{6}$ & $1.145272\times 10^{-9}$ & $0.9999999994$ \\
        $22$ & $0.90$ & $2D_{5}:\rho_{7}$ & $3.583196\times 10^{-4}$ & $0.9998$ \\
        $23$ & $0.90$ & $C_{4}:\rho_{3}\oplus\rho_{4}$ & $2.625900\times 10^{-3}$ & $0.9986$ \\
        $24$ & $0.90$ & $C_{4}:\rho_{1}\oplus\rho_{2}$ & $9.012402\times 10^{-4}$ & $0.9995$ \\
        $25$ & $0.90$ & $2O:\rho_{4}$ & $1.109436\times 10^{-3}$ & $0.9994$ \\
        \hline
    \end{tabular}
    \caption{Symmetry classification of the selected high-transmissivity sweep winners. We use $\eta=0.90$ except at $n=7,15$, where $\eta=0.95$. When comparably accurate subgroup descriptions exist, we report the larger group. At $n=15$, the $2T$ description is approximate, with $\Omega_{2T}(P)=0.975127087$; the closer $C_6$ subgroup fit has overlap $0.997486117$. Overlap entries are truncated after the first decimal digit other than $9$, retaining at least four decimal places.}
    \label{tab:numerical-code-symmetry-classification}
\end{table}

\subsubsection{Numerical comparisons with pure-irrep constructions}

We compare pure-irrep constructions at photon numbers $n=4,7,10,13,15,16,17,19,21,22,25$ with the selected numerical codes. Each pair is aligned by an $SU(2)$ rotation and, when required, complex conjugation in the occupation-number basis. These transformations preserve optimized fidelity. The plots show the actual numerical and constructed projectors, including their residual structural differences, and the captions report their fidelities at the same noise strength. All caption fidelities are evaluated at the stated transmissivity using the saved encoder and recovery; when the saved result uses a different transmissivity, we reweight its normalized photon-number-sector fidelities by the corresponding loss probabilities.

The $n=4$ construction is the two-mode Chuang--Leung--Yamamoto (CLY) code~\cite{chuangBosonicQuantumCodes1997}. Its binary-octahedral $2O$ symmetry was also noted in a recent work on bosonic cyclic codes~\cite{wetherbeeBosonicCyclicCodes2026}.
\FloatBarrier
\begin{figure}[H]
    \centering
    \begin{subfigure}{0.48\textwidth}
        \centering
        \includegraphics[width=\linewidth]{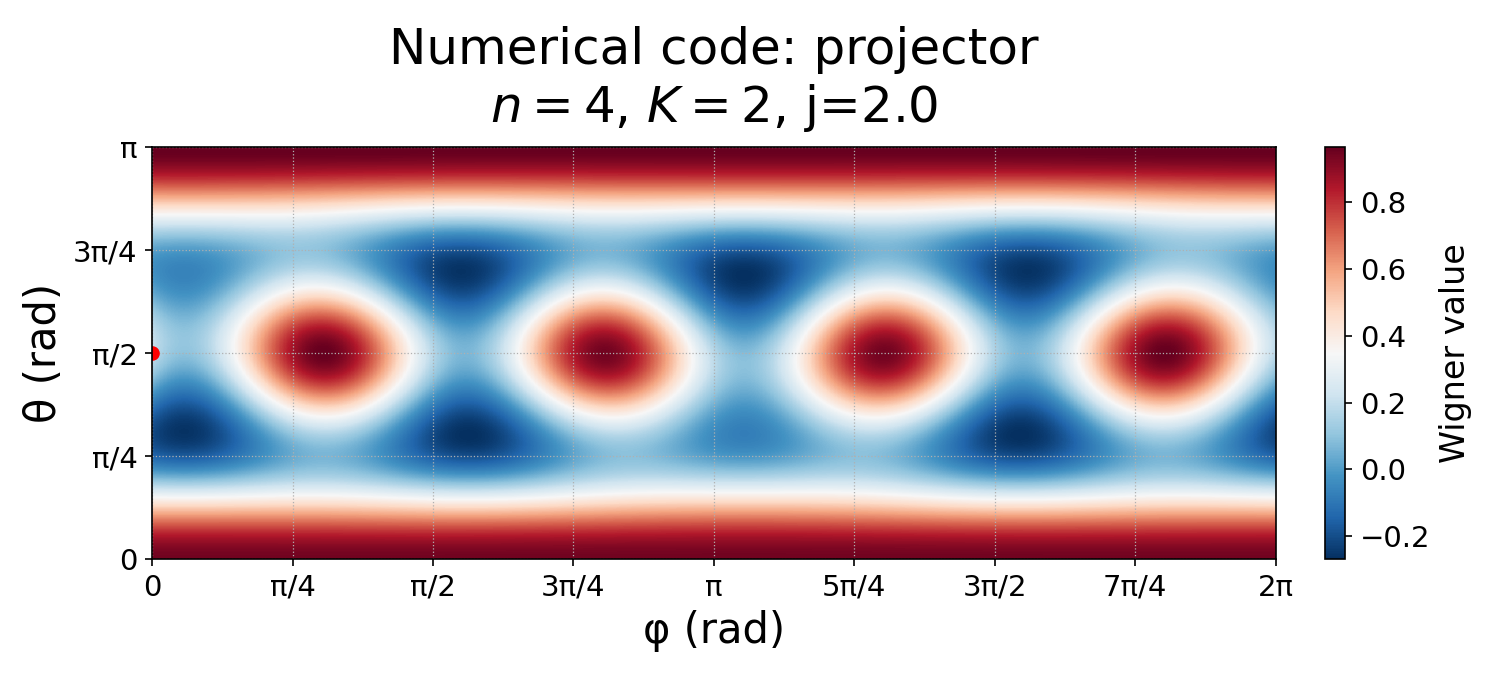}
        \caption{Numerically optimized code.}
    \end{subfigure}
    \hfill
    \begin{subfigure}{0.48\textwidth}
        \centering
        \includegraphics[width=\linewidth]{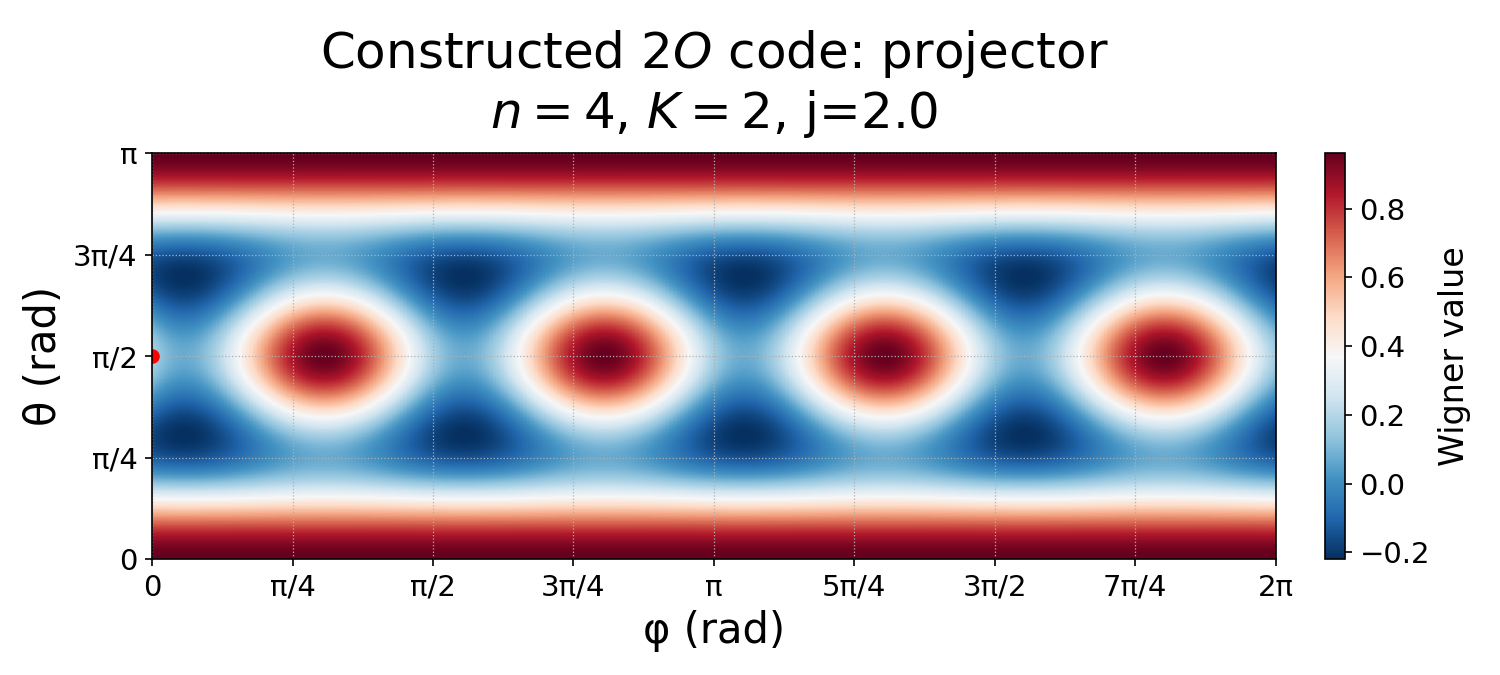}
        \caption{Constructed CLY ($2O$) code.}
    \end{subfigure}
    \caption{Best numerical code and constructed $2O$ reference at $n=4$. At $\eta=0.90$, their respective entanglement fidelities are $F_e^{\mathrm{num}}=0.973273$ and $F_e^{\mathrm{con}}=0.972925$.}
    \label{fig:pure-irrep-n4}
\end{figure}

\begin{figure}[H]
    \centering
    \begin{subfigure}{0.48\textwidth}
        \centering
        \includegraphics[width=\linewidth]{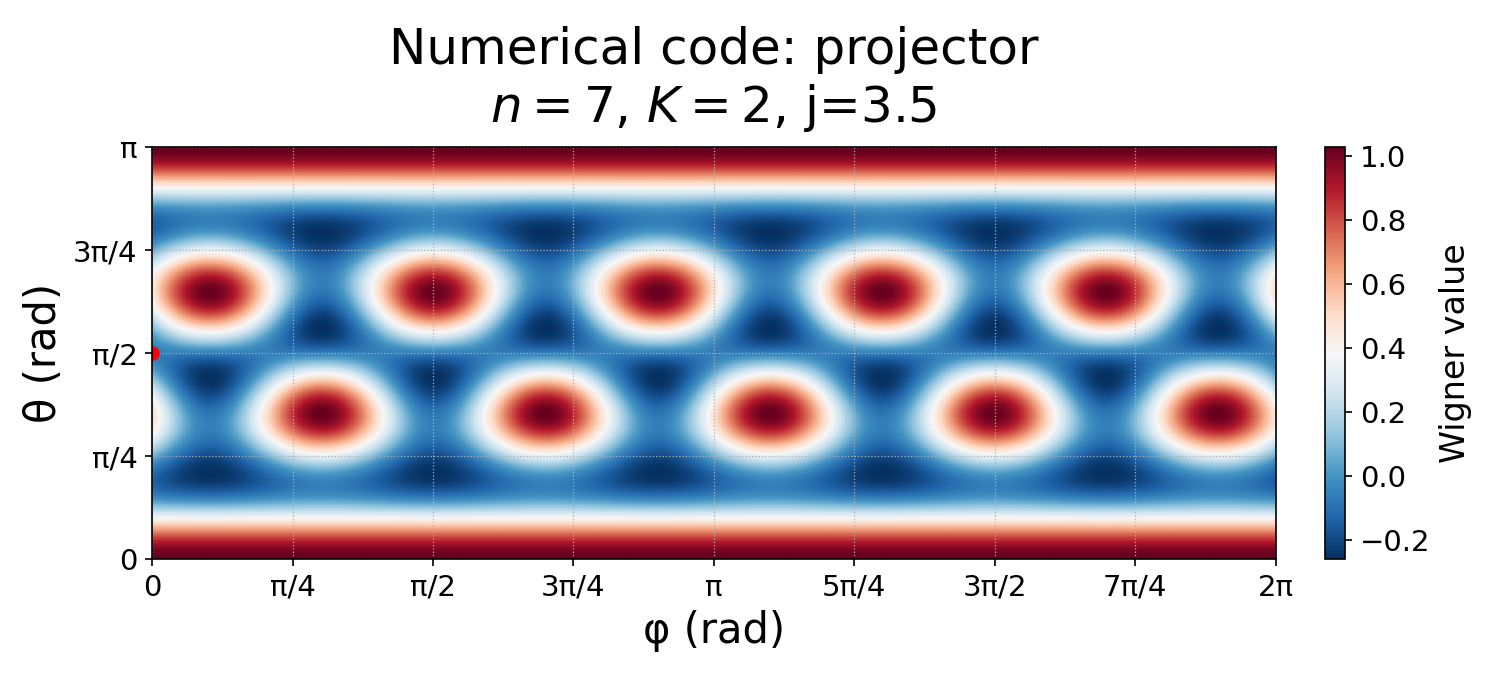}
        \caption{Numerically optimized code.}
    \end{subfigure}
    \hfill
    \begin{subfigure}{0.48\textwidth}
        \centering
        \includegraphics[width=\linewidth]{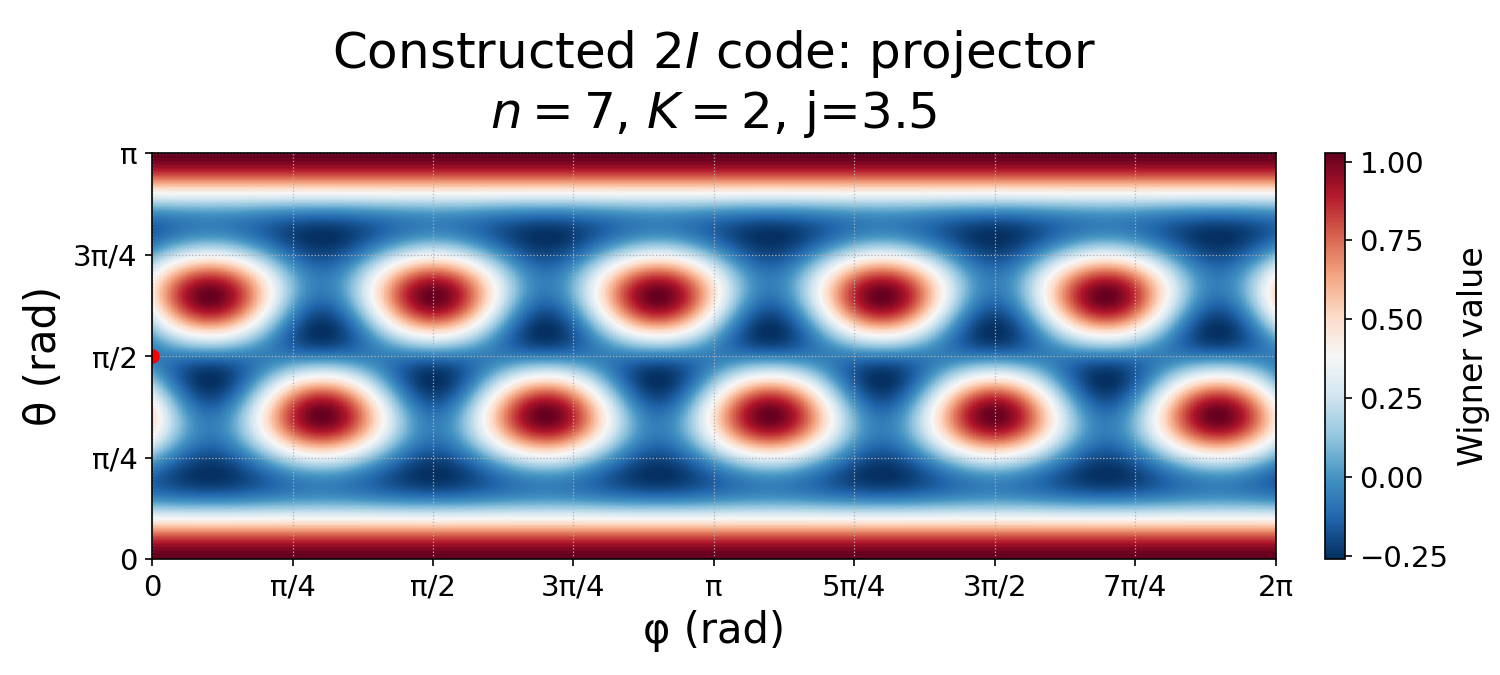}
        \caption{Constructed $2I$ code.}
    \end{subfigure}
    \caption{Best numerical code and constructed $2I$ reference at $n=7$. At $\eta=0.95$, their respective entanglement fidelities are $F_e^{\mathrm{num}}=0.998547$ and $F_e^{\mathrm{con}}=0.998547$.}
    \label{fig:pure-irrep-n7}
\end{figure}

\begin{figure}[H]
    \centering
    \begin{subfigure}{0.48\textwidth}
        \centering
        \includegraphics[width=\linewidth]{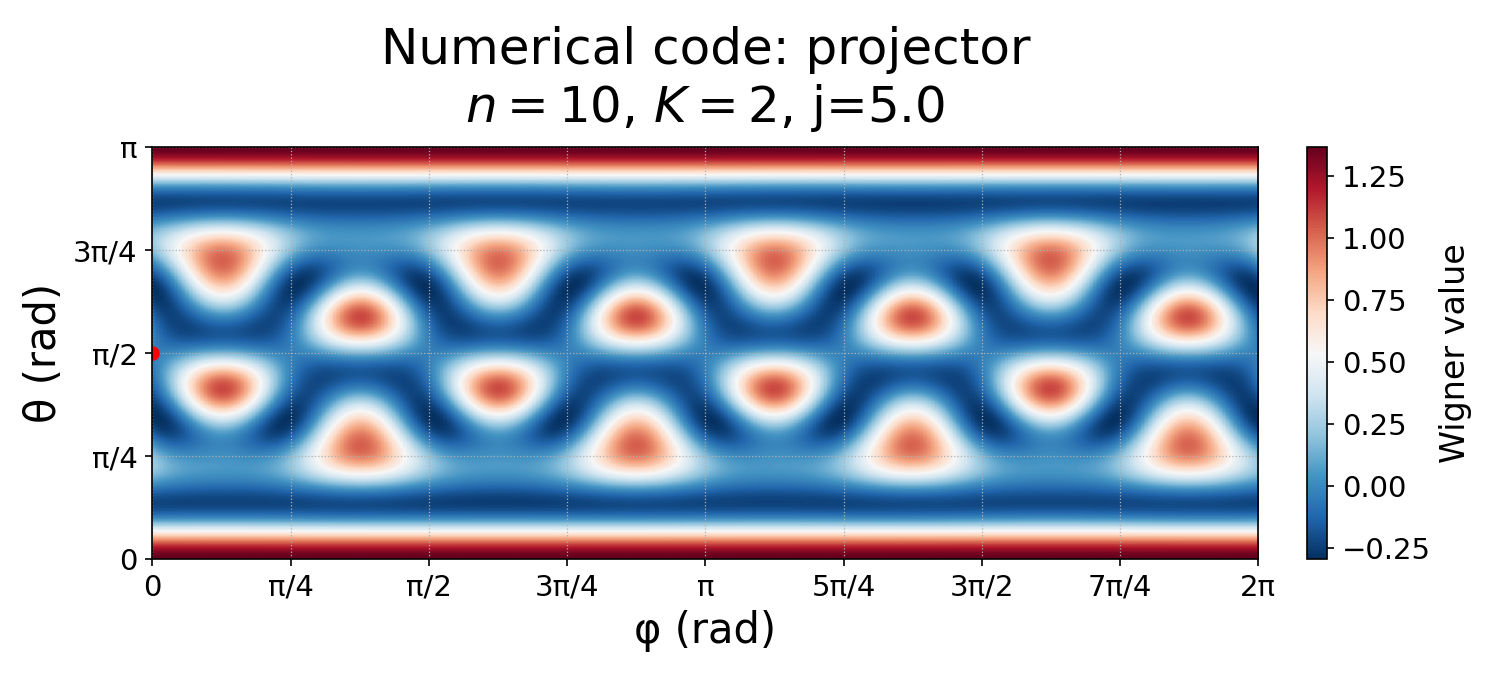}
        \caption{Numerically optimized code.}
    \end{subfigure}
    \hfill
    \begin{subfigure}{0.48\textwidth}
        \centering
        \includegraphics[width=\linewidth]{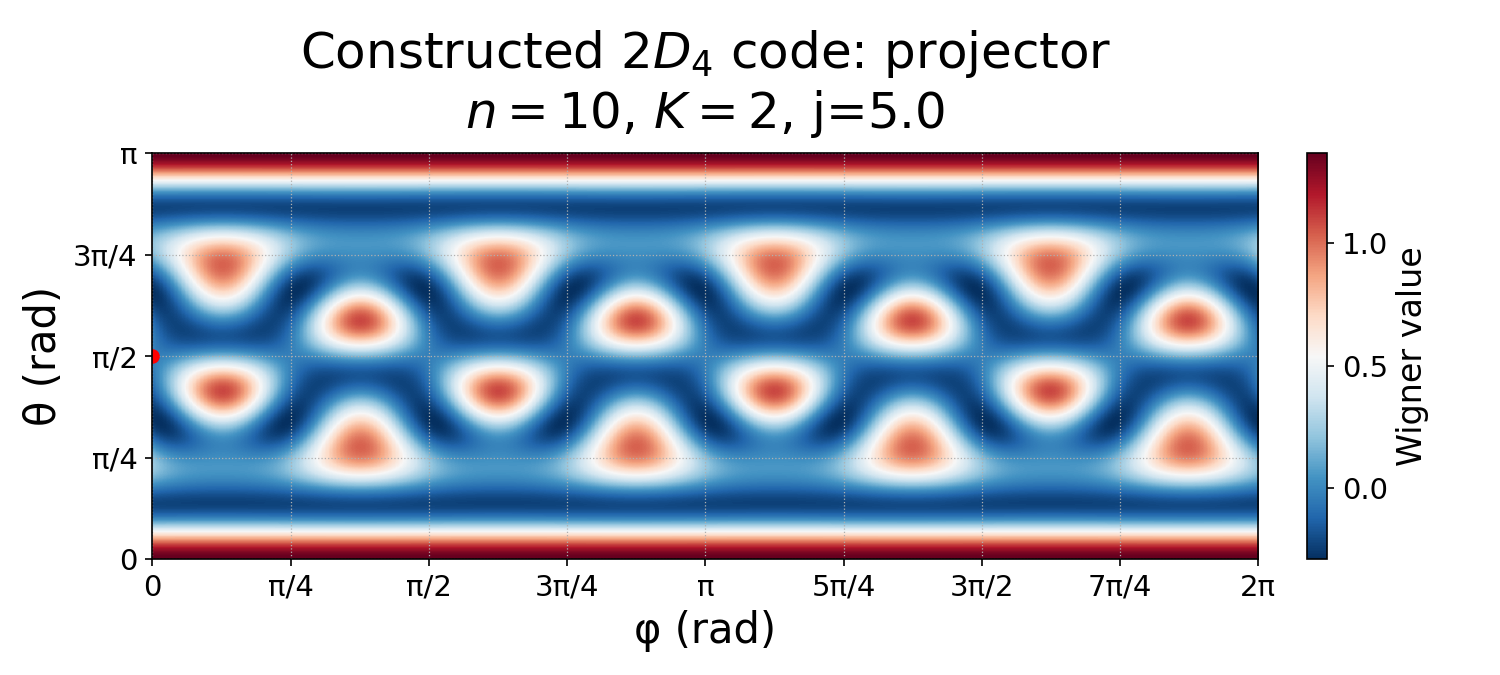}
        \caption{Constructed $2D_4$ code.}
    \end{subfigure}
    \caption{Best numerical code and constructed $2D_{4}$ reference at $n=10$. At $\eta=0.90$, their respective entanglement fidelities are $F_e^{\mathrm{num}}=0.994345$ and $F_e^{\mathrm{con}}=0.994343$.}
    \label{fig:pure-irrep-n10}
\end{figure}

\begin{figure}[H]
    \centering
    \begin{subfigure}{0.48\textwidth}
        \centering
        \includegraphics[width=\linewidth]{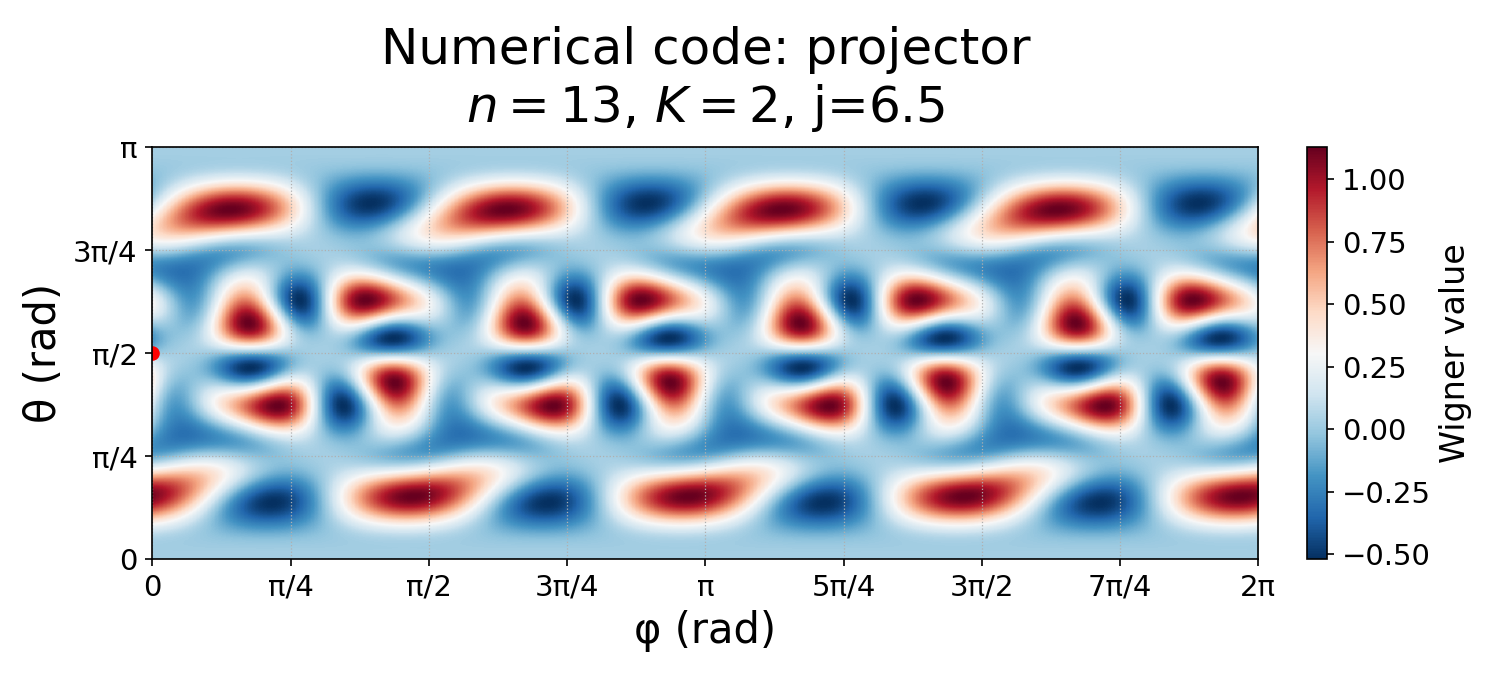}
        \caption{Numerically optimized code.}
    \end{subfigure}
    \hfill
    \begin{subfigure}{0.48\textwidth}
        \centering
        \includegraphics[width=\linewidth]{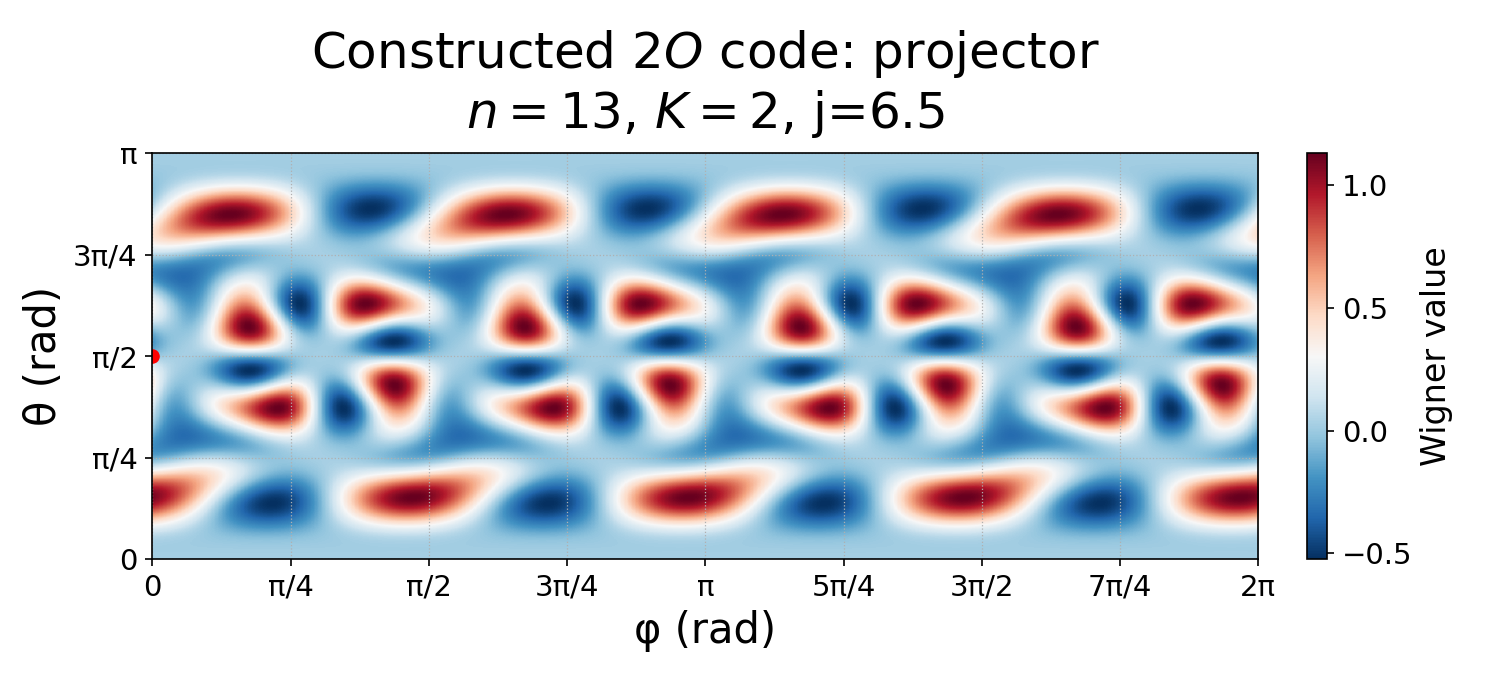}
        \caption{Constructed $2O$ code.}
    \end{subfigure}
    \caption{Best numerical code and constructed $2O$ reference at $n=13$. At $\eta=0.90$, their respective entanglement fidelities are $F_e^{\mathrm{num}}=0.997469$ and $F_e^{\mathrm{con}}=0.997466$.}
    \label{fig:pure-irrep-n13}
\end{figure}

\begin{figure}[H]
    \centering
    \begin{subfigure}{0.48\textwidth}
        \centering
        \includegraphics[width=\linewidth]{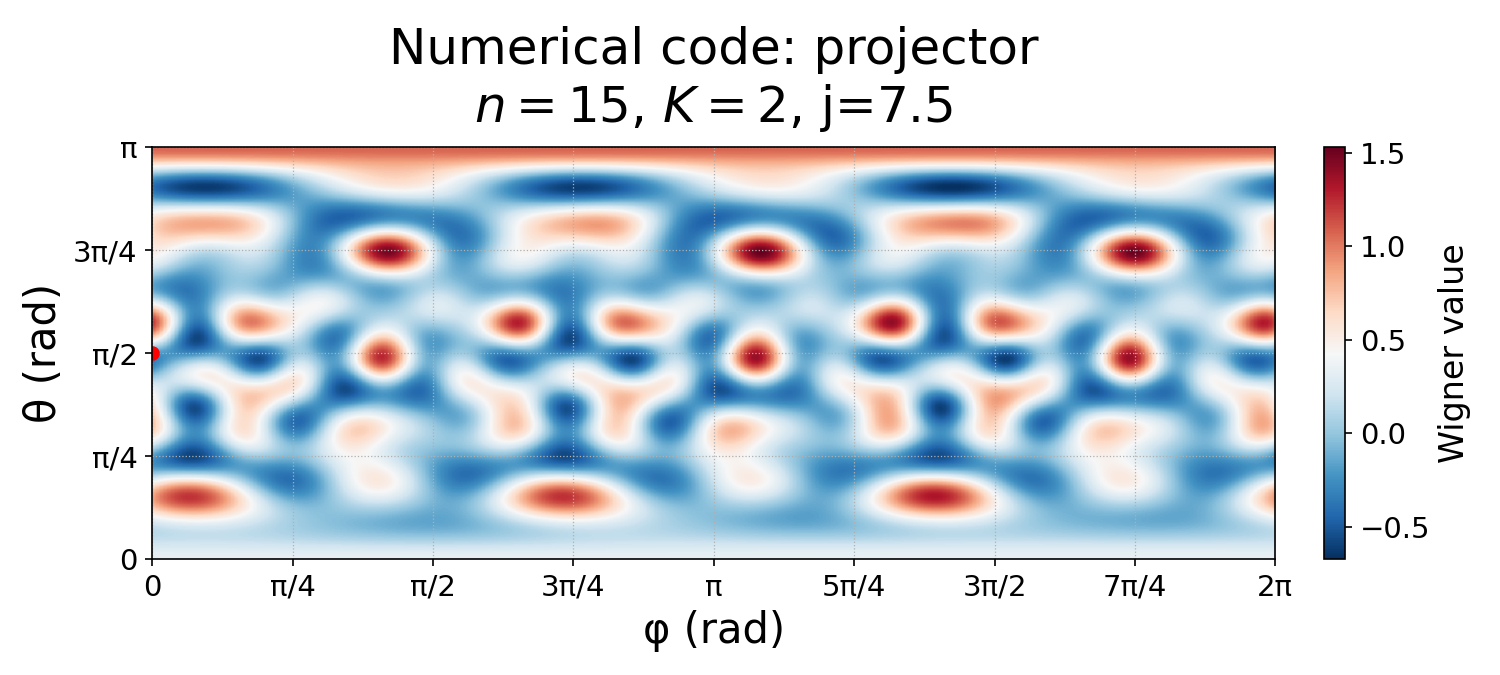}
        \caption{Best random-start numerical code.}
    \end{subfigure}
    \hfill
    \begin{subfigure}{0.48\textwidth}
        \centering
        \includegraphics[width=\linewidth]{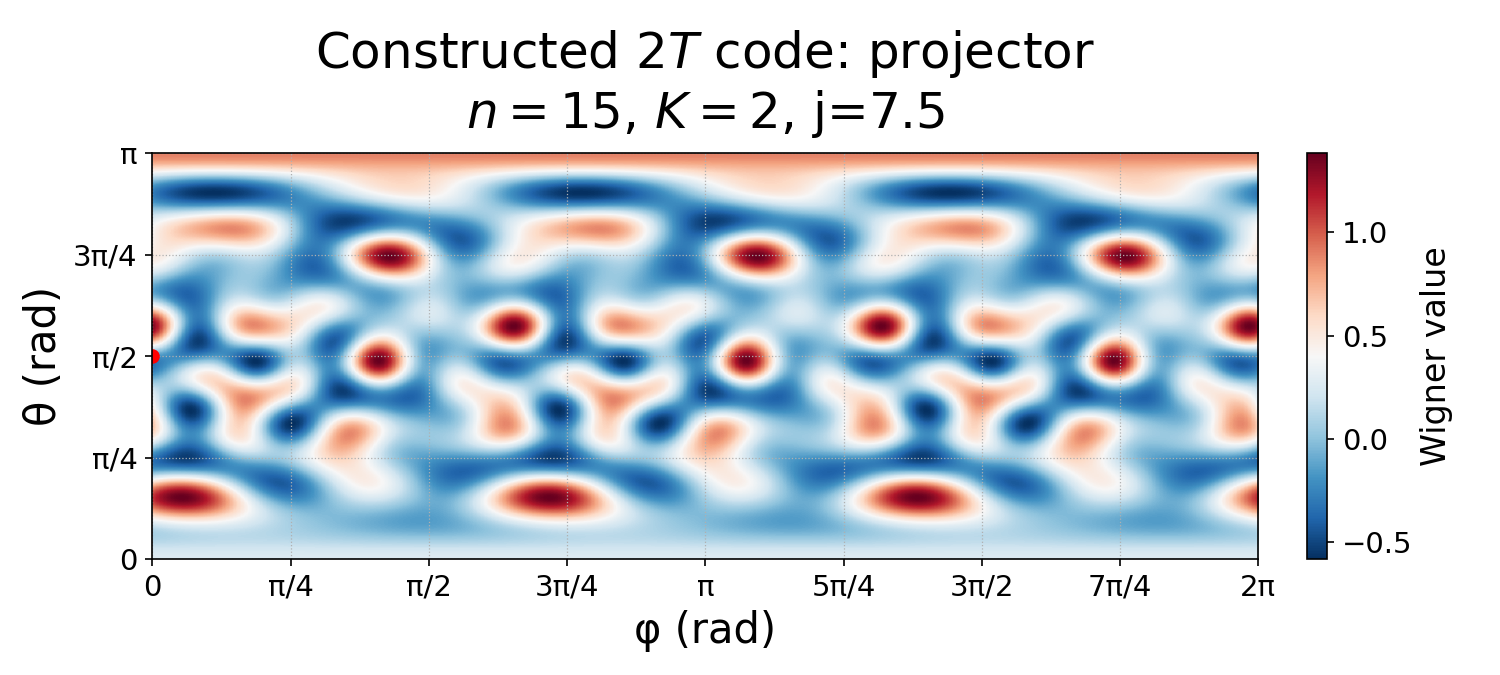}
        \caption{Constructed $2T$ code.}
    \end{subfigure}
    \caption{Best numerical code and constructed $2T$ reference at $n=15$. At $\eta=0.95$, their respective entanglement fidelities are $F_e^{\mathrm{num}}=0.999869004$ and $F_e^{\mathrm{con}}=0.999869001$. The numerical projector differs from the exact construction; the fidelity difference is within numerical decoder tolerances.}
    \label{fig:pure-irrep-n15}
\end{figure}

\begin{figure}[H]
    \centering
    \begin{subfigure}{0.48\textwidth}
        \centering
        \includegraphics[width=\linewidth]{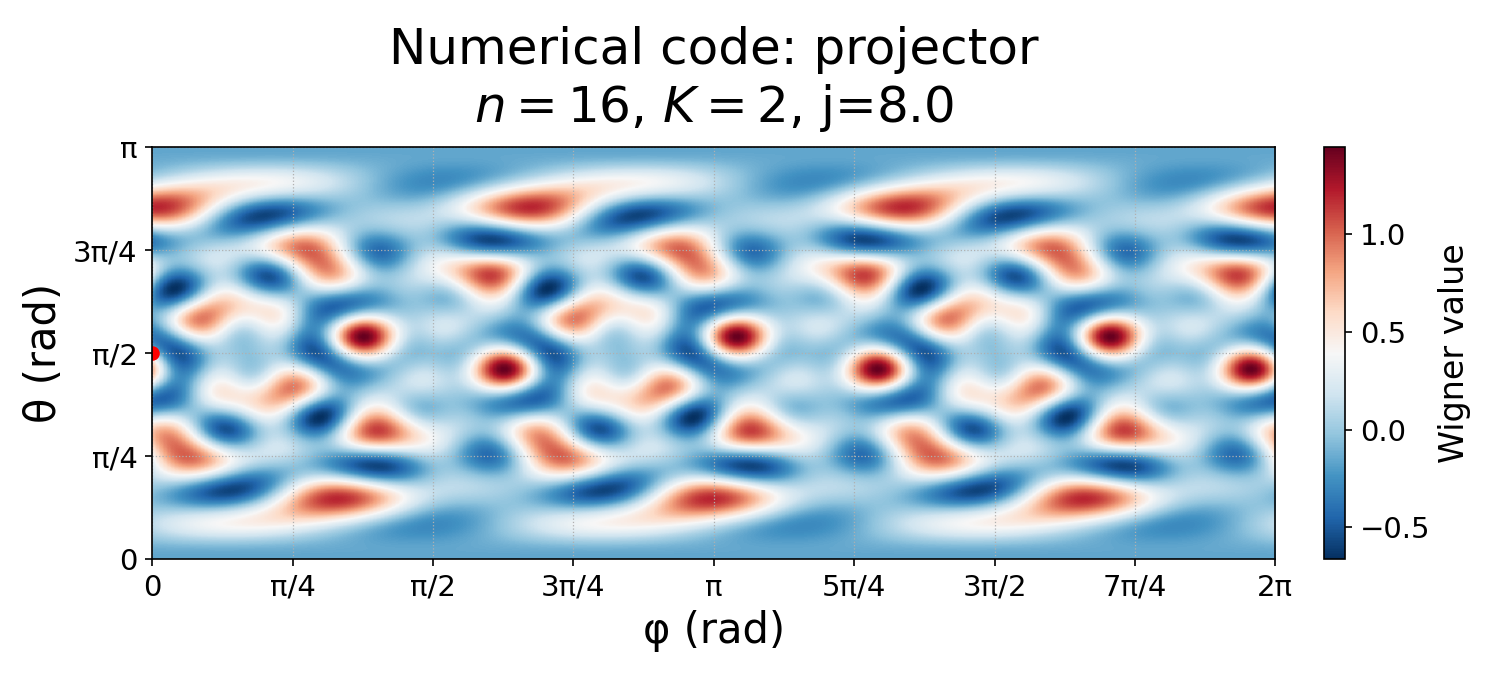}
        \caption{Numerically optimized code.}
    \end{subfigure}
    \hfill
    \begin{subfigure}{0.48\textwidth}
        \centering
        \includegraphics[width=\linewidth]{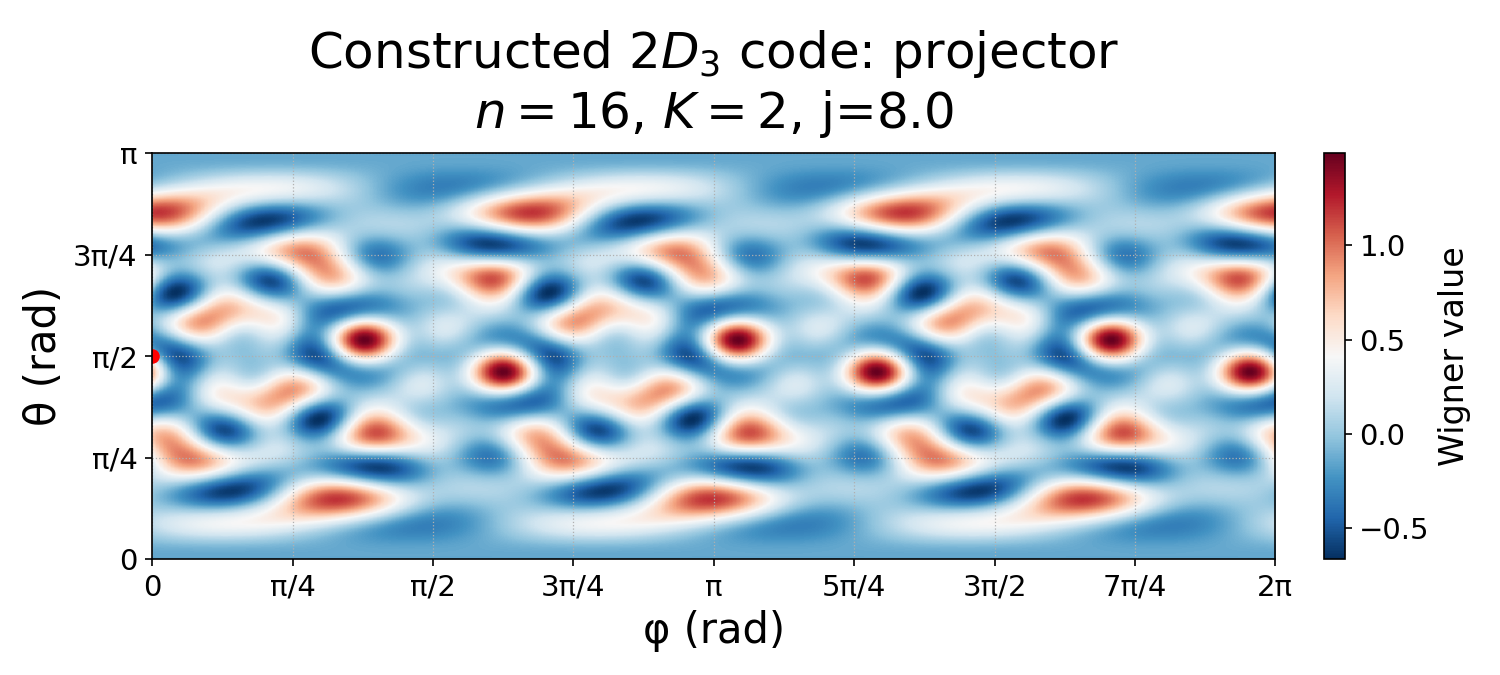}
        \caption{Constructed $2D_3$ code.}
    \end{subfigure}
    \caption{Best numerical code and constructed $2D_{3}$ reference at $n=16$. At $\eta=0.90$, their respective entanglement fidelities are $F_e^{\mathrm{num}}=0.998315$ and $F_e^{\mathrm{con}}=0.998257$.}
    \label{fig:pure-irrep-n16}
\end{figure}

\begin{figure}[H]
    \centering
    \begin{subfigure}{0.48\textwidth}
        \centering
        \includegraphics[width=\linewidth]{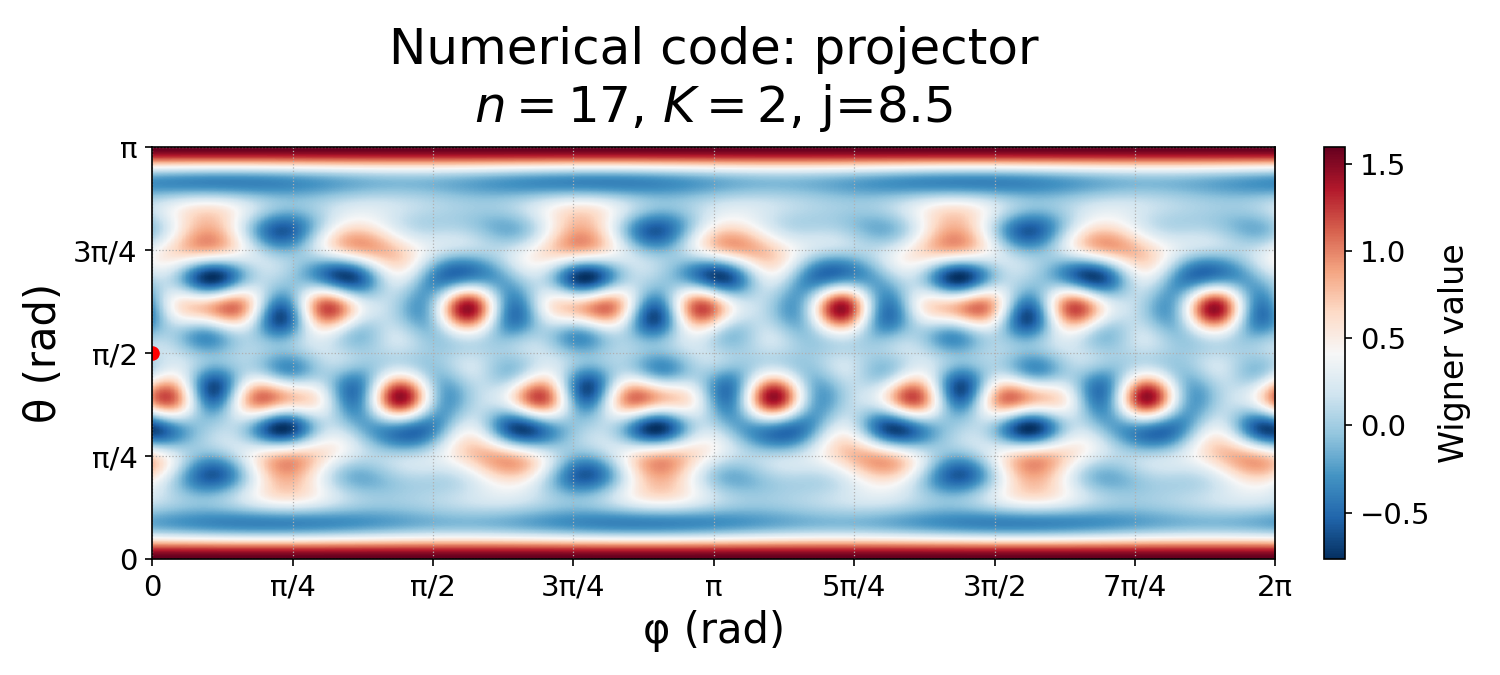}
        \caption{Numerically optimized code.}
    \end{subfigure}
    \hfill
    \begin{subfigure}{0.48\textwidth}
        \centering
        \includegraphics[width=\linewidth]{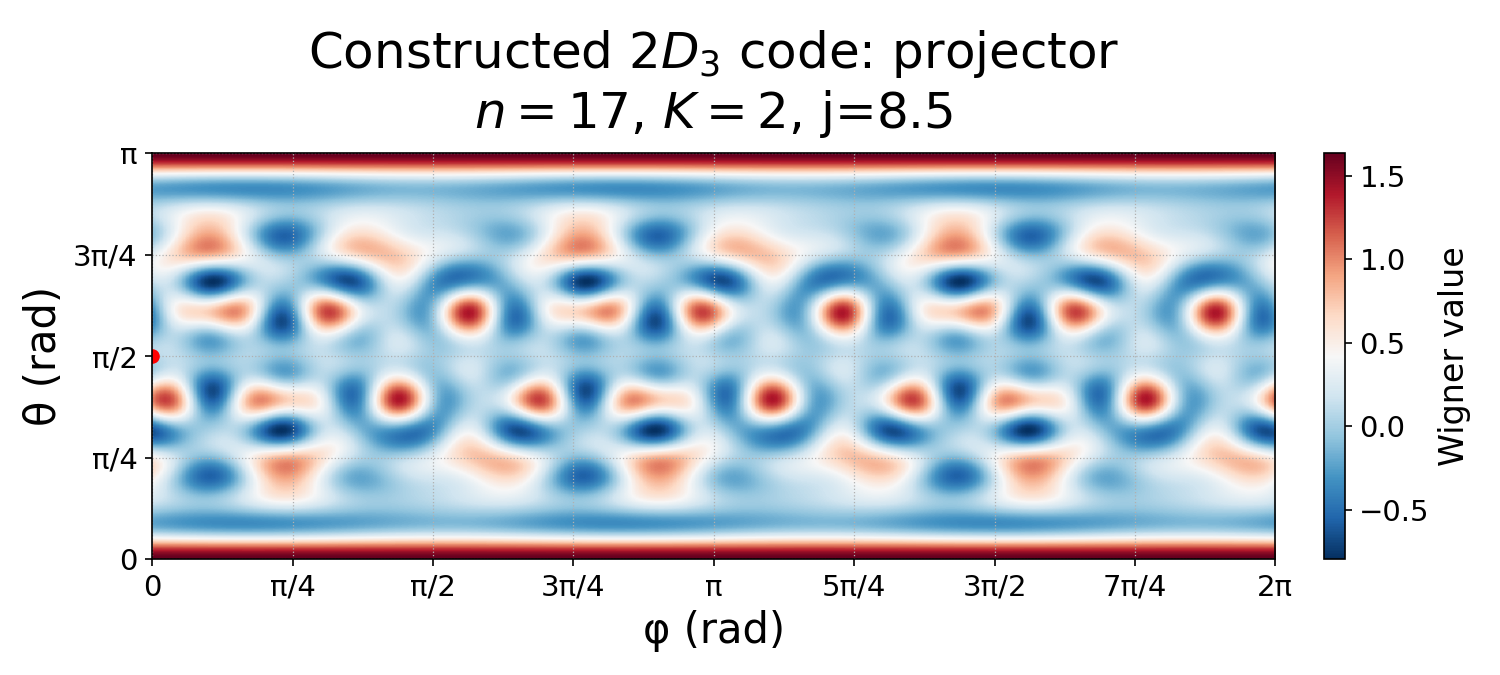}
        \caption{Constructed $2D_3$ code.}
    \end{subfigure}
    \caption{Best numerical code and constructed $2D_{3}$ reference at $n=17$. At $\eta=0.90$, their respective entanglement fidelities are $F_e^{\mathrm{num}}=0.998527$ and $F_e^{\mathrm{con}}=0.998488$.}
    \label{fig:pure-irrep-n17}
\end{figure}

\begin{figure}[H]
    \centering
    \begin{subfigure}{0.48\textwidth}
        \centering
        \includegraphics[width=\linewidth]{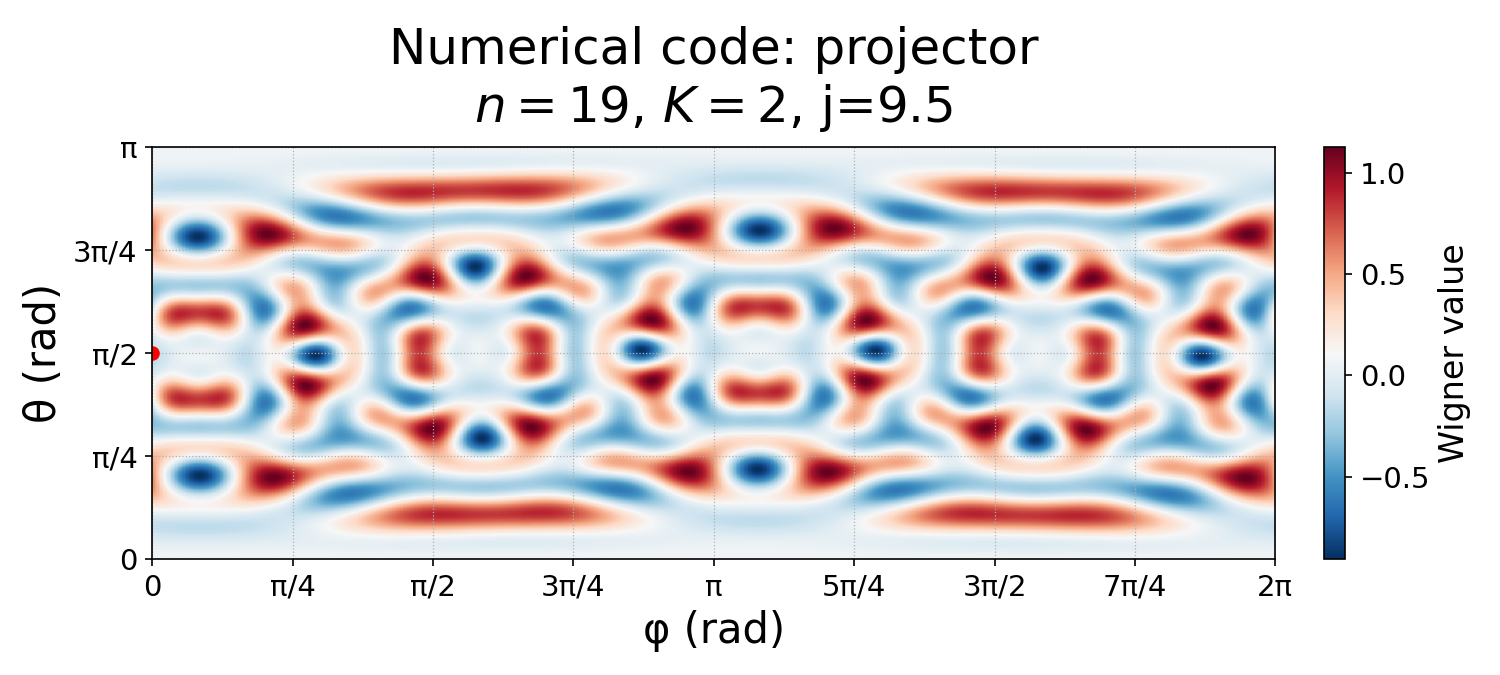}
        \caption{Numerically optimized code.}
    \end{subfigure}
    \hfill
    \begin{subfigure}{0.48\textwidth}
        \centering
        \includegraphics[width=\linewidth]{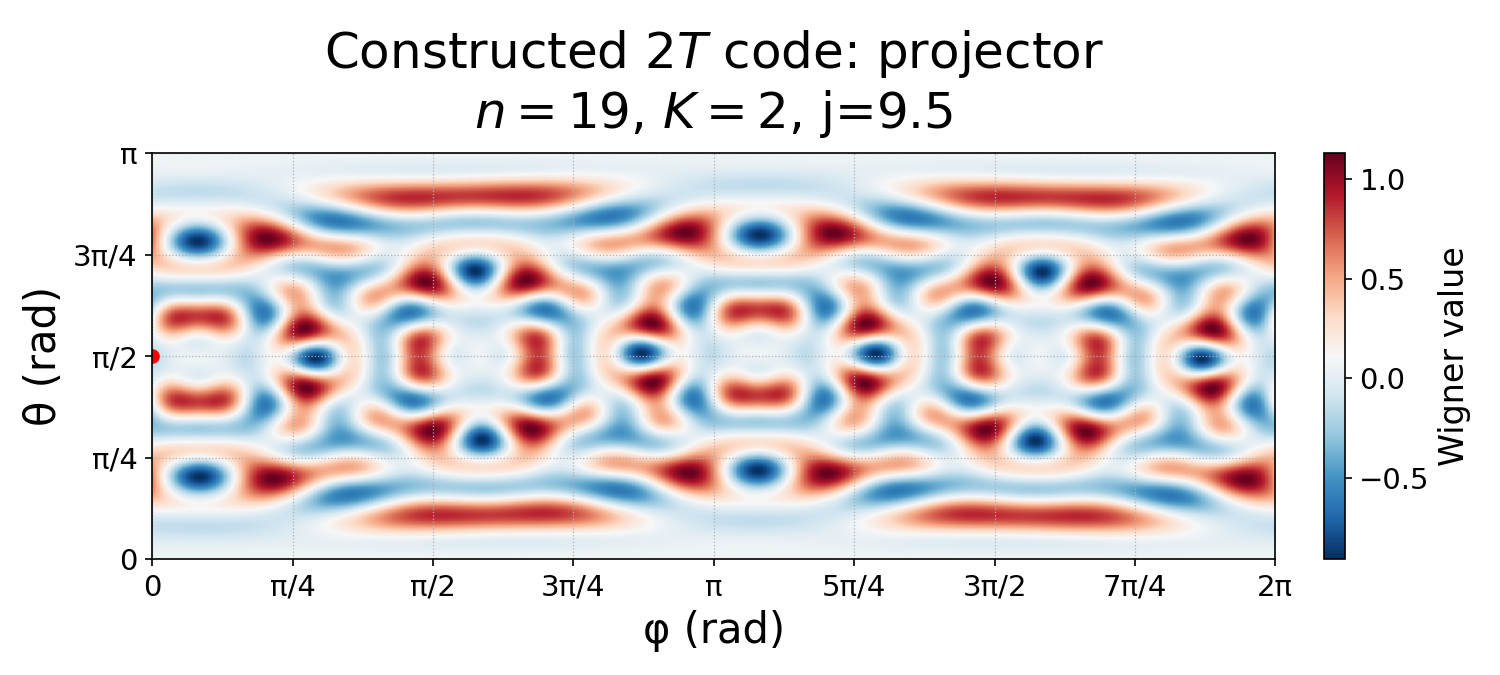}
        \caption{Constructed $2T$ code.}
    \end{subfigure}
    \caption{Best numerical code and constructed $2T$ reference at $n=19$. At $\eta=0.90$, their respective entanglement fidelities are $F_e^{\mathrm{num}}=0.999078$ and $F_e^{\mathrm{con}}=0.999077$.}
    \label{fig:pure-irrep-n19}
\end{figure}

\begin{figure}[H]
    \centering
    \begin{subfigure}{0.48\textwidth}
        \centering
        \includegraphics[width=\linewidth]{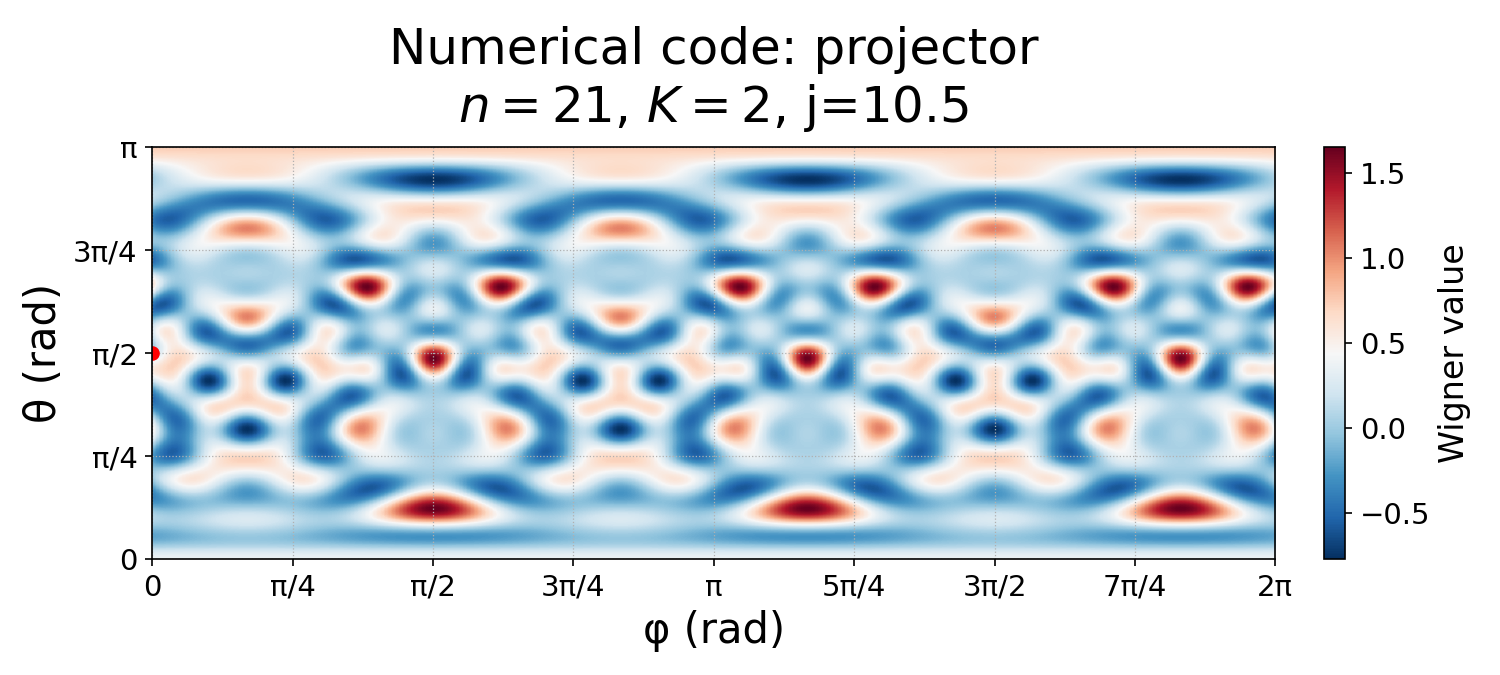}
        \caption{Numerically optimized code.}
    \end{subfigure}
    \hfill
    \begin{subfigure}{0.48\textwidth}
        \centering
        \includegraphics[width=\linewidth]{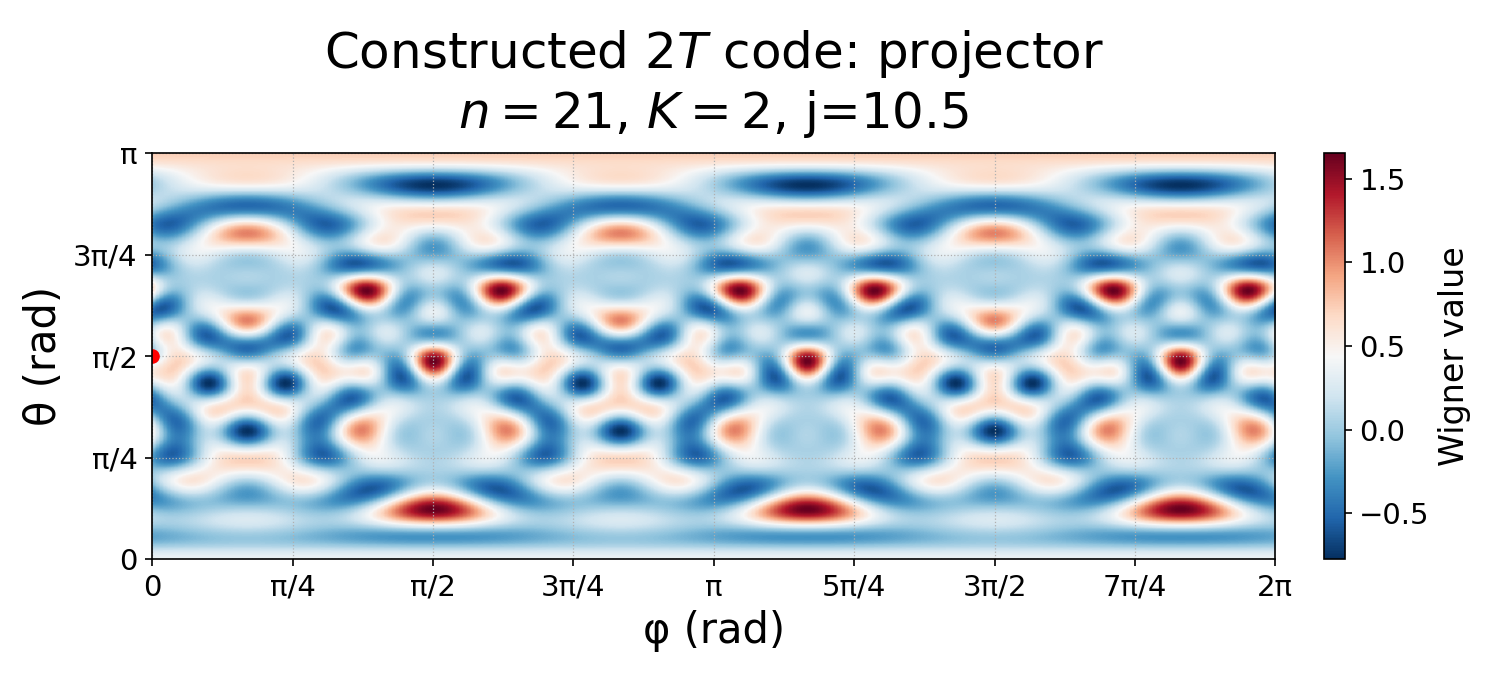}
        \caption{Constructed $2T$ code.}
    \end{subfigure}
    \caption{Best numerical code and constructed $2T$ reference at $n=21$. At $\eta=0.90$, their respective entanglement fidelities are $F_e^{\mathrm{num}}=0.999217$ and $F_e^{\mathrm{con}}=0.999217$.}
    \label{fig:pure-irrep-n21}
\end{figure}

\begin{figure}[H]
    \centering
    \begin{subfigure}{0.48\textwidth}
        \centering
        \includegraphics[width=0.85\linewidth]{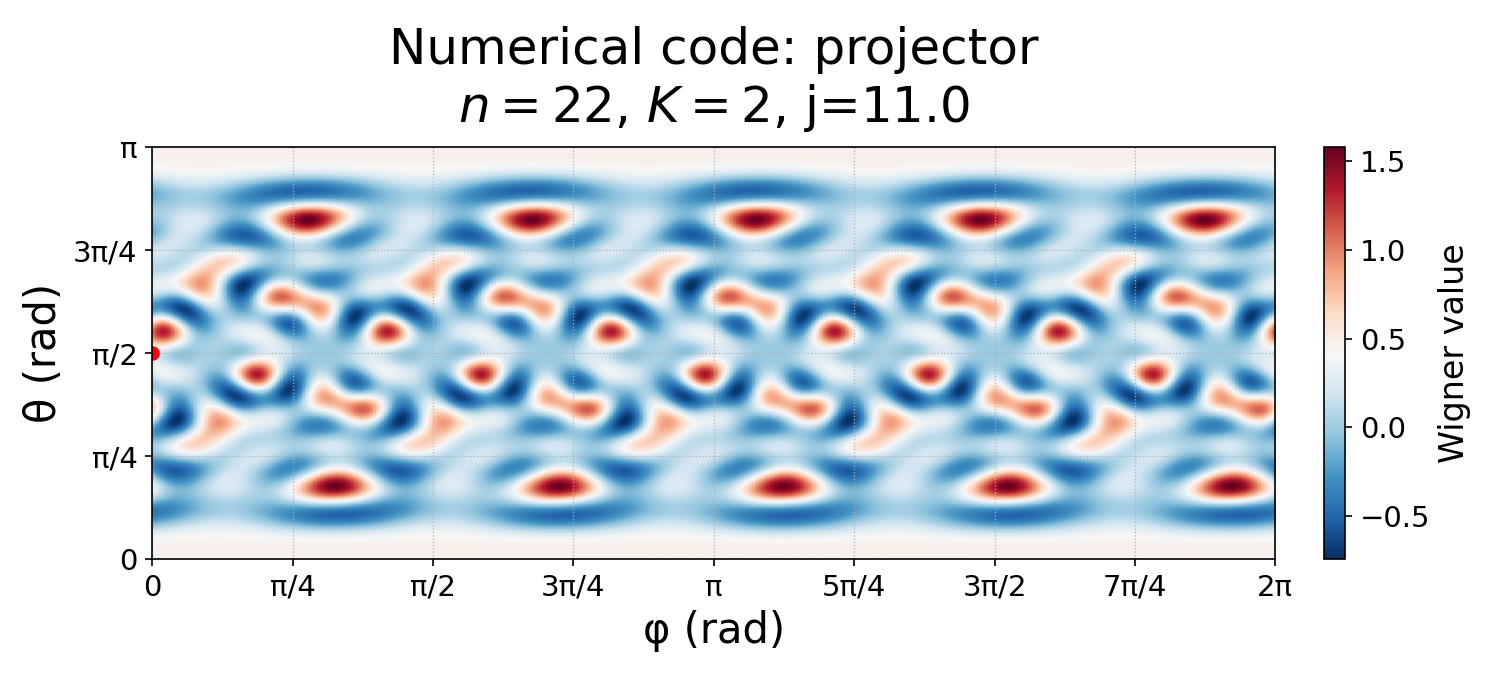}
        \caption{Numerically optimized code.}
    \end{subfigure}
    \hfill
    \begin{subfigure}{0.48\textwidth}
        \centering
        \includegraphics[width=0.85\linewidth]{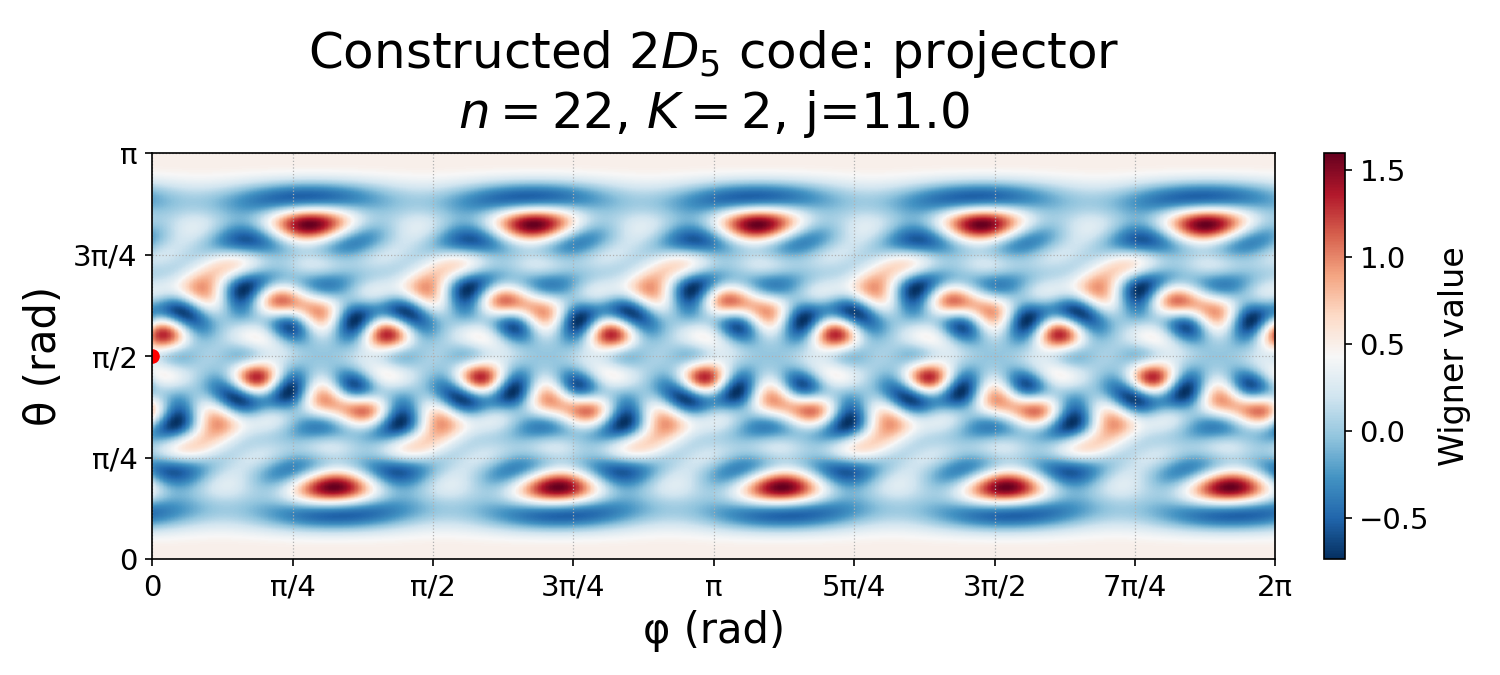}
        \caption{Constructed $2D_5$ code.}
    \end{subfigure}
    \caption{Best numerical code and constructed $2D_{5}$ reference at $n=22$. At $\eta=0.90$, their respective entanglement fidelities are $F_e^{\mathrm{num}}=0.999355$ and $F_e^{\mathrm{con}}=0.999349$.}
    \label{fig:pure-irrep-n22}
\end{figure}

\begin{figure}[H]
    \centering
    \begin{subfigure}{0.48\textwidth}
        \centering
        \includegraphics[width=0.85\linewidth]{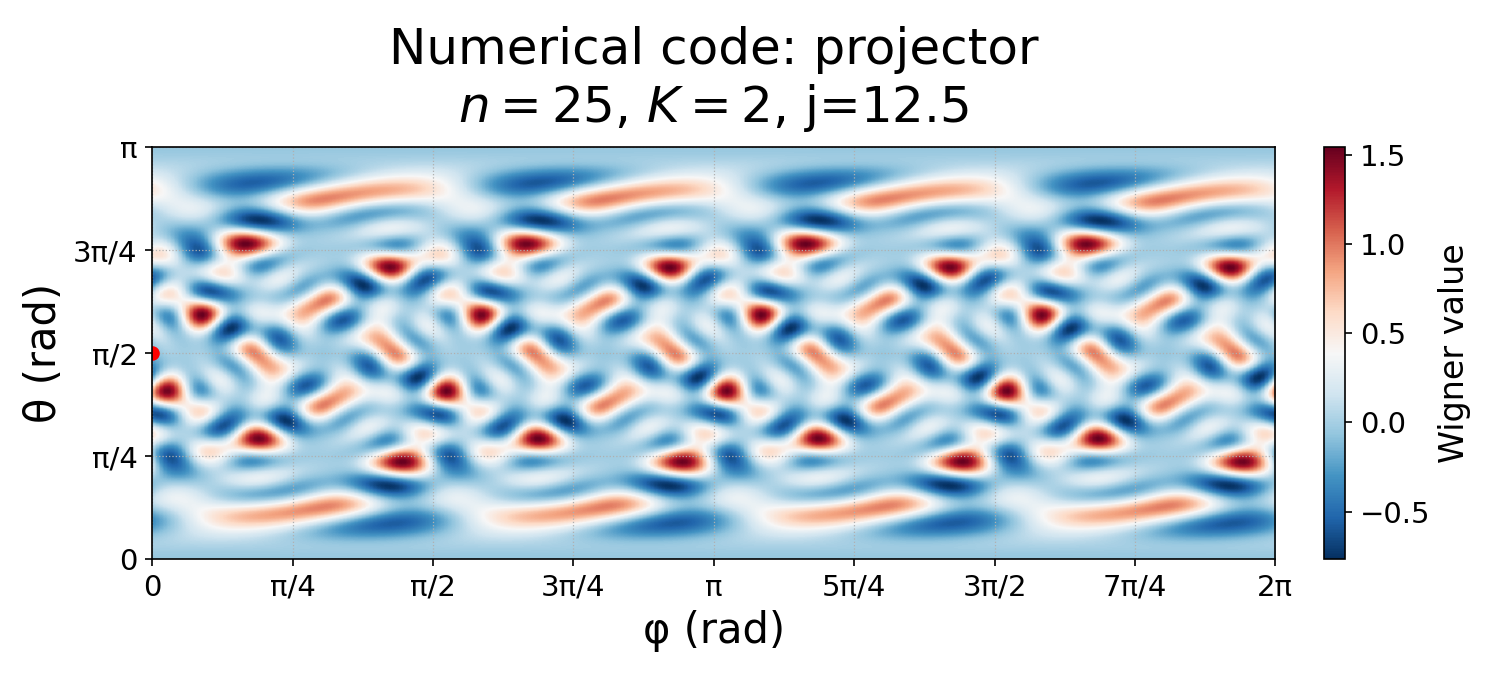}
        \caption{Numerically optimized code.}
    \end{subfigure}
    \hfill
    \begin{subfigure}{0.48\textwidth}
        \centering
        \includegraphics[width=0.85\linewidth]{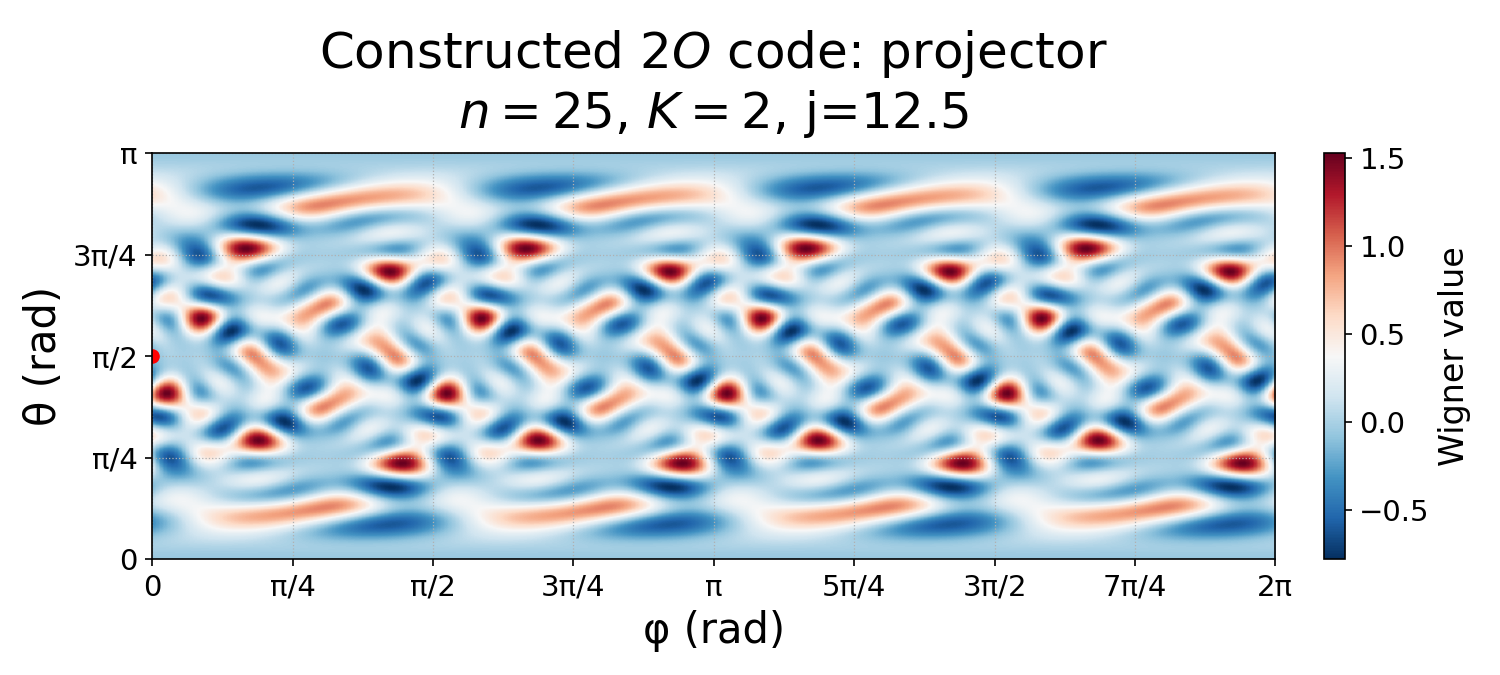}
        \caption{Constructed $2O$ code.}
    \end{subfigure}
    \caption{Best numerical code and constructed $2O$ reference at $n=25$. At $\eta=0.90$, their respective entanglement fidelities are $F_e^{\mathrm{num}}=0.999561$ and $F_e^{\mathrm{con}}=0.999547$.}
    \label{fig:pure-irrep-n25}
\end{figure}

We summarize the novelty of these codes as follows. Beyond the $n=4$ CLY code described above, the $n=7$ binary icosahedral code is known and was first constructed in \cite{grossDesigningCodesInteractions2021}. Likewise, our $n=13$ $2O$ code is an instance of the octahedral code in \cite{grossDesigningCodesInteractions2021}. However, in the past, it's been described as a distance-$3$ code (equivalently correcting up to $2$ losses); here, we show that it can actually be made distance $4$, which also makes it achieve distance $4$ at the smallest possible $n$, according to our Farkas certificates. An $n=25$ $2O$ code that detects spin tensors through rank $4$ was also constructed in \cite{omanakuttanMultispinCliffordCodes2023}. Nevertheless, our multiplicity-space vector choice is different, and results in a significant reduction of finite-$\eta$ infidelity. Specifically, at $\eta = 0.9$, our construction has fidelity $F_e = 0.999547$, whereas the Omanakuttan-Gross code (with their listed approximate codewords) gives $F_e = 0.999114 $. This is a small difference in absolute terms, but corresponds to a $\approx 49\%$ infidelity reduction, and the corresponding codes have a best-found overlap of only $52.2\%$. Finally,  the $2T$ codes at $n=15,19,21$, as well as the specific binary dihedral codes ($2D_3$ at $n=16,17$, $2D_4$ at $n = 10$, and $2D_5$ at $n=22$), are novel to the best of our knowledge.

\subsubsection{More pure-irrep codes}

We also present a few additional pure-irrep constructions that appear in Table~\ref{tab:farkas-constructed-distances}, as well as the $n =37$ instance of our binary polyhedral family. Unlike the preceding examples, these constructions need not agree with the finite-$\eta$ numerical optimum at the same photon number (we haven't optimized at $n =37$; this statement applies to the $n \leq 25$ codes). However, some of them are used to achieve the optimal distance at their $n$, so they're still significant. Their projector spin-Wigner functions are shown in Figure~\ref{fig:additional-pure-irrep-codes}. 

\begin{figure}[H]
    \centering
    \begin{subfigure}[t]{0.48\textwidth}
        \centering
        \includegraphics[width=\linewidth]{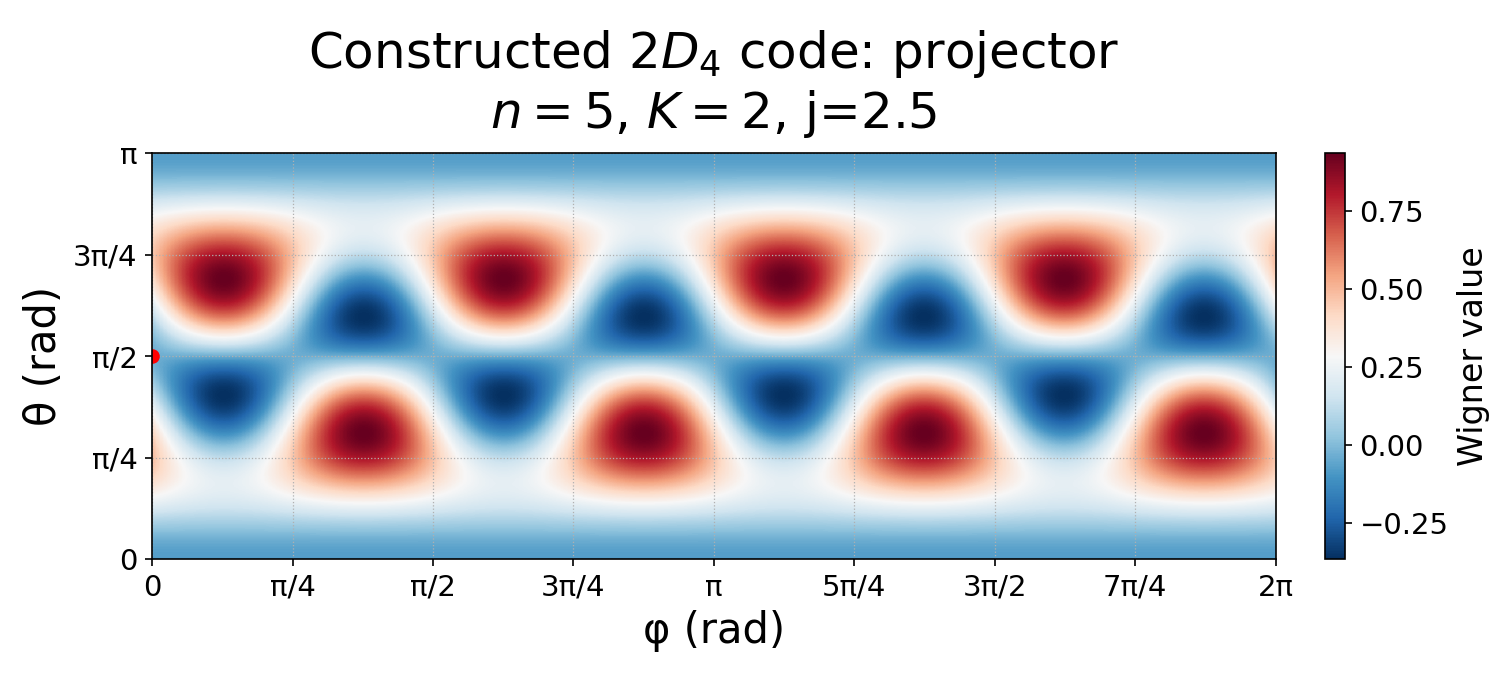}
        \caption{$n=5$, $2D_4:\rho_7$. At $\eta=0.9$, $F_e=0.977109$.}
    \end{subfigure}
    \hfill
    \begin{subfigure}[t]{0.48\textwidth}
        \centering
        \includegraphics[width=\linewidth]{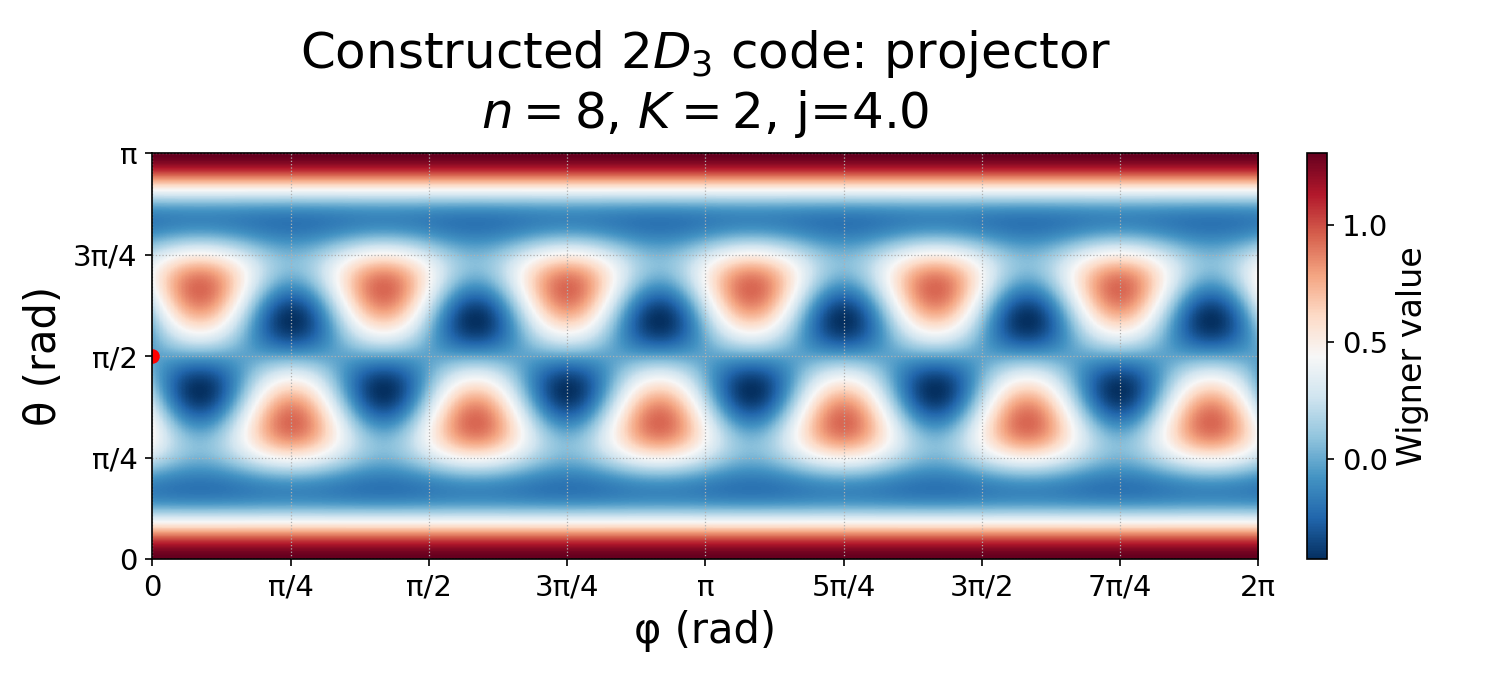}
        \caption{$n=8$, $2D_3:\rho_5$. At $\eta=0.9$, $F_e=0.990021$.}
    \end{subfigure}

    \medskip

    \begin{subfigure}[t]{0.48\textwidth}
        \centering
        \includegraphics[width=\linewidth]{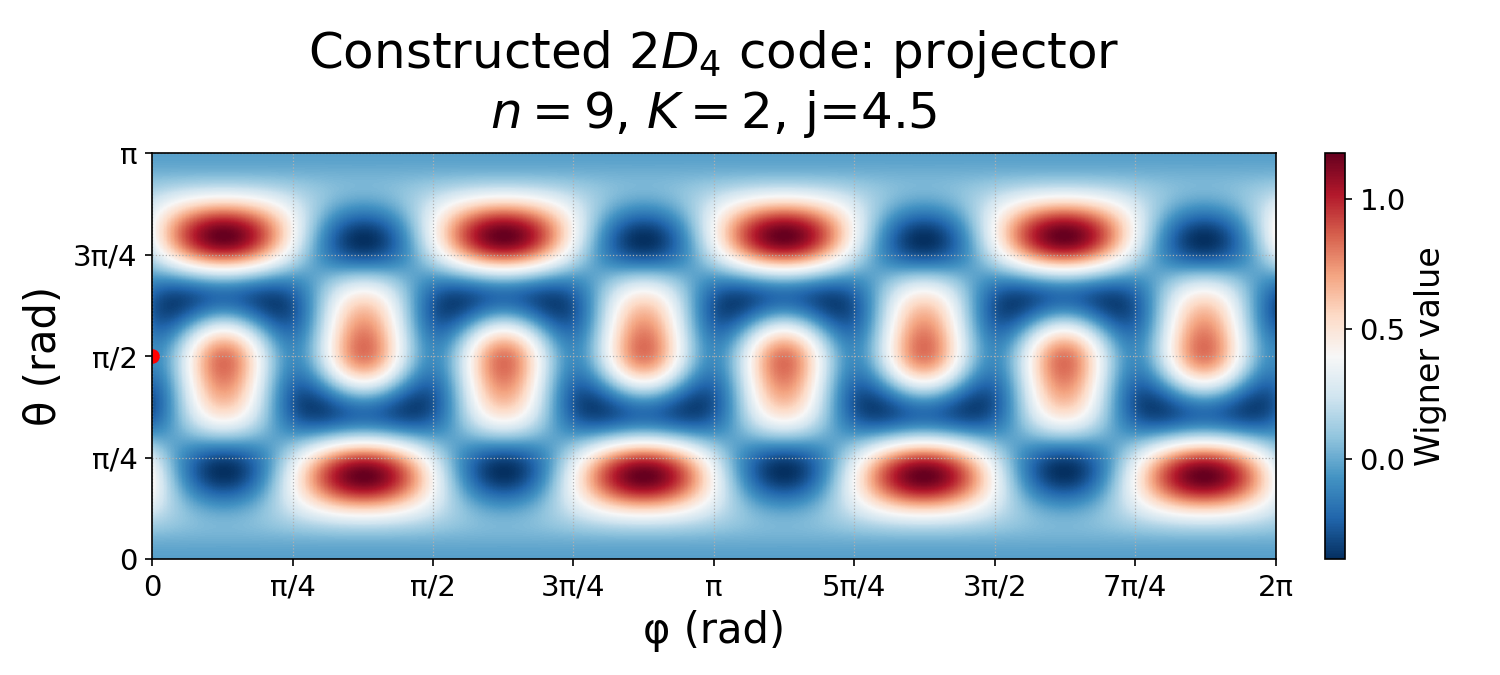}
        \caption{$n=9$, $2D_4:\rho_6$. At $\eta=0.9$, $F_e=0.991155$.}
    \end{subfigure}
    \hfill
    \begin{subfigure}[t]{0.48\textwidth}
        \centering
        \includegraphics[width=\linewidth]{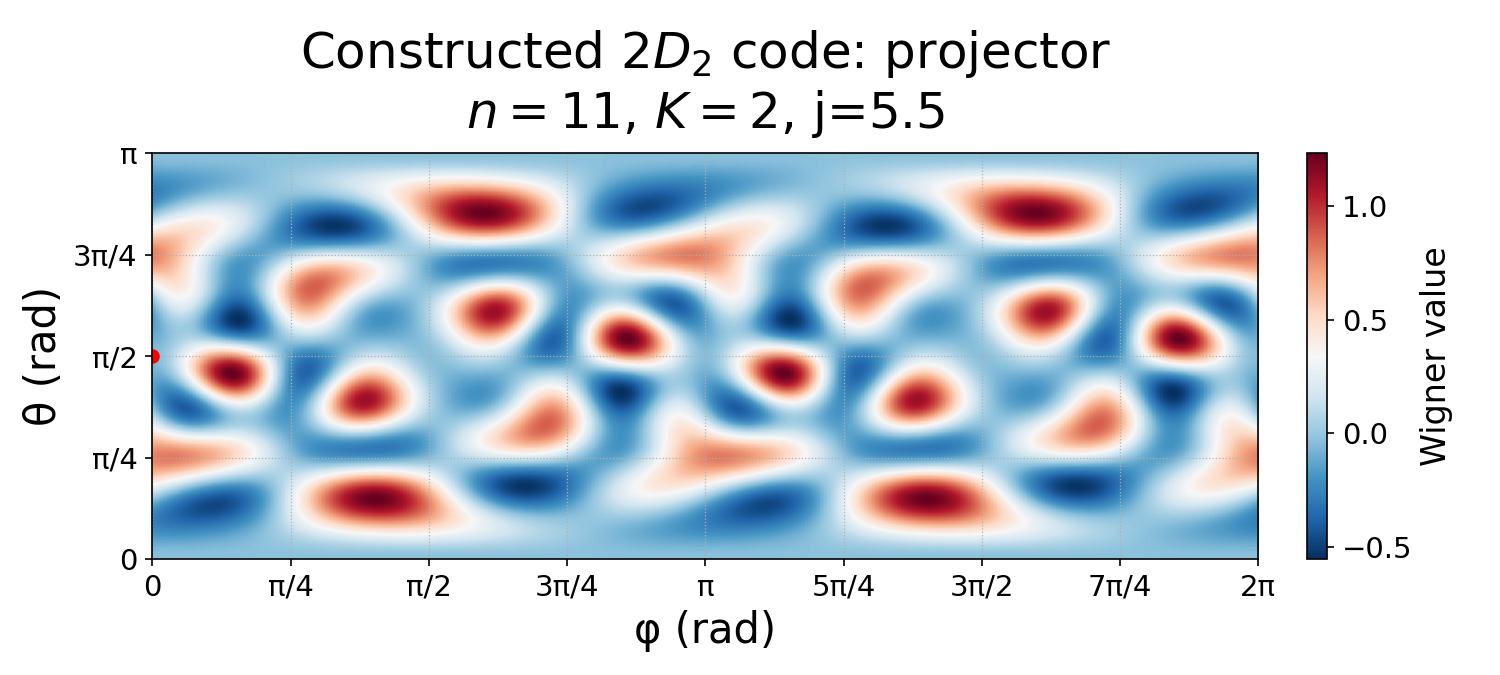}
        \caption{$n=11$, $2D_2:\rho_5$. At $\eta=0.95$, $F_e=0.999459$.}
    \end{subfigure}
    \caption{Projector spin-Wigner functions of additional pure-irrep constructions in their canonical symmetry orientations.}
    \label{fig:additional-pure-irrep-codes}
\end{figure}

\begin{figure}[H]
    \ContinuedFloat
    \centering

    \begin{subfigure}[t]{0.48\textwidth}
        \centering
        \includegraphics[width=\linewidth]{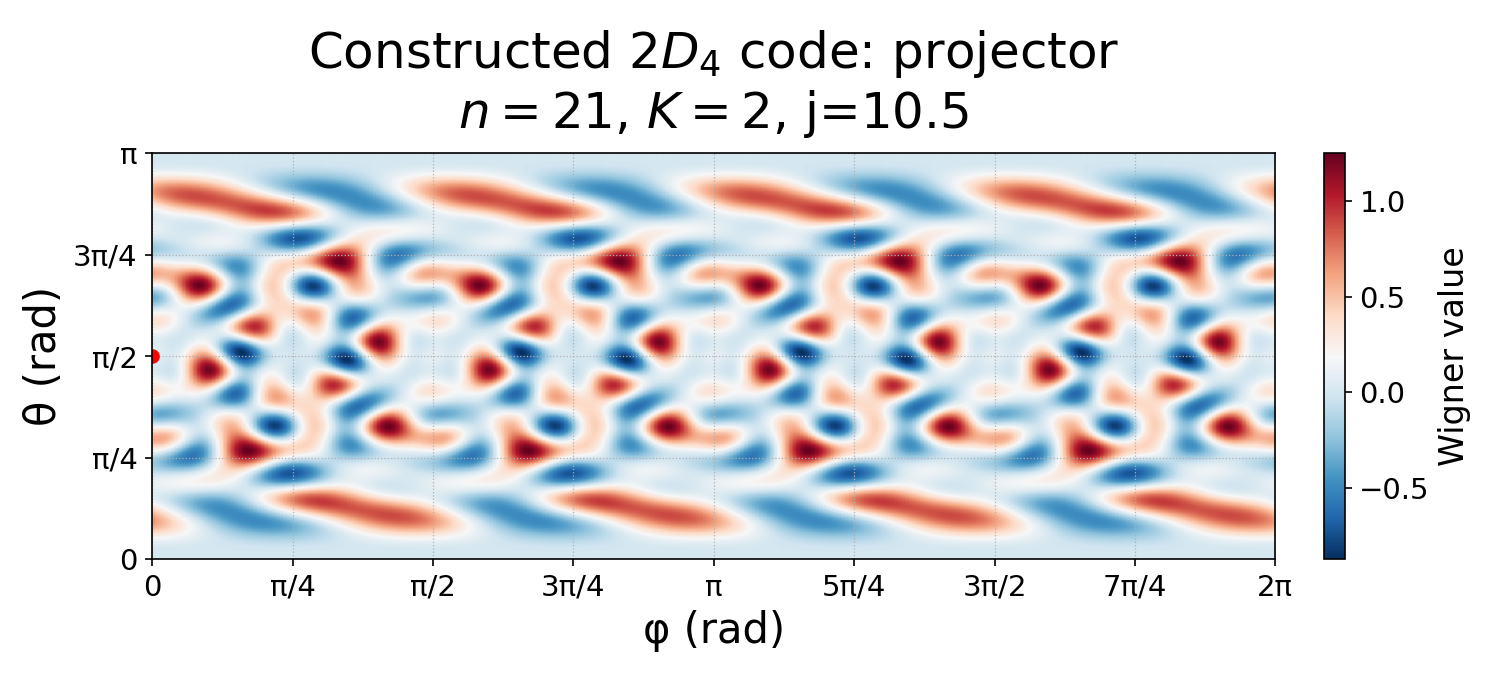}
        \caption{$n=21$, $2D_4:\rho_7$. At $\eta=0.9$, $F_e=0.999179$.}
    \end{subfigure}
    \hfill
    \begin{subfigure}[t]{0.48\textwidth}
        \centering
        \includegraphics[width=\linewidth]{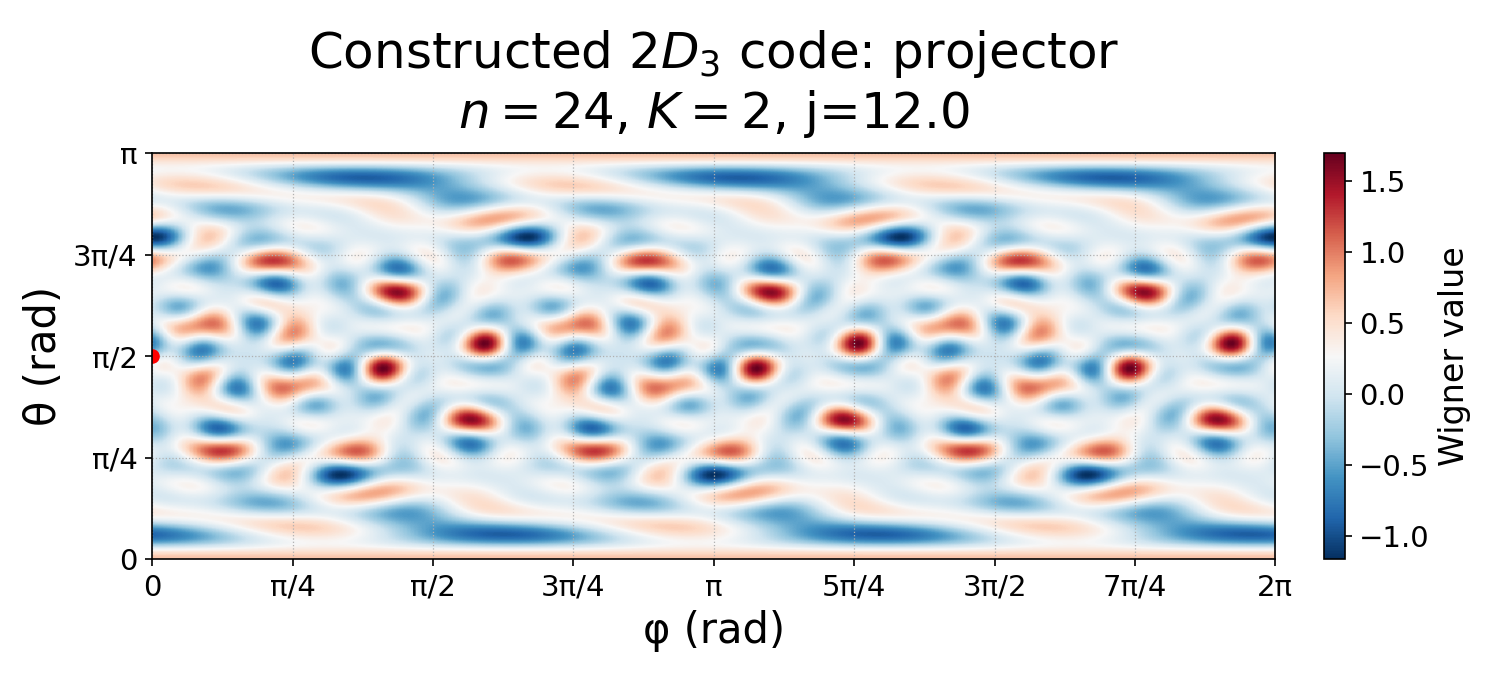}
        \caption{$n=24$, $2D_3:\rho_5$. At $\eta=0.9$, $F_e=0.999300$.}
    \end{subfigure}
    \medskip

    \begin{subfigure}[t]{0.48\textwidth}
        \centering
        \includegraphics[width=\linewidth,height=0.19\textheight,keepaspectratio]{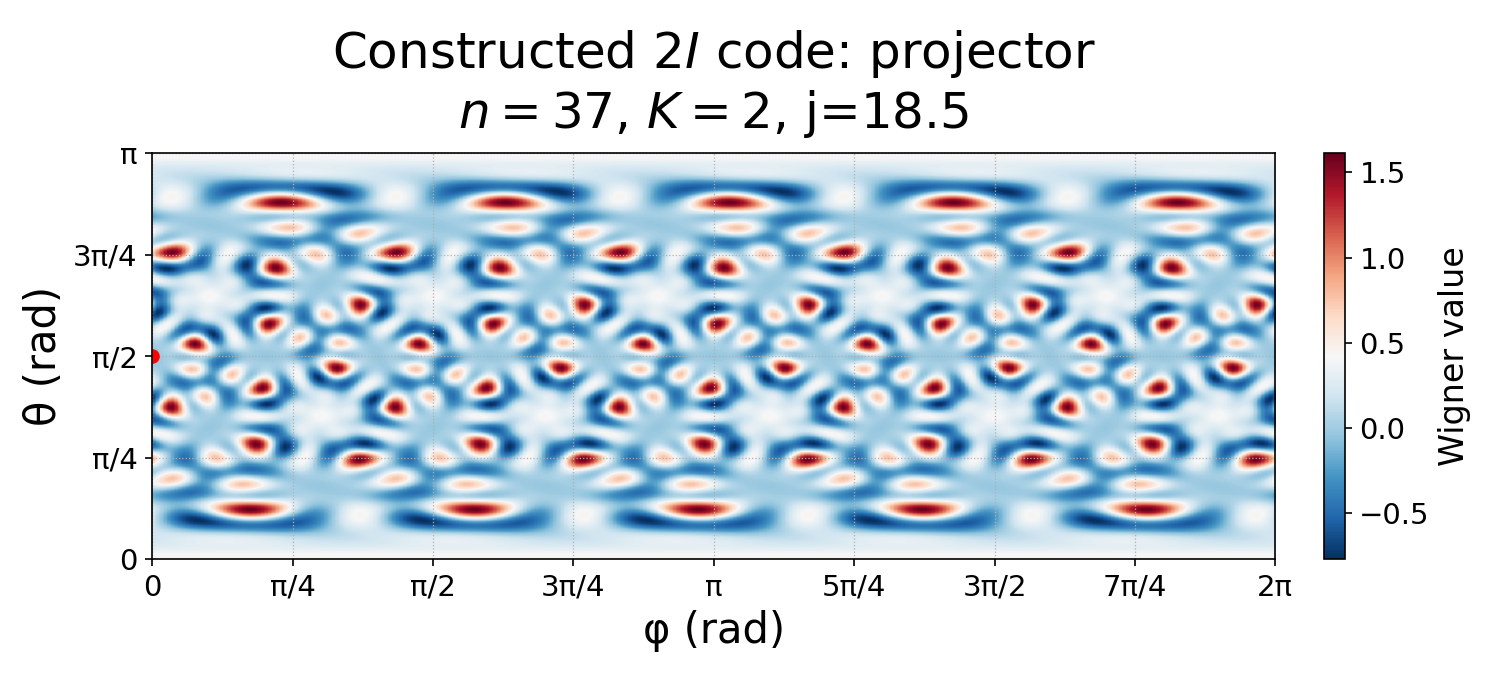}
        \caption{$n=37$, $2I:\rho_3$. At $\eta=0.9$, $F_e=0.999873$.}
    \end{subfigure}
    \hfill
    \begin{subfigure}[t]{0.48\textwidth}
        \centering
        \includegraphics[width=\linewidth,height=0.19\textheight,keepaspectratio]{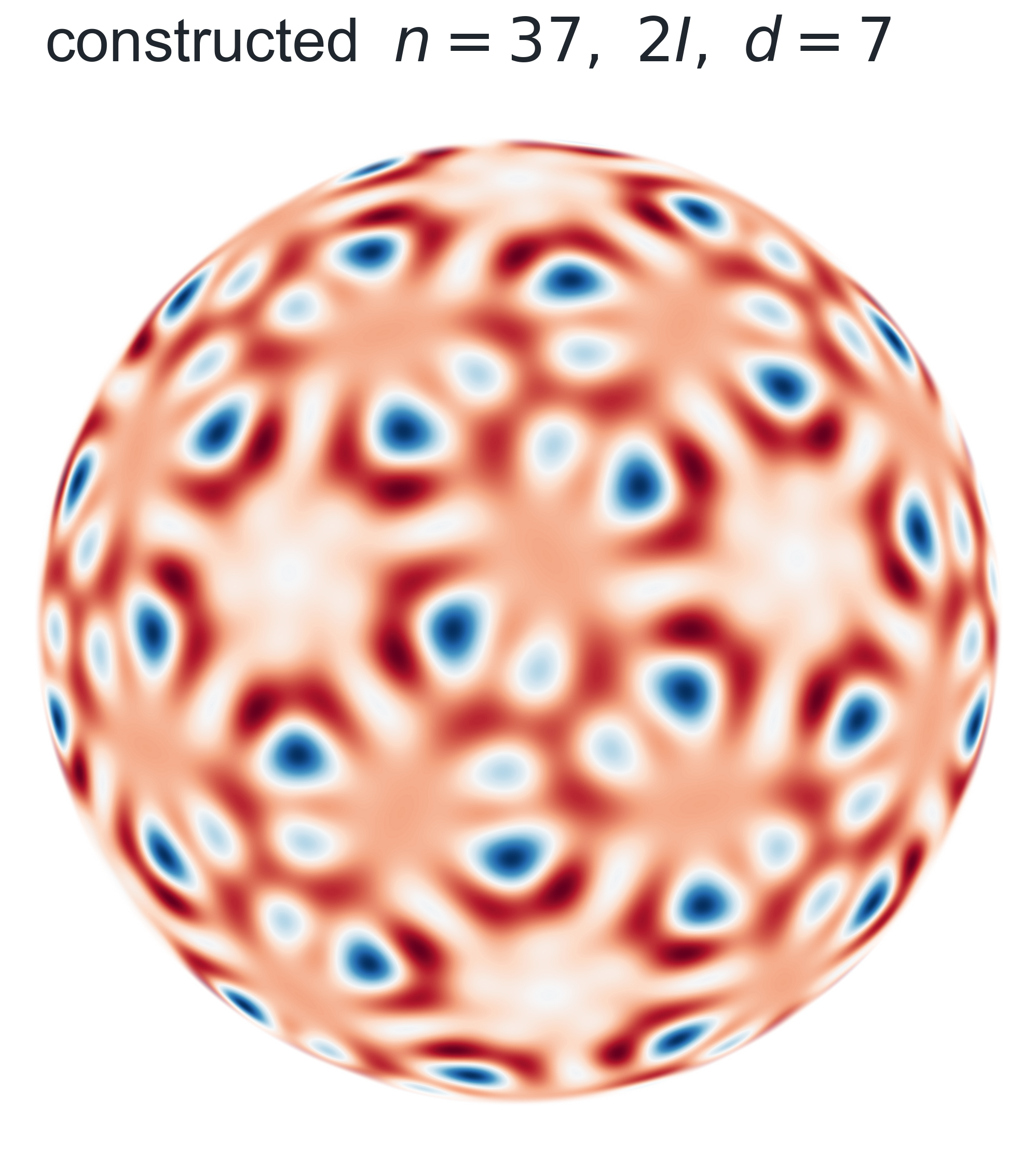}
        \caption{Spherical view of the same $n=37$, $2I:\rho_3$ projector; blue denotes larger Wigner values.}
    \end{subfigure}
    \caption[]{(Continued.) The $n=37$ construction is shown in both plane and spherical views. The spherical panel reverses the color convention of the plane panels.}
\end{figure}

\subsubsection{Numerical comparisons with reducible constructions}
We compare the remaining selected numerical codes with the corresponding reducible constructions. Each analytic reference is shown in its canonical symmetry frame, and the numerical projector is aligned to it. At $n=11$, we align the common $C_4$ symmetry axes; maximizing overlap without this restriction tilts the axes apart. These visual alignments do not imply equality of the codes; non-negligible differences remain, particularly for some cyclic-group constructions. Irrep labels in Table~\ref{tab:numerical-code-symmetry-classification} refer to the recorded classification frame and can change under an equivalent rotation or conjugation.

\begin{figure}[H]
    \centering
    \begin{subfigure}{0.48\textwidth}
        \centering
        \includegraphics[width=\linewidth]{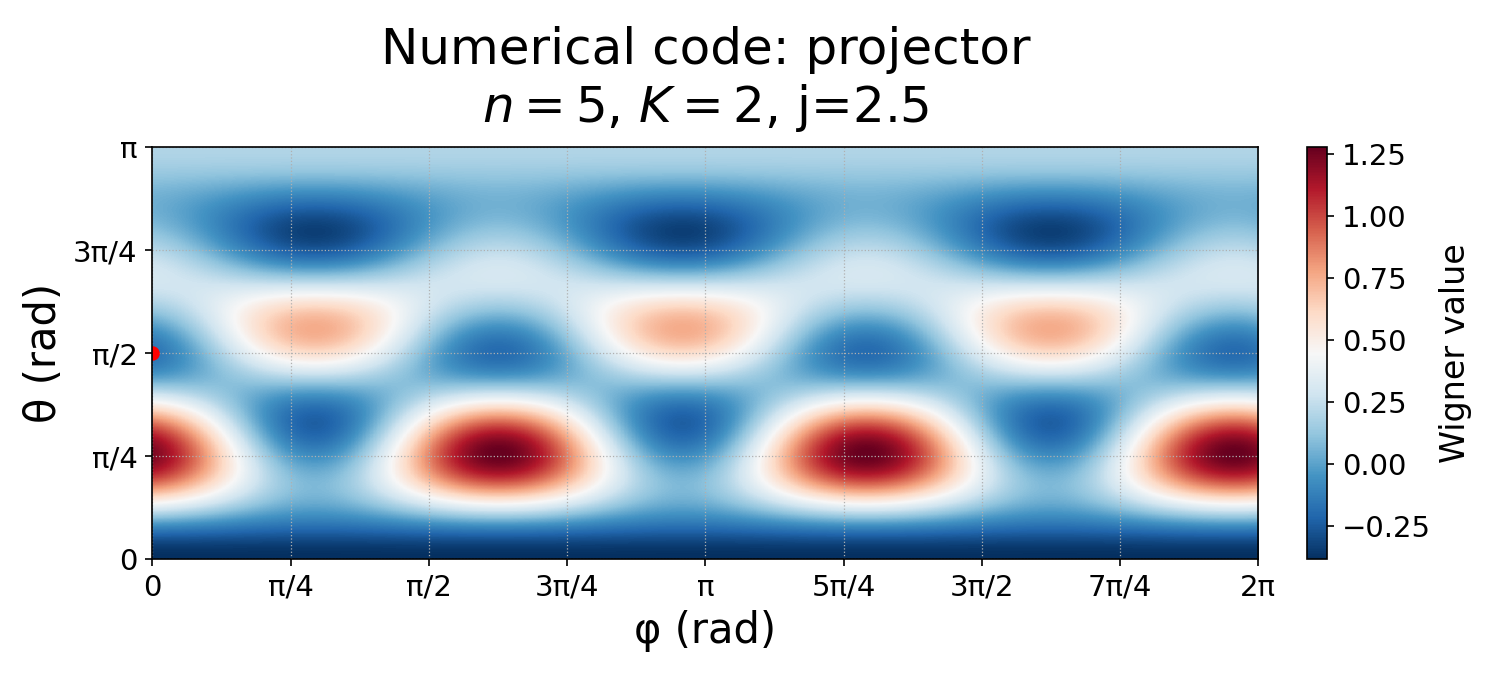}
        \caption{Numerically optimized code.}
    \end{subfigure}
    \hfill
    \begin{subfigure}{0.48\textwidth}
        \centering
        \includegraphics[width=\linewidth]{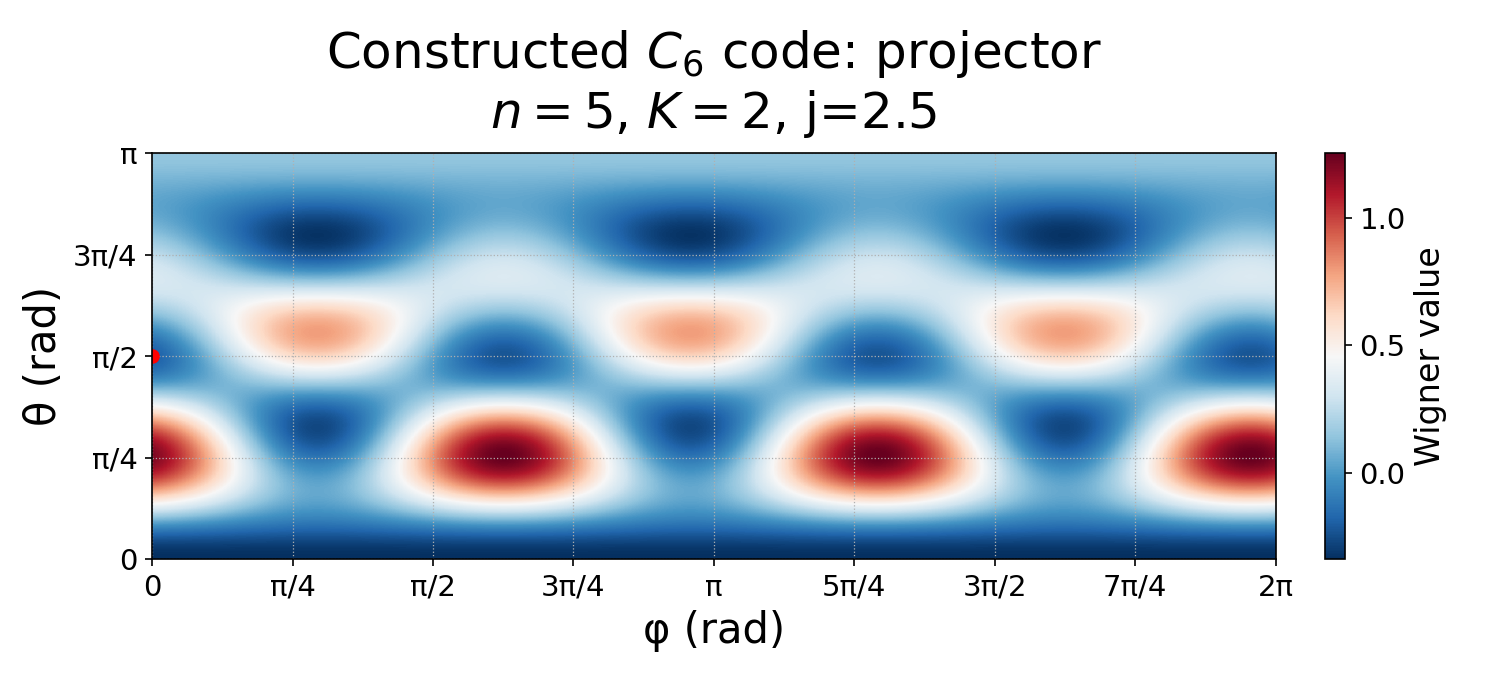}
        \caption{Constructed $C_6:\rho_2\oplus\rho_6$ code.}
    \end{subfigure}
    \caption{Best numerical code and constructed $C_{6}$ reference at $n=5$. At $\eta=0.90$, their respective entanglement fidelities are $F_e^{\mathrm{num}}=0.979710$ and $F_e^{\mathrm{con}}=0.979326$.}
    \label{fig:reducible-n5}
\end{figure}

\begin{figure}[H]
    \centering
    \begin{subfigure}{0.48\textwidth}
        \centering
        \includegraphics[width=\linewidth]{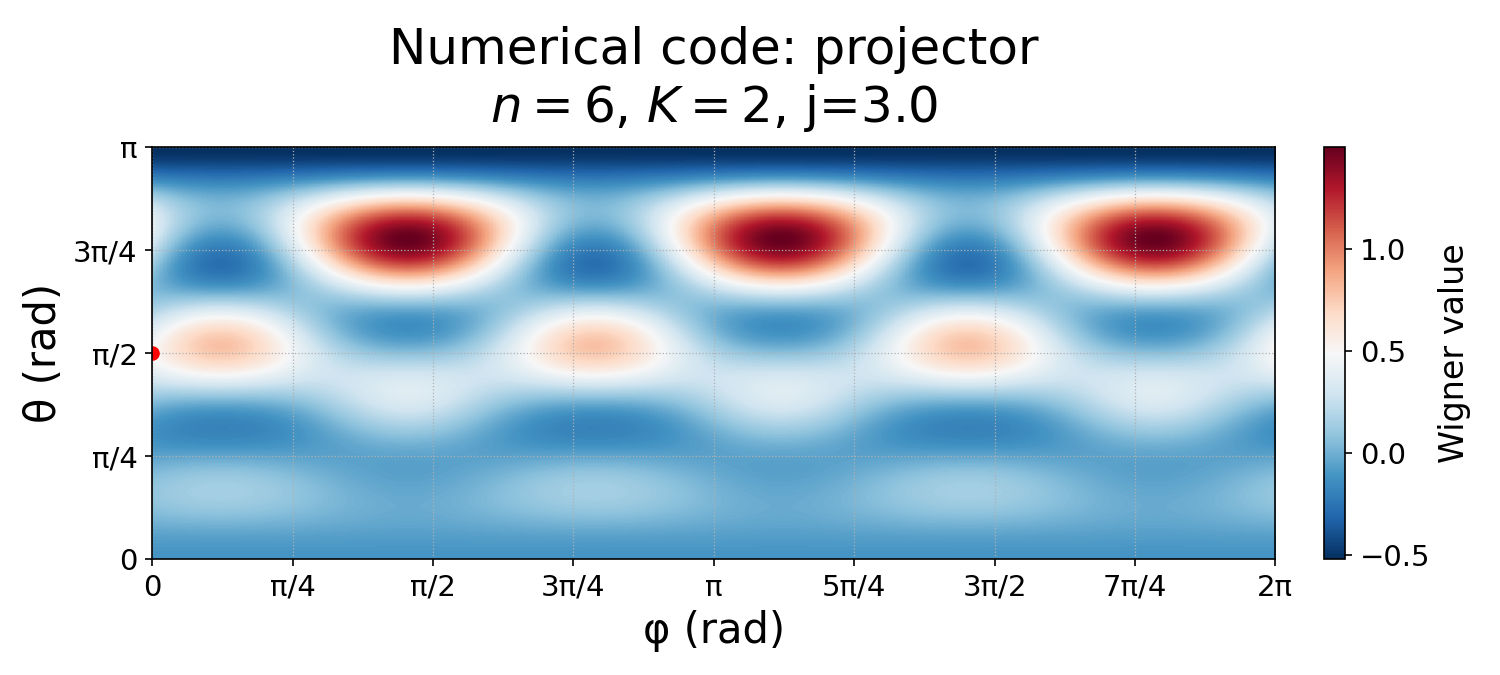}
        \caption{Numerically optimized code.}
    \end{subfigure}
    \hfill
    \begin{subfigure}{0.48\textwidth}
        \centering
        \includegraphics[width=\linewidth]{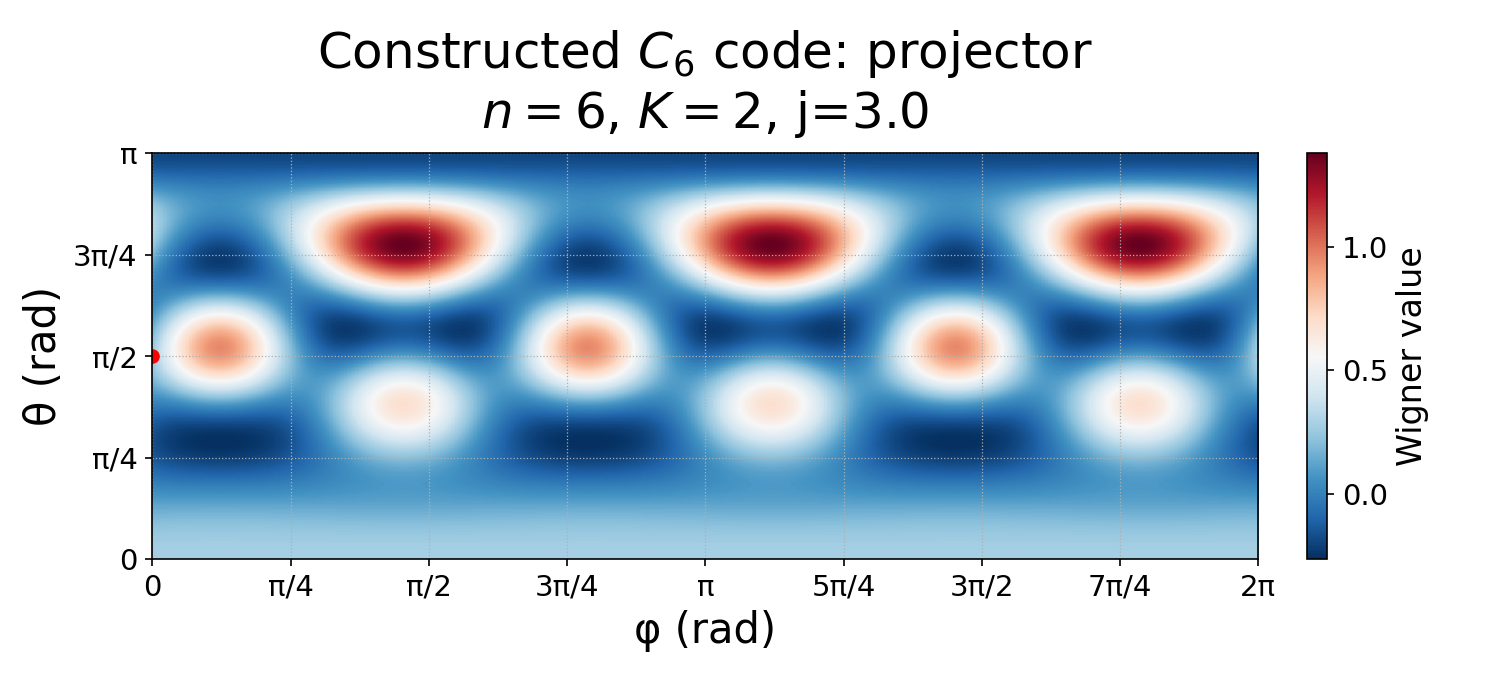}
        \caption{Constructed $C_6:\rho_1\oplus\rho_3$ code.}
    \end{subfigure}
    \caption{Best numerical code and constructed $C_{6}$ reference at $n=6$. At $\eta=0.90$, their respective entanglement fidelities are $F_e^{\mathrm{num}}=0.981849$ and $F_e^{\mathrm{con}}=0.979284$.}
    \label{fig:reducible-n6}
\end{figure}

\begin{figure}[H]
    \centering
    \begin{subfigure}{0.48\textwidth}
        \centering
        \includegraphics[width=\linewidth]{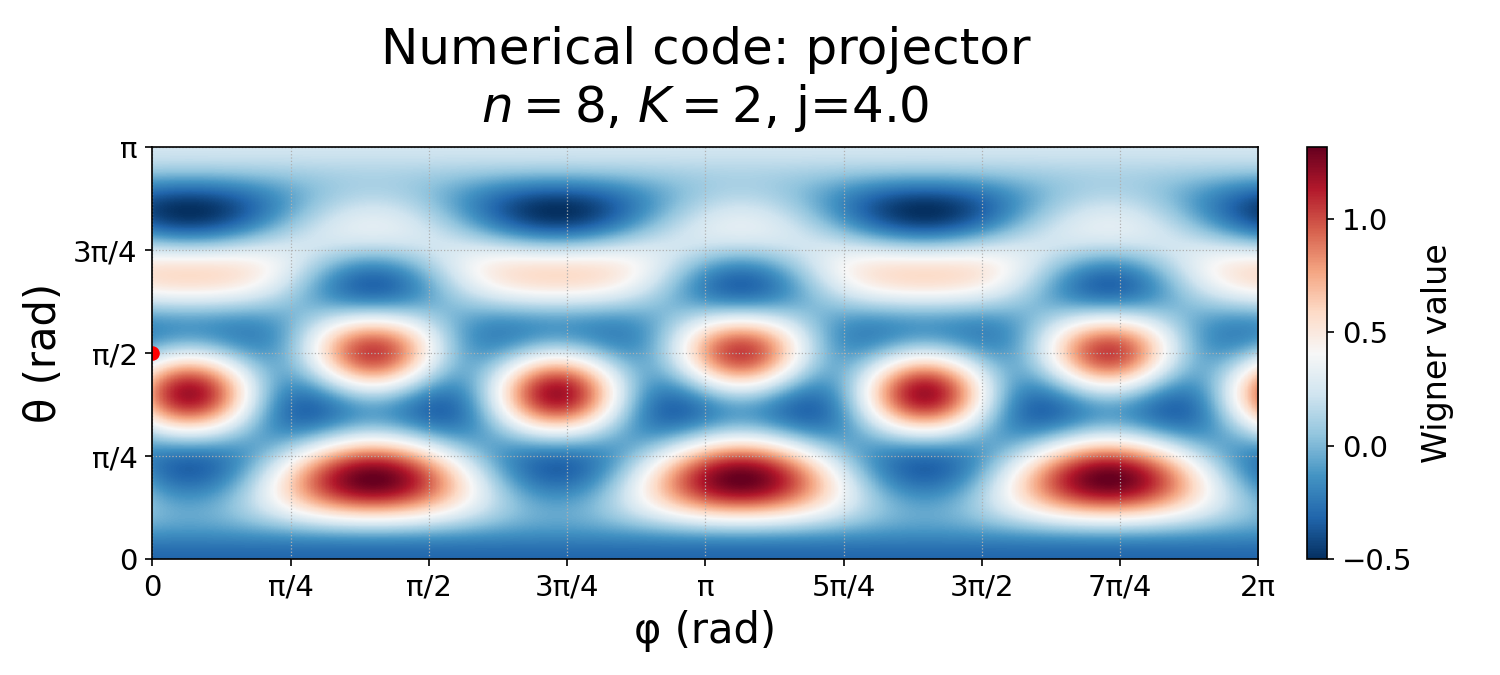}
        \caption{Numerically optimized code.}
    \end{subfigure}
    \hfill
    \begin{subfigure}{0.48\textwidth}
        \centering
        \includegraphics[width=\linewidth]{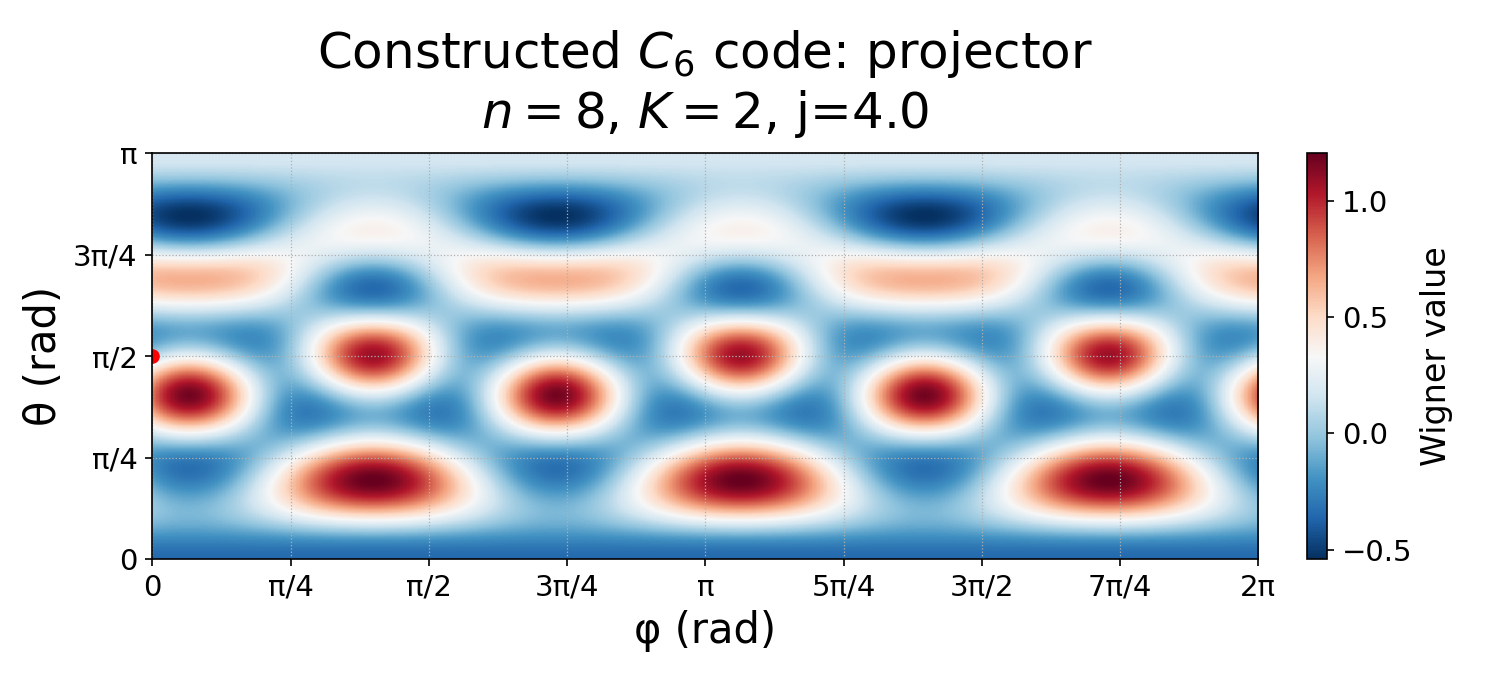}
        \caption{Constructed $C_6:\rho_1\oplus\rho_3$ code.}
    \end{subfigure}
    \caption{Best numerical code and constructed $C_{6}$ reference at $n=8$. At $\eta=0.90$, their respective entanglement fidelities are $F_e^{\mathrm{num}}=0.992062$ and $F_e^{\mathrm{con}}=0.991885$.}
    \label{fig:reducible-n8}
\end{figure}

\begin{figure}[H]
    \centering
    \begin{subfigure}{0.48\textwidth}
        \centering
        \includegraphics[width=\linewidth]{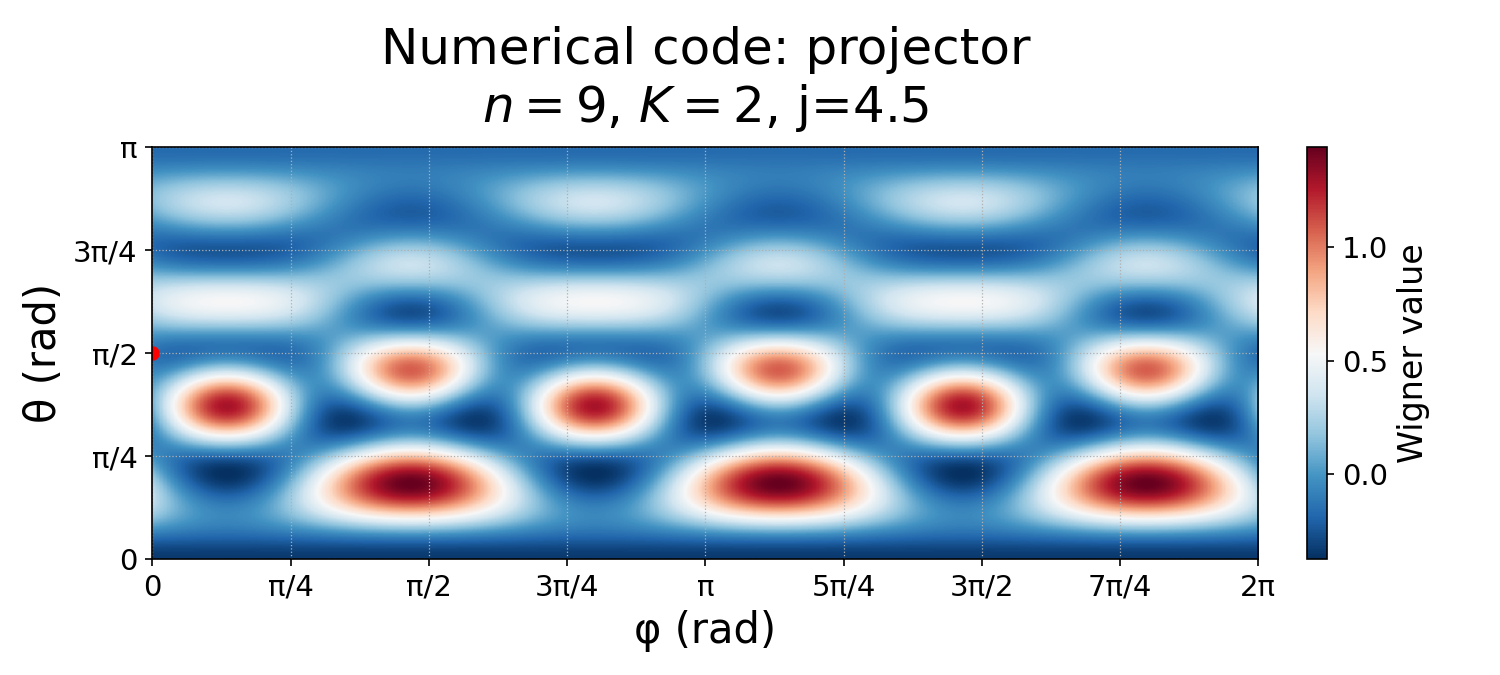}
        \caption{Numerically optimized code.}
    \end{subfigure}
    \hfill
    \begin{subfigure}{0.48\textwidth}
        \centering
        \includegraphics[width=\linewidth]{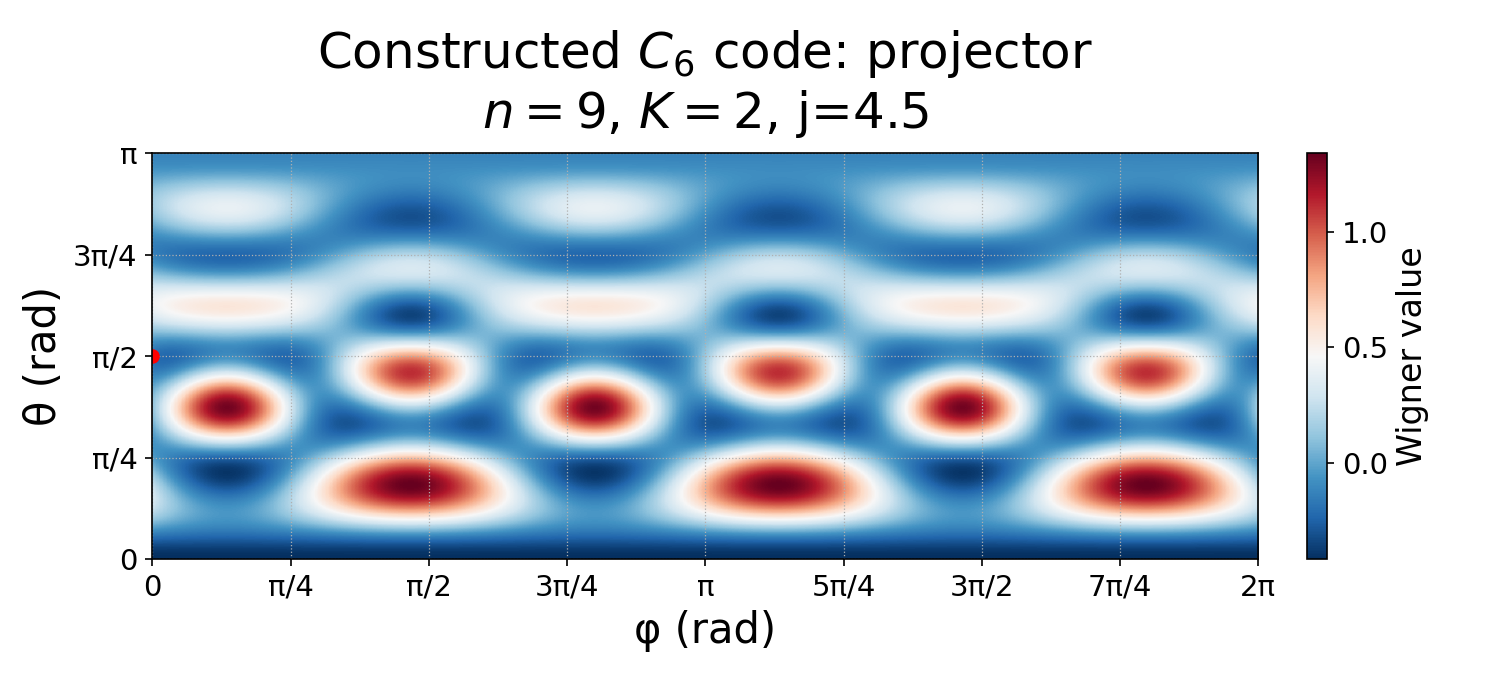}
        \caption{Constructed $C_6:\rho_2\oplus\rho_4$ code.}
    \end{subfigure}
    \caption{Best numerical code and constructed $C_{6}$ reference at $n=9$. At $\eta=0.90$, their respective entanglement fidelities are $F_e^{\mathrm{num}}=0.992889$ and $F_e^{\mathrm{con}}=0.992701$.}
    \label{fig:reducible-n9}
\end{figure}

\begin{figure}[H]
    \centering
    \begin{subfigure}{0.48\textwidth}
        \centering
        \includegraphics[width=\linewidth]{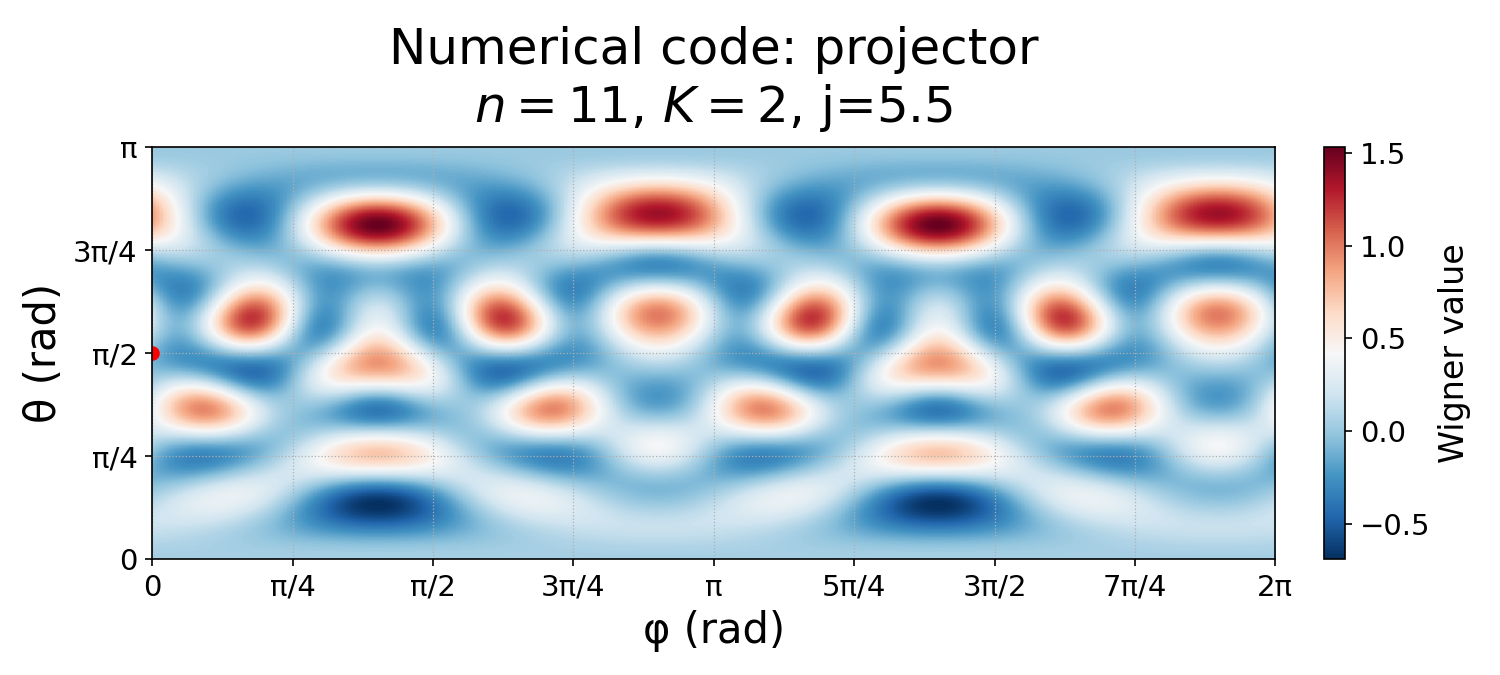}
        \caption{Numerically optimized code.}
    \end{subfigure}
    \hfill
    \begin{subfigure}{0.48\textwidth}
        \centering
        \includegraphics[width=\linewidth]{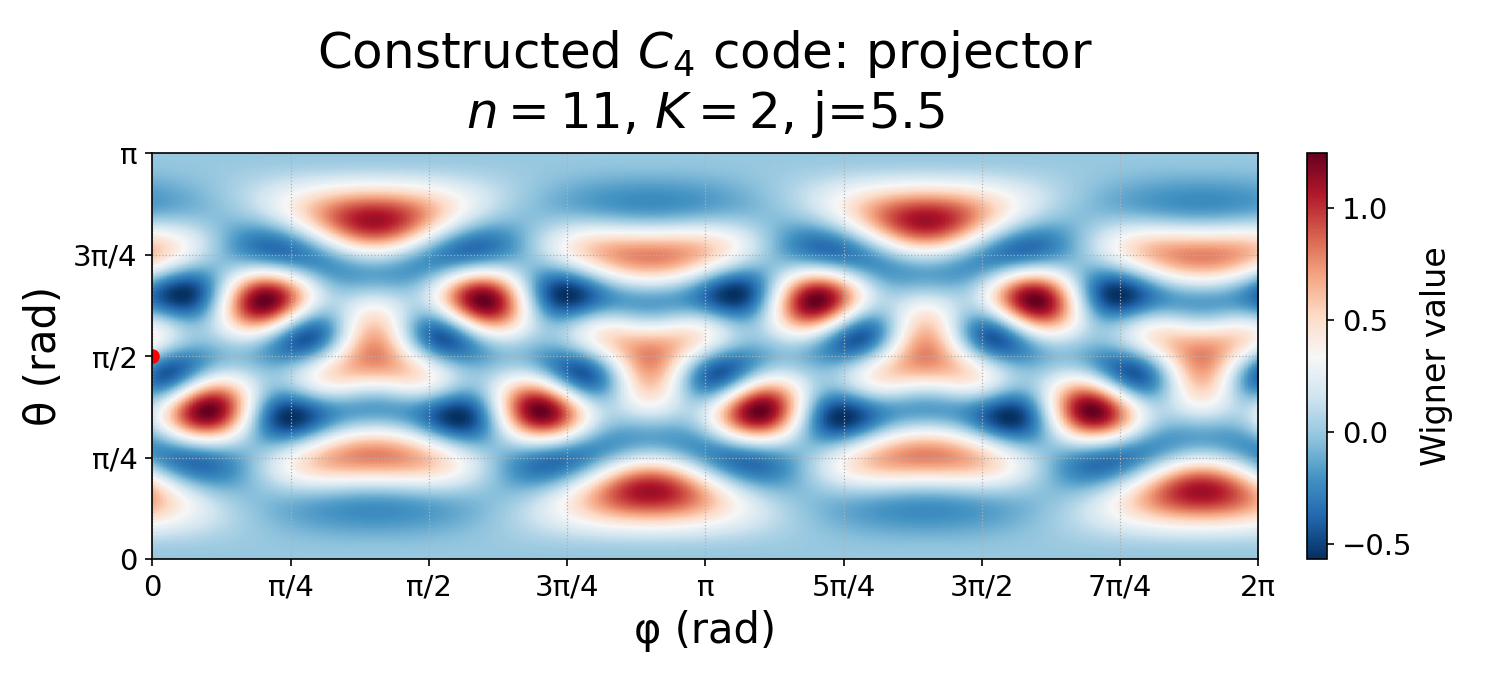}
        \caption{Constructed $C_4:\rho_3\oplus\rho_4$ code.}
    \end{subfigure}
    \caption{Best numerical code and constructed $C_{4}$ reference at $n=11$. At $\eta=0.90$, their respective entanglement fidelities are $F_e^{\mathrm{num}}=0.995075$ and $F_e^{\mathrm{con}}=0.994748$.}
    \label{fig:reducible-n11}
\end{figure}

\begin{figure}[H]
    \centering
    \begin{subfigure}{0.48\textwidth}
        \centering
        \includegraphics[width=\linewidth]{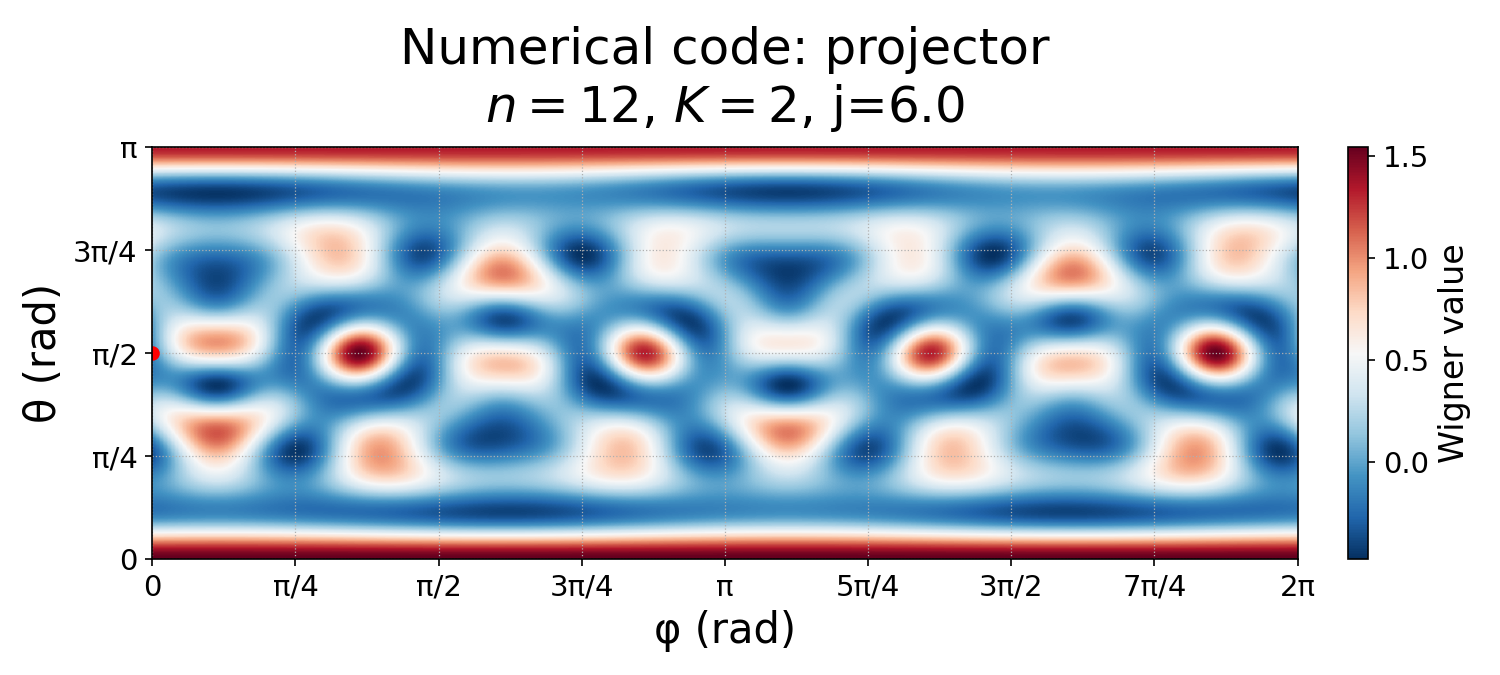}
        \caption{Numerically optimized code.}
    \end{subfigure}
    \hfill
    \begin{subfigure}{0.48\textwidth}
        \centering
        \includegraphics[width=\linewidth]{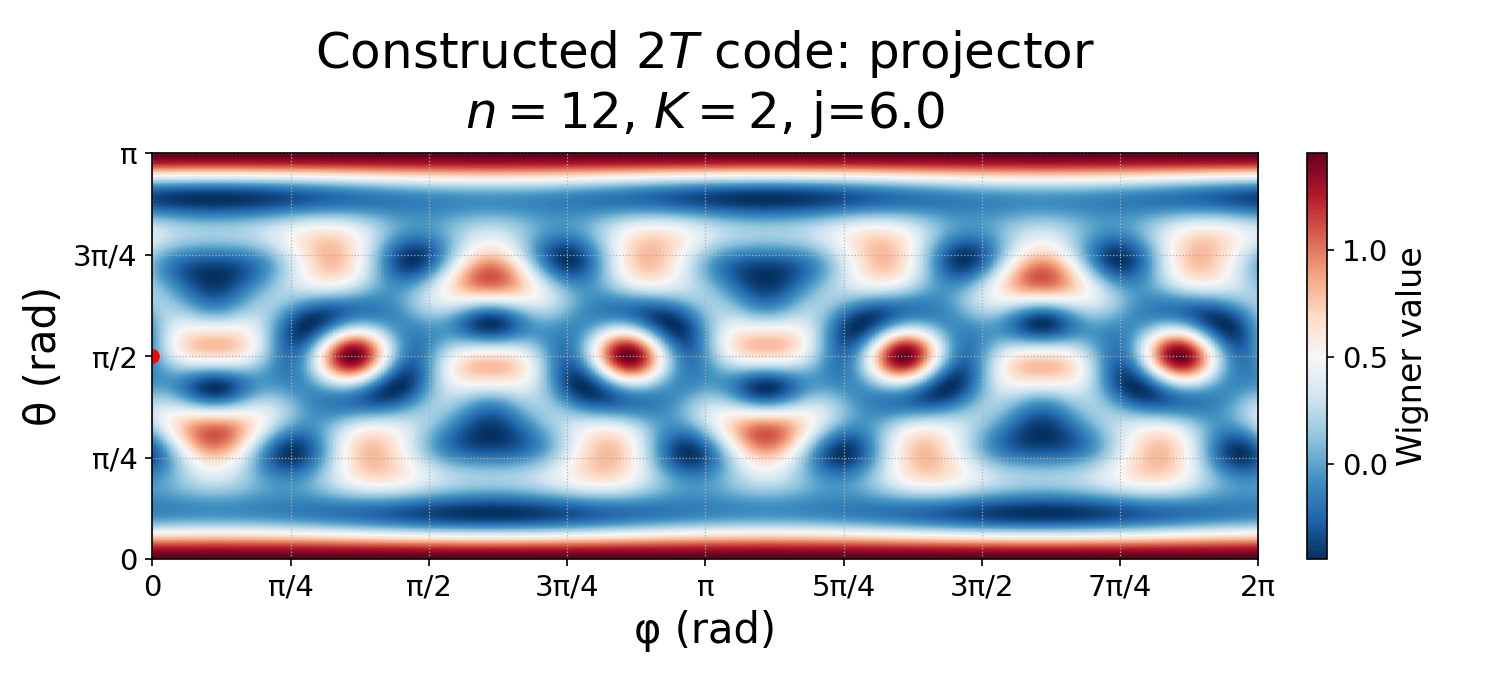}
        \caption{Constructed $2T:\rho_1\oplus\rho_2$ code.}
    \end{subfigure}
    \caption{Best numerical code and constructed $2T$ reference at $n=12$. At $\eta=0.90$, their respective entanglement fidelities are $F_e^{\mathrm{num}}=0.995795$ and $F_e^{\mathrm{con}}=0.995762$.}
    \label{fig:reducible-n12}
\end{figure}

\begin{figure}[H]
    \centering
    \begin{subfigure}{0.48\textwidth}
        \centering
        \includegraphics[width=\linewidth]{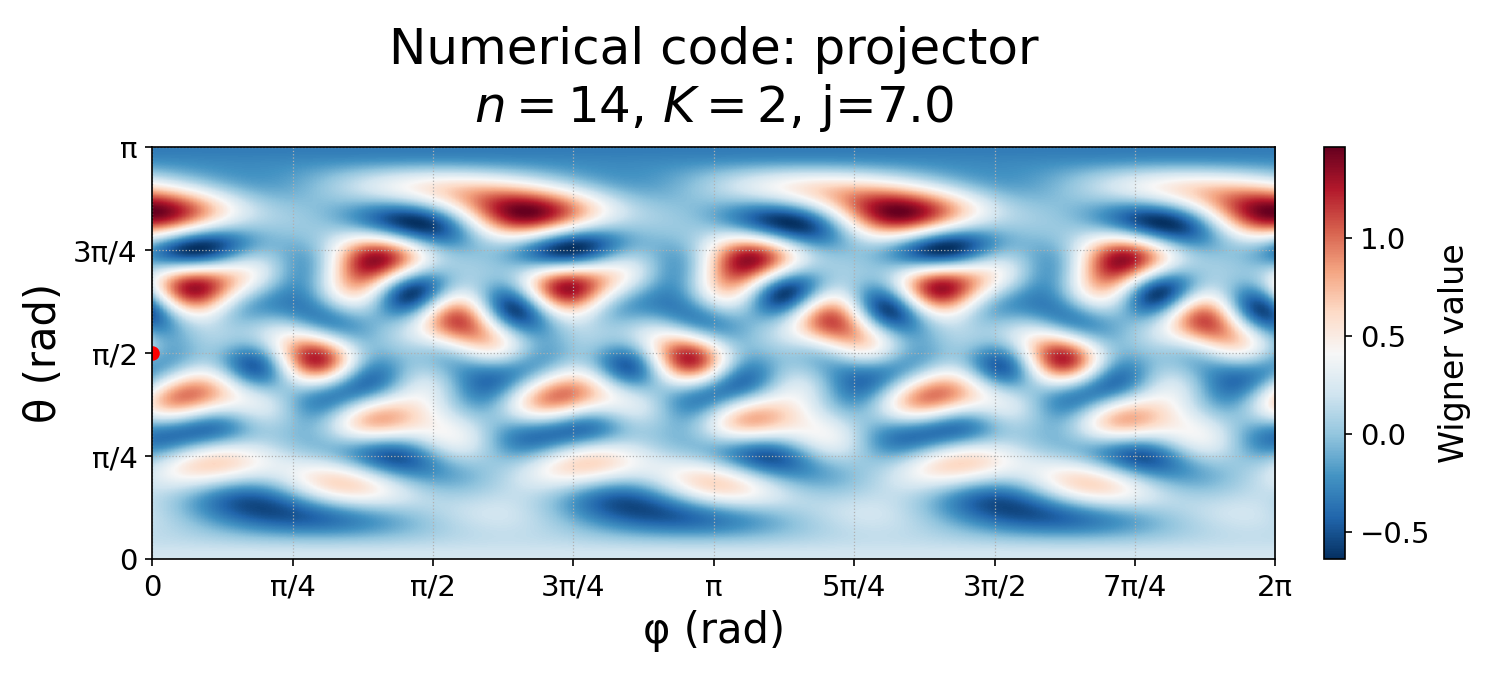}
        \caption{Numerically optimized code.}
    \end{subfigure}
    \hfill
    \begin{subfigure}{0.48\textwidth}
        \centering
        \includegraphics[width=\linewidth]{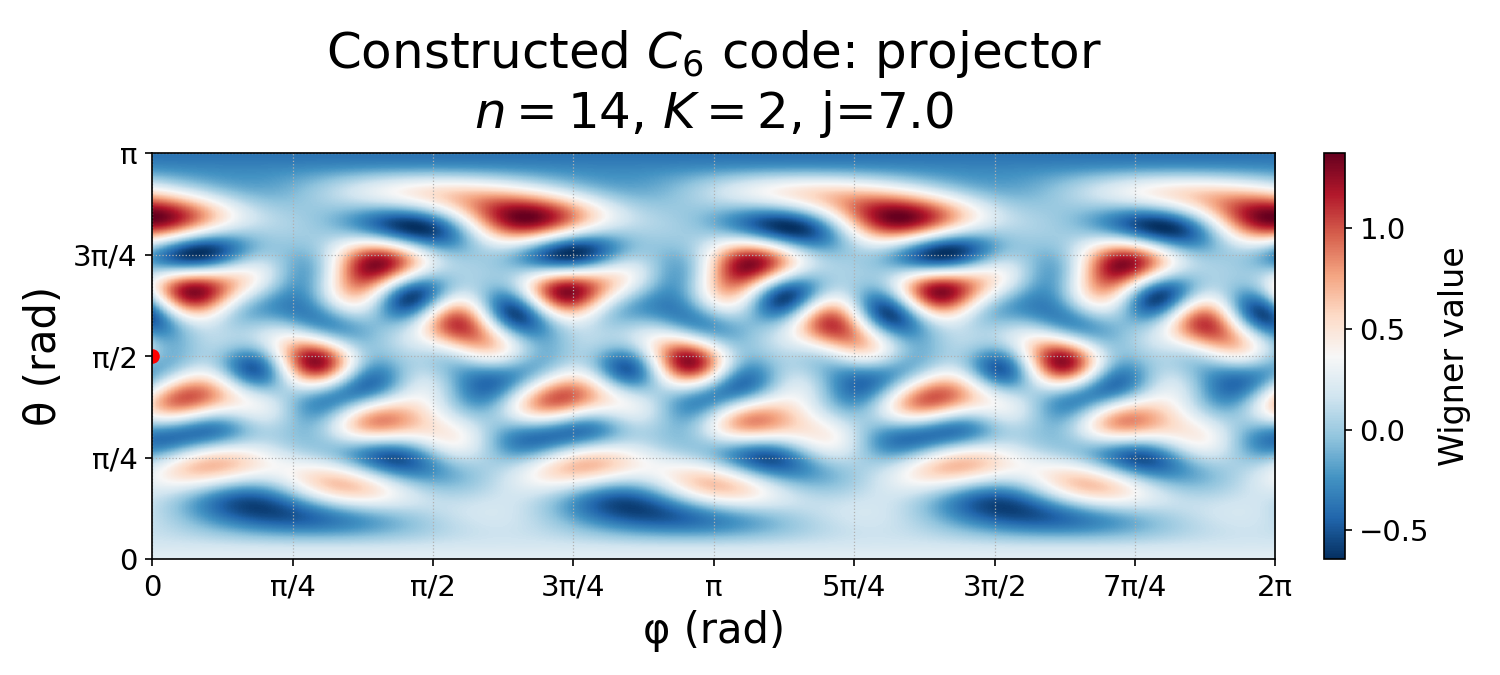}
        \caption{Constructed $C_6:\rho_1\oplus\rho_5$ code.}
    \end{subfigure}
    \caption{Best numerical code and constructed $C_{6}$ reference at $n=14$. At $\eta=0.90$, their respective entanglement fidelities are $F_e^{\mathrm{num}}=0.997669$ and $F_e^{\mathrm{con}}=0.997642$.}
    \label{fig:reducible-n14}
\end{figure}

\begin{figure}[H]
    \centering
    \begin{subfigure}{0.48\textwidth}
        \centering
        \includegraphics[width=\linewidth]{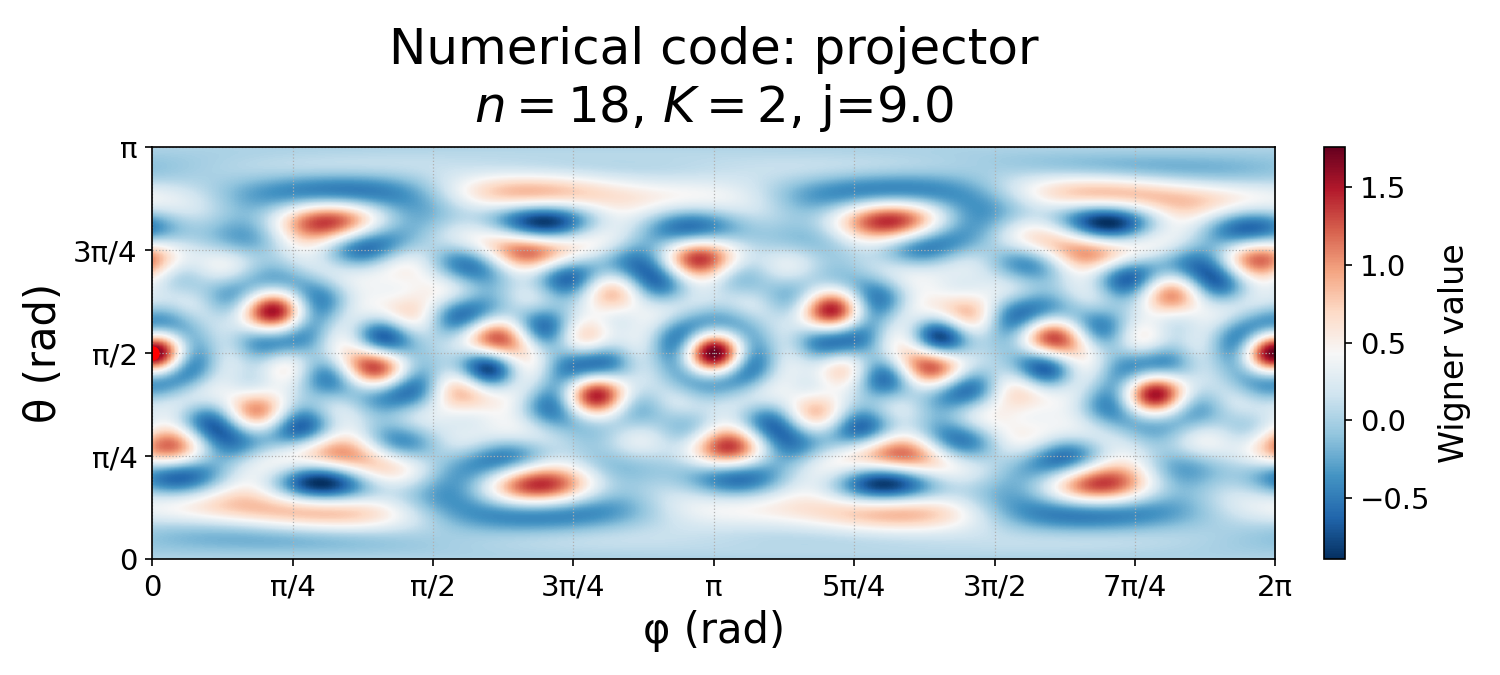}
        \caption{Numerically optimized code.}
    \end{subfigure}
    \hfill
    \begin{subfigure}{0.48\textwidth}
        \centering
        \includegraphics[width=\linewidth]{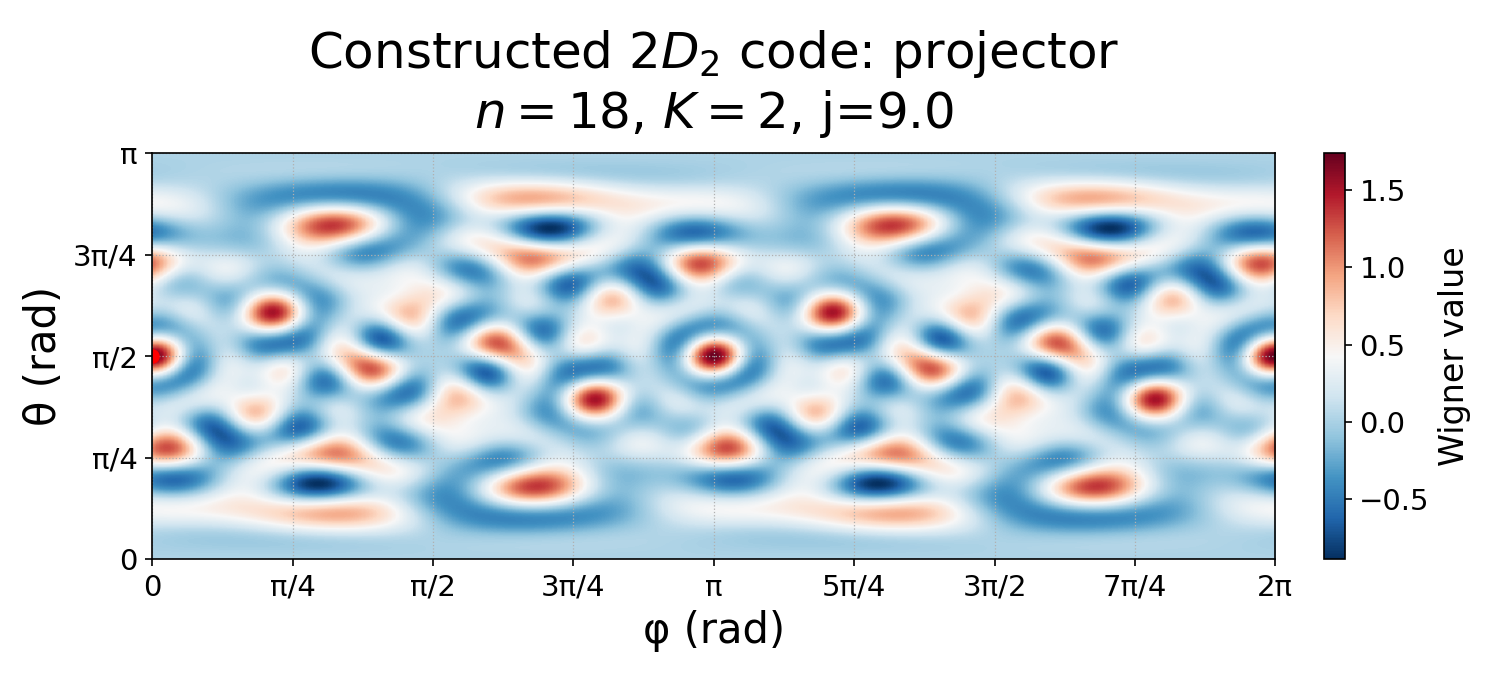}
        \caption{Constructed $2D_2:\rho_3\oplus\rho_4$ code.}
    \end{subfigure}
    \caption{Best numerical code and constructed $2D_{2}$ reference at $n=18$. At $\eta=0.90$, their respective entanglement fidelities are $F_e^{\mathrm{num}}=0.998739$ and $F_e^{\mathrm{con}}=0.998679$.}
    \label{fig:reducible-n18}
\end{figure}

\begin{figure}[H]
    \centering
    \begin{subfigure}{0.48\textwidth}
        \centering
        \includegraphics[width=\linewidth]{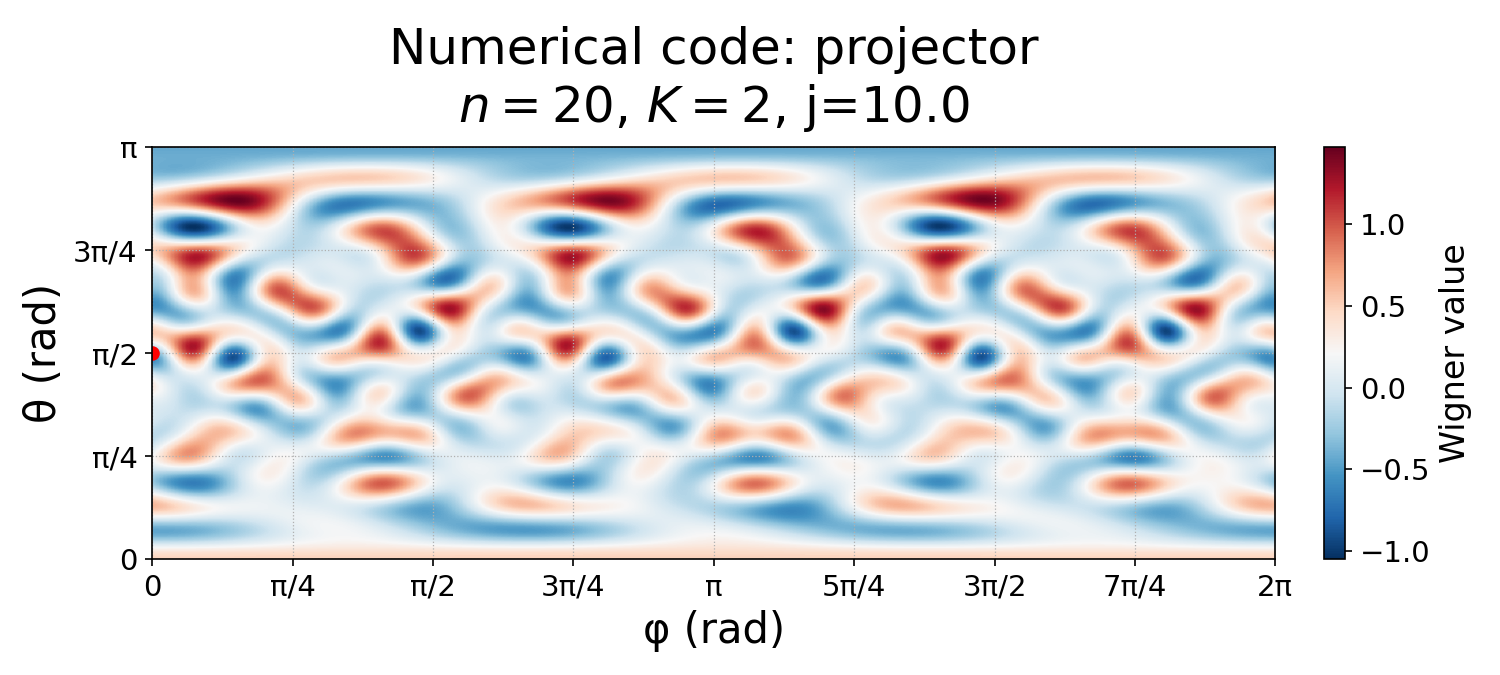}
        \caption{Numerically optimized code.}
    \end{subfigure}
    \hfill
    \begin{subfigure}{0.48\textwidth}
        \centering
        \includegraphics[width=\linewidth]{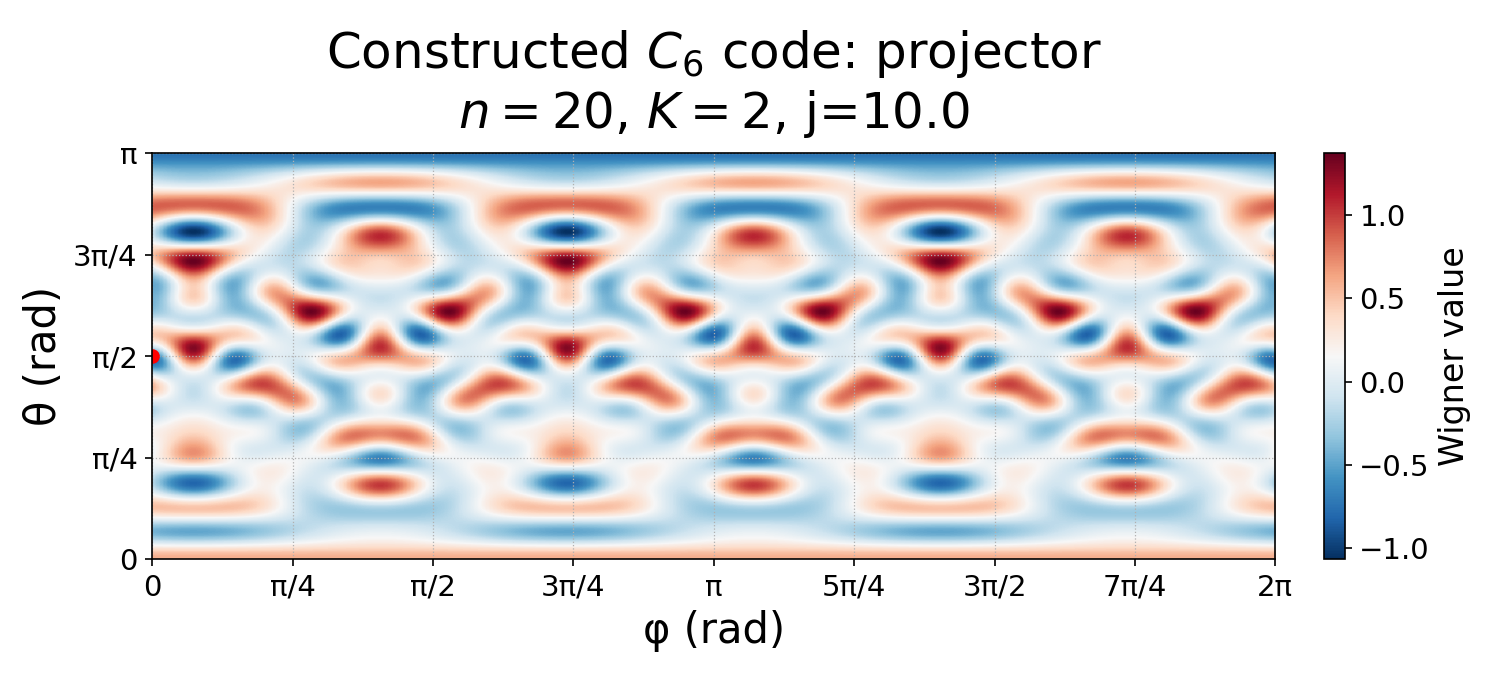}
        \caption{Constructed $C_6:\rho_1\oplus\rho_5$ code.}
    \end{subfigure}
    \caption{Best numerical code and constructed $C_{6}$ reference at $n=20$. At $\eta=0.90$, their respective entanglement fidelities are $F_e^{\mathrm{num}}=0.999063$ and $F_e^{\mathrm{con}}=0.999046$.}
    \label{fig:reducible-n20}
\end{figure}

\begin{figure}[H]
    \centering
    \begin{subfigure}{0.48\textwidth}
        \centering
        \includegraphics[width=\linewidth]{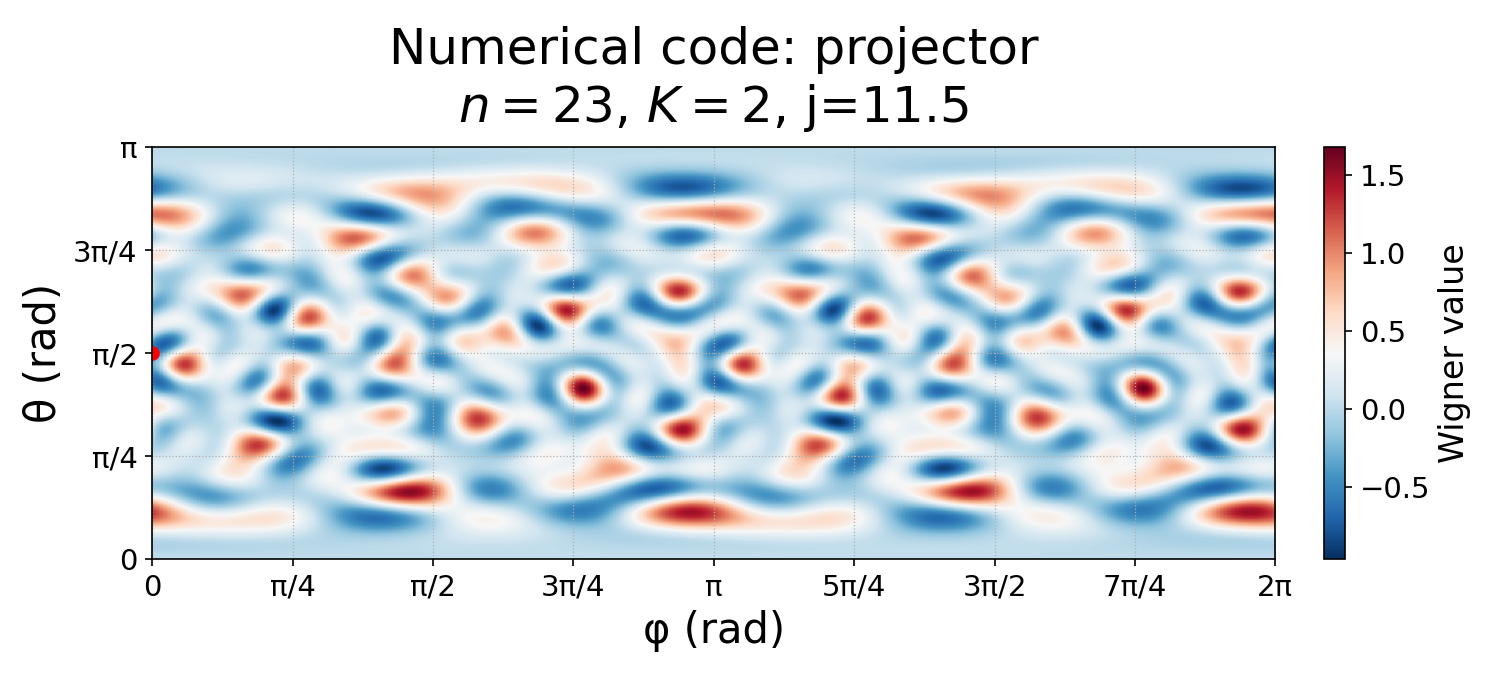}
        \caption{Numerically optimized code.}
    \end{subfigure}
    \hfill
    \begin{subfigure}{0.48\textwidth}
        \centering
        \includegraphics[width=\linewidth]{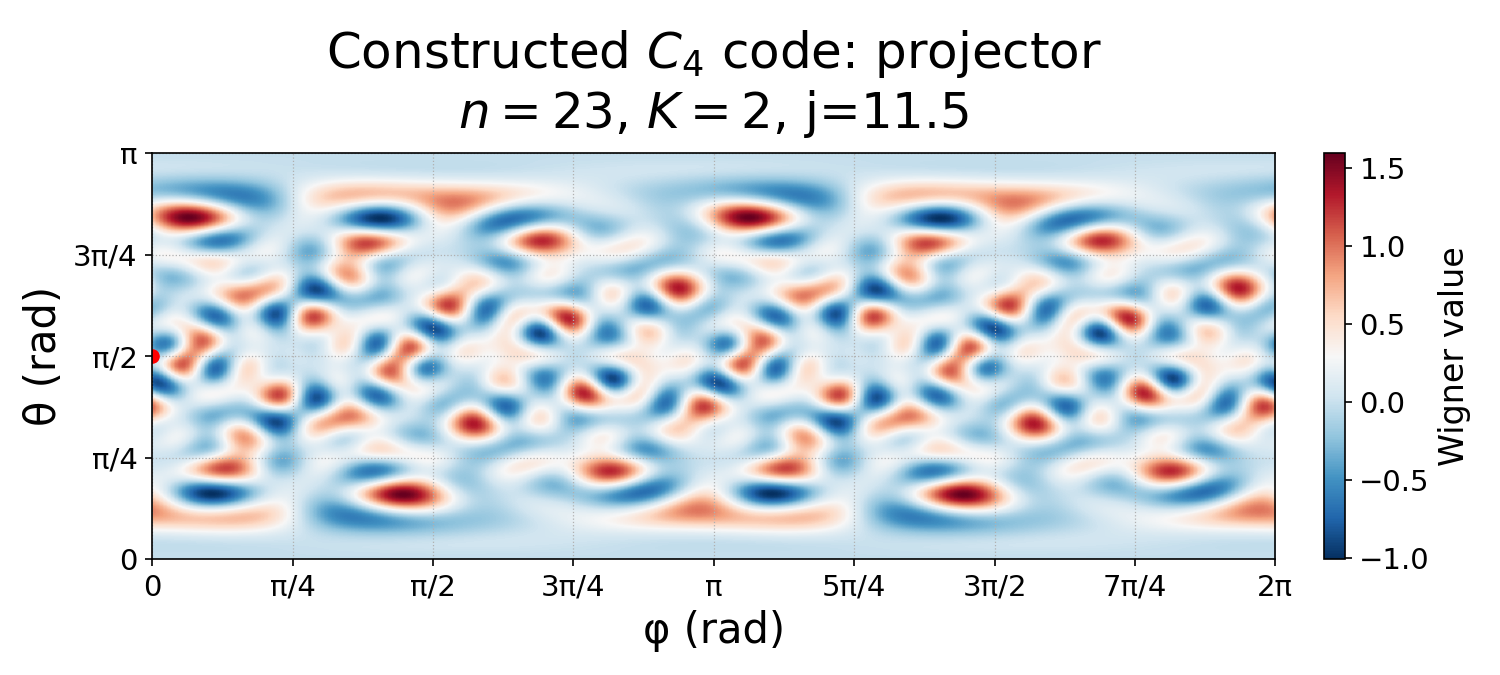}
        \caption{Constructed $C_4:\rho_3\oplus\rho_4$ code.}
    \end{subfigure}
    \caption{Best numerical code and constructed $C_{4}$ reference at $n=23$. At $\eta=0.90$, their respective entanglement fidelities are $F_e^{\mathrm{num}}=0.999359$ and $F_e^{\mathrm{con}}=0.999338$.}
    \label{fig:reducible-n23}
\end{figure}

\begin{figure}[H]
    \centering
    \begin{subfigure}{0.48\textwidth}
        \centering
        \includegraphics[width=\linewidth]{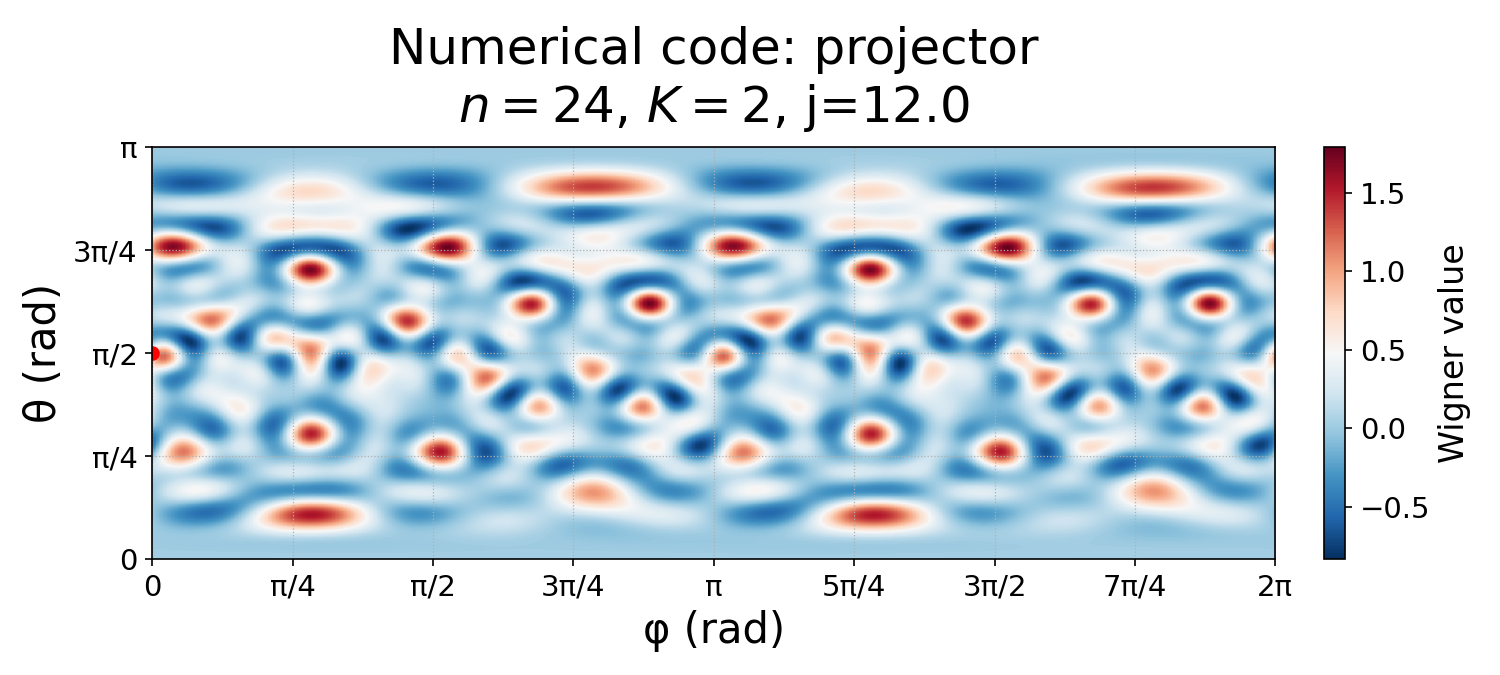}
        \caption{Numerically optimized code.}
    \end{subfigure}
    \hfill
    \begin{subfigure}{0.48\textwidth}
        \centering
        \includegraphics[width=\linewidth]{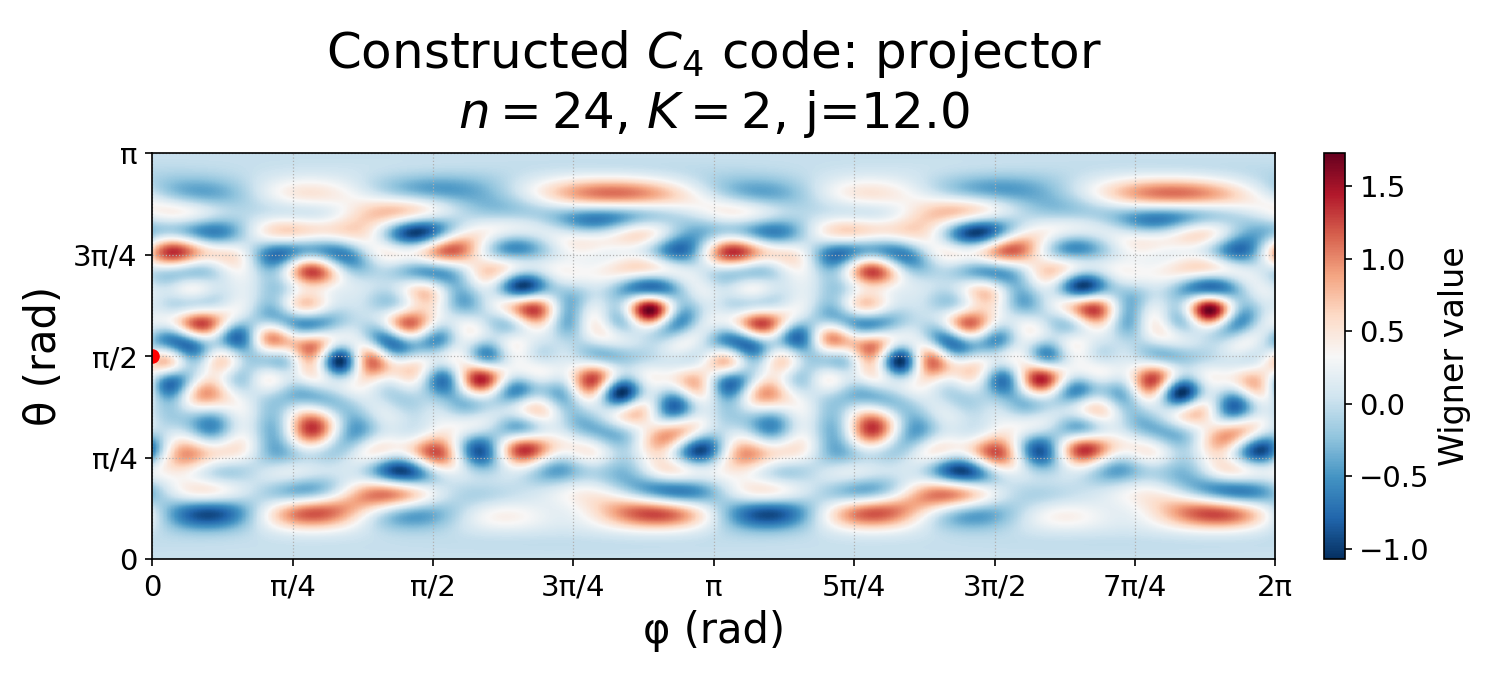}
        \caption{Constructed $C_4:\rho_1\oplus\rho_2$ code.}
    \end{subfigure}
    \caption{Best numerical code and constructed $C_{4}$ reference at $n=24$. At $\eta=0.90$, their respective entanglement fidelities are $F_e^{\mathrm{num}}=0.999422$ and $F_e^{\mathrm{con}}=0.999391$.}
    \label{fig:reducible-n24}
\end{figure}

\FloatBarrier

\subsection{Lower-transmissivity symmetry plateaus}

For $6\leq n\leq13$, we also examine selected parameter points in the range $0.70\leq\eta\leq0.90$. At each displayed point, we use the highest saved-fidelity run from our optimization. Table~\ref{tab:lower-eta-numerical-code-symmetry-classification} records the resulting symmetry assignments, including changes of the preferred group as $\eta$ varies.

\begin{table}[H]
    \centering
    \footnotesize
    \renewcommand{\arraystretch}{1.08}
    \begin{tabular}{c c l c c}
        \hline
        $n$ & $\eta$ & Group and irrep content & $\operatorname{Comm}_L(P,G)$ & $\Omega_G(P)$ \\
        \hline
        $6$ & $0.70$ & $C_{4}:\rho_{1}\oplus\rho_{2}$ & $3.481143\times 10^{-10}$ & $0.9999999998$ \\
        $6$ & $0.80$ & $C_{6}:\rho_{1}\oplus\rho_{5}$ & $3.059112\times 10^{-12}$ & $0.999999999998$ \\
        $7$ & $0.70$ & $C_{4}:\rho_{3}\oplus\rho_{4}$ & $5.132143\times 10^{-12}$ & $0.999999999997$ \\
        $7$ & $0.75$ & $C_{6}:\rho_{4}\oplus\rho_{6}$ & $1.141150\times 10^{-12}$ & $0.9999999999994$ \\
        $7$ & $0.80$ & $C_{4}:\rho_{3}\oplus\rho_{4}$ & $9.810392\times 10^{-11}$ & $0.99999999995$ \\
        $7$ & $0.85$ & $C_{10}:\rho_{6}\oplus\rho_{8}$ & $1.089124\times 10^{-6}$ & $0.9999994$ \\
        $8$ & $0.70$ & $C_{6}:\rho_{1}\oplus\rho_{5}$ & $2.079371\times 10^{-12}$ & $0.999999999998$ \\
        $9$ & $0.70$ & $C_{6}:\rho_{2}\oplus\rho_{6}$ & $2.646346\times 10^{-11}$ & $0.99999999998$ \\
        $10$ & $0.70$ & $C_{6}:\rho_{3}\oplus\rho_{5}$ & $3.510255\times 10^{-9}$ & $0.999999998$ \\
        $10$ & $0.80$ & $C_{4}:\rho_{1}\oplus\rho_{2}$ & $7.645625\times 10^{-12}$ & $0.999999999996$ \\
        $11$ & $0.70$ & $C_{4}:\rho_{3}\oplus\rho_{4}$ & $1.814410\times 10^{-12}$ & $0.9999999999990$ \\
        $11$ & $0.75$ & $C_{10}:\rho_{4}\oplus\rho_{6}$ & $1.595834\times 10^{-2}$ & $0.9919$ \\
        $12$ & $0.70$ & $C_{10}:\rho_{5}\oplus\rho_{7}$ & $1.811684\times 10^{-3}$ & $0.9990$ \\
        $12$ & $0.80$ & $C_{8}:\rho_{2}\oplus\rho_{4}$ & $3.053948\times 10^{-12}$ & $0.999999999998$ \\
        $13$ & $0.70$ & $C_{10}:\rho_{8}\oplus\rho_{10}$ & $7.876081\times 10^{-3}$ & $0.9960$ \\
        \hline
    \end{tabular}
    \caption{Symmetry classification of the sweep winners at the displayed lower-transmissivity points. We include the $\eta=0.70$ endpoint for every $6\leq n\leq13$ and selected additional symmetry branches. Reversing a cyclic symmetry axis or conjugating the code can change the displayed irrep labels; the labels here use the recorded classification frame. Overlap entries are truncated after the first decimal digit other than $9$, retaining at least four decimal places.}
    \label{tab:lower-eta-numerical-code-symmetry-classification}
\end{table}

We finally present the projector spin-Wigner functions of all lower-transmissivity numerical projectors listed in Table~\ref{tab:lower-eta-numerical-code-symmetry-classification} in Figures~\ref{fig:low-eta-c4}--\ref{fig:low-eta-remaining}. 

\begin{figure}[H]
    \centering
    \begin{subfigure}[t]{0.48\textwidth}
        \centering
        \includegraphics[width=\linewidth]{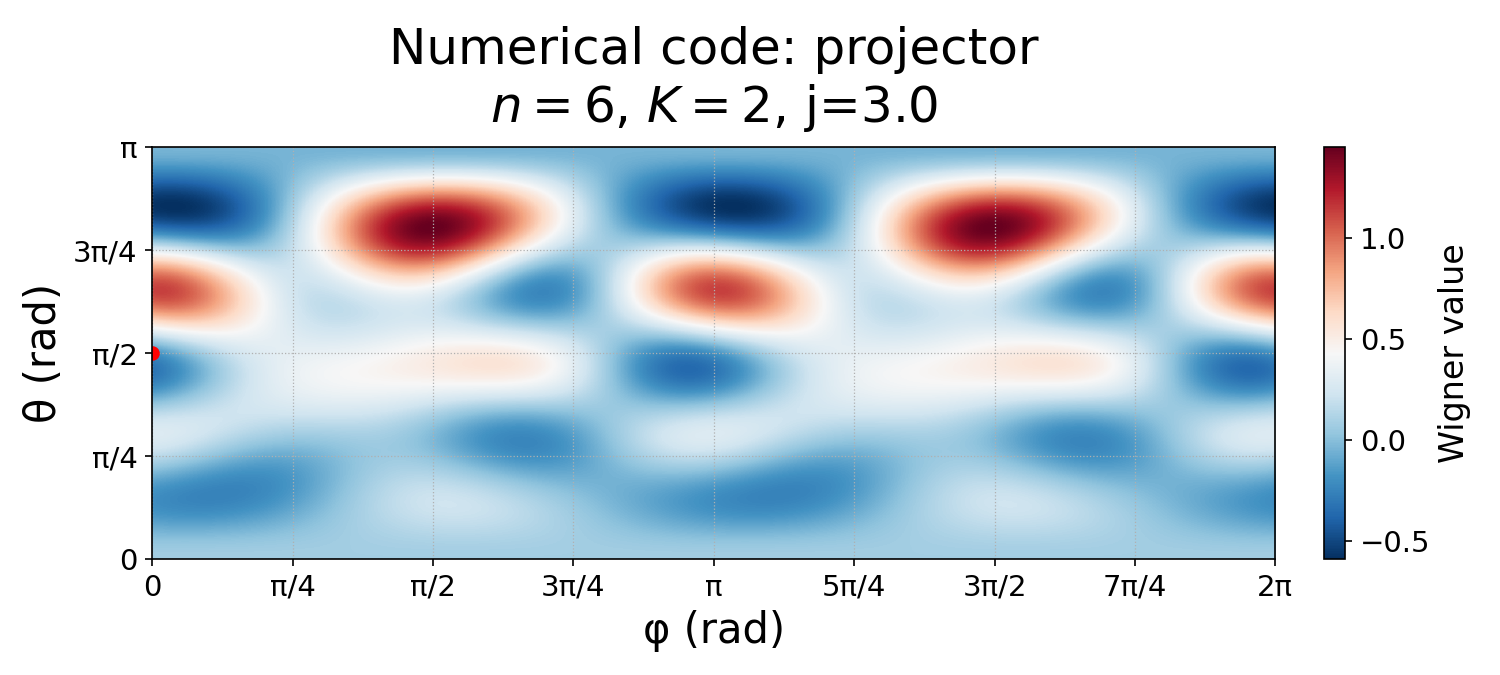}
        \caption{$n=6$, $\eta=0.70$, $C_{4}:\rho_{1}\oplus\rho_{2}$, $F_e=0.858964$.}
    \end{subfigure}
    \hfill
    \begin{subfigure}[t]{0.48\textwidth}
        \centering
        \includegraphics[width=\linewidth]{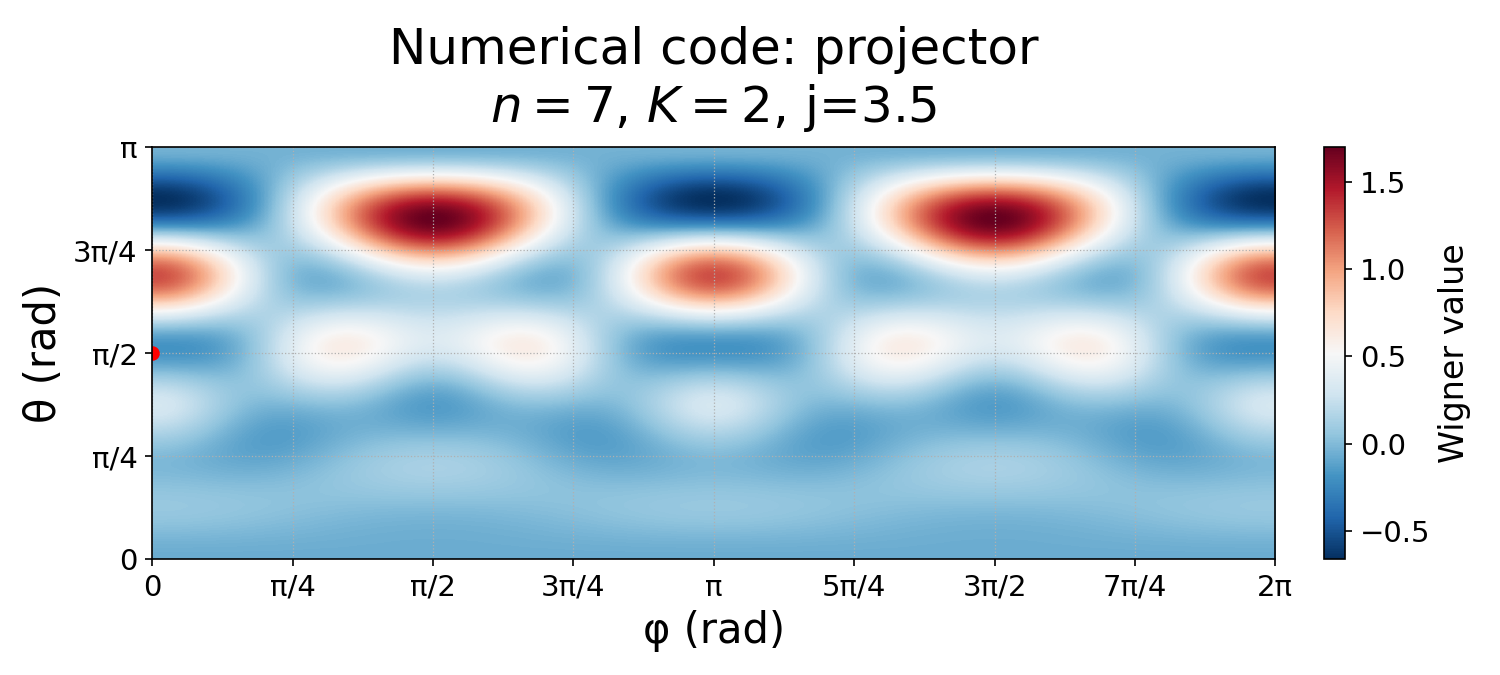}
        \caption{$n=7$, $\eta=0.70$, $C_{4}:\rho_{3}\oplus\rho_{4}$, $F_e=0.865803$.}
    \end{subfigure}

    \medskip

    \begin{subfigure}[t]{0.48\textwidth}
        \centering
        \includegraphics[width=\linewidth]{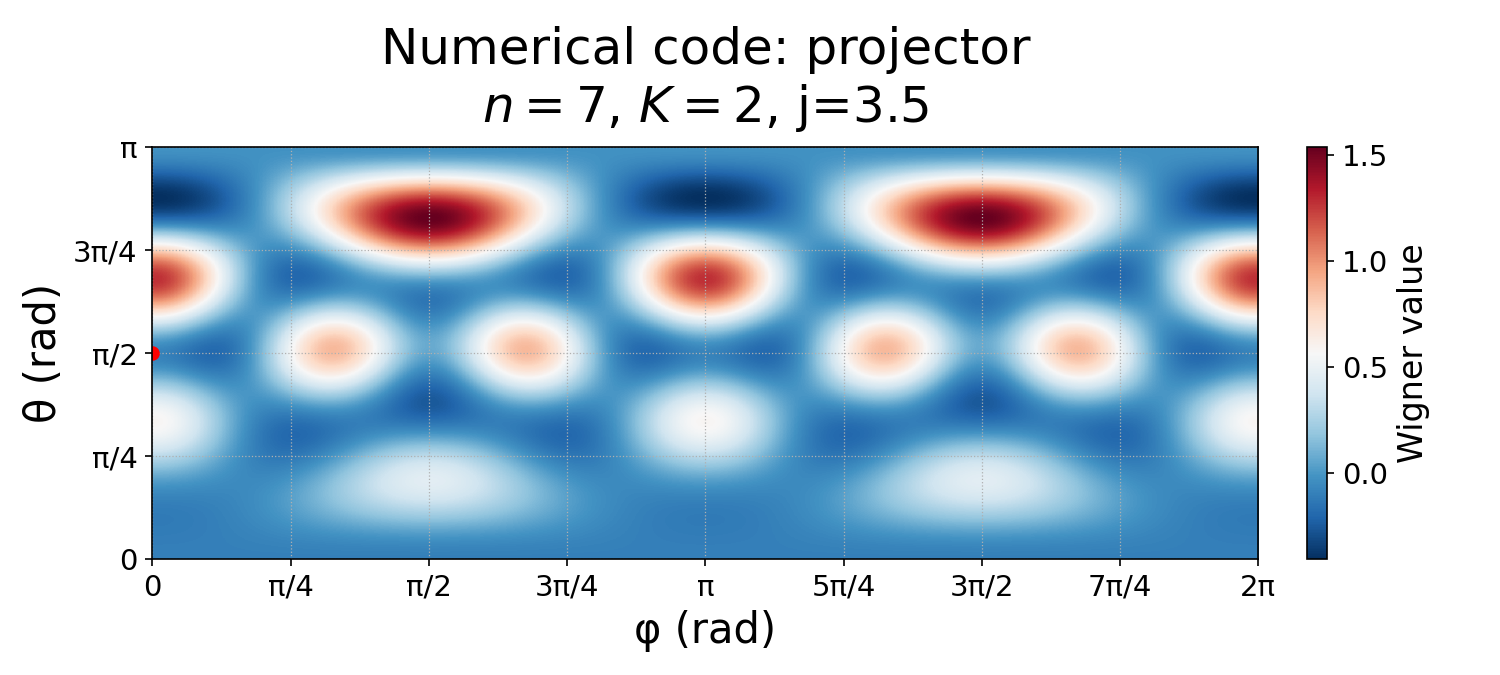}
        \caption{$n=7$, $\eta=0.80$, $C_{4}:\rho_{3}\oplus\rho_{4}$, $F_e=0.941701$.}
    \end{subfigure}
    \hfill
    \begin{subfigure}[t]{0.48\textwidth}
        \centering
        \includegraphics[width=\linewidth]{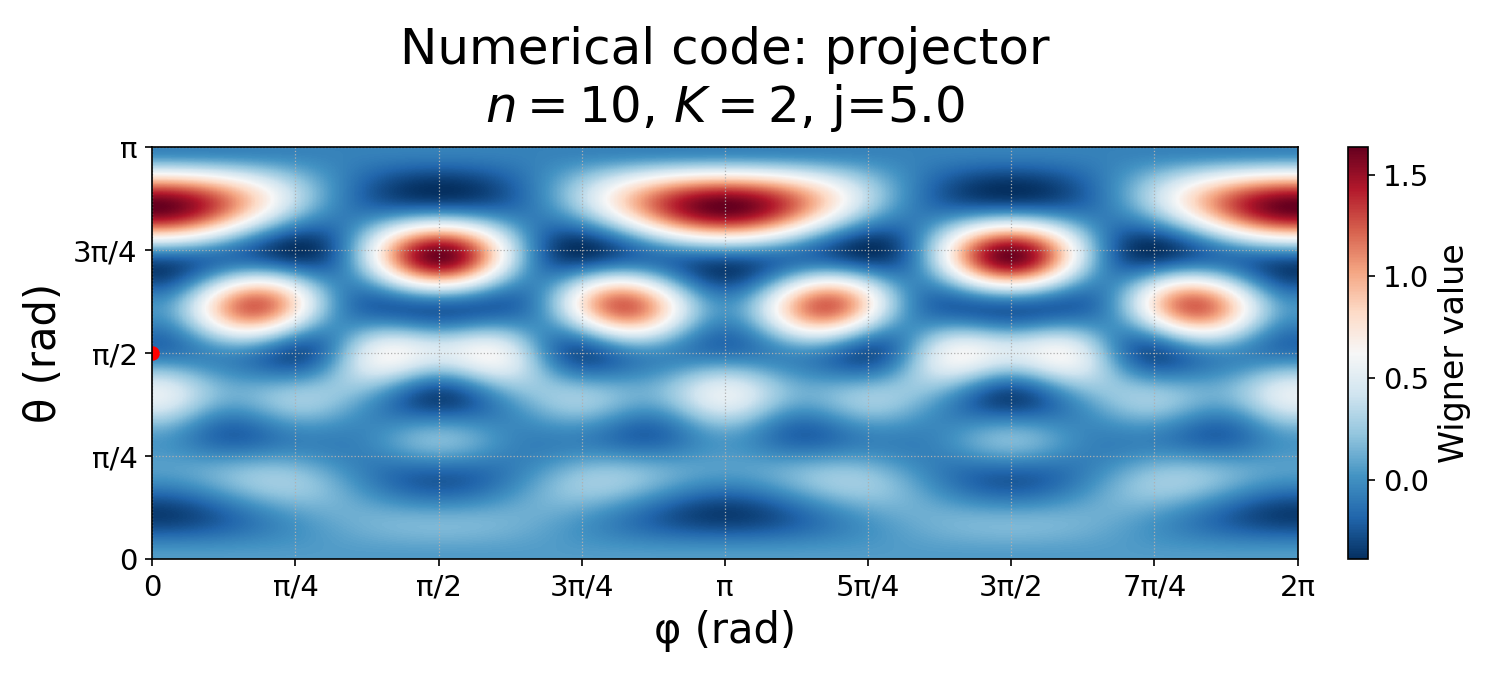}
        \caption{$n=10$, $\eta=0.80$, $C_{4}:\rho_{1}\oplus\rho_{2}$, $F_e=0.956684$.}
    \end{subfigure}
    \caption{Group-aligned projector spin-Wigner functions of four selected low-$\eta$ numerical codes with $C_4$ symmetry.}
    \label{fig:low-eta-c4}
\end{figure}

\begin{figure}[H]
    \centering
    \begin{subfigure}[t]{0.48\textwidth}
        \centering
        \includegraphics[width=\linewidth]{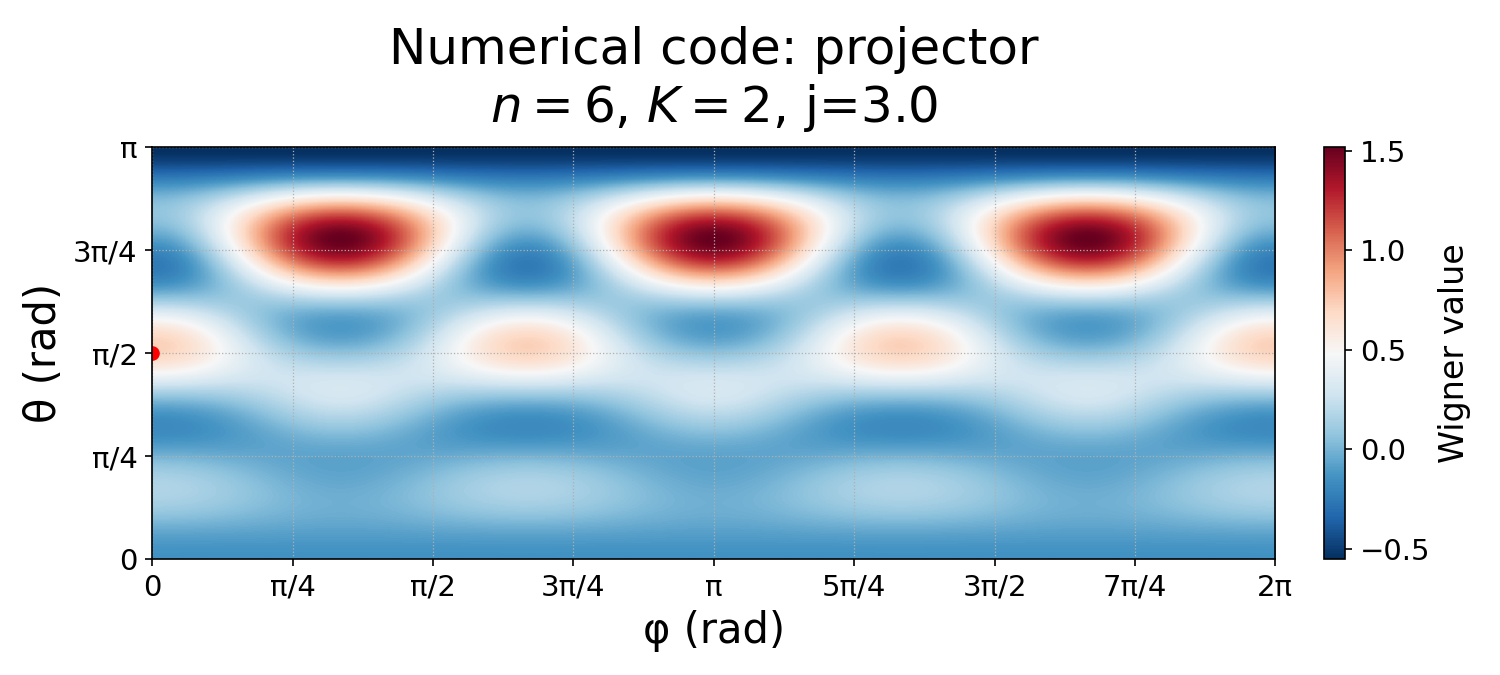}
        \caption{$n=6$, $\eta=0.80$, $C_{6}:\rho_{1}\oplus\rho_{5}$, $F_e=0.931881$.}
    \end{subfigure}
    \hfill
    \begin{subfigure}[t]{0.48\textwidth}
        \centering
        \includegraphics[width=\linewidth]{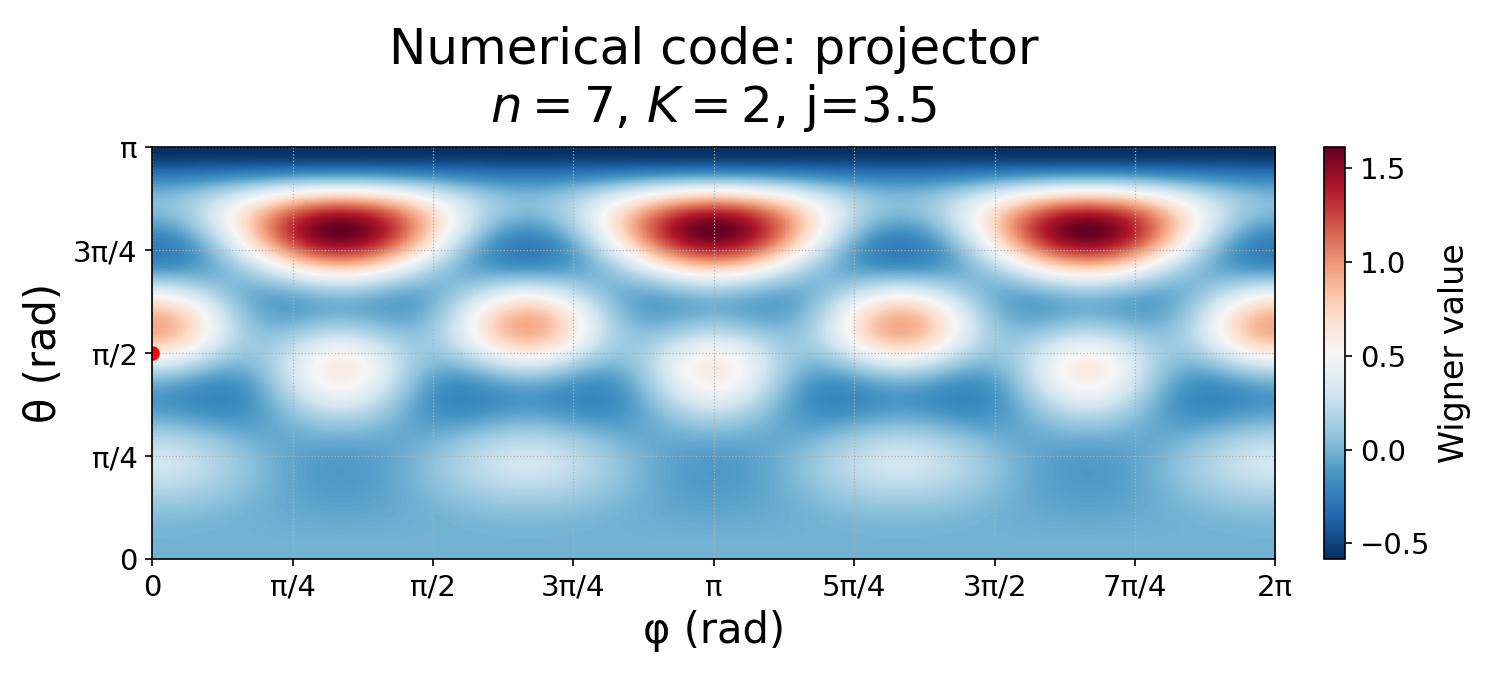}
        \caption{$n=7$, $\eta=0.75$, $C_{6}:\rho_{4}\oplus\rho_{6}$, $F_e=0.905719$.}
    \end{subfigure}

    \medskip

    \begin{subfigure}[t]{0.48\textwidth}
        \centering
        \includegraphics[width=\linewidth]{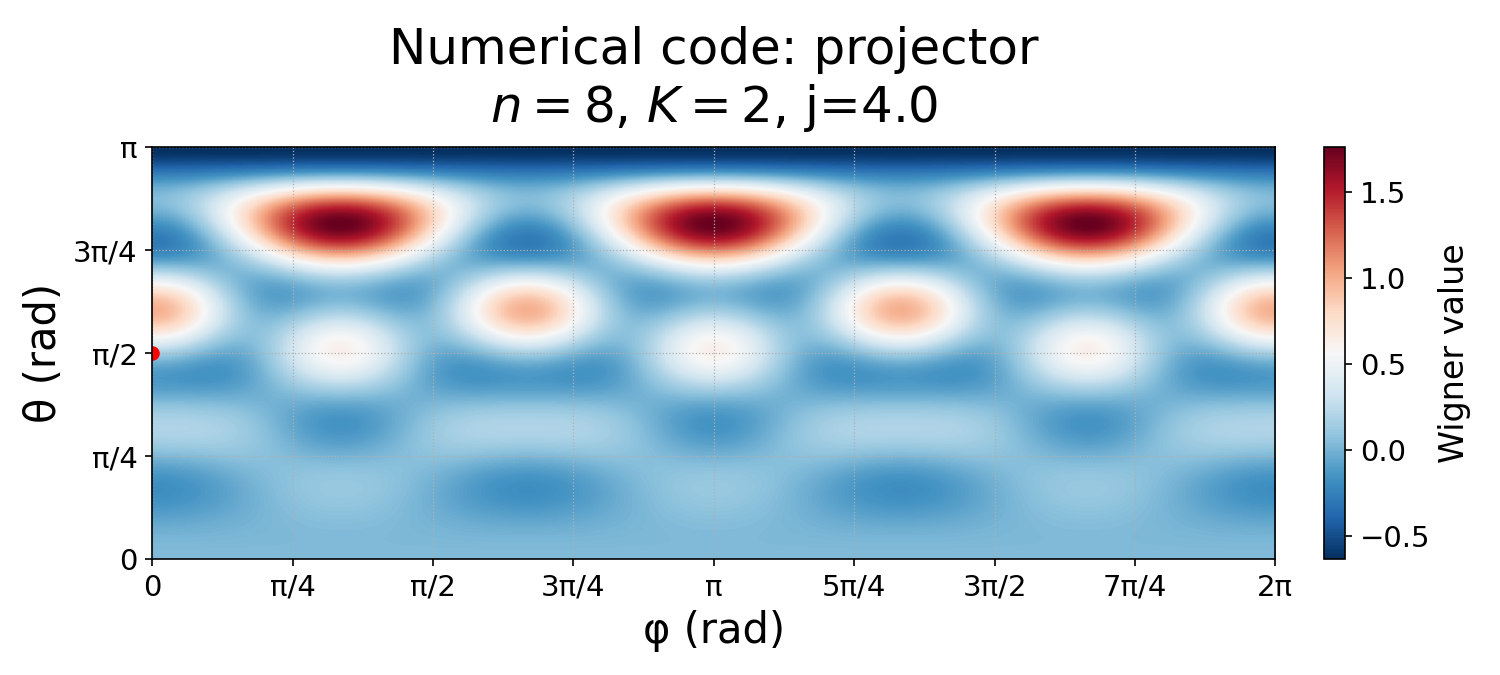}
        \caption{$n=8$, $\eta=0.70$, $C_{6}:\rho_{1}\oplus\rho_{5}$, $F_e=0.872531$.}
    \end{subfigure}
    \hfill
    \begin{subfigure}[t]{0.48\textwidth}
        \centering
        \includegraphics[width=\linewidth]{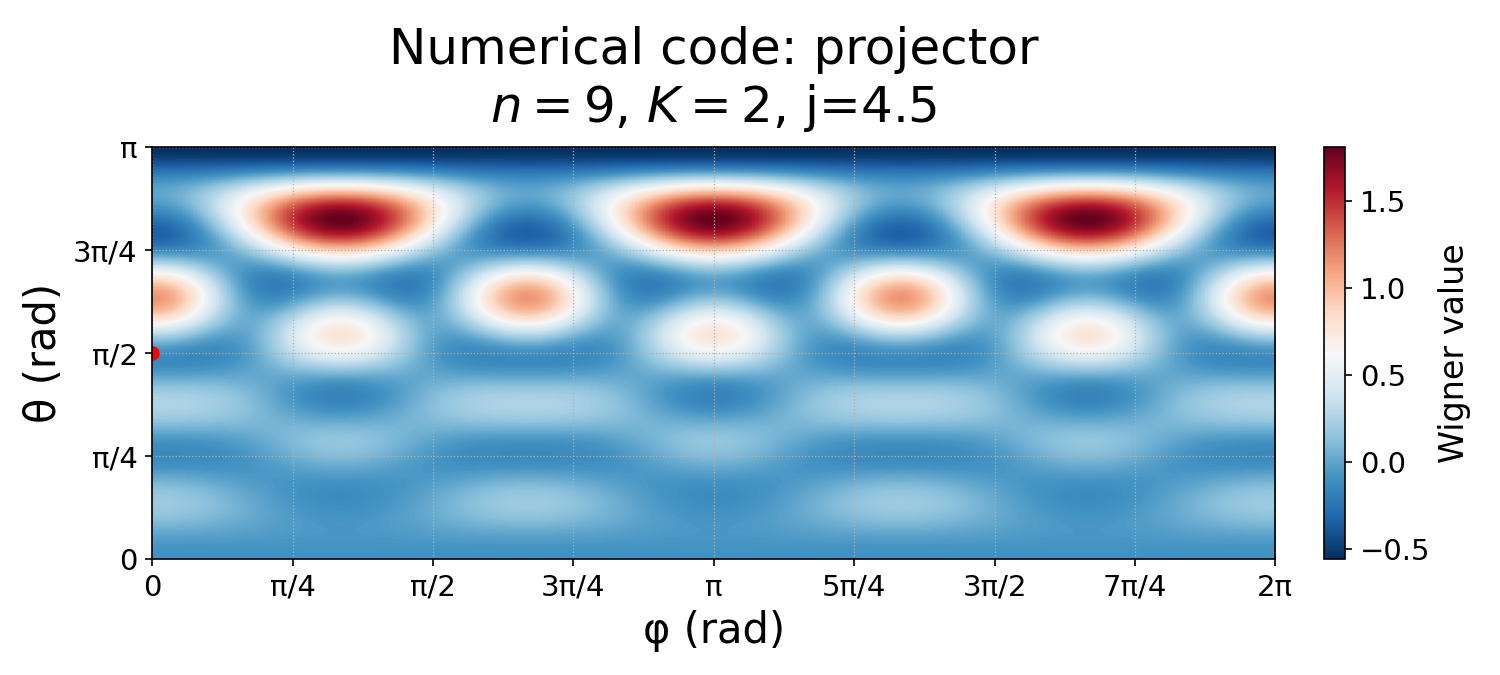}
        \caption{$n=9$, $\eta=0.70$, $C_{6}:\rho_{2}\oplus\rho_{6}$, $F_e=0.878327$.}
    \end{subfigure}
    \caption{Group-aligned projector spin-Wigner functions of four selected low-$\eta$ numerical codes with $C_6$ symmetry.}
    \label{fig:low-eta-c6}
\end{figure}

\begin{figure}[H]
    \centering
    \begin{subfigure}[t]{0.48\textwidth}
        \centering
        \includegraphics[width=\linewidth]{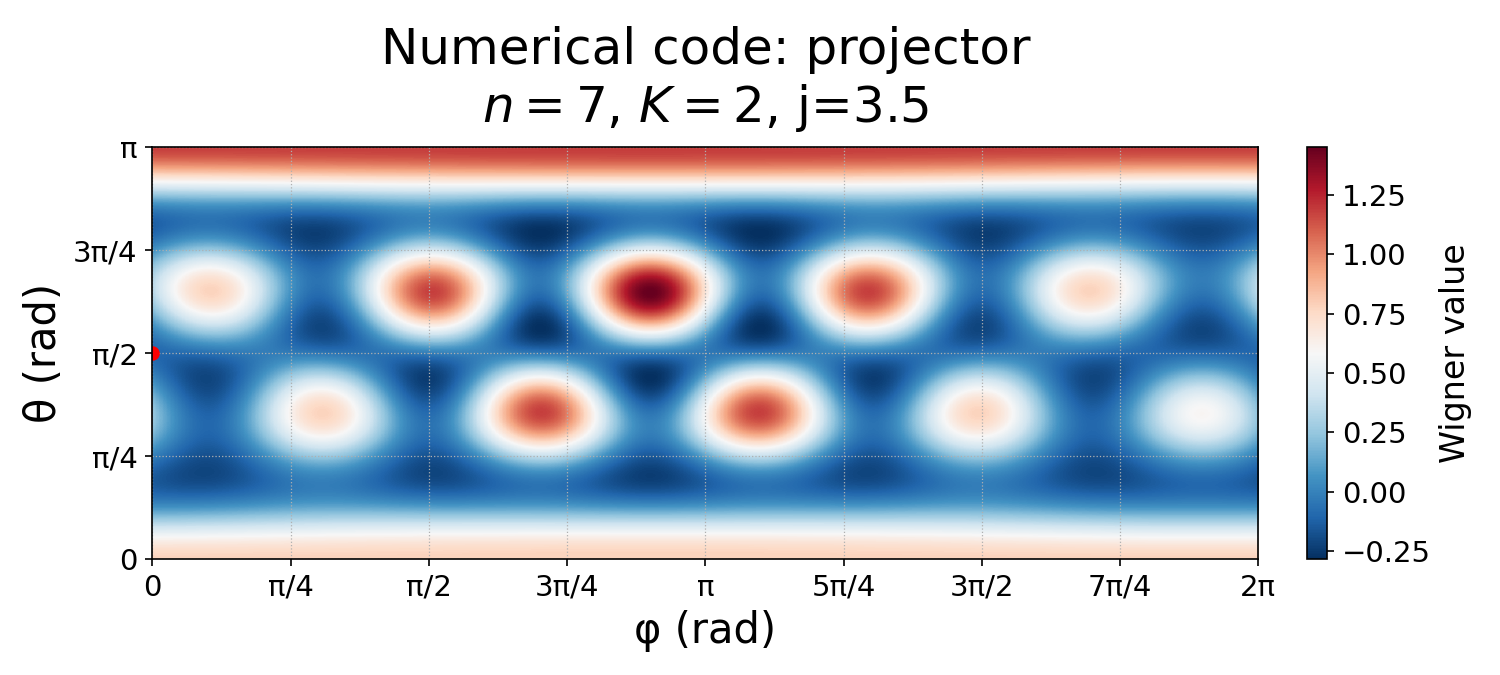}
        \caption{$n=7$, $\eta=0.85$, $C_{10}:\rho_{6}\oplus\rho_{8}$, $F_e=0.970657$.}
    \end{subfigure}
    \hfill
    \begin{subfigure}[t]{0.48\textwidth}
        \centering
        \includegraphics[width=\linewidth]{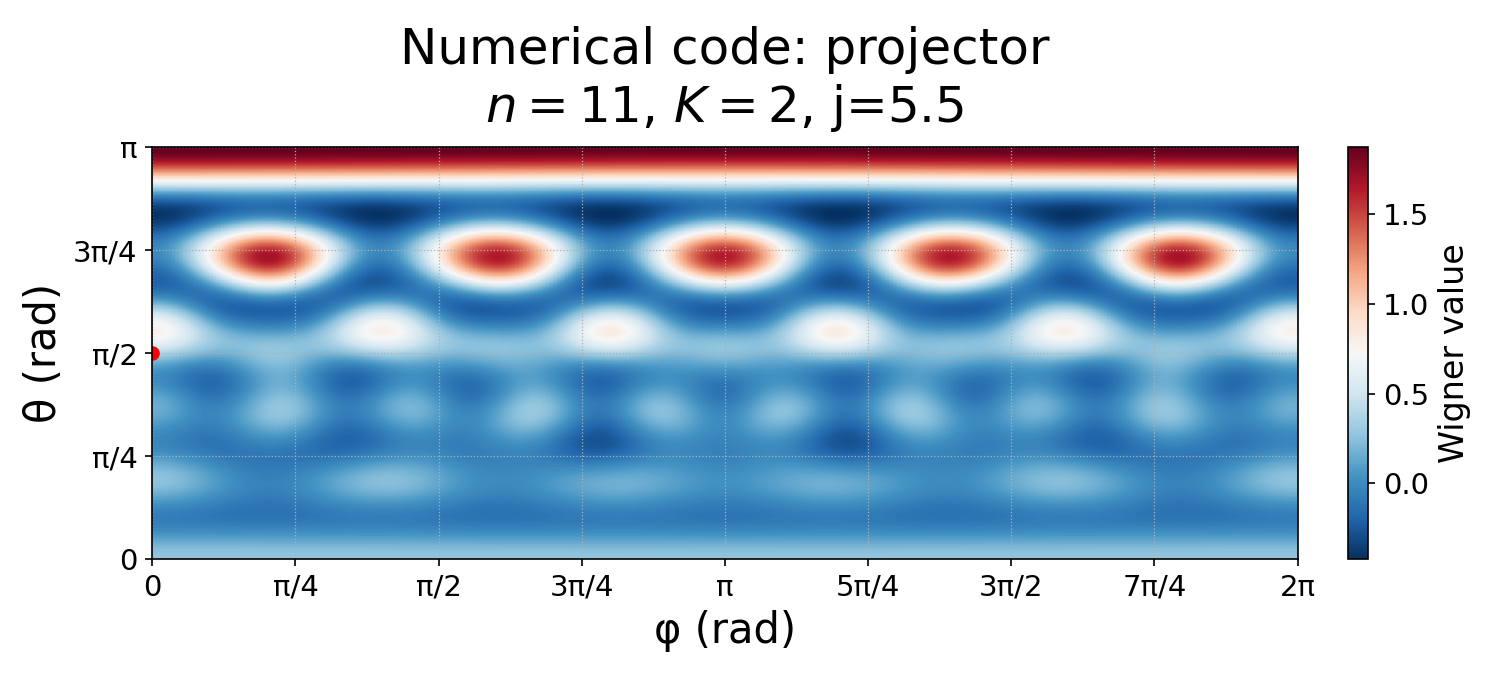}
        \caption{$n=11$, $\eta=0.75$, $C_{10}:\rho_{4}\oplus\rho_{6}$, $F_e=0.926577$.}
    \end{subfigure}

    \medskip

    \begin{subfigure}[t]{0.48\textwidth}
        \centering
        \includegraphics[width=\linewidth]{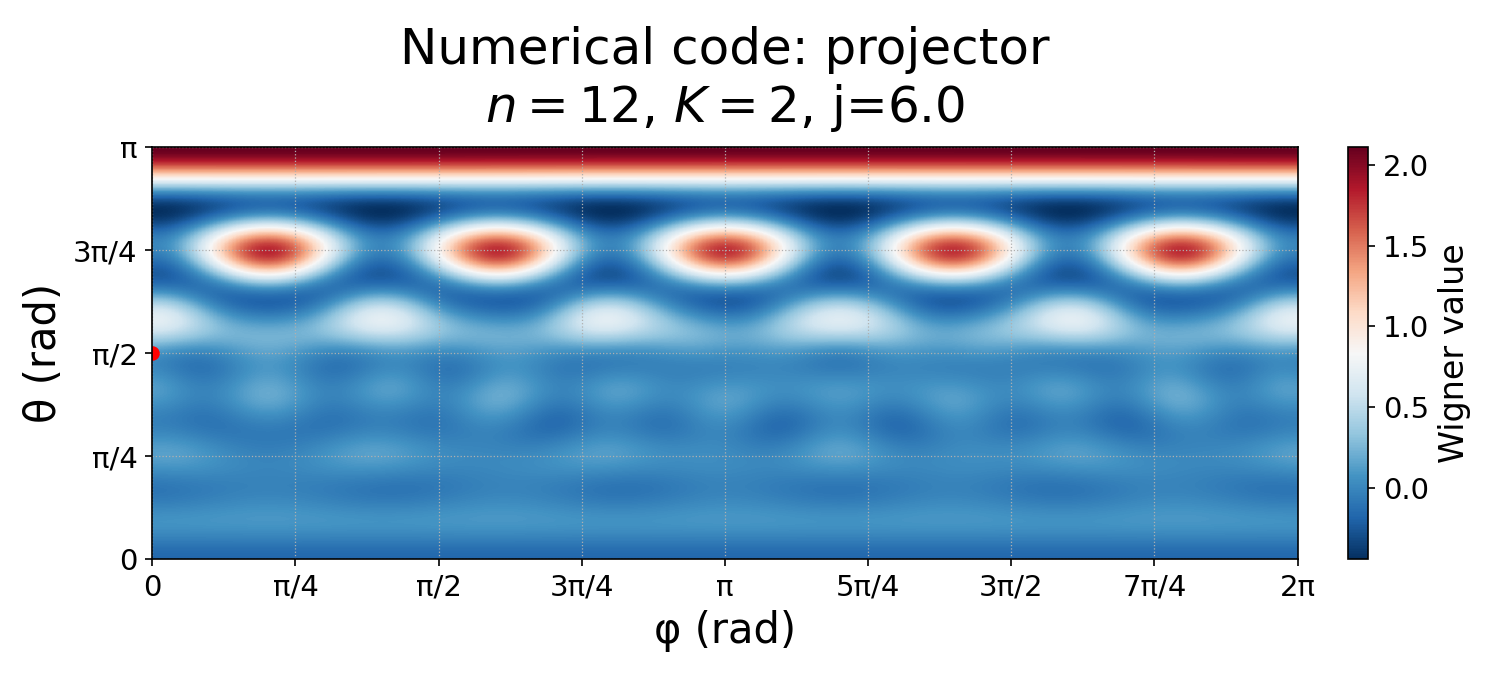}
        \caption{$n=12$, $\eta=0.70$, $C_{10}:\rho_{5}\oplus\rho_{7}$, $F_e=0.888776$.}
    \end{subfigure}
    \hfill
    \begin{subfigure}[t]{0.48\textwidth}
        \centering
        \includegraphics[width=\linewidth]{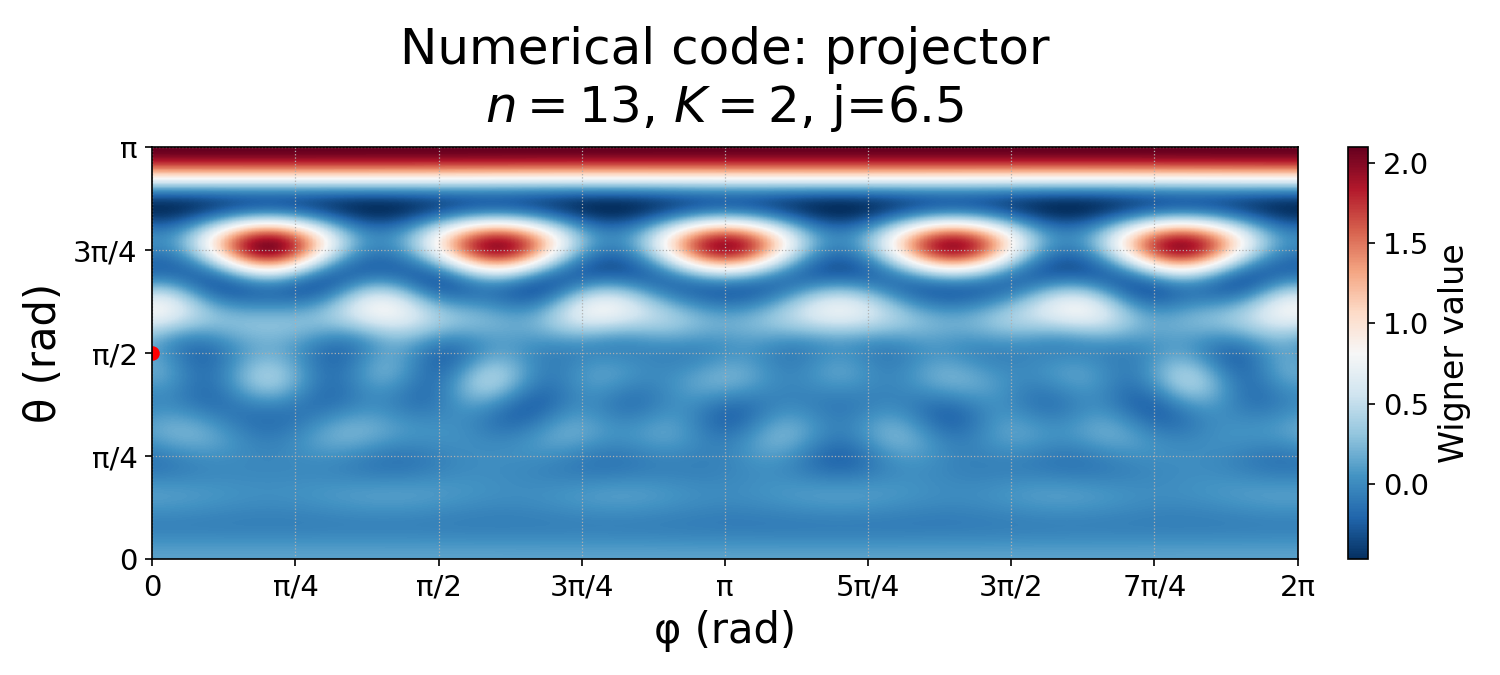}
        \caption{$n=13$, $\eta=0.70$, $C_{10}:\rho_{8}\oplus\rho_{10}$, $F_e=0.891057$.}
    \end{subfigure}
    \caption{Group-aligned projector spin-Wigner functions of the selected low-$\eta$ numerical codes with $C_{10}$ symmetry.}
    \label{fig:low-eta-c10}
\end{figure}

\begin{figure}[H]
    \centering
    \begin{subfigure}[t]{0.48\textwidth}
        \centering
        \includegraphics[width=\linewidth]{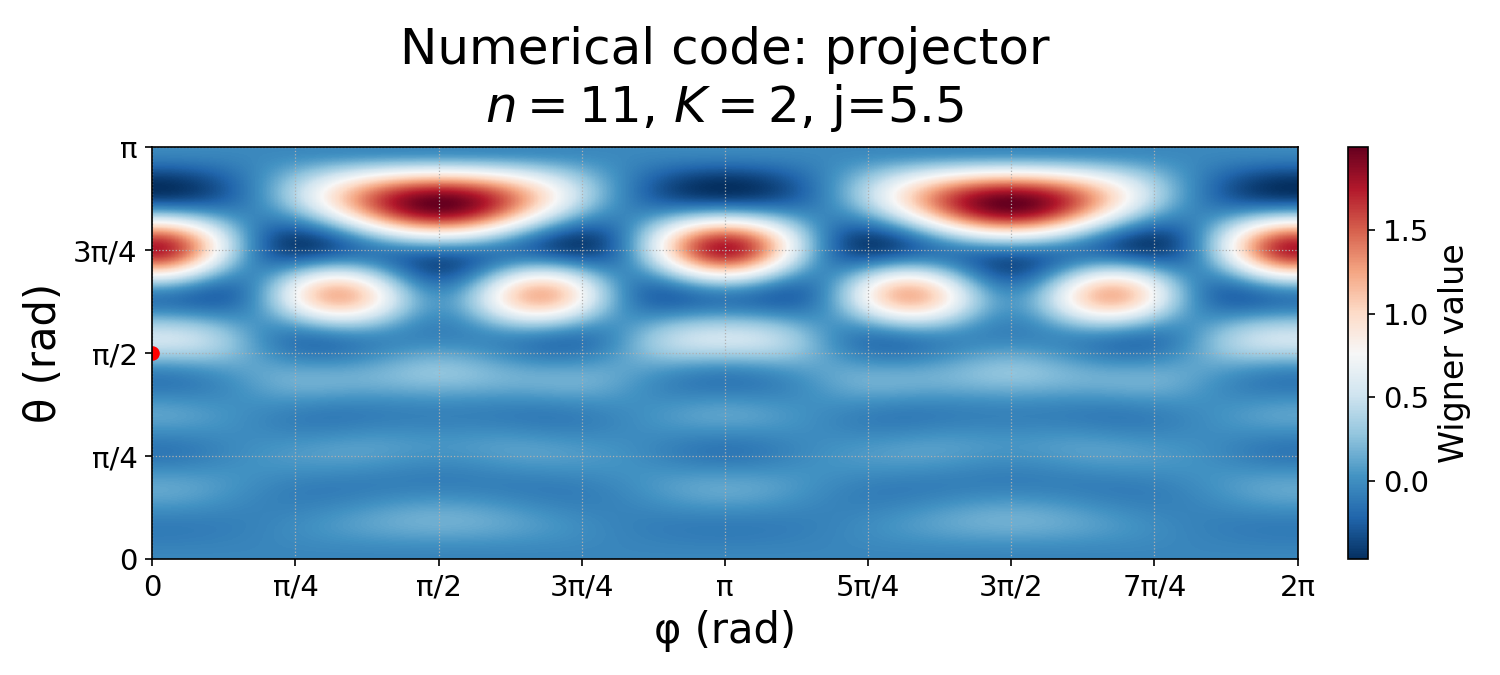}
        \caption{$n=11$, $\eta=0.70$, $C_{4}:\rho_{3}\oplus\rho_{4}$, $F_e=0.886079$.}
    \end{subfigure}
    \hfill
    \begin{subfigure}[t]{0.48\textwidth}
        \centering
        \includegraphics[width=\linewidth]{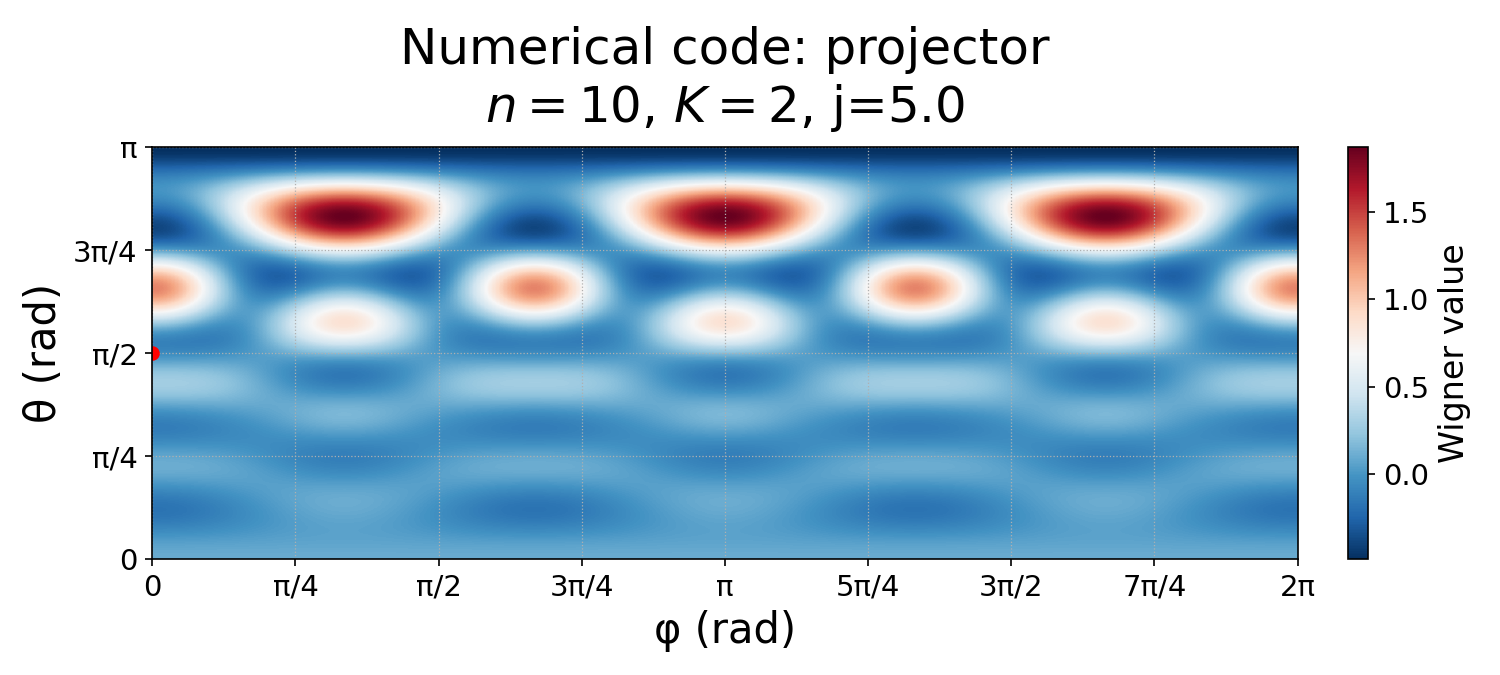}
        \caption{$n=10$, $\eta=0.70$, $C_{6}:\rho_{3}\oplus\rho_{5}$, $F_e=0.882753$.}
    \end{subfigure}

    \medskip

    \begin{subfigure}[t]{0.48\textwidth}
        \centering
        \includegraphics[width=\linewidth]{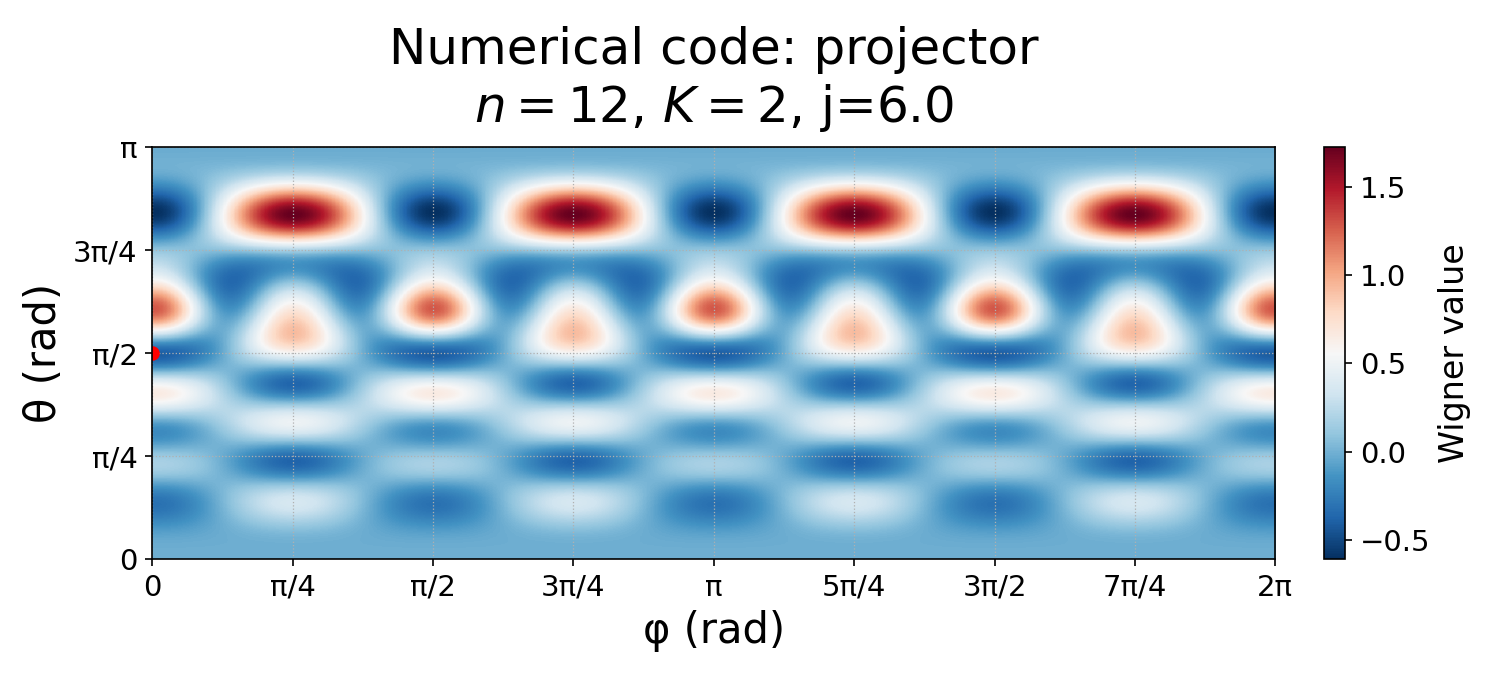}
        \caption{$n=12$, $\eta=0.80$, $C_{8}:\rho_{2}\oplus\rho_{4}$, $F_e=0.962246$.}
    \end{subfigure}
    \caption{Group-aligned projector spin-Wigner functions of the remaining selected low-$\eta$ numerical codes.}
    \label{fig:low-eta-remaining}
\end{figure}

\FloatBarrier

\subsection{Secondary basins}
\label{sec:secondary-basins}

Different random starts can sometimes reach distinct code structures at the same $(n,\eta)$. At $n=11$ and $\eta=0.90$, three of the five starts yield the $C_4$ branch shown in Figure~\ref{fig:reducible-n11}, while two yield the $C_8$ branch in Figure~\ref{fig:secondary-n11}. The two $C_8$ numerical projectors have best-found normalized overlap $0.9999982$ after $SU(2)$ alignment and occupation-basis conjugation, whereas the best saved representatives of the $C_4$ and $C_8$ branches have best-found overlap $0.7592$. Both $C_8$ runs are saved timeout checkpoints, so their agreement demonstrates reproducible structure without certifying convergence to a local optimum.

\begin{figure}[H]
    \centering
    \begin{subfigure}{0.48\textwidth}
        \centering
        \includegraphics[width=\linewidth]{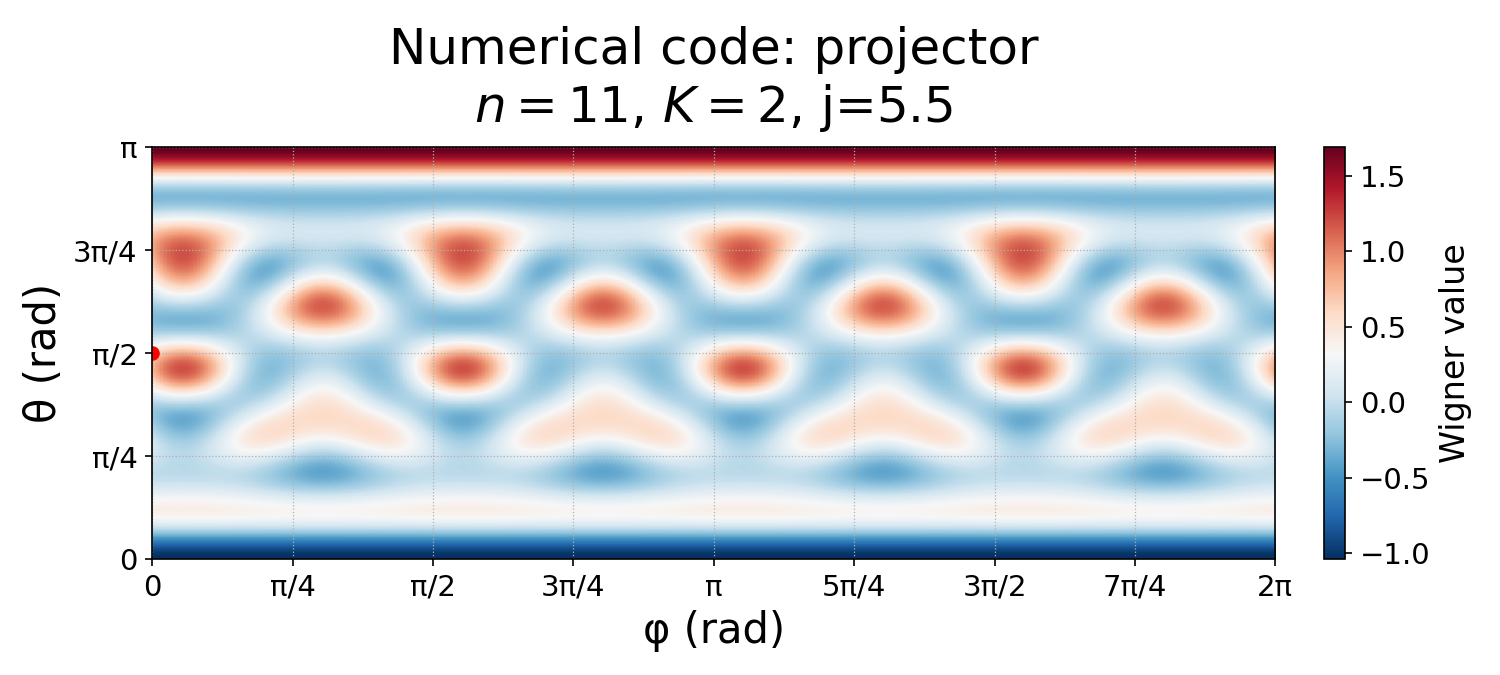}
        \caption{Numerical $C_8$ branch.}
    \end{subfigure}
    \hfill
    \begin{subfigure}{0.48\textwidth}
        \centering
        \includegraphics[width=\linewidth]{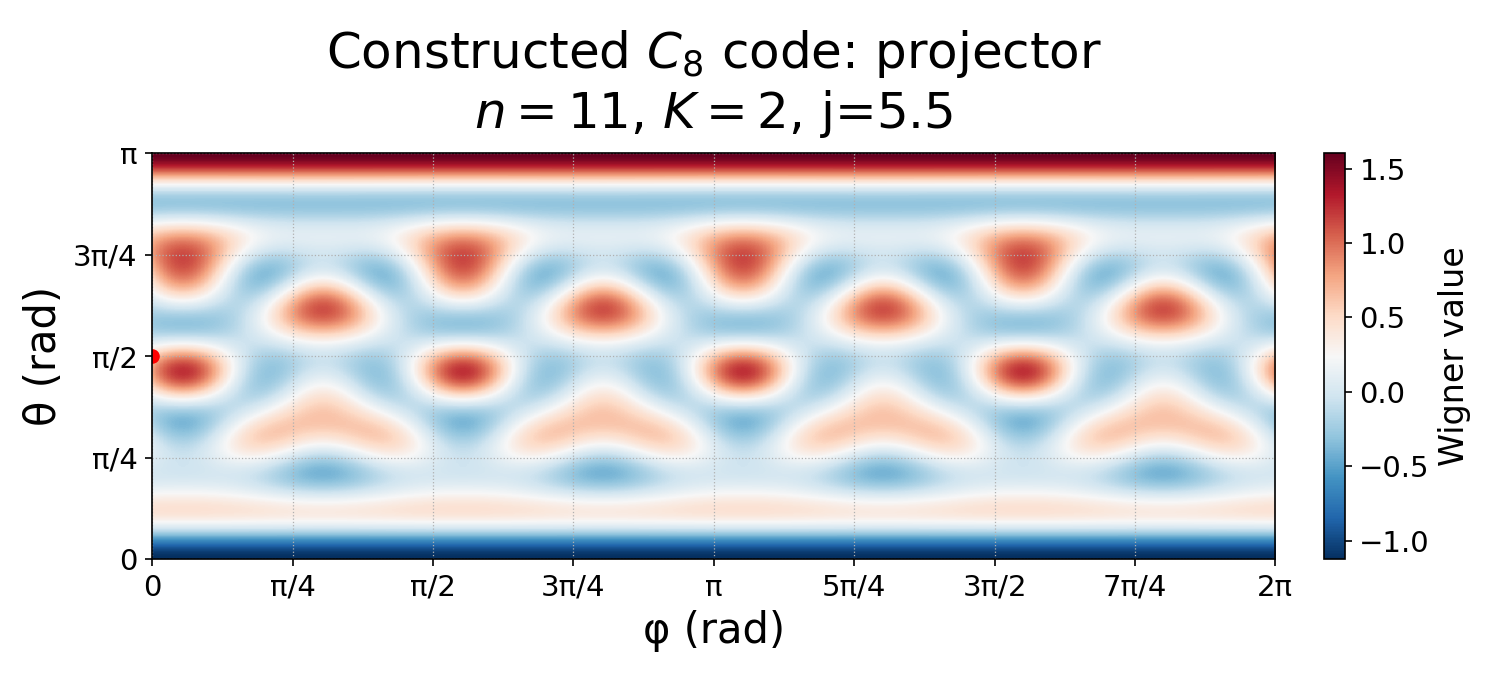}
        \caption{Constructed $C_8:\rho_5\oplus\rho_6$ code.}
    \end{subfigure}
    \caption{Secondary numerical branch and its analytic construction at $n=11$, $\eta=0.90$. The numerical panel uses the higher saved-fidelity member of the $C_8$ branch (seed $0$); the respective saved fidelities are $F_e^{\mathrm{num}}=0.995072$ and $F_e^{\mathrm{con}}=0.995002$. Their aligned projector overlap is $0.9983$. The best saved-fidelity run at this parameter point belongs to the $C_4$ branch, with $F_e=0.995075$.}
    \label{fig:secondary-n11}
\end{figure}

For the full sweep, we group runs into structural families using chains of best-found normalized projector overlaps of at least $0.95$, allowing $SU(2)$ rotations and occupation-basis conjugation. This is a heuristic structural criterion, not a convergence test. Of the $48$ families additional to the family containing the best saved-fidelity run at each parameter point, $32$ occur at $n\geq17$ and $\eta=0.95$. We do not interpret all these separations as evidence for distinct converged optima. At $n=17$ and $n=19,\ldots,25$, taking the best encoder found at $\eta=0.90$ and reoptimizing only its decoder at $\eta=0.95$ gives higher computed fidelity than the $\eta=0.95$ sweep winner, with both decoders evaluated at solver tolerance $10^{-9}$. At $n=18$, decoder-only reevaluation is a near-tie, while ten encoder--decoder rounds from the $\eta=0.90$ winner followed by final polishing give a small improvement. These checks indicate incomplete exploration of the high-transmissivity fidelity landscape. A plausible explanation is that at high transmissivity, structurally different encoders can have nearly identical fidelities, particularly when they correct the dominant loss processes similarly. The larger encoder search space at higher photon number may compound the resulting difficulty of resolving small improvements within a finite optimization budget. 

\FloatBarrier

\endgroup

\section{MacWilliams--Farkas distance bounds and LP monotonicity}

\subsection{The \texorpdfstring{$SU(2)$}{SU(2)} MacWilliams transform and exact Farkas certificates}

MacWilliams identities and linear-programming bounds for spin codes were developed by Bumgardner \cite{bumgardnerCodesWmetricSpaces2012a} and later generalized by Okada \cite{okadaQuantumAnalogDelsartes2025a}. We use the more recent formulation of Kubischta and Teixeira \cite{kubischtaMacWilliamsIdentitiesIntrinsic2026}, as stated in their Theorem~A1.

\begin{lemma}[Kubischta--Teixeira, Theorem~A1]
    Let $G$ a group with unitary representation $\pi_V: G\to \operatorname{GL}(V)$ on vector space $V$. Suppose that 

    \begin{equation}
\operatorname{End}(V) \cong \bigoplus_\ell W_\ell
\end{equation}

    is the decomposition of the endomorphism space of $V$ in terms of irreps of $G$ under the usual conjugation action, and further assume that each irrep appears with multiplicity at most $1$. Now, consider a codespace $\mathcal{C} \subseteq V$ defined by a projector $P$ with dimension $\operatorname{rank} P = K$ which detects every sector $W_\ell$ with $\ell \in \mathcal{E}$. Then, there exist real numbers $A_\ell, B_\ell$ such that 

    \begin{equation}
        \begin{aligned}
            A_\ell &\geq 0
                && \quad \forall \ell, \\
            A_\ell &\leq K B_\ell
                && \quad \forall \ell, \\
            \sum_\ell A_\ell &= K, \\
            \sum_\ell B_\ell &= K^2, \\
            B_\ell &= \sum_t M_{\ell t} A_t
                && \quad \forall \ell, \\
            A_\ell &= K B_\ell
                && \quad \forall \ell \in \mathcal{E}.
        \end{aligned}
    \end{equation}

    where $M = (M_{\ell t})$ is the relevant MacWilliams transform.  \label{kubischta-LP-theorem}
\end{lemma}

We use the application of this result to $G = SU(2)$, with $V = H_n \cong S_{n/2}$ where $H_n$ is the two-mode subspace with exactly $n$ photons, as defined earlier, and $S_{n/2}$ the spin-$n/2$ irrep of $SU(2)$. This is exactly the space on which all our code projectors are defined. Furthermore, from Clebsch-Gordan (Lemma~\ref{Clebsch-gordan}) and Lemma~\ref{su2-self-dual},

\begin{equation}
\operatorname{End}(S_{n/2}) \cong S_n \oplus S_{n-1} \oplus \dots \oplus S_0
\end{equation}

i.e., our vector space satisfies the multiplicity-free condition. The $SU(2)$ MacWilliams transform goes back to Bumgardner \cite{bumgardnerCodesWmetricSpaces2012a}. In the normalization used here, its matrix is given by Corollary~A3 of Ref.~\cite{kubischtaMacWilliamsIdentitiesIntrinsic2026}:

\begin{equation}
   [M_{n}]_{\ell t} = (-1)^{n + \ell + t}(2\ell + 1) \begin{Bmatrix}
        n/2 & n/2 & \ell \\ 
        n/2 & n/2 & t
    \end{Bmatrix} \label{su2-macwilliams}
\end{equation}

which is written in terms of the Wigner $6j$- symbols. Finally, as argued in Corollary \ref{loss-rotation-equivalence}, correcting loss up to order $k$ is equivalent to detecting the elementary tensors $W_0 = \mathcal{T}_n^0, \ W_1 = \mathcal{T}_n^1, \ \dots, W_k = \mathcal{T}_n^k$. Thus, we get the following corollary

\begin{corollary}
    If there exists a $K$-dimensional two-mode code on $H_n$ given by a projector $P$ that corrects $d-1$ loss errors, then there exist $A_\ell, B_\ell \in \mathbb{R}$ that satisfy the constraints of Lemma \ref{kubischta-LP-theorem}, with $\mathcal{E} = \{0, 1,2,\dots, d-1 \}$ and MacWilliams matrix $M_{n}$ given by equation \ref{su2-macwilliams}. \label{su2-LP-corollary} 
\end{corollary}

The following will be used to systematically show that LPs obtained from MacWilliams identities are infeasible \cite{boydConvexOptimization2023}.

\begin{lemma}[Farkas]
    Let $A \in \mathbb{R}^{m \times n}$ and $c \in \mathbb{R}^n$. The systems 
    \begin{align}
        S_1 &=\left\{
        \begin{aligned}
            Ax &\leq 0, \\
            c^T x &< 0,
        \end{aligned}
        \right.
        \label{system-1-farkas} \\[0.5em]
        S_2&=\left\{
         \begin{aligned}
            A^T y + c &= 0, \\
            y &\geq 0,
        \end{aligned}
        \right.
        \label{system-2-farkas}
    \end{align}
    are strong alternatives. 
    \label{Farkas}
\end{lemma}

Finally, we need to transform the above LP in the form of Farkas' Lemma. Specifically, we're going to define appropriate variables to write the LP in the form of the second Farkas system \ref{system-2-farkas}. We prove the following Lemma. 

\begin{lemma}
    Let $\mathbf{e}_{n+1}, \ \mathbf{0}_{n+1}$ the all-one, all-zero vectors in $\mathbb{R}^{n+1}$, then $I^{(n+1)}, \mathbf{O}_{n+1}$ the identity, all-zero  matrices in $\mathbb{R}^{(n+1) \times (n+1)}$, respectively, and finally $[I_d^{(n+1)}]_{ij} = \delta_{ij} \mathbbm{1}_{\{i \leq d\}}, \ 1 \leq i,j \leq n+1$, where $\mathbbm{1}$ the indicator function. It follows that if there exists a vector $\mathbf{x} \in \mathbb{R}^{3n + 5}$ satisfying the system

    \begin{equation}
    S = \begin{cases}
        U \mathbf{x} \leq 0 \\ 
        \mathbf{c}^T \mathbf{x} < 0
    \end{cases} \quad \text{where} \quad  U = \begin{pmatrix}
        -I^{(n+1)} & \mathbf{e}_{n+1} & \mathbf{0}_{n+1} & M_n^T & -I_d^{(n+1)} \\ 
        KI^{(n+1)}& \mathbf{0}_{n+1} & \mathbf{e}_{n+1} & -I^{(n+1)} & KI_{d}^{(n+1)} \\
         -I^{(n+1)} &  \mathbf{0}_{n+1} & \mathbf{0}_{n+1} & \mathbf{O}_{n+1} & \mathbf{O}_{n+1} 
    \end{pmatrix}, \quad \mathbf{c} = - \begin{pmatrix} \mathbf{0}_{n+1} \\ 
        K \\ K^2 \\ \mathbf{0}_{n+1} \\ \mathbf{0}_{n+1}
    \end{pmatrix} 
\end{equation}

then there exists no QEC code on $H_n$ with projector of rank $K$ that can correct up to $d-1$ losses.  \label{loss-distance-certificate}
\end{lemma}

\begin{proof}
    
    Let 
\begin{equation}
    \mathbf{A} = \begin{pmatrix}
        A_0 \\ 
        A_1 \\ 
        \vdots \\
        A_n
    \end{pmatrix}, \quad \mathbf{B} = \begin{pmatrix}
        B_0 \\ 
        B_1 \\
        \vdots \\ 
        B_n
    \end{pmatrix}, \quad \mathbf{s} = K \mathbf{B} - \mathbf{A}, \quad \mathbf{z} = \begin{pmatrix}
        \mathbf{A} \\ \mathbf{B} \\ \mathbf{s}
    \end{pmatrix} 
\end{equation}

It's clear that $\mathbf{A} \geq 0, \ \mathbf{B} \geq 0, \ K \mathbf{B} \geq \mathbf{A} \implies \mathbf{s} \geq 0 \implies \mathbf{z} \geq 0$. This last equation capture the two inequalities of the LP; we now need only encode the $4$ linear equations in $\mathbf{z}$, plus the relationship $\mathbf{s} = K \mathbf{B} - \mathbf{A}$. Indeed, we see we can write 

\begin{equation}
     \begin{pmatrix} 
        -I^{(n+1)} & K I^{(n+1)} & - I^{(n+1)} \\
        \mathbf{e}_{n+1}^T & \mathbf{0}_{n+1}^T  & \mathbf{0}_{n+1}^T\\ 
        \mathbf{0}_{n+1}^T & \mathbf{e}_{n+1}^T & \mathbf{0}_{n+1}^T \\ 
       M_{n} & -I^{(n+1)}  & \mathbf{O}_{n+1} \\ 
       -I^{(n+1)}_d & KI^{(n+1)}_{d} & \mathbf{O}_{n+1} 
     \end{pmatrix} \begin{pmatrix}
        \mathbf{A} \\ \mathbf{B} \\ \mathbf{s}
    \end{pmatrix} = \begin{pmatrix} \mathbf{0}_{n+1} \\ 
    K \\ 
    K^2  \\ 
    \mathbf{0}_{n+1} \\
    \mathbf{0}_{n+1}
     \end{pmatrix}
\end{equation}

Thus, in the notation of Farkas' lemma, we have the system 

\begin{equation}
    U = \begin{pmatrix}
        -I^{(n+1)} & \mathbf{e}_{n+1} & \mathbf{0}_{n+1} & M_n^T & -I_d^{(n+1)} \\ 
        KI^{(n+1)}& \mathbf{0}_{n+1} & \mathbf{e}_{n+1} & -I^{(n+1)} & KI_{d}^{(n+1)} \\
         -I^{(n+1)} &  \mathbf{0}_{n+1} & \mathbf{0}_{n+1} & \mathbf{O}_{n+1} & \mathbf{O}_{n+1} 
    \end{pmatrix}, \quad \mathbf{c} = - \begin{pmatrix} \mathbf{0}_{n+1} \\ 
        K \\ K^2 \\ \mathbf{0}_{n+1} \\ \mathbf{0}_{n+1}
    \end{pmatrix}, \quad \text{with equations} \quad S = \begin{cases}
        U^T \mathbf{z} + \mathbf{c} = 0 \\ 
        \mathbf{z} \geq 0
    \end{cases}
\end{equation}

So, from Farkas' lemma \ref{Farkas}, we get that if there exists an $\mathbf{x} \in \mathbb{R}^{3n + 5}$ such that 

\begin{equation}
    \mathbf{c}^T \mathbf{x}< 0, \ U \mathbf{x} \leq 0, 
\end{equation}

then from Corollary \ref{su2-LP-corollary} there exists no two-mode code on $H_n$ with dimension $K$ that corrects $d-1$ loss errors, as desired. 
\end{proof}

Finally, we note that the Wigner-$6j$ symbols are given by the Racah formula that takes the form \cite{varsalovicQuantumTheoryAngular1989}

\begin{equation}
    \begin{aligned}
        \begin{Bmatrix}
            a & b & c\\
            d & e & f
        \end{Bmatrix}
        &=
        \Delta(abc)\Delta(cde)\Delta(aef)\Delta(bdf)
        \times \text{rational sum},\\
        \Delta(abc)
        &=
        \left(
        \frac{
            (a+b-c)!(a-b+c)!(-a+b+c)!
        }{
            (a+b+c+1)!
        }
        \right)^{1/2}.
    \end{aligned}
\end{equation}

In our case, we have $a = b = d = e = n/2 = j, \ c = \ell, \ f = t$, so we get the product $\Delta(jj\ell) \Delta(\ell jj) \Delta(jjt) \Delta(jjt) = \Delta(jj\ell)^2 \Delta(jjt)^2 = \frac{(n-\ell)! (\ell!)^2 }{(n + \ell + 1)!} \cdot \frac{(n-t)!(t!)^2}{(n+t+1)!}$ using the fact that $\Delta$ is clearly invariant under cyclic permutations. So, we get that the relevant $6j$ symbols $\begin{Bmatrix}
    n/2 & n/2 & \ell \\ 
    n/2 & n/2 & t
\end{Bmatrix} \in \mathbb{Q}$, which means we can use exact symbolic algebra to find certificates, avoiding potential numerical tolerance issues.

\subsection{Monotonicity of LP feasibility in photon number}

We finally prove the monotonicity of LP feasibility. 

\begin{theorem}[Monotonicity of LP feasibility]
    If the loss linear program is feasible for some $(n,K,d)$, then it is also feasible for $(n+1,K,d)$. Consequently, if distance $d$ is infeasible for a given $(n,K)$, then it is also infeasible for $(n-1,K)$.
    \label{LP-feasibility-monotonicity}
\end{theorem}

 The motivation for this was originally to be able to speed up the process of obtaining the distance certificates, since it implies that if we've found a solution to the $\operatorname{LP}(n,\,K,\,d)$ and we know that  $\operatorname{LP}(n,\,K,\,d+1)$ is infeasible, then we can start the search for the maximum possible distance at $d+1$ for $n+1$ (instead of starting at distance $1$). Nevertheless, this result is meaningful progress towards the bigger question of whether the existence of an $(n,K,d)$ code implies the existence of an $(n+1, K,d)$ code.

The proof follows from the following lemmas.

\begin{lemma}
    Suppose the loss linear program is feasible for some $(n,K,d)$ with solutions $\mathbf{A}, \mathbf{B}$. Further assume there exists an intertwiner $T_n \in \mathbb{R}^{(n+2)\times (n+1)}$ such that:
    \begin{itemize}
        \item $M_{n+1} T_n = T_n M_n$ 
        \item $T_n \geq 0$
        \item $\mathbf{1}^T T_n = \mathbf{1}^T$
        \item $[T_n]_{\ell t} = 0$ for $\ell < t$. 
    \end{itemize}
    Then, $\mathbf{A'} = T_n \mathbf{A}, \, \mathbf{B'} = T_n \mathbf{B}$ satisfy the loss LP for $(n+1, K,d)$. \label{LP-monotonicity-condition}
\end{lemma}

\begin{proof}
    Firstly, since $T_n \geq 0, \, \mathbf{A} \geq 0$, it immediately follows that $\mathbf{A'} \geq 0$. Furthermore, we have that $\mathbf{A} \leq K \mathbf{B}$, and for the new variables we similarly see
    \begin{equation}
K\mathbf{B'} - \mathbf{A'} = T_n (K\mathbf{B} - \mathbf{A}) \geq 0
\end{equation}

    We also start with $\mathbf{1}^T \mathbf{A} = K, \ \mathbf{1}^T \mathbf{B} = K^2$, so we get
    \begin{equation}
\mathbf{1}^T \mathbf{A'} = \mathbf{1}^T T_n \mathbf{A} = \mathbf{1}^T \mathbf{A} = K
\end{equation}
    and likewise 
    \begin{equation}
\mathbf{1}^T \mathbf{B'} = \mathbf{1}^T T_n \mathbf{B} = \mathbf{1}^T \mathbf{B} = K^2
\end{equation}
    
    We also satisfy the new MacWilliams identity 
    \begin{equation}
\mathbf{B'} = T_n \mathbf{B} = T_n M_n \mathbf{A} = M_{n+1} T_n \mathbf{A} = M_{n+1} \mathbf{A'}
\end{equation}

    and finally, for $0 \leq \ell \leq d-1$ we have $A_\ell = KB_\ell$, so 

    \begin{equation}
\mathbf{A'}_\ell = (T_n\mathbf{A})_\ell = \sum_{k =0}^{\ell} [T_n]_{\ell k} A_k = \sum_{k =0}^\ell [T_n]_{\ell k} K B_k = K \sum_{k =0}^\ell [T_n]_{\ell k}  B_k  = K \mathbf{B'}_\ell
\end{equation}

    So, $\mathbf{A'},\, \mathbf{B'}$ indeed satisfy the LP for $(n+1, K,d)$ as desired. 
\end{proof}

Thus, the goal is to find a $T_n$ satisfying these conditions.  We begin by simplifying the form of the $SU(2)$ MacWilliams transform \eqref{su2-macwilliams}.

\begin{lemma}
    The $SU(2)$ MacWilliams matrix can be written as

    \begin{equation}
        [M_n]_{\ell t} = (2\ell+1) \cdot \sum_{k = 0}^{\min(\ell,t)} (-1)^k \frac{(t + k)! (\ell+k)! (n-k)! }{(t-k)! (\ell-k)! (n+1+k)! (k!)^2}
    \end{equation}
    
    where $n \in \mathbb{Z}_{\geq 0}, \ 0 \leq \ell,t \leq n$. 
    \label{su2-macwilliams-simplified}
\end{lemma}

\begin{proof}
    We start by writing the Wigner-$6j$ symbol in terms of hypergeometric ${}_4F_3$ functions. Indeed \cite{varsalovicQuantumTheoryAngular1989}

\begin{equation}
    \begin{aligned}
        \begin{Bmatrix}
        a & b & c \\ 
        d & e & f
    \end{Bmatrix} &=  \frac{ (-1)^{a + b + d + e}\Delta(abc)\Delta(cde)\Delta(aef)\Delta(bdf) (a + e + f + 1)! (b + d + f + 1)! (a + c + d - f)!}{(a + b - c)! (a - b + c)! (-c + d + e)! (c + d - e)! (a + e - f)! (b + d - f)! (-a + c - d + f)! (2f + 1)!}   \\ 
    &\times {}_4F_3 \ \left[ \begin{matrix}
        -a - e + f,\, -b - d + f,\, -a + e + f + 1,\, b - d + f + 1 \\ 
        -a - c - d + f,\, -a + c - d + f + 1,\, 2f + 2
    \end{matrix} \middle| 1 \right]
    \end{aligned}
\end{equation}

In our case, $a = c = d = f = n/2$, $b= \ell,\, e = t$ (using the fact that $6j$ symbols are column-permutation invariant), this simplifies to

\begin{equation}
    \begin{aligned}
        \begin{Bmatrix}
            n/2 & \ell & n/2 \\ 
            n/2 & t & n/2
        \end{Bmatrix} &= \frac{(-1)^{n + \ell + t} (n-\ell)! (n-t)! (\ell! t!)^2}{(n + \ell + 1)!(n + t + 1)!} \times \frac{(n + t + 1)! (n + \ell + 1)! n! }{ (n-\ell)! (n-t)! (n + 1)! (t! \ell!)^2} \\ 
        &\times {}_4F_3 \left[ \begin{matrix}
       -t,\, -\ell ,\, t + 1 ,\,\ell + 1    \\ 
       -n,\, 1 ,\, n+2 
    \end{matrix} \middle| 1 \right] \\ 
        &= \frac{(-1)^{n + \ell + t}}{n+1} \times  {}_4F_3 \left[ \begin{matrix}
        -t,\, -\ell ,\, t + 1 ,\,\ell + 1    \\ 
        -n,\, 1 ,\, n+2 
        \end{matrix} \middle| 1 \right] 
    \end{aligned}
\end{equation}

Now, using \cite[\href{https://dlmf.nist.gov/16.2.E1}{Eq.~(16.2.1)}]{NISTDigitalLibrary},  and the convention for terminating hypergeometric series with negative-integer parameters \cite{gesselFindingIdentitiesWZ1995}, we get

\begin{equation}
    \begin{Bmatrix}
        n/2 & n/2 & \ell \\ 
        n/2 & n/2 & t
    \end{Bmatrix} = \frac{(-1)^{n + \ell + t}}{n + 1} \times {}_4 F_3 \left[ \begin{matrix}
        -t,\, -\ell ,\, t + 1 ,\,\ell + 1    \\ 
        -n,\, 1 ,\, n+2 
        \end{matrix} \middle| 1 \right]  = \frac{(-1)^{n + \ell + t}}{n+1} \times \sum^{ \min(\ell,t)}_{k =0} \frac{(-t)_k (-\ell)_k (t + 1)_k (\ell+1)_k}{(-n)_k (1)_k (n+2)_k  \cdot k!}
\end{equation}

where $(a)_r = a (a+1) \dots (a + r-1)$ for $a \in \mathbb{C},\,  r \in \mathbb{Z}^+$, with $(a)_0 = 1$.  So for $m \in \mathbb{Z}_{\geq 0}$ we have $(m)_r = m(m+1)(m+2) \dots(m+r-1) , \ (-m)_r = \begin{cases}
    \frac{(-1)^r m!}{(m-r)!}, & 0 \leq r \leq m \\ 0, & r > m 
\end{cases}$; see \cite[\href{https://dlmf.nist.gov/5.2}{Eq.~(5.2.4, 5.2.7 )}]{NISTDigitalLibrary}. So, our formula becomes

\begin{equation}
    \begin{aligned}
        \begin{Bmatrix}
        n/2 & n/2 & \ell \\
        n/2 & n/2 & t
        \end{Bmatrix} &= \frac{(-1)^{n + \ell + t}}{n + 1} \times \sum_{k =0}^{\min(\ell,t)}  (-1)^k \frac{t! \ell! (t + k)! (\ell+k)! }{(t-k)! (\ell-k)! t! \ell! k! } \cdot \frac{(n-k)! (n+1)!}{n! k! (n + 1 + k)!} \\ 
        &=  (-1)^{n + \ell + t} \sum_{k =0}^{\min(\ell,t)} (-1)^k  \frac{(t + k)! (\ell+k)! (n-k)! }{(t-k)! (\ell-k)! (n+1+k)! (k!)^2}
    \end{aligned}
\end{equation}

Substituting into equation \eqref{su2-macwilliams} gives the desired formula. 

\end{proof}

We now show how we can factor the MacWilliams matrix. 

\begin{lemma}
    Define the matrices 

    \begin{equation}
        [\mathsf{R}_n]_{\ell j} = \begin{cases}
            \frac{(2\ell+1)(\ell+j)! (n-j)!}{j!(\ell-j)! (n-\ell)! (n+\ell+1)!}, & j \leq \ell \\
            0, & j > \ell 
        \end{cases}, \quad [\mathsf{U}_n]_{ij} = \begin{cases}
            (-1)^i \binom{j}{i} , & i \leq j \\ 
            0, & i > j 
        \end{cases}
    \end{equation}

    with $0 \leq \ell,j \leq n$. Then, it follows that $\mathsf{R}_n$ is invertible, and we can write 

    \begin{equation}
        M_n = \mathsf{R}_n \mathsf{U}_n \mathsf{R}_n^{-1} 
    \end{equation} \label{su2-macwilliams-matrix-factorization}
\end{lemma}

\begin{proof}
    Indeed, we see
    
    \begin{equation}
        \begin{aligned}
        [\mathsf{R}_n \mathsf{U}_n]_{\ell t} &=  \sum_k [\mathsf{R}_n]_{\ell k} [\mathsf{U}_n]_{kt}  =  \sum_{k \leq \min(\ell,t)} \frac{(2\ell + 1) (\ell+k)! (n-k)!}{k! (\ell-k)! (n-\ell)! (n+\ell+1)!} \cdot (-1)^k \frac{t!}{k! (t-k)!} \\ 
        &= \sum_{k = 0}^{\min(\ell,t)} (-1)^k  \frac{(2\ell+1) (\ell+k)! (n-k)! t!}{(k!)^2 (t-k)! (\ell-k)! (n-\ell)! (n+\ell+1)!}
    \end{aligned}, \label{R_nU_n}
    \end{equation}

    and similarly, 

    \begin{equation}
        \begin{aligned}
            [M_n \mathsf{R}_n]_{\ell t}
            &= \sum_{k} [M_n]_{\ell k} [\mathsf{R}_n]_{kt} \\
            &= (2\ell + 1)  \sum_{k=t}^n \frac{(2k+1)(k+t)! (n-t)!}{t!(k-t)! (n-k)! (n+k + 1)!} \\
            &\quad\times \sum_{j = 0}^{\min(\ell,k)} (-1)^j \frac{(k + j)! (\ell + j)! (n-j)!}{(k-j)! (\ell-j)! (n+1+j)! (j!)^2} 
        \end{aligned}
    \end{equation}

    Now, we're going to utilize the Whipple identity \cite{koepfHypergeometricSummationAlgorithmic2014} (we specialize to the case $e = -N$)
    
    \begin{equation}
        {}_4F_3 \left[ \begin{matrix}
        a,\, 1 + a/2,\, d,\, -N \\ 
        a/2,\, 1 + a-d,\, 1 + a+N
    \end{matrix}  \middle| -1 \right] = \frac{(1+a)_N}{(1+a-d)_N} \label{whipple}
    \end{equation}

    Note that the sum is defined for $t \leq k \leq n,\, 0 \leq j \leq \ell, \ 0 \leq j \leq k$. We want to flip the order of the finite sum i.e., fix $j$ first, which gives us $t,j \leq k \leq n, \ 0 \leq j \leq \ell$. Thus, the sum becomes 

    \begin{equation}
        \begin{aligned}
            [M_n \mathsf{R}_n]_{\ell t} &= (2\ell+1) \sum_{j = 0}^\ell  (-1)^j \frac{(\ell+j)! (n-j)! (n-t)!}{t! (\ell-j)! (n+1+j)! (j!)^2} \sum_{k = \max(t,j)}^n  \frac{(k+j)! (2k+1) (k+t)!}{(k-t)! (n-k)! (n+k+1)! (k-j)!}
        \end{aligned}
    \end{equation}

    Re-writing the inner sum with $r = n-k \implies k = n-r$, we have 

    \begin{equation}
        \begin{aligned}
            &\sum_{r = 0}^{n-\max(t,j)}  \frac{(n-r + j)! (2n - 2r+1) (n-r+t)!}{(n-r-t)! (r)! (2n-r+1)! (n-r-j)!}  \\ 
            = \frac{(n+j)!(n+t)!}{(n-t)! (n-j)! (2n)! } &\sum_{r =0}^{n- \max(j,t)} \frac{(n-r + j)! (n-j)!}{r! (n+j)! (n-r-j)!} \cdot \frac{(n-r + t)! (n-t)!}{(n + t)! (n-r - t)!}  \cdot \frac{(2n+1)!}{(2n - r + 1)!}  \cdot  \frac{(2n-2r + 1)}{2n+1} \\ 
             = \frac{(n+j)!(n+t)!}{(n-t)! (n-j)! (2n)!} &\sum_{r =0}^{n- \max(j,t)} \frac{(j-n)_r}{(-j-n)_r} \cdot \frac{(t-n)_r}{(-t -n)_r} \cdot (-1)^r (-2n - 1)_r \cdot \frac{(1/2 - n)_r}{(-1/2 - n)_r} \cdot \frac{1}{r!} \\
             &= \frac{(n+j)!(n+t)!}{(n-t)! (n-j)! (2n)!} {}_4F_3 \left[\begin{matrix}
                -2n - 1,\, 1 - (2n + 1)/2,\, t-n,\, j-n \\ 
                -n-1/2,\, -n -t,\, -n-j
             \end{matrix} \middle| -1 \right] 
        \end{aligned}
    \end{equation}

    and we see this is equation \eqref{whipple} with $a = -2n - 1$, $d = t-n$, $N = n-j$. So, we can just substitute

    \begin{equation}
        \begin{aligned}
            [M_n \mathsf{R}_n]_{\ell t} &= (2\ell+1) \sum_{j = 0}^\ell  (-1)^j \frac{(\ell+j)! (n-j)! (n-t)!}{t! (\ell-j)! (n+1+j)! (j!)^2}  \cdot \frac{(n+j)! (n+t)!}{(n-t)! (n-j)! (2n)!} \frac{(2n)! (t+j)!}{(n + j)! (n+t)!} \\
            &= (2\ell+1) \sum_{j = 0}^\ell  (-1)^j  \frac{(\ell+j)! (t+j)!}{t!(\ell-j)!(n+1+j)! (j!)^2} \\ 
            &= \frac{(2\ell+1)}{(n+1)!} {}_3 F_2 \left[\begin{matrix}
                \ell+1,\, t+1,\, -\ell \\
                n+2,\, 1
            \end{matrix}\middle| 1 \right]
        \end{aligned}
    \end{equation}
    using $(-\ell)_j (\ell+1)_j = (-1)^j \frac{(\ell+j)!}{(\ell-j)!}$.  We now apply the terminating Thomae transformation \cite{raoGroupTheoreticalBasis1992} \begin{equation}
{}_3F_2 \left[\begin{matrix}
    a,\, b,\, -N \\ d,\, e
    \end{matrix} \middle| 1 \right] = \frac{(d-a)_N}{(d)_N}{}_3F_2 \left[\begin{matrix}
        a,\, e-b,\, -N \\ 1+a-d-N,\, e
    \end{matrix} \middle| 1 \right].
\end{equation}

    In our notation we have $a = \ell+1$, $b = t+1$, $N = \ell$, $d = n+2$, $e = 1$. So, this gives us 

    \begin{equation}
        \begin{aligned}
            [M_n \mathsf{R}_n]_{\ell t} &= \frac{(2\ell+1)}{(n+1)!} {}_3 F_2 \left[\begin{matrix}
                \ell+1,\, t+1,\, -\ell \\
                n+2,\, 1
            \end{matrix}\middle| 1 \right] \\ 
            &=  \frac{(2\ell+1)n!}{(n-\ell)! (n + 1 + \ell )!} {}_3F_2  \left[\begin{matrix}
                \ell+1,\, -t,\, -\ell \\
                 -n,\, 1
            \end{matrix}\middle| 1 \right] \\ 
            &= \frac{(2\ell +1)n!}{(n-\ell)! (n+1+\ell)!} \sum_{k=0}^{\min(\ell,t)} \frac{(\ell+1)_k (-t)_k (-\ell)_k}{(-n)_k (1)_k } \frac{1}{k!} \\ 
            &= \frac{(2\ell +1)}{(n-\ell)! (n+1+\ell)!} \sum_{k=0}^{\min(\ell,t)} (-1)^k \frac{(\ell+k)! t! (n-k)! }{ (t-k)! (\ell-k)! (k!)^2} = [\mathsf{R}_n \mathsf{U}_n]_{\ell t}
        \end{aligned} 
    \end{equation}

    which was to be shown. Furthermore, since $\mathsf{R}_n$ is triangular, and all its diagonal entries are non-zero, $\mathsf{R}_n$ is invertible, so we conclude $M_n = \mathsf{R}_n \mathsf{U}_n \mathsf{R}_n^{-1}$, as desired. 

\end{proof}

Now, we find an intertwiner $\mathsf{Q}_n$ that satisfies the desired commutation in the rotated basis of the MacWilliams matrix. 

\begin{lemma}
    
Let

\begin{equation}
    \begin{aligned}
        a_{n,t} = \frac{(n+1-t)(n+2t + 2)}{n+2}, \quad  b_{n,t} &= \frac{4t^2 + (2-3n)t - 2n}{n+2}, \quad  c_{n,t} = \frac{2(t+1)(n-t)}{n+2} \\ 
        \mathsf{Q}_n \in \mathbb{R}^{(n+2)\times (n+1)}, \quad [\mathsf{Q}_n]_{it} &= \begin{cases}
    a_{n,t}, & i = t, \\ 
    b_{n,t}, & i = t+1, \\ 
    c_{n,t}, & i = t + 2, \\ 
    0, & \text{otherwise}
\end{cases}
    \end{aligned} \label{Qn-definition}
\end{equation}

It follows that $\mathsf{U}_{n+1} \mathsf{Q}_n = \mathsf{Q}_n \mathsf{U}_n$. 
\end{lemma}

\begin{proof}
    Indeed, with the convention that $\binom{j}{i} = 0$ for $i > j,\, i < 0$, we have

\begin{equation}
    \begin{aligned}
        [\mathsf{U}_{n+1} \mathsf{Q}_n]_{ij}
        &= \sum_t [\mathsf{U}_{n+1}]_{it} [\mathsf{Q}_n]_{tj} \\
        &= (-1)^i \left[
            \binom{j}{i} a_{n,j}
            + \binom{j + 1}{i} b_{n,j}
            + \binom{j + 2}{i} c_{n,j}
        \right], \\
        [\mathsf{Q}_n \mathsf{U}_n]_{ij}
        &= \sum_t [\mathsf{Q}_n]_{it} [\mathsf{U}_n]_{tj} \\
        &= (-1)^i \left[
            a_{n,i} \binom{j}{i}
            - b_{n,i-1} \binom{j}{i-1}
            + c_{n,i-2} \binom{j}{i-2}
        \right].
    \end{aligned}
\end{equation}

Now, Pascal's identity give us $\binom{n}{k} = \binom{n-1}{k-1} + \binom{n-1}{k}$, so we have $\binom{j+1}{i} = \binom{j}{i-1} + \binom{j}{i},\, \binom{j+2}{i} = \binom{j+1}{i} + \binom{j+1}{i-1} = \binom{j}{i-2} +  2\binom{j}{i-1} + \binom{j}{i}$. So, we see 

\begin{equation}
    \begin{aligned}
        \binom{j}{i} a_{n,j}
        + \binom{j + 1}{i} b_{n,j}
        + \binom{j + 2}{i} c_{n,j}
        &= \binom{j}{i}(a_{n,j}+b_{n,j}+c_{n,j}) \\
        &\quad + \binom{j}{i-1}(b_{n,j}+2c_{n,j})
        + \binom{j}{i-2}c_{n,j}.
    \end{aligned}
\end{equation}

We can simplify $a_{n,j} + b_{n,j} + c_{n,j} = n + 1, \ b_{n,j} + 2c_{n,j} = \frac{-2j + n(2+j)}{n+2} $. Now, we want to somehow go from parameters indexed by $j$ to parameters indexed by $i$. Note that $i \binom{j}{i} = i \frac{j!}{i! (j-i)!} = \frac{j!}{(i-1)!(j-i)!} = (j-i + 1)\binom{j}{i-1} $, and likewise $(i-1) \binom{j}{i-1} = (j-i + 2) \binom{j}{i-2}$. Thus, we see

\begin{equation}
    \begin{aligned}
        (-1)^i[\mathsf{U}_{n+1} \mathsf{Q}_n-\mathsf{Q}_n \mathsf{U}_n]_{ij}
        &= \binom{j}{i}
        \left[n+1-\frac{(n+1-i)(n+2i+2)}{n+2}\right] \\
        &\quad + \binom{j}{i-1}
        \left[
            \frac{-2j+n(2+j)}{n+2}
            + \frac{4(i-1)^2+(2-3n)(i-1)-2n}{n+2}
        \right] \\
        &\quad + \binom{j}{i-2}
        \left[
            \frac{2(j+1)(n-j)-2(i-1)(n-i+2)}{n+2}
        \right] \\
        &= \binom{j}{i}\frac{i(2i-n)}{n+2} \\
        &\quad + \binom{j}{i-1}
        \frac{2(i-1)(i+j-n-1)-(2i-n)(j-i+1)}{n+2} \\
        &\quad + \binom{j}{i-2}
        \frac{2(i-j-2)(i+j-n-1)}{2+n} \\
        &= (n+2)^{-1}\Biggl\{
            (2i-n)\left[
                i \binom{j}{i}-(j-i+1)\binom{j}{i-1}
            \right] \\
        &\qquad\qquad
            +2(i+j-n-1)\left[
                (i-1)\binom{j}{i-1}-(j-i+2)\binom{j}{i-2}
            \right]
        \Biggr\} \\
        &= 0.
    \end{aligned}
\end{equation}

using the two binomial coefficient identities. Thus, we conclude that 

\begin{equation}
    \mathsf{U}_{n+1} \mathsf{Q}_n = \mathsf{Q}_n \mathsf{U}_n \label{U_nQ_n_intertwine}
\end{equation}

which was to be shown.
\end{proof}

\begin{lemma}
    Define the map 

\begin{equation}
    T_n = \mathsf{R}_{n+1} \mathsf{Q}_n \mathsf{R}_{n}^{-1} \label{MacWilliams-intertwiner}
\end{equation}

This map satisfies the conditions of Lemma \ref{LP-monotonicity-condition}, and the monotonicity of LP feasibility follows. 
\end{lemma}

\begin{proof}
    From the proven intertwining relation \eqref{U_nQ_n_intertwine}, we see that 

\begin{equation}
    \begin{aligned}
        M_{n+1} T_n &= \mathsf{R}_{n+1} \mathsf{U}_{n+1} \mathsf{R}_{n+1}^{-1} \mathsf{R}_{n+1} \mathsf{Q}_n \mathsf{R}_n^{-1} \\ 
        &= \mathsf{R}_{n+1} \mathsf{U}_{n+1} \mathsf{Q}_n \mathsf{R}_{n}^{-1} \\ 
        &= \mathsf{R}_{n+1} \mathsf{Q}_n \mathsf{U}_n \mathsf{R}_n^{-1} \\ 
        &= \mathsf{R}_{n+1} \mathsf{Q}_n \mathsf{R}_n^{-1} \mathsf{R}_n \mathsf{U}_n \mathsf{R}_{n}^{-1} \\ 
        &= T_n M_n 
    \end{aligned} \label{M_nT_n_intertwine}
\end{equation}

that is, $T_n$ intertwines MacWilliams transforms at $n$, $n+1$. 

Now, we show that $T_n \geq 0$ and lower-triangular. We're specifically going to show that 

\begin{equation}
    [T_n]_{\ell t} = \begin{cases}
        0, & \ell < t \\ 
        \frac{(n+1-t)(n+ 2t + 2)}{(n+2)(n + t + 2)}, & \ell = t \\ 
        \frac{(2\ell + 1)(n+2 + 2t(t+1))(n-t)! (n+t+1)!}{(n+2)(n+1-\ell)!(n+\ell+2)!}, & \ell > t
    \end{cases} \label{explicit_T_n}
\end{equation}

Instead of inverting $\mathsf{R}_n$, we're going to show that this form satisfies $T_n \mathsf{R}_n = \mathsf{R}_{n+1} \mathsf{Q}_n$, and then the result follows from invertibility. Indeed, we check for $0 \leq r \leq n,\, r \leq \ell \leq n+1$, (note that the first term only exists for $\ell \leq n$, so the term is written to vanish at $\ell = n + 1$)

\begin{equation}
    \begin{aligned}
        [T_n \mathsf{R}_n]_{\ell r}
        &= \sum_t [T_n]_{\ell t} [\mathsf{R}_n]_{tr} \\
        &= \frac{(n+1-\ell)^2(n+2\ell+2)}{(n+2)(n+\ell+2)}
        \frac{(2\ell+1)(\ell+r)!(n-r)!}{r! (\ell-r)! (n+1-\ell)! (n+\ell+1)!} \\
        &\quad + \sum_{t = r}^{\ell-1} \frac{(2\ell+1)(n+2 + 2t(t+1))(n-t)! (n+t+1)!}{(n+2)(n+1-\ell)! (n+\ell+2)!}
        \cdot \frac{(2t+1)(t+r)!(n-r)!}{r! (t-r)! (n-t)! (n+t+1)!} \\ 
        &= \frac{(2\ell+1)(n-r)!}{(n+2)(n+\ell+2)! (n+1-\ell)! r! }
        \left( \frac{(\ell+r)!(n+2\ell+2) (n+1-\ell)^2}{(\ell-r)!} \right. \\
        &\qquad\left. + \sum_{t =r}^{\ell-1} \frac{(2t+1)(n + 2 + 2t(t+1)) (t+r)!}{(t-r)!}  \right)
    \end{aligned}\label{T_nR_n_derivation}
\end{equation}

Now, we show inductively that 

\begin{equation}
\sum_{t =r}^{\ell-1} \frac{(2t+1)(n + 2 + 2t(t+1)) (t+r)!}{(t-r)!} = \frac{\ell(\ell-r)(nr +  2n + 2(r+1)\ell^2 + 2)(\ell+r)!}{(r+1)(r+2)(\ell-r)!}
\end{equation}

Indeed, when $\ell = r$, both sides are trivially zero. Now, suppose the identity holds for some $\ell \geq r$. Then, we have 

\begin{equation}
    \begin{aligned}
        &\sum_{t =r}^{\ell} \frac{(2t+1)(n + 2 + 2t(t+1)) (t+r)!}{(t-r)!} \\
        &= \frac{\ell(\ell-r)(nr +  2n + 2(r+1)\ell^2 + 2)(\ell+r)!}{(r+1)(r+2)(\ell-r)!}
        + \frac{(2\ell+1)(n+2+2\ell(\ell+1))(\ell+r)!}{(\ell-r)!} \\
        &= \frac{(\ell+r)!}{(\ell-r)! (r+1)(r+2)}  \left[ \ell(\ell-r) (nr + 2n + 2(r+1)\ell^2 + 2) + (2\ell+1)(n+2 + 2\ell(\ell+1))(r+1)(r+2) \right] \\ 
        &= \frac{(\ell+r)!}{(\ell-r)! (r+1)(r+2)} \cdot \left[(\ell+1)(1+\ell+r)\left(4\ell(1+r) + 2\ell^2(1+r) + (2+n)(r+2) \right) \right] \\ 
        &= \frac{(\ell-r+1) (\ell+1) (\ell+1+r)! }{(r+1) (r+2) (\ell-r+1)!} \left(2(r+1)(\ell+1)^2 + nr + 2n + 2 \right) 
    \end{aligned}
\end{equation}

where the algebraic simplification was checked with Mathematica. This concludes the inductive proof. So, we get

\begin{equation}
    \begin{aligned}
        [T_n \mathsf{R}_n]_{\ell r}
        &= \frac{(2\ell+1)(n-r)! (\ell+r)!}{(n+2)(n+\ell+2)! (n+1-\ell)! r! (\ell-r)!} \\
        &\quad\times \left( (n+2\ell+2) (n+1-\ell)^2
        + \frac{\ell(\ell-r)(nr +  2n + 2(r+1)\ell^2 + 2)}{(r+1)(r+2)} \right)
    \end{aligned} \label{T_nR_nResult}
\end{equation}

Similarly, we see for $0 \leq r \leq n$,\, $r \leq \ell \leq n+1$

\begin{equation}
    \begin{aligned}
        [\mathsf{R}_{n+1} \mathsf{Q}_n]_{\ell r} &=  \sum_{k = 0}^\ell [\mathsf{R}_{n+1}]_{\ell k} [\mathsf{Q}_n]_{kr}\\
         &=  \frac{(2\ell+1)(\ell+r)!(n+1-r)!}{r!(\ell-r)! (n+1-\ell)!(n+\ell+2)!}
         \left[ a_{n,r} + \frac{(\ell+r+1)  (\ell-r)}{(r+1) (n+1-r)} b_{n,r} \right. \\
         &\qquad\left. + \frac{2(\ell+r+1)(\ell+r+2)(\ell-r)(\ell-r-1)}{(r+2) (n+1-r)(n+2)} \right] \\ 
         &= \frac{(2\ell+1)(\ell+r)!(n-r)!}{(n+2)r!(\ell-r)! (n+1-\ell)!(n+\ell+2)!} (n+2)(n+1-r) \\
         &\quad\times \left[ a_{n,r} + \frac{(\ell+r+1)  (\ell-r)}{(r+1) (n+1-r)} b_{n,r} \right. \\
         &\qquad\left. + \frac{2(\ell+r+1)(\ell+r+2)(\ell-r)(\ell-r-1)}{(r+2) (n+1-r)(n+2)} \right]
    \end{aligned}
\end{equation}

So we check 

\begin{equation}
    \begin{aligned}
        (n+2)(n+1-r) \cdot \left[ a_{n,r} + \frac{(\ell+r+1)  (\ell-r)}{(r+1) (n+1-r)} b_{n,r} + \frac{2(\ell+r+1)(\ell+r+2)(\ell-r)(\ell-r-1)}{(r+2) (n+1-r)(n+2)} \right] \\ 
        = (n+2r + 2)(n+1-r)^2 + \frac{(\ell+r+1)(\ell-r)(4r^2 + (2-3n)r -2n)}{(r+1)} + \frac{2(\ell+r+1)(\ell+r+2)(\ell-r)(\ell-r-1)}{(r+2)} \\
        = (n+2\ell+2) (n+1-\ell)^2  + \frac{\ell(\ell-r)(nr +  2n + 2(r+1)\ell^2 + 2)}{(r+1)(r+2)}
    \end{aligned}
\end{equation}

where the algebraic simplification was once again checked with Mathematica. Combining with \eqref{T_nR_nResult}, as well as with the fact that for $r > \ell$, $[T_n \mathsf{R}_n]_{\ell r} = 0 = [\mathsf{R}_{n+1}\mathsf{Q}_n]_{\ell r}$ since first sum requires $t \leq \ell$, $t \geq r$ to be non-zero in equations \eqref{T_nR_n_derivation} (and similarly for the other sum with $k$), we conclude that $T_n \mathsf{R}_n = \mathsf{R}_{n+1} \mathsf{Q}_n$, and since $\mathsf{R}_n$ are invertible, that confirms equation \eqref{explicit_T_n} as the unique formula for $T_n$. And since all entries of $T_n$ are clearly non-negative, we conclude $T_n \geq 0$. 

Finally, we want to check $\mathbf{1}^T T_n = \mathbf{1}^T$. Indeed, we see

\begin{equation}
    \begin{aligned}
        [\mathbf{1}^T T_n]_{j} &= \frac{(n+1-j)(n+2j + 2)}{(n+2)(n+2+j)} + \frac{(n+2+2j(j+1))(n-j)! (n+j+1)!}{(n+2)}  \sum_{k=j+1}^{n+1} \frac{(2k+1)}{(n+1-k)! (n+k+2)!}  
    \end{aligned}
\end{equation}

The sum becomes 

\begin{equation}
    \begin{aligned}
        \sum_{k=j+1}^{n+1} \frac{(2k+1)}{(n+1-k)! (n+k+2)!}   &= \sum_{k=j+1}^{n} \frac{1}{(n+1-k)!(n+k+1)!} - \frac{1}{(n-k)!(n+k+2)!}  + \frac{1}{(2n+2)!}  \\ 
        &= \frac{1}{(n-j)!(n+j+2)!} - \frac{1}{(2n+2)!} + \frac{1}{(2n+2)!} = \frac{1}{(n-j)!(n+j+2)!}
    \end{aligned}
\end{equation}

from the telescoping sum. Substituting, in the previous equation, we get

\begin{equation}
    \begin{aligned}
        [\mathbf{1}^T T_n]_j &= \frac{(n+1-j)(n+2j + 2)}{(n+2)(n+2+j)} + \frac{(n+2+2j(j+1))(n-j)! (n+j+1)!}{(n+2)}  \cdot \frac{1}{(n-j)!(n+j+2)!} \\ 
        &=  \frac{(n+1-j)(n+2j + 2)}{(n+2)(n+2+j)} + \frac{(n+2+2j(j+1)) }{(n+2)(n+j+2)} \\ 
        &=  \frac{(n+2 + 2j)(n+2-j) + 2j^2}{(n+2)(n+2 +j)} = \frac{(n+2)^2 - (n+2)j + 2j(n+2) - 2j^2 + 2j^2}{(n+2)(n+2+j)} \\ 
        &= \frac{(n+2)(n+2+j)}{(n+2)(n+2+j)} = 1
    \end{aligned}
\end{equation}

so we have 

\begin{equation}
    \mathbf{1}^T T_n = \mathbf{1}^T \label{normalization-T_n}
\end{equation}

as desired. 
\end{proof}

\subsection{Certified and constructed distances}

 We present the maximum $\operatorname{LP}$-feasible distance at each total photon number $1 \leq n \leq 25$ for a qubit encoding ($K = 2$) obtained from explicit rational Farkas certificates, alongside the maximum distance of constructed codes at each $n$. Note that we're able to achieve the maximum attainable photon loss distance at every $n$ apart from $n = 18$ (with high-precision constructions, and the exact algebraic constructions mentioned earlier). Whether that difference is a result of the LP conditions not being tight, or the existence of a better $n = 18$ code, remains an open question. To see the relevant distance benchmarks for our binary-polyhedral family, see Table \ref{tab:constraint-balanced-family-codes}. Also note that related spin-code LP bounds were tabulated by Okada \cite{okadaQuantumAnalogDelsartes2025a}, which are consistent with our distance bound. 

\begin{table}[H]
    \centering
    \small
    \renewcommand{\arraystretch}{1.08}
    \begin{tabular}{c c c l l}
        \hline
        $n$ & Farkas bound $d_{\mathrm{F}}$ & Constructed distance $d_{\mathrm{con}}$ & Constructed code & Numerical classification \\
        \hline
        $1$  & $1$ & $1$ & Trivial rank-two subspace & Trivial \\
        $2$  & $1$ & $1$ & Trivial rank-two subspace & Trivial \\
        $3$  & $1$ & $1$ & Trivial rank-two subspace & Trivial \\
        $4$  & $2$ & $2$ & $2O:\rho_3$ & $2O:\rho_3$ \\
        $5$  & $2$ & $2$ & $C_6: \rho_2 \oplus \rho_6,\, \left[ 2D_4:\rho_7\right]$ & $C_6:\rho_2\oplus\rho_6$ \\
        $6$  & $2$ & $2$ & $C_6:\rho_1\oplus\rho_3$ & $C_6:\rho_1\oplus\rho_3$ \\
        $7$  & $3$ & $3$ & $2I:\rho_3$ & $2I:\rho_3$ \\
        $8$  & $3$ & $3$ & $2D_3:\rho_5$ & $C_6:\rho_1\oplus\rho_3$ \\
        $9$  & $3$ & $3$ & $2D_4:\rho_6$ & $C_6:\rho_2\oplus\rho_4$ \\
        $10$ & $3$ & $3$ & $2D_4:\rho_5$ & $2D_4:\rho_5$ \\
        $11$ & $3$ & $3$ & $2D_2:\rho_5$ & $C_4:\rho_3\oplus\rho_4$ \\
        $12$ & $3$ & $3$ & $2T:\rho_1\oplus\rho_2$ & $2T:\rho_1\oplus\rho_2$ \\
        $13$ & $4$ & $4$ & $2O:\rho_5$ & $2O:\rho_5$ \\
        $14$ & $4$ & $4$ & $C_6:\rho_1\oplus\rho_5$ & $C_6:\rho_1\oplus\rho_5$ \\
        $15$ & $4$ & $4$ & $2T:\rho_5$ & $2T:\rho_5$ \\
        $16$ & $4$ & $4$ & $2D_3:\rho_5$ & $2D_3:\rho_5$ \\
        $17$ & $4$ & $4$ & $2D_3:\rho_6$ & $2D_3:\rho_6$ \\
        $18$ & $5$ & $4$ & $2D_2:\rho_3\oplus\rho_4$ & $2D_2:\rho_3\oplus\rho_4$ \\
        $19$ & $5$ & $5$ & $2T:\rho_4$ & $2T:\rho_4$ \\
        $20$ & $5$ & $5$ & $2D_6:\rho_6$ & $C_6:\rho_1\oplus\rho_5$ \\
        $21$ & $5$ & $5$ & $2D_4:\rho_7$ & $2T:\rho_6$ \\
        $22$ & $5$ & $5$ & $2D_5:\rho_7$ & $2D_5:\rho_7$ \\
        $23$ & $5$ & $5$ & $C_4:\rho_3\oplus\rho_4$ & $C_4:\rho_3\oplus\rho_4$ \\
        $24$ & $5$ & $5$ & $C_4: \rho_1 \oplus \rho_2, \, \left[2D_3:\rho_5\right]$ & $C_4:\rho_1\oplus\rho_2$ \\
        $25$ & $5$ & $5$ & $2O:\rho_4$ & $2O:\rho_4$ \\
        \hline
    \end{tabular}
    \caption{Farkas upper bounds, the best saved rank-two constructions, and the symmetry classifications of the best numerical codes from Table~\ref{tab:numerical-code-symmetry-classification} for $1\leq n\leq 25$. A code of distance $d$ corrects up to $d-1$ losses. The codes in $[]$ are examples of alternative codes achieving the same distance.}
    \label{tab:farkas-constructed-distances}
\end{table}
 
\FloatBarrier

\bibliographystyle{apsrev4-2}
\bibliography{Candidacy}

\clearpage
\endgroup
\end{document}